\documentclass[11pt]{article}
\usepackage{jheppub} 

\usepackage{amsfonts,amsmath,amssymb,mathrsfs}
\usepackage{mathtools}
\usepackage{amsthm}
\usepackage{xcolor}
\usepackage{booktabs}
\usepackage{braket}
\usepackage{comment}
\usepackage{breqn}
\usepackage{booktabs}
\usepackage{float}  
\usepackage{longtable}
\usepackage{array}

\newcolumntype{C}[1]{>{\centering\arraybackslash}m{#1}}

\newcommand\norm[1]{\left\lVert#1\right\rVert}
\newcommand{\mH}[0]{\mathcal{H}}
\newcommand{\mL}[0]{\mathcal{L}}
\newcommand{\mO}[0]{\mathcal{O}}
\newcommand{\bd}[0]{\bar{d}}
\newcommand{\lrp}[1]{\left(#1\right)}

\DeclareMathOperator*{\argmax}{arg\,max}
\DeclareMathOperator{\rank}{rank}

\DeclareMathOperator{\Tr}{Tr}

\newtheorem{theorem}{Theorem}[section]
\newtheorem{corollary}{Corollary}[section]

\newtheorem{definition}{Definition}[section]

\newtheorem{lemma}{Lemma}[section]

\allowdisplaybreaks

\newenvironment{eqns}
{\begin{equation}
\begin{aligned}
}
{ 
\end{aligned} 
\end{equation}
}

\title{State-dependent geometries from magic-enriched quantum codes}
\author[a,b]{ChunJun Cao}
\author[a,b]{Gong Cheng}
\author[a,b]{Krishnanand Karthikeyan}
\author[a,b]{Cathy Li}
\author[c]{John Preskill}

\affiliation[a]{Department of Physics, Virginia Tech, Blacksburg, VA 24061, USA}
\affiliation[b]{Virginia Tech Center for Quantum Information Science and Engineering, Blacksburg, VA 24061, USA}
\affiliation[c]{Institute for Quantum Information and Matter,
California Institute of Technology,
Pasadena, CA 91125, USA}

\abstract{Quantum error-correcting codes provide a powerful framework for emergent spacetime, yet existing subsystem erasure-correcting holographic code models describe only quantum fields on a fixed background: in such codes, the entropic area term is state independent and cannot capture gravitational backreaction.  We argue that this limitation is intrinsic to exact subsystem complementary recovery and that incorporating backreaction naturally leads to approximate quantum error correction. We introduce a Ryu-Takayanagi-like entropy decomposition for approximate subsystem erasure-correcting codes, defining bulk matter entropy via optimal recovery and a complementary proto-area entropy as the difference between boundary entropy and recoverable bulk entropy. For a broad class of skewed quantum codes obtained by small nonlocal perturbations of exact codes, the proto-area increases monotonically with bulk entropy, closely aligning with the changes in extremal surface areas in gravitational systems.  We identify the origin of this response as a form of tripartite non-local magic in the Choi state of the encoding map, which vanishes in stabilizer codes and controls the leading matter–geometry coupling in approximate subsystem erasure-correcting codes. }

\begin{document}

\maketitle

\section{Introduction}
Recent developments in quantum gravity suggest that key aspects of spacetime dynamics are governed by universal quantum information principles. In AdS/CFT, entanglement entropy is related to bulk geometry through the Ryu–Takayanagi (RT) formulas \cite{Ryu_2006,Ryu:2006bv,Faulkner:2013ana,Hubeny_2007,Engelhardt_2015}, while more general arguments based on entanglement equilibrium and the entanglement first law indicate that Einstein’s equations themselves may emerge from entropic constraints \cite{Faulkner_2014,Faulkner_2017,Czech:2016tqr,Swingle:2014uza,Lashkari:2014kda,Lashkari:2015hha,Lashkari:2016idm}, even in systems that are not asymptotically AdS \cite{Jacobson_1995,Jacobson_2016,Cao_2017,Cao_2018}. These insights motivate the search for models of emergent gravity that rely only on information-theoretic structure, rather than specific microscopic dynamics or conformal symmetry.

Quantum error-correcting codes (QECCs) provide a natural framework for this program. In holography, the low-energy sector of quantum gravity behaves as a code subspace, with bulk degrees of freedom redundantly encoded in boundary variables \cite{Almheiri_2015}. This perspective has led to a wide class of tensor-network and QECC toy models that successfully reproduce kinematic features of holography, including RT-like entropy formulas and entanglement wedge reconstruction \cite{Pastawski:2015qua,Hayden_2016,Cao_2021,Steinberg_2023,Harlow:2016vwg,Harris_2018,Dolev_2022}. This perspective also generalizes to other geometries \cite{Hayden_2016,Harlow:2016vwg,Cao_2018,Cao:2021fyk}, providing a crucial ``emergence map'' for separating matter from geometry in complex quantum systems. 

However, existing QECC models do not capture the emergence of gravity. This deficiency is particularly pronounced in exact subsystem erasure-correcting codes (including some stabilizer codes and many other holographic tensor networks), where the encoded state can always be locally decoded into a product of a logical bulk state and a fixed entangled resource \cite{Harlow:2016vwg}. Consequently, the entropic area term is independent of the logical state. This structure is well-suited for quantum field theory on curved spacetime but fails to capture a defining feature of gravity: changes in matter must backreact on geometry. Recent no-go theorems formalize this limitation, showing that stabilizer codes and their local-unitary deformations necessarily admit only trivial, state-independent area operators \cite{Cao:2023mzo}.

This observation indicates that gravity cannot emerge from \textit{exact subsystem erasure correction codes}. Two complementary strategies are commonly discussed to overcome this deficiency. The first is to consider exact subalgebra quantum error correcting codes that contain a non-trivial area operator. This approach, originally suggested by \cite{Harlow:2016vwg}, is based on an algebraic generalization of the subsystem erasure correction codes where the recoverable algebra from a physical subsystem (or a boundary subregion) need not be a factor. The different eigenvalues of the area operator, which lives in the center of the recoverable algebra, then correspond to the different extremal surface areas that supposedly arise from the different background geometries when coupled to different bulk matter configurations. A few examples of such codes are later realized in tensor network toy models \cite{Donnelly_2017,Cao_2021, Dolev_2022, Dong:2023kyr,Qi_2022}. In such an approach, one is required to extend beyond the common stabilizer code framework \cite{Cao:2023mzo}.  

The second approach, in contrast, is to consider \textit{approximate} erasure correction, which relaxes the rigid separation between logical information and the code's entangled resource, allowing correlations between bulk matter and geometry. Such correlations are unavoidable in a proper gravitational theory, where matter excitations source metric perturbations and therefore become entangled with the gravitational degrees of freedom \cite{Harlow:2016vwg,Cao:2023mzo,Pollack_2022}. In this case, what is recoverable from a boundary subregion need not even form a subalgebra. Empirically, approximate QECCs have indeed been successful in recovering many features analogous to gravity \cite{Hayden_2016,Cao_2021}, sometimes with additional gauge constraints \cite{Qi_2022,Dong:2023kyr,Dolev_2022,Cao_2021} that links them to the first approach. However, a full explanation of the underlying mechanisms that give rise to these features is still unclear, which may involve the complex interplay of multiple properties of approximate codes that are absent in exact subsystem erasure correction codes.

In this work, we will focus on the second approach by developing a framework for emergent geometry in codes with approximate subsystem complementary recovery. We introduce an RT-like entropy decomposition that remains well defined beyond exact recovery. The matter entropy is defined as the entropy of the optimally recoverable bulk state, obtained by maximizing coherent information over recovery channels. This choice is motivated by the interpretation of bulk matter as the degrees of freedom that can be reconstructed with maximal fidelity from a boundary subregion, even when recovery is imperfect.

The complementary contribution, which we call the \textit{proto-area entropy}, is defined as the difference between the boundary entropy and the recoverable bulk entropy. This quantity captures the residual entanglement that cannot be attributed to bulk matter and is therefore naturally interpreted as geometric. The resulting decomposition reduces to the standard and Faulkner–Lewkowycz-Maldacena (FLM) formula \cite{Faulkner:2013ana} in exact codes, while allowing the geometric term to become state dependent when recovery is approximate.

We study a broad class of \textit{skewed quantum codes} \cite{Cao_2021,Cao:2021wrb}, obtained by perturbing exact subsystem erasure-correcting codes with small nonlocal unitary deformations. For these codes, we show that the proto-area entropy typically increases monotonically with bulk entropy for mixed bulk states. When the bulk state is pure, the proto-area entropy increases with the entanglement between the entanglement wedge and its bulk complement, which plays the role of bulk entropy in the discussion of entanglement wedge reconstructions \cite{Engelhardt_2015,Akers:2021fut}. These features closely align with the response of quantum extremal surfaces to bulk excitations or the response of classical extremal surface areas in backreactions. The stated behavior is generic in the sense that it holds for almost all small nonlocal perturbations of an exact code, with violations occurring only for finely tuned or symmetry-restricted deformations. 

Crucially, we identify the quantum resource that controls this matter–geometry coupling. We show that the strength of the proto-area response is governed by \textit{tripartite non-local magic} in the Choi state of the encoding map. This state involves three tensor factors: the logical reference system, the recoverable bulk degrees of freedom, and the geometric entanglement. In a manner directly analogous to the quantum extremal surface construction, correlations between the bulk matter and geometry depend on the logical state. This form of magic is irreducibly non-local in the sense that it cannot be generated or removed by any two-subsystem unitary operations. Using stabilizer R\'enyi entropies as a diagnostic, we show that this non-local magic vanishes in stabilizer codes and quantitatively controls the leading dependence of the proto-area on bulk entanglement in approximate complementary recovery codes. 

Our results clarify why exact subsystem QECC models fail to reproduce gravitational backreaction and identify a tripartite non-local magic as an important ingredient enabling geometry to respond to matter in stabilizer codes. This provides an explicit prescription for how non-local magic can enter in gravitationally relevant codes~\cite{Cao:2023mzo}.  More broadly, they establish a concrete information-theoretic mechanism for emergent gravity in generic quantum systems, linking spacetime dynamics to entanglement, approximate quantum error correction, and non-local non-Clifford encoding. 

We organize this paper as follows. In Sec.~\ref{sec:2}, we motivate the necessary structural changes needed by gravity when deforming exact erasure correction codes. In Sec.~\ref{sec:3}, we extend the RT-like entropic formula for QECCs to approximate erasure correction codes and prove general properties linking the change of area-like entropies to that of ``bulk/matter entropies'' in Sec.~\ref{sec:4}. Finally, we show in Sec.~\ref{sec:5} that the strength of this interlink is determined by a form of perturbative non-local tripartite magic, thus establishing a precise connection between non-local magic and gravity-like conditions in approximate erasure correction codes. Since we focus exclusively on subsystem codes and perturbations and henceforth we will often drop the subsystem modifier. However, the statements should be understood in this context.

\section{Quantum codes and emergent gravity}
\label{sec:2}
\begin{figure}
    \centering
    \includegraphics[width=0.8\linewidth]{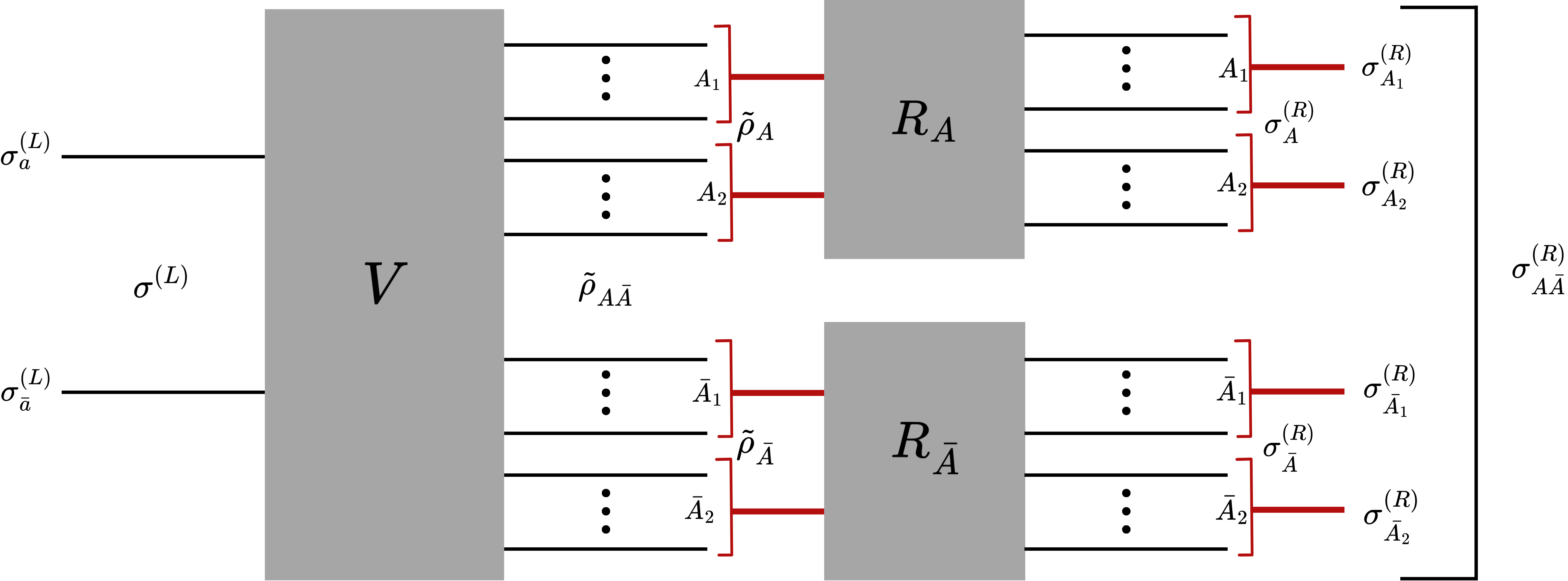}
    \caption{ Schematic of the encode–recover process. A logical input $\sigma^{(L)}_{a\bar a}\equiv \sigma^{(L)}$ is supplied to the code, where $\sigma_{a}^{(L)}=\Tr_{\bar a}(\sigma^{(L)}_{a\bar a})$ and $\sigma_{\bar a}^{(L)}=\Tr_{a}(\sigma^{(L)}_{a\bar a})$ denote the logical marginals on $a$ and $\bar a$, respectively. The isometry $V$ encodes $\sigma^{(L)}_{a\bar a}$ into boundary degrees of freedom, producing the collective encoded state $\tilde\rho_{A\bar A}$. We group boundary output legs as $A \equiv A_1\cup A_2$ and $\bar A \equiv \bar A_1\cup\bar A_2$. Intermediate encoded marginals are obtained by partial traces, $\tilde\rho_{A}=\Tr_{\bar A}(\tilde\rho_{A\bar A})$ and $\tilde\rho_{\bar A}=\Tr_{A}(\tilde\rho_{A\bar A})$. Local recovery unitaries $R_A$ and $R_{\bar A}$ act independently on $A$ and $\bar A$, producing the recovered components $\sigma^{(R)}_{A_1},\,\sigma^{(R)}_{A_2},\,\sigma^{(R)}_{\bar A_1},\,\sigma^{(R)}_{\bar A_2}$, which together form the final recovered boundary state $\sigma^{(R)}_{A\bar A}$. }
    \label{fig:comp_unencode}
\end{figure}

\subsection{Codes and complementary recovery}
We begin by reviewing the aspects of quantum gravity that existing subsystem QECCs models do capture. A quantum code is specified by a code subspace $\mathcal{C}$ within the physical Hilbert space $\mH_P$. It is often convenient to introduce a logical Hilbert space $\mH_L$, which is isomorphic to the code subspace $\mathcal{C}$. The encoding is implemented by an isometric map
\begin{equation}
        V: \mH_L \rightarrow \mH_P,
\end{equation}
whose image is the code subspace. 

Suppose $\mH_P$ is factorizable and consider a bipartition of the physical system into subsystems $\mH_P = \mH_A\otimes \mH_{\bar A}$. The code is said to exhibit \textit{complementary subsystem erasure correction}\footnote{In this work, we often refer to such codes as exact subsystem erasure-correcting codes.} \cite{Harlow:2016vwg} if the logical space $\mH_L$ admits a factorization
\begin{eqns}
    \mH_L=\mathcal L_{a}\otimes \mathcal L_{\bar a},
\end{eqns}
such that  all the logical operators $O_a$ acting on the encoded information in $\mathcal L_a$ has a representation in the physical Hilbert space as $\tilde{O}_a={O}_A\otimes I_{\bar A}$, and  all logical operators $O_{\bar a}$ acting on the encoded information in the complement $\mathcal L_{\bar a}$ has the representation $\tilde{O}_{\bar a}=I_A\otimes {O}_{\bar A}$. For convenience, we will drop all identity operators and denote all operators with support only on $A$ as $O_A$.

A code with subsystem complementary recovery~\cite{Harlow:2016vwg} possesses the property that any codeword can be decoded by local unitaries supported on complementary regions $A$ and $\bar A$. 
Explicitly, let $|\tilde{\psi}\rangle\in \mathcal{C}$ be any codeword. There exist local decoding unitaries $R_A$ and $R_{\bar A}$ independent of the codeword such that 
\begin{equation}
    R_AR_{\bar A}{\ket{\tilde{\psi}}}=\ket{\psi}_{A_1\bar{A}_1}\ket{\chi}_{A_2\bar{A}_2}.
    \label{eqn:unitary_rec}
\end{equation}
Here we take $\mH_{A_1}$ ($\mH_{\bar{A}_1}$) to be isomorphic to $\mathcal L_a$ ($\mathcal L_{\bar{a}}$). In other words, the recovery unitary extracts the encoded logical information ${\sigma}_a=\Tr_{\bar{a}}[\ket{{\psi}}\bra{ \psi}]$  and ${\sigma}_{\bar a}=\Tr_{{a}}[\ket{{\psi}}\bra{ \psi}]$  and places it in the physical subsystems $\mH_{A_1}\cup \mH_{\bar{A}_1}$. A shared entangled state $\ket{\chi}$ between the remaining auxiliary subsystems $A_2$ and $\bar A_2$ is then left over which contains the entanglement that cannot be removed by local unitaries $R_A\otimes R_{\bar{A}}$. Operationally, subsystem complementary recovery means that the information on $A$ (or $\bar{A}$) can be recovered by unencoding a subspace $A_1$ (or $\bar{A_1}$) of that region with auxiliary system $A_2 \bar{A_2}$. 
When we combine the parts recoverable from $A$ and that from $\bar{A}$, they make up the entire encoded state. Many stabilizer codes and more general symplectic codes are known to have this property \cite{Pastawski:2015qua,Pollack_2022,Cao:2025iec}. The encoding, recovery, and the notation of states at different stages are shown in Figure \ref{fig:comp_unencode}.

\subsection{RT and FLM formulas in exact erasure correction codes}
By tracing out the subsystem on $\bar{A}$, a code like the above satisfies a Ryu-Takayanagi (RT)-like formula, or more precisely, a Faulkner-Lewkowycz-Maldacena (FLM)-like formula such that the ``boundary entropy'' of a subregion $A$ is
\begin{eqns}
   S({\rho}_A)=S(\rho_{A_1})+S(\chi)= S({\sigma}_a^{(L)})+\Tr[{\mathcal A}{\sigma}_a^{(L)}],
\label{eq:QEC_entorpy}
\end{eqns}
where $\rho_A=\Tr_{\bar{A}}[|\tilde{\psi}\rangle\langle\tilde{\psi}|_{A\bar{A}}]$, $\rho_{A_1}=\Tr_{\bar A_1}[|\psi\rangle\langle\psi|_{A_1\bar{A}_1}]$, and $\chi = \Tr_{\bar{A}_2}[|\chi\rangle\langle\chi|_{A_2\bar{A}_2}]$. The second equality follows because $\rho_{A_1}$ (respectively $\rho_{\bar{A}_1}$) has the same matrix elements as ${\sigma}_a^{(L)}$ (respectively $\sigma_{\bar a}^{(L)}$). An area operator ${\mathcal{A}}= S(\chi)I_a$ can be defined to rewrite $S(\chi)$ as $\Tr[{\mathcal A}{\sigma}_a^{(L)}]$ 
to be directly analogous to the FLM relation in AdS/CFT \cite{Harlow:2016vwg},
\begin{eqns}
    S_{\rm CFT}(A)=S(\rho_{\rm EW(A)})+\frac{\langle\mathcal{A}\rangle}{4G}
\end{eqns}
where $S(\rho_A)\leftrightarrow S_{\rm CFT}(A)$ is the boundary entropy, $ S( \sigma_a^{(L)})\leftrightarrow S(\sigma_{\rm EW(A)})$ is the entropy associated with the effective field theory degrees of freedom restricted to the bulk entanglement wedge (EW) of $A$, and $\langle\mathcal{A}\rangle/4G\leftrightarrow \Tr[\sigma_a^{(L)} {\mathcal A}]=S(\chi)$ is the entropy tied to the area $\langle\mathcal{A}\rangle$ of a bulk minimal/extremal surface.
Post recovery, therefore, it is natural to link $A_1$ and $\bar{A}_1$ as containing the logical and hence the bulk matter field information while $A_2$ and $\bar{A}_2$ containing the information needed to emerge the background geometry through its link to the minimal surface area.

A canonical example is the HaPPY code \cite{Pastawski:2015qua}, illustrated in Figure~\ref{fig:TN-HaPPY}, where logical qubits reside in the bulk and physical qubits live on the boundary. The entangled state $\chi$ corresponds to the EPR pairs along the minimum cut ``geodesic'' through the tensor network.

  Much of the above results can also be generalized to operator algebraic quantum error correction codes when the code subspace do not admit a nice tensor factorization \cite{Harlow:2016vwg,Donnelly_2017,Dolev_2022,Cao_2021}. 
  However, for the sake of simplicity, we focus our presentation on subsystem codes and their variants even though some results below apply more generally.

\begin{figure}
    \centering
    \includegraphics[width=0.5\linewidth]{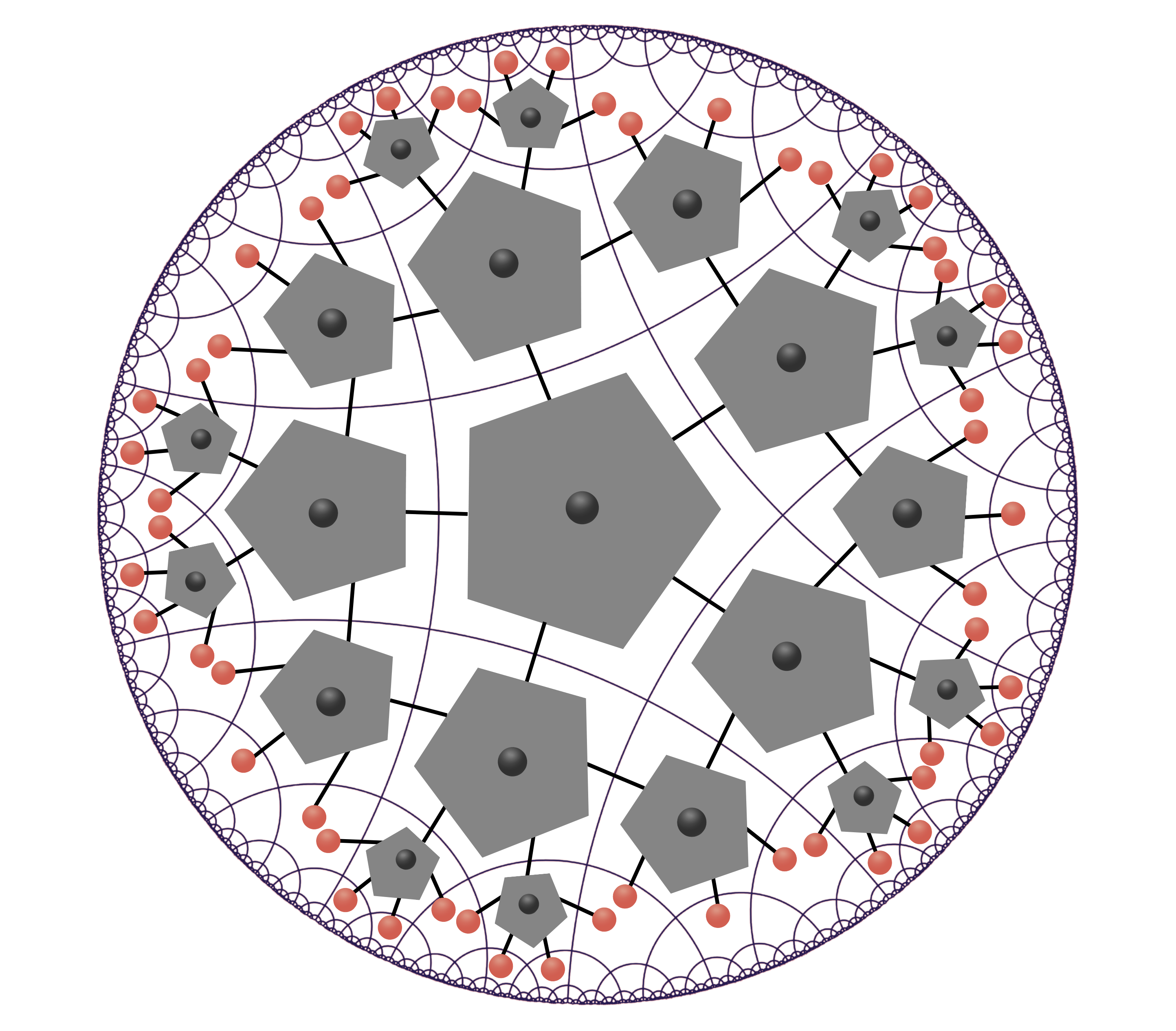}
    \caption{HaPPY code network shown as a hyperbolic tiling of perfect pentagon tensors. Each pentagon corresponds to a perfect tensor ($[[5,1,3]]$ code), with the black dot inside indicating the logical input leg. These logical legs point inward and should be understood as inputs. The red dots on the boundary represent the physical qubits where the encoded state appears.}
    \label{fig:TN-HaPPY}
\end{figure}

Going beyond AdS/CFT, Ref. \cite{Cao_2018} shows that near-flat geometries can also emerge from the entanglement data $S(\chi)$ for codewords $|\tilde\psi\rangle$ in a factorizable physical Hilbert space with suitable entanglement patterns. 
When such background geometries are well-defined \cite{Cao_2020}, Ref.~\cite{Cao_2018} argues that the matter entropy can in general be identified with the logical entropy while the remaining piece $S(\chi)$ is the geometric entropy, which is equal to the area of an extremal surface, i.e., hyperplane, in the background geometry with spatial cutoff. In this context, the emergent geometry in AdS/CFT can then be thought of as a special case in which the entanglement pattern from $S(\chi)$ across different partitions $A, \bar{A}$ gives rise to a hyperbolic rather than a flat spatial geometry.

In the same proposal \cite{Cao_2018}, it is shown that if the linearized Einstein's equations on an emergent flat background holds, then an entropic equation must be satisfied when the entanglement patterns of the state $|\tilde{\psi}\rangle$ are perturbed. That is, consider $|\tilde \psi\rangle\rightarrow |\tilde \psi\rangle+\delta |\tilde\phi\rangle$, then the respective entropy perturbations $\delta S(\chi), \delta S(\sigma_a)$ must satisfy an entropic equation that is equivalent to the linearized Hamiltonian constraint such that non-trivial changes of $\delta S(\rho_a)$ must lead to non-trivial changes in $\delta S(\chi)$ under certain assumptions.  In AdS/CFT, a similar entropic relation is satisfied by a combination of the entanglement first law and the bulk-boundary dictionary \cite{Faulkner_2014,Czech:2016tqr}. More generally, Jacobson \cite{Jacobson_2016} shows in a different setup where a similar entropic relation called entanglement equilibrium will yield the full non-linear Einstein's equations with additional assumptions. Indeed, these conditions make intuitive sense ---  generic changes in the bulk matter configuration and hence its entropy should be reflected on the underlying geometry or geometric entropy $S(\chi)$ in General Relativity. Therefore, identifying QECCs where similar conditions can be reproduced are natural next steps for emerging gravity in and beyond AdS using tensor network models.

\subsection{Beyond Exact QECC: the need for matter-geometry correlation}

However, it is easy to see that a code that satisfies exact subsystem complementary recovery in equation \eqref{eqn:unitary_rec} does not reproduce such conditions because the recovery is completely independent of the codeword and always factorizes the state into $|\chi\rangle$ and $|\psi\rangle$. This implies that $S(\chi)$ is independent of $|\psi\rangle$ and therefore its entropies. The same problem persists in operator algebra quantum error correction codes that have trivial area operators. In particular, \cite{Cao:2023mzo} established a no-go theorem showing that for all stabilizer codes --- and for any local unitary deformation thereof --- the area operators are trivial, whereas a non-trivial area operator is believed to be required for gravitational backreaction \cite{Harlow:2016vwg,Cao_2021,Cao:2021fyk}. Because the geometric entropy is always independent of the logical state, such toy models are analogs of quantum field theories on curved spacetime, which are indeed expected to satisfy the FLM formula. This contrasts with genuine holographic systems, in which the area term $\mathcal{A}$ depends on the bulk state.

For the rest of this work, we will often refer to the portion of the entropy $S(\chi)$ that can give rise to emergent geometry as the geometric entropy $S_{\rm geom}$ and $S(\sigma_a^{(L)})$ as the matter entropy $S_{\rm matter}$. In instances where the codewords are low energy states of a quantum field theory, $S(\sigma_a^{(L)})$ corresponds to the vacuum subtracted entropy \cite{Cao_2018,Casini_2008,boussobound,boussobound2}. This naming is motivated by the observed connections above between entropy and emergent gravity. However, it should be clear from context that they are simply labels of convenience because a quantum code does not generally admit a geometric or an effective field theory description. 

The goal of this paper is to construct and understand QECCs where $\delta S_{\rm matter}$ can trigger nonzero $\delta S_{\rm geom}$. As required by perturbative quantum gravity (see Appendix \ref{app:semiclassical}) and non-trivial state dependence \cite{Pollack_2022,Cao:2021fyk,Cao:2023mzo}, 
we need to consider codes whose best recovered  ``matter'' and ``geometric'' sectors are coupled, e.g.

\begin{equation}
    R_A R_{\bar A} |\tilde{\psi}\rangle  = \sum_{i} c_i|\psi_i\rangle_{A_1\bar{A}_1}|\chi_i\rangle_{A_2\bar{A}_2}
    \label{eqn:statedep}
\end{equation}
where formally the $|\chi_i\rangle$ can support different entanglement structures unlike those in stabilizer codes \footnote{The state written in Eq.~\eqref{eqn:statedep} is not identical to the most general structure of an operator-algebra QECC with a non-trivial area operator, where different $\alpha$-blocks may involve distinct direct-sum decompositions of the Hilbert spaces associated with $A_1$, $A_2$, and their complements \cite{Harlow:2016vwg, Akers_2019}. It may be possible to relate certain codes of the form in Eq.~(2.6) to such subalgebra codes by introducing additional ancillary degrees of freedom and restricting to an appropriate subspace \cite{Pollack_2022}. However, the precise scope of this correspondence is not yet well understood, and we do not assume that the two structures are generally equivalent.}

\subsection{Non-local magic enables gravity}
To produce such codes, a different type of quantum resource called non-stabilizerness or \emph{magic} is required. Magic is a notion of quantumness distinct from entanglement. The concept of magic comes from fault-tolerant quantum computation and error correction schemes based on the stabilizer formalism \cite{Gottesman:1998hu}. Broadly, magic is closely connected with Wigner negativity and the hardness of classical simulations, such as stabilizer simulations \cite{Aaronson_2004}, tensor network \cite{Cao:2024nrx}, and Monte Carlo \cite{magic_MC} methods. 

In quantum gravity, it was shown that magic is essential for the emergence of gravitational backreaction \cite{Cao:2024nrx}, for the construction of non-trivial area operators in QECCs \cite{Cao:2023mzo}, and is abundant in CFTs \cite{White:2020zoz,Oliviero_2022,Hoshino:2025jko}. 

Importantly, the type of magic needed for emergent gravity also has to be non-local, i.e. it cannot be removed by local unitaries acting on separate subsystems. Ref. \cite{Cao:2024nrx,Cao:2023mzo} define the (bipartite) non-local magic of a bipartite pure state to be 
\begin{equation}
    \min_{U_A\otimes U_B}\mathcal{M}(U_A\otimes U_B|\psi\rangle_{AB})
\end{equation}
where $\mathcal M$ is any magic measure and $U_A, U_B$ are unitaries. For example, although Haar random states and Haar random tensor networks \cite{Hayden_2016} are highly magical, they have near vanishing non-local magic \cite{Cao:2024nrx}. In contrast, CFTs have not only large total magic, but also a large amount of non-local magic that increases with its central charge \cite{Cao:2024nrx}.

\begin{figure}
    \centering
    \includegraphics[width=0.5\linewidth]{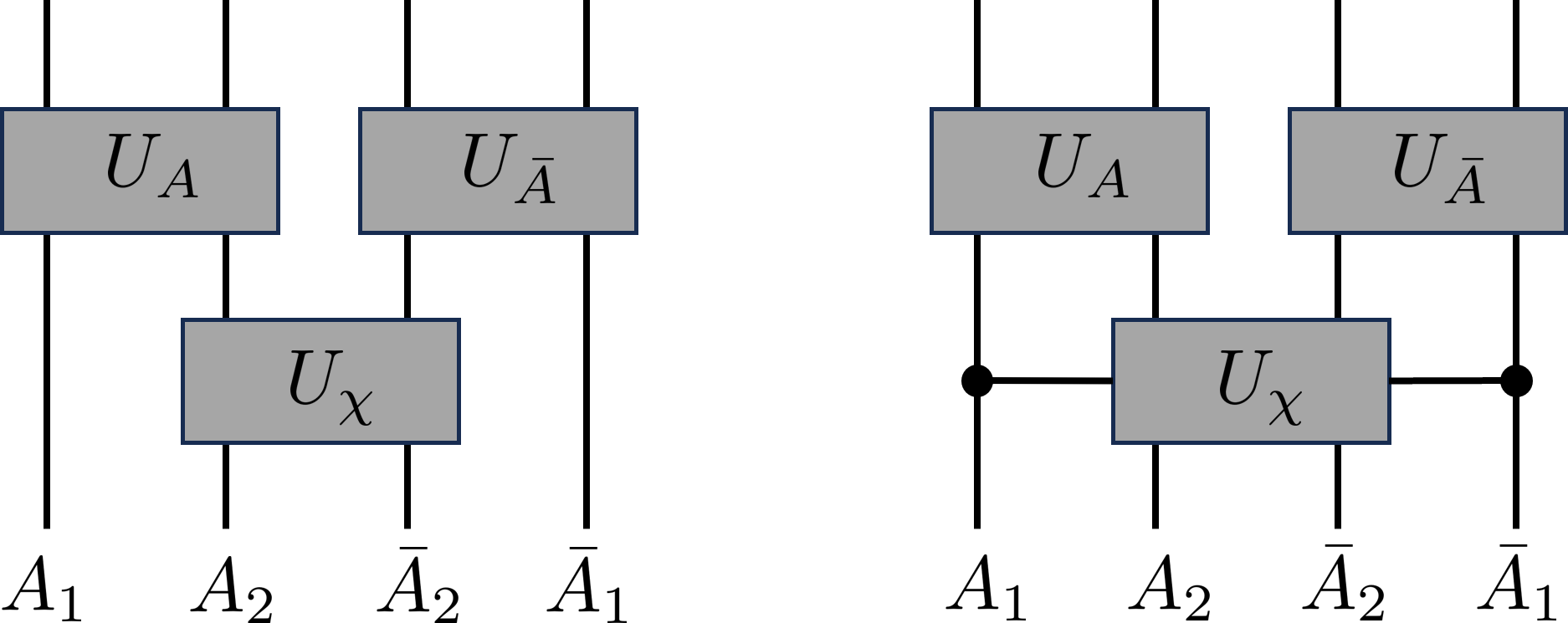}
    \caption{Left: an exact subsystem erasure-correcting code can be written such that an entangling state is first generated by $U_\chi$ on $A_2,\bar{A}_2$ followed by local encoding $U_A\otimes U_{\bar{A}}$, which can not produce matter-geometry correlation. Right: A controlled-$\chi$ unitary is needed to generate different amounts of entanglement based on the logical information on $A_1, \bar{A}_1$, so as to produce states of the form in \eqref{eqn:statedep}. For example, when $U_\chi = CX$, then it is a multi-controlled gate, which is non-Clifford. }
    \label{fig:CCX_circuit}
\end{figure}

A precise definition of non-local magic for quantum codes is not yet known in literature, but we can intuitively see why it is needed from relation \eqref{eqn:statedep}. To construct such a state,  the encoding unitary needs to generate different amounts of entanglement in $\chi$ based on the logical state $\psi$ (Figure~\ref{fig:CCX_circuit}). Such states are produced by quantum circuits with controlled-entangling unitaries, for example, Toffoli gates. Such unitaries are non-Clifford, thus inject magic, and cannot be synthesized by mere local unitaries on individual qubits $U_A\otimes U_B\otimes U_C$, or any depth-1 bipartite unitaries like $I_{A}\otimes U_{BC}$ \cite{toffoli}. Therefore, this suggests a form of tripartite non-local magic. We will make this notion precise in Sec.~\ref{sec:5} and comment on the details why local or bipartite forms of magic do not provide a state-dependent proto-area in Appendix~\ref{app:magic}.

\section{Entropy formula in an approximate erasure correction code}\label{sec:3}
In previous works \cite{Harlow:2016vwg,Cao_2018}, the matter and geometric entropies (or area operators) are only well-defined when the QECC satisfies exact erasure correction, i.e., the recovery channel can fully extract all the encoded information by acting only on subregions $A$ and $\bar{A}$ and hence there is a clean factorization of quantum states describing matter and geometry degrees of freedom like in equation \eqref{eqn:unitary_rec} \footnote{More generally, the structure of a subalgebra is required for these definitions, which is also absent in approximate codes. }. However, this is no longer true for an approximate erasure correction code where equation \eqref{eqn:statedep} holds.
For this, we will need a general definition of $S_{\rm geom}$ and $S_{\rm matter}$ that holds for approximate QECCs which may lack exact factorization or algebraic structures. In the following, we show that it is physically well-motivated to define them respectively as the entropy of the best-recoverable state under the erasure channel and the residual entanglement needed to satisfy a RT-like formula.

Consider the encoding map $V$, which is an isometry but does not yield a code with the exact subsystem erasure correction property. Denote the encoded state as 
\begin{equation}
    \ket{\tilde \psi}=V\ket{\psi},
\end{equation}
for $\ket{\psi}\in \mH_L$. More generally, we consider the encoding of a logical state $\sigma^{(L)} \in \mathcal{L}(\mH_L)$ through $\tilde{\rho}_{A \bar{A}} = V \sigma^{(L)} V^\dagger$ and the recovery $R_A R_{\bar{A}}$ thereof by
\begin{equation}
    \sigma^{(R)}_{A_1\bar{A}_1} 
    = \Tr_{A_2\bar A_2} (R_A R_{\bar A} \tilde\rho_{A\bar{A}} R_A^\dagger R_{\bar A}^\dagger) \equiv \mathcal{N}_R(\sigma^{(L)}),
    \label{eqn:recov_state}
\end{equation}
where we treat the encoding-recovery as a quantum channel $\mathcal{N}_R: \mathcal{L}(\mH_L)\rightarrow \mathcal{L}(\mH_{A_1\bar{A}_1})$  such that $\mathcal{N}_R(\sigma^{(L)})=\sigma^{(R)}\approx \sigma^{(L)}$. To clarify the notations: $\sigma^{(L/R)}_M$ denotes the logical state. The superscript ($L$) and ($R$) indicates that the state is pre-encoding and post-recovery respectively, and the subscript $M$ denotes the subsystem the state is supported on. The subscript is often omitted if the state is over the entire physical or logical Hilbert space as opposed to being restricted to a subsystem.

When the encoding $V$ produces an approximate QECC, the best possible recovery $R^*$ is the one that maximizes the standard coherent information $I_c(\mathcal{N})$ over all possible channels\footnote{$R^*$ need not be unique.}, i.e.,
\begin{equation}\label{eq:Rstar}
R^* := \argmax_R I_c(\mathcal{N}_R)
\end{equation}
where the coherent information of the channel $\mathcal{N}_R$ is defined as
\begin{equation}
    I_c(\mathcal{N}_R) = S(\Tr_r[\mathcal{N}_R\otimes I_r(|\Phi\rangle\langle\Phi|)]) - S(\mathcal{N}_R\otimes I_r(|\Phi\rangle\langle\Phi|)) = S(\sigma_{A_1\bar{A}_1}^{(R)})-S(\sigma_{A_1\bar{A}_1r}^{(R)}).
    \label{ref-coherentinfo}
\end{equation}
Here,
\begin{equation}
    \ket{\Phi}=\frac{1}{\sqrt{d_L}}\sum_{i=1}^{d_L} \ket{i}_L\ket{i}_r
\end{equation}
is a maximally entangled state between $\mathcal H_L$ and a reference system $r$ that is isomorphic to $\mathcal H_L$. $d_L$ is dimension of $\mH_L$.

For an exact subsystem erasure-correcting code, the maximal coherent information reaches $\log d_L$,  corresponding to perfect recovery of the logical subspace.  For a general encoding isometry, the maximal value is smaller, reflecting imperfect recovery. When the residual reconstruction error is small --- or, equivalently, when $I_c(\mathcal N)$ is close to maximal --- the code is an approximate QECC. Now that we have defined the approximated QECC as a quantum channel $\mathcal{N}_R$, we need to define the analog matter and geometry entropies as the entropies associated with the channel outputs in order to verify that the approximate recovery is possible.

\begin{definition}
    The analog matter entropies of $a$ and respectively $\bar{a}$ in a code with encoding map $V$ and logical Hilbert space $\mathcal{H}_L=\mathcal L_a\otimes \mathcal L_{\bar{a}}$ that satisfies approximate subsystem complementary recovery on $A$ and $\bar{A}$ is 
    \begin{align}
           S(\sigma_{A_1}^{(R^*)}) &:= S(\Tr_{\bar{A}_1}[\sigma_{A_1\bar{A}_1}^{(R^*)} ])\\
           S(\sigma_{\bar{A}_1}^{(R^*)})&:=S(\Tr_{A_1}[\sigma_{A_1\bar{A}_1}^{(R^*)}])
    \end{align}
    with the understanding that $\mathcal L_a\cong \mH_{A_1}$ and $\mathcal L_{\bar{a}}\cong \mH_{\bar{A}_1}$.
\end{definition}

Intuitively, one can treat these quantities as the entropies of the matter fields in disjoint subregions on a spacetime background.
It is easy to check that this definition reduces to the standard definition of bulk or matter entropy by \cite{Harlow:2016vwg} in the usual RT/FLM formula when $V$ is an exact erasure correction code. 
Here we formally write the matter entropy as the entropy of a state $\sigma_{A_1}^{(R^*)}$ that is best recoverable from the ``boundary''. It is not to be confused with the entropy defined directly over the logical information \cite{Akers:2021fut}, namely, what $\sigma_a^{(L)}$ contains. Suppose $\mH_L= \mathcal L_a\otimes \mathcal L_{\bar{a}}$ and $V$ defines the isomorphism between the code subspace and the logical Hilbert space such that $\mathcal{C}=\operatorname{Im}(V)$, then the above entropy are distinct from $S(\sigma_a^{(L)})=S(\Tr_{\mathcal L_{\bar{a}}}|\psi\rangle\langle\psi]),~ S(\sigma_{\bar a}^{(L)})=S(\Tr_{\mathcal L_{a}}|\psi\rangle\langle\psi])$. Although they are equal when the recovery is perfect, $S(\sigma_a^{(L)})\ne S(\sigma_{A_1}^{(R^*)}), ~S(\sigma_{\bar{a}}^{(L)})\ne S(\sigma_{\bar A_1}^{(R^*)})$ in general.

Similarly, the geometric part of the entropy associated with a subsystem $A$ needs to be modified for an approximate QECC. There is no formal definition of the area or related area operator when the code and related subalgebra becomes approximate. Here we define a notion of area assuming that an RT-like formula continues to hold when the code is approximate.

\begin{definition}
The geometric entropy, which we call the proto-area entropy, is defined as \begin{equation}
    S_{\rm PA}(V,\sigma^{(L)}, A):=S(\rho_A) - S(\sigma_{A_1}^{(R^*)})
\end{equation}
\end{definition}

One notes that the proto-area entropy depends on a triple of quantities --- the encoding map $V$, the logical state $\sigma^{(L)}$, and the bipartition of the physical degrees of freedom $A, \bar A$. For simplicity, we will sometimes drop the arguments it depends on, but it will be clear from context.

It is clear that $S(\rho_A)$ is analogous to the boundary entropy. In the limit of exact subsystem complementary recovery, the difference $S(\rho_A)-S(\sigma_a^{(L)})$ recovers the area term in \cite{Harlow:2016vwg}. 
The same conclusion holds even when the code satisfies only exact subalgebra complementary recovery and has trivial area operator. The recovery unitary $R$ that maximizes the coherent information  $I_c(\mathcal{N}_R)$ defines a channel $\mathcal{N}_R$  that correctly reconstructs the algebraic state $\sigma_Q^{(R^*)}$ of \cite{Harlow:2016vwg}  over some von Neumann algebra $Q$. In this setting, the area term is given by $\langle\mathcal{A}\rangle=S(\rho_A)-S(\sigma_Q^{(R^*)})$, 
which precisely matches the definition of PA entropy.

Heuristically, we can also arrive at the same definition by implicitly treating $A_1\bar A_1$ as the IR or low energy Hilbert space in which the matter field degrees of freedom live and $A_2\bar A_2$ as the UV or high energy subspace that capture the quantum gravity degrees of freedom which build up the background geometry.

To define a purely geometry quantity $S_{\rm geom}$ that is analogous to the one in the RT formula, it is natural to subtract the correlation $I(A_1: A_2)$ between the two sectors from $S(\sigma_{A_2}^{(R^*)})$, where $I(A_1:A_2)=S(\sigma_{A_1}^{(R^*)})+S(\sigma_{A_2}^{(R^*)})-S(\sigma_{A_1A_2}^{(R^*)})$.
It then follows that the analog quantity for area should be
\begin{equation}
    S(\sigma_{A_2}^{(R^*)})-I(A_1:A_2) = S(\sigma_{A_1A_2}^{(R^*)}) -S(\sigma_{A_1}^{(R^*)}) = S(\rho_A)-S(\sigma_{A_1}^{(R^*)}) = S_{\rm PA}.
\end{equation}

Before we proceed, it is helpful to build up some intuition for these entropic quantities by comparing them with terms in the generalized second law and the quantum extremal surface (QES) formula \cite{Engelhardt_2015}

\begin{equation}
    S_A = \mathrm{Ext}_{\Gamma_A} \left[ \frac{A(\Gamma_A)}{4 G_N} + S_{\rm bulk}\right],
    \label{eq: QES formula}
\end{equation}
where $\Gamma_A$ extremizes the sum of its surface area and the ``matter'' entropy contribution $S_{\rm bulk}$.

In AdS/CFT, the QES formula decomposes boundary entropy into bulk entropy within the entanglement wedge and an area term evaluated on a state-dependent extremal surface. 
In our framework, the entropy of a physical subsystem $A$ plays the role of generalized entropy while the optimally recoverable bulk state plays the role of the bulk entropy. The proto-area captures the geometrical contribution to the boundary entropy, which will consist of two parts: a fixed background area term which is always present, and an additional state-dependent correction in the case where bulk reconstruction is imperfect. This state dependence captures gravitational effects that allow the extremal surface area to vary, in close analogy to QES formalism \footnote{We emphasize that this is merely an analogy to build up a more concrete mental image --- the connection between the proto-area and the actual area term in the QES formula in holography is still  far from precise, as we will discuss in the next section.}.

Note that generally $S_{PA}(V, \sigma^{(L)},A)$ need not coincide with $S_{PA}(V, \sigma^{(L)},\bar{A})$ and is expected when complementary recovery is broken by the no-man's land when discussing QES. Such effects have also been observed in specific models \cite{Cao_2021,Steinberg_2023,Cao:2021fyk}.  We also emphasize that while we invoke the QES as a helpful scaffold to build up intuition, we do not claim that the proto-area entropy is the QES area.

\section{General properties of the Proto-area entropy}
\label{sec:4}

We will now prove that deviations from exact erasure correction codes will yield monotonic dependence of proto-area on bulk entanglement or bulk entropy depending on whether the encoded state is pure or mixed. 
Specifically, the averaged $S_{PA}$ over random local unitary encoding $R_A, R_{\bar{A}}$ is a monotonically increasing function of the ``bulk entropy'', which is equal to the entanglement entropy of the logical state when it is pure and the thermal entropy when it is mixed. This dependence is qualitatively consistent with what one expects extremal surface areas to change in the presence of gravity.

We begin with a general subsystem exact erasure-correcting code, where the encoded state is given by
\begin{eqns}
    \ket{\tilde\psi^{(0)}}=V^{(0)}\ket{\psi},
\end{eqns}
for $\ket{\psi} \in \mH_L$. The subsystem erasure-correcting property ensures that, for any bipartition into subsystems $A$ and $\bar A$, this encoded state admits the decomposition
\begin{eqns}\label{eq:exact-code}
    \ket{\tilde \psi^{(0)}}= R_A^{(0)\dagger}R_{\bar A}^{(0)\dagger}\ket{\psi}_{A_1\bar A_1}\ket{\chi}_{A_2\bar A_2},
\end{eqns}
for some recovery unitaries $R_A^{(0)}$ and $R_{\bar A}^{(0)}$. We now consider a skewed code \cite{Cao_2021} obtained by perturbing the encoding isometry $V^{(0)}$ to $V^{(\epsilon)}$.

By the Stinespring dilation theorem, the isometry $V^{(0)}$ can be built using a unitary encoding circuit $U^{(0)}$ with the addition of ancillae initialized at some fixed state, e.g. $|0\rangle$. Without loss of generality, assume that the circuit has a structure $U^{(0)}=\prod_{k=1}^K U_k$, where $k$ labels the layer of encoding circuit. The deformation is implemented by inserting a unitary $e^{i\epsilon W_k}$ close to the identity at each layer of the circuit.
Since any local unitary deformation can be propagated through the circuit, all such operators can be equivalently pushed to the output layer, yielding
\begin{eqns}
    U^{(\epsilon)}=&\prod_{k=1}^Ke^{i\epsilon W_k}U_k\\
    =&\prod_{k=1}^K \left(\prod_{i=1}^{k-1}U_{i}\ e^{i\epsilon W_{k}}(\prod_{i=1}^{k-1}U_{i})^{\dagger}\right) U^{(0)}\\
    :=&e^{i\epsilon W} U^{(0)}.
\end{eqns}
where $W$ is a Hermitian operator that captures the effective deformation operator acting on the boundary  with norm $\sqrt{\Tr[W^{\dagger}W]}\leq d_\chi$ and  $d_\chi$ is the Schmidt rank of $|\chi\rangle$. 
Accordingly, the encoded state of the skewed code takes the form
\begin{eqns}\label{eq:skew}
    \ket{\tilde\psi}=e^{i\epsilon W} V^{(0)}\ket{\psi}. 
\end{eqns}

To recover the logical information from the skewed code, the original recovery unitaries $R_A^{(0)}$ and $R_{\bar A}^{(0)}$ of the exact subsystem erasure-correcting code are no longer sufficient. We denote the optimal recovery unitaries of the deformed code as $R_A^{(\epsilon)}$ and $R_{\bar A}^{(\epsilon)}$. 

After applying recoveries the boundary state becomes
\begin{eqns}\label{eq:recoverbd}
    R_A^{(\epsilon)}R_{\bar A}^{(\epsilon)}\ket{\tilde \psi}=&R_A^{(\epsilon)}R_{\bar A}^{(\epsilon)}e^{i\epsilon W}\ket{\tilde \psi^{(0)}}\\
    =&R_A^{(\epsilon)}R_{\bar A}^{(\epsilon)}e^{i\epsilon W}R_A^{(0)\dagger}R_{\bar A}^{(0)\dagger}\ket{\psi}_{A_1\bar A_1}\ket{\chi}_{A_2\bar A_2}\\
    :=&e^{i\epsilon W_R}\ket{\psi}_{A_1\bar A_1}\ket{\chi}_{A_2\bar A_2},
\end{eqns}
where we have combined the perturbation with the recovery into a single operator $e^{i\epsilon W_R}$.

Taking the outer product of Eq.~\eqref{eq:recoverbd} gives
\begin{equation}
\sigma^{(R^{(\epsilon)})}_{A\bar A}
=
e^{i\epsilon W_R}
\Bigl(|\psi\rangle\langle\psi|_{A_1\bar A_1}\otimes |\chi\rangle\langle\chi|_{A_2\bar A_2}\Bigr)
e^{-i\epsilon W_R}.
\end{equation}

The boundary state on $A$ is obtained by tracing out $\bar A$,
\begin{equation}
\sigma^{(R^{(\epsilon)})}_{A_1A_2}
:=\Tr_{\bar A}\!\left(\sigma^{(R^{(\epsilon)})}_{A\bar A}\right).
\label{eq:sigma_bd}
\end{equation}
The recovered bulk state (supported on $A_1$) is obtained by further tracing out the auxiliary subsystem $A_2$,
\begin{equation}
\sigma^{(R^{(\epsilon)})}_{A_1}
:=\Tr_{A_2\bar A}\!\left(\sigma^{(R^{(\epsilon)})}_{A\bar A}\right)
=\Tr_{A_2}\!\left(\sigma^{(R^{(\epsilon)})}_{A}\right).
\label{eq:sigma_bk}
\end{equation}

For simplicity, in the remainder of this section we restrict our analysis to the case where the entanglement spectrum of $\ket{\chi}$ across $A_2\bar A_2$ is flat \footnote{It is helpful to think of the undeformed code as a stabilizer code, but the results apply to slightly more general codes that have flat spectra. The entanglement spectrum of $|\chi\rangle$ in stabilizer codes not only has to be flat, but also has to have Schmidt rank $q^\ell$ for codes over qudits with local dimension $q$.}, and defer the discussion of the more general case to Sec.~\ref{sec:generalize}.

First we show that the proto-area (PA) entropy is related to the relative entropy difference between boundary and bulk state, as in the following theorem.  
\begin{theorem}\label{thm:relative}
   Suppose the state $|\chi\rangle$ of the undeformed code has a flat entanglement spectrum across $A_2\bar A_2$ for a bipartition $A$ and $\bar A$ of the physical degrees of freedom. Respectively, let $\sigma^{(R^{(\epsilon)})}_{A_1A_2}$ and $\sigma^{(R^{(\epsilon)})}_{A_1}$ denote the recovered boundary state on $A_1A_2$ and the recovered bulk state on $A_1$ for an subsystem erasure-correcting code whose encoding isometry $V^{(0)}$ is perturbed by a small unitary $e^{i\epsilon W}$, so that $V^{(\epsilon)}=e^{i\epsilon W}V^{(0)}$. Let $\sigma^{(R^{(0)})}_{A_1A_2}$ and $\sigma^{(R^{(0)})}_{A_1}$ denote the corresponding recovered states of the exact (undeformed) code, which are the states defined in \eqref{eq:sigma_bd} and \eqref{eq:sigma_bk} by taking $\epsilon\rightarrow0$.

Then the proto-area entropy satisfies
\begin{eqns}
    S_{PA}(V^{(\epsilon)}, \sigma^{(L)}, A)=S(\chi)-S_{corr}(V^{(\epsilon)}, \sigma^{(L)}, A),
\end{eqns}
where 
\begin{equation}
        S_{\rm corr}(V^{(\epsilon)}, \sigma^{(L)}, A)= D\left(\sigma_{A_1A_2}^{({R^*}^{(\epsilon)})}||\sigma_{A_1A_2}^{({R^*}^{(0)})}\right)-D\left(\sigma_{A_1}^{({R^*}^{(\epsilon)})}||\sigma_{A_1}^{({R^*}^{(0)})}\right)
\end{equation}
$\chi=\Tr_{\bar{A}_2}[|\chi\rangle\langle\chi|]$, and  $D(\rho||\xi)$ denotes the relative entropy between  states $\rho$ and $\xi$.
\end{theorem}
\begin{proof}
   See Appendix \ref{app:relative}.
\end{proof}

Now we rewrite the equation above as
\begin{eqns}
D(\sigma_{A_1A_2}^{({R^*}^{(\epsilon)})}||\sigma_{A_1A_2}^{({R^*}^{(0)})})=S_{PA}^{(0)}(A)-S_{PA}^{(\epsilon)}(A)+D(\sigma_{A_1}^{({R^*}^{(\epsilon)})}||\sigma_{A_1}^{({R^*}^{(0)})})
\end{eqns}
Note that this equation is reminiscent of the JLMS relation in AdS/CFT, which states that the relative entropy of boundary state variations equals that of the corresponding bulk state variations at leading order in $\mathcal{O}(G_{N}^{-1})$. The relation was later generalized to the quantum JLMS (qJLMS) relation by incorporating quantum corrections \cite{Dong2017EntropyEE}:
\begin{eqns}
    D(\rho_{boundary}||\xi_{boundary})= \left\langle \frac{\mathcal A^{X_\xi}}{4G_N} - \frac{\mathcal A^{X_\rho}}{4G_N}
+ K^{X_\xi}_{\mathrm{bulk},\xi} - K^{X_\rho}_{\mathrm{bulk},\rho} \right\rangle_{\rho},
\end{eqns}
where $\rho_{boundary}$ and $\xi_{boundary}$ are two distinct boundary states. $K_{\text{bulk},\xi}^{X_{\xi}} $ and $K_{\text{bulk},\rho}^{X_{\rho}}$ are the bulk modular Hamiltonians, with the partition determined by the   coordinate $X_{\rho/\xi}$  of the quantum extremal surface. 
Despite the structural similarity, an essential distinction is that in the qJLMS relation, the variation  originates from changes of the state in the code subspace, whereas in our case, the variation is induced by a skewing of the code subspace itself.

In general, $S_{PA}$ is a complicated functional of the bulk state. Since our primary interest is to understand how the area term depends on bulk entanglement between the ``entanglement wedges'' $a$ and $\bar{a}$ in the bulk and not details the state within each subsystem,  we will average over the logical states $U_aU_{\bar a}\ket{\psi}_L$ over the Haar ensemble that leave the bulk entanglement invariant, and we define the averaged proto-area entropy $\langle S_{PA}\rangle$ by taking the ensemble average over such local unitaries. In other words, the optimized recovery now only depends on the bulk entanglement structure between the wedges, as the details of the state is integrated out by the averaging \footnote{Note the similarity of entanglement-dependent definitions here and in \cite{Akers:2021fut}.}.

To examine how $\langle S_{PA}\rangle$ varies as a function of bulk entanglement, we consider two representative scenarios: In the first case, the bulk qubits are entangled with external reference systems, which is equivalent to having the bulk state as mixed. In the second case, the bulk degrees of freedom form a pure state shared between the two bulk subregions $a$ and $\bar a$.

\subsection{Mixed bulk state}

In this section, we first consider a special case where all the bulk qubits are approximately recoverable from $A$, and are entangled with some external reference system. In other words, the encoded state is mixed (Figure.~\ref{fig:mixed_circuit}).  
Then we show that the proto-area entropy has the property detailed in the following theorem. 
\begin{figure}[t]
    \centering
    \includegraphics[width=1\linewidth]{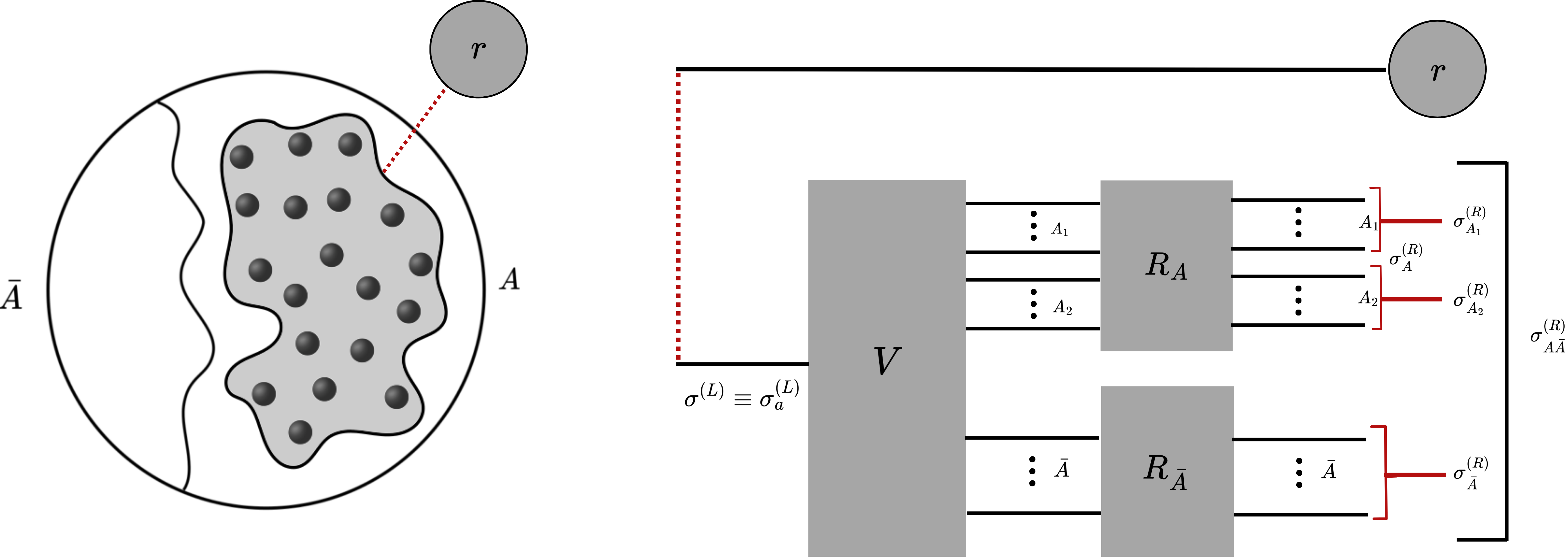}
    \caption{
Left: AdS picture for a \emph{mixed} bulk state supported in the entanglement wedge $\mathrm{EW}(A)$.
The shaded region denotes bulk degrees of freedom in $\mathrm{EW}(A)$, and the red dotted line to the external system $r$
indicates that these bulk degrees of freedom are mixed because they are entangled with $r$. The complementary wedge $\mathrm{EW}(\bar A)$ is taken to contain no relevant bulk degrees of
freedom in this setup. We assume the bulk information in $\mathrm{EW}(A)$ is approximately recoverable from the boundary region $A$. Right: Circuit representation of the same setting. The logical state $\sigma^{(L)}$ is a mixed state on $\mathcal{H}_L$, encoded by the isometry $V$ into $\mathcal{H}_A \otimes \mathcal{H}_{\bar A}$. Independent recovery maps $R_A$ and $R_{\bar A}$ act on the boundary regions, producing the recovered reduced states $\sigma^{(R)}_{A_1}$. We focus on the recovery of bulk information from $A$, treating the resulting map as an effective channel from $\mathcal{H}_L$ to $\mathcal{H}_{A_1}$, with the joint state $\sigma^{(R)}_{A\bar A}$ encoding the correlations with the complementary region.}
\label{fig:mixed_circuit}
\end{figure}
\begin{theorem}\label{thm:monotonic-mixed}
    Assume that $|\chi\rangle$ has flat entanglement spectrum in the undeformed code following a bipartition into $A$ and $\bar{A}$. Consider encoding a mixed state $\sigma$ into the bulk which can be purified into Bell-like states $(\sum_i \sqrt{\lambda_i}|i\rangle|i\rangle)^{\otimes \ell}$ for some local basis $\{|i\rangle\}$ with the addition of a reference $r$. In leading order of $\epsilon$, the bulk-unitary-averaged correction to the proto-area entropy, $\langle S_{corr}\rangle$, is non-negative and decreases monotonically with the amount of bulk entropy. 
\end{theorem}

\begin{proof}
By averaging the proto-area entropy over the bulk local unitaries, we derive the following equation for the proto-area entropy correction up to $\mO(\epsilon^2)$ (The details of derivation are in Appendix~\ref{app:mixed}.),
\begin{eqns}\label{eq:Scorr}
    \langle S_{corr}\rangle =\frac{\epsilon^2}{2}\left(c_1f_1(\lambda)+c_2f_2(\lambda)\right) + c_3,
\end{eqns} 
where $c_1$, $c_2$, $c_3$ are non-negative coefficients only depending on $W_R$ in Eq.~\eqref{eq:recoverbd}. Both $f_1(\lambda)$ and $f_2(\lambda)$ are function of the eigenvalues $\lambda_i$ of the input logical mixed state.  They are monotonically decreasing functions of the matter entropy $S_{\rm matter}(\lambda)$. 
 
\end{proof}

\begin{corollary}
    In leading order perturbation, $\langle S_{PA}\rangle$ increases monotonically with bulk entropy.
\end{corollary}

\begin{figure}[t]
    \centering
    \includegraphics[width=0.7\linewidth]{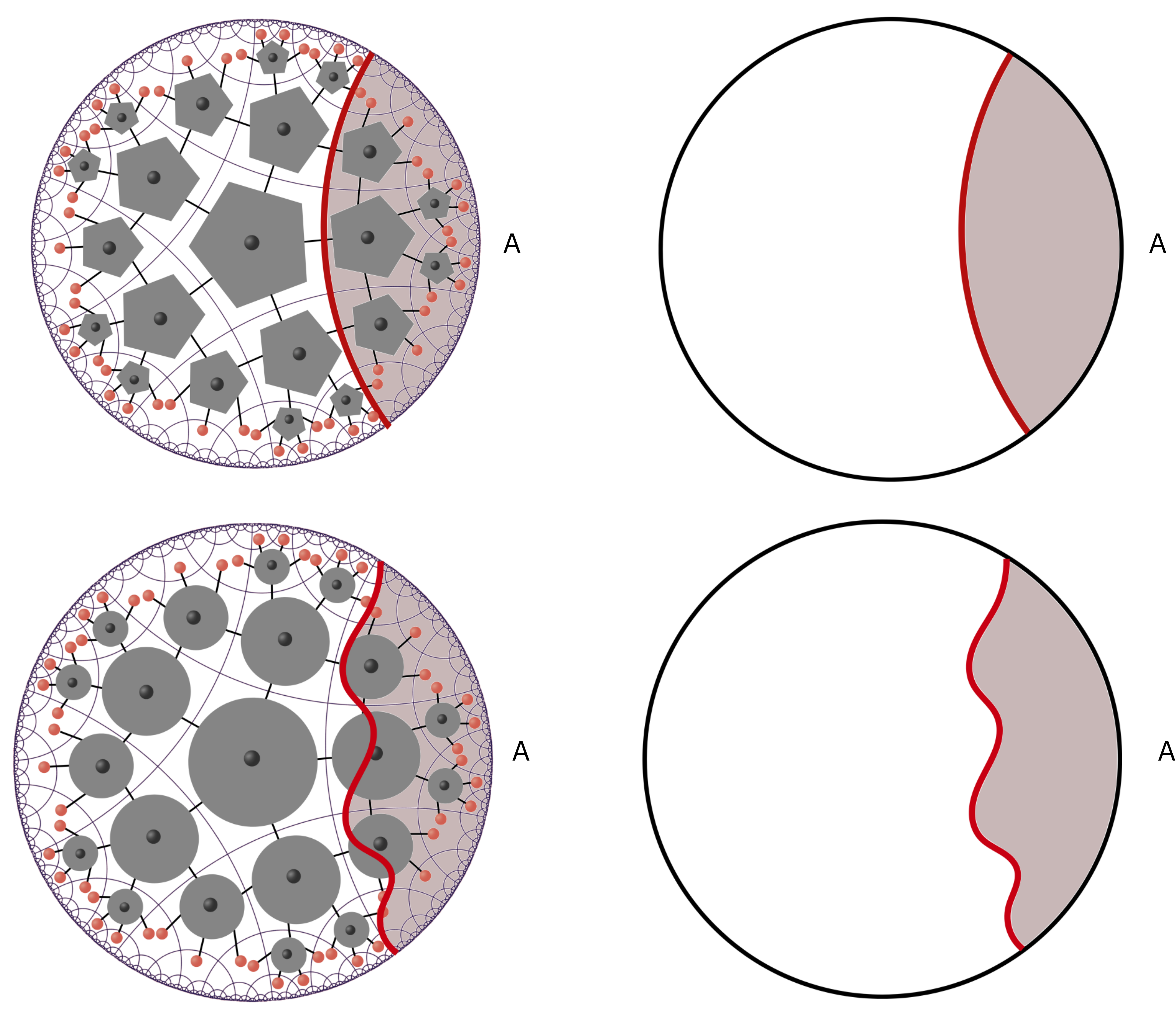}
    \caption{
    Top: The proto-area surface in a holographic stabilizer codes is unchanged by the encoded logical information, similar to the RT surface in holography where the bulk background geometry is fixed. Bottom: a magic-enriched code where the encoding circuits are skewed away from the exact encoding maps now permits the area of some (extremal) surface to become state-dependent, an effect observed in QES and in systems with gravitational backreaction in holography. {Note that area changes from gravitational backreactions and shifts in the QES are distinct concepts and should not be confused with each other even though they both contribute to state dependence.}}
    \label{fig:qes}
\end{figure}

Note that this increase in the ``area'' of a surface in the presence of bulk or matter entropy is analogous to how the area of the quantum extremal surface could shift in the presence of bulk entanglement \cite{Engelhardt_2015,Akers:2019wxj} (Figure~\ref{fig:qes}). A qualitatively similar shift of a different physical origin in classical extremal surface area can also be observed when gravitational back-reaction is incorporated \cite{Porrati_2004,Ghosh_2016,Ryu_2006}. In Appendix~\ref{app:perfect_fluid}, we also explicitly reproduce one such scenario in $AdS_3$ where the classical extremal surface area will increase in response to an increase of bulk entropy. Such behaviors have also been observed in some tensor network models such as \cite{Hayden_2016,Cao_2021} which are approximate quantum codes where it has been argued to mimic gravitational features. However, we hasten to point out that this change in the quantum information theoretic quantity does not yet have a clear correspondence with either effect beyond the current qualitative similarity.  This is because it is unclear if the extremization in the usual quantum extremal surface can be related to our optimal recovery condition. A similar correspondence with backreaction also requires further analysis to clarify its connection with, e.g., the linearized Hamiltonian constraint, in specific families of the code. We note that the lack of a more precise correspondence with spacetime and gravity is difficult at this level of generality because not all quantum codes admit spacetime descriptions. We will leave a rigorous analysis of its physical meaning in specialized quantum codes to future work.

The optimal recovery is obtained by maximizing the coherent information between the reference system and the recovered bulk state over the choice of local recovery unitaries $R_A^{(\epsilon)}$ and $R_{\bar A}^{(\epsilon)}$.  We prove the following lemma regarding optimization:
\begin{lemma}\label{lemma:optimization:mixed}
    The optimal recovery, defined in Eq.~\eqref{eq:Rstar}, 
is achieved when $c_1(W_R)=0$ in Eq.~\eqref{eq:Scorr}. 
\end{lemma}

The proof is given in the appendix~\eqref{app:optimization:mixed}. There we show (i)~by an appropriate choice of the local recovery, one can always set $c_1(W_R)$ to zero, achieving the optimum, and (ii)~$c_2(W_R)$ is invariant under such variations of the local recovery, and hence can be regarded as function of the encoding unitary, $c_2(V)$.
Therefore, in the case where the bulk is in a mixed state, the PA entropy takes on a universal behavior where it only depends on the perturbation via the parameter $c_2(V)$. We will show later that it is precisely connected to non-local magic when the unperturbed system is a stabilizer code.

\subsection{Pure bulk state}

The bulk can also take on more general quantum states. Next, we consider the case where the logical state is any pure state $|\psi\rangle$ which may be entangled between the bulk subregions $a$ and $\bar a$. Any such state can be written in the form 
\begin{equation}
    \label{eq:bulk_state_pure_case}\ket{\psi}_L=\sum_{i=1}^d \sqrt{\lambda_i}\ket{i}_a\ket{i}_{\bar a}
\end{equation}
for some Schmidt basis and  coefficients $\sqrt{\lambda_i}$ (Figure~\ref{fig:pure_circuit}).
We discover a similar dependence of the PA entropy on the entanglement spectrum $\{\lambda_i\}$. The corresponding encoded boundary state takes the form of Eq.~\eqref{eq:recoverbd}, partitioned into the subsystems $A_1$, $\bar A_1$ (supporting the recovered state) and  $A_2$, $\bar A_2$ (representing the geometric part). 
\begin{figure}[t]
    \centering
    \includegraphics[width=1\linewidth]{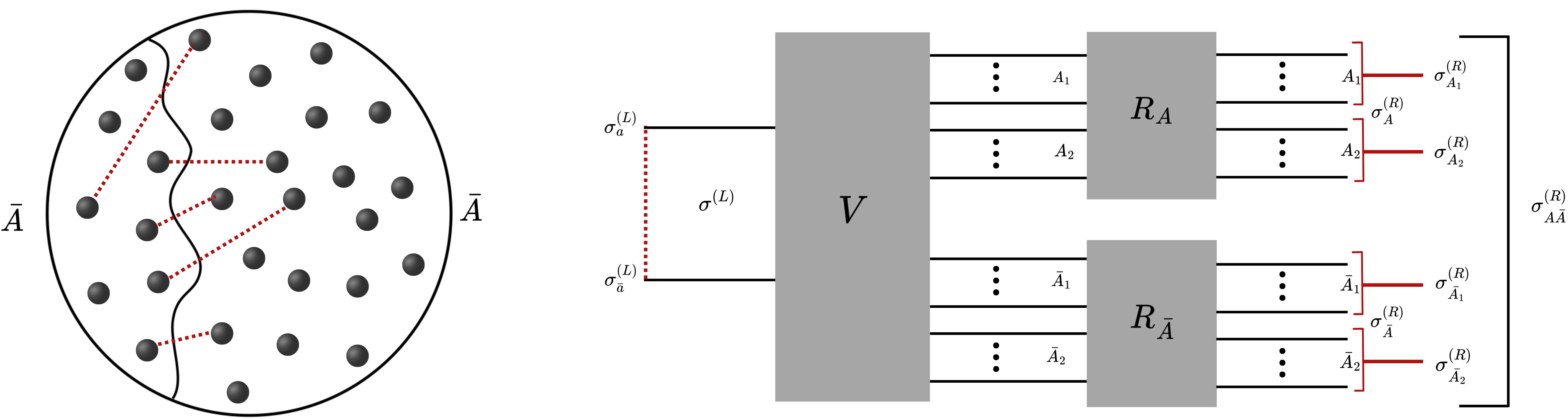}
 \caption{
Left: AdS picture, where the bulk degrees of freedom are in a pure but generally entangled state. Internal entanglement within $a, \bar{a}$ are permitted, but they do not contribute to the entropy dependence. Therefore they are not shown in the figure.
Right: Circuit representation of the same setup. The logical state $\sigma^{(L)}$ is a pure state on $\mathcal{H}_a \otimes \mathcal{H}_{\bar a}$, with the red dotted lines indicating entanglement between $a$ and $\bar a$. This state is encoded by the isometry $V$ into $\mathcal{H}_A \otimes \mathcal{H}_{\bar A}$, followed by independent recovery maps $R_A$ and $R_{\bar A}$. The outputs define the recovered reduced states $\sigma_{A_1}^{(R)}$, $\sigma_{\bar A_1}^{(R)}$, and the joint state $\sigma_{A\bar A}^{(R)}$.
}
\label{fig:pure_circuit}
\end{figure}
Similar to the previous setting, we study typical input logical states by averaging over logical local unitaries over $a$ and $\bar{a}$ (or $A_1$ and $\bar{A}_1$ via an isomorphism). Let $W$ denote the skewing matrix defined in Eq.~\eqref{eq:skew}. We then establish the following result for the skewed erasure-correcting code:
\begin{theorem}\label{thm:monotonic-pure}
        Assume that $\ket{\chi}$ has flat spectrum in the undeformed code. For a random choice of the skewing matrix $W$, drawn from the Gaussian Unitary Ensemble (GUE), the averaged PA entropy of subregion $A$ is typically a monotonically increasing function of the bulk entanglement in leading order of $\epsilon$. The probability that the PA entropy exhibits this monotonic behavior scales as $\mathcal{O}(1-e^{-d^2})$, where $d$ is the dimension of bulk Hilbert space $\mathcal{L}_a$. 
\end{theorem}
\begin{proof}
Let $\dim \mathcal{L}_a= d, \dim \mathcal{L}_{\bar{a}}=\bar{d}$. As detailed in Appendix~\ref{app:monotonic-pure}, we obtain the following expression up to $\mO(\epsilon^2)$ for the correction to the PA entropy of subregion $A$:
\begin{eqns}\label{eq:Scorrpure}
    \langle S_{corr}\rangle=&\frac{\epsilon^2}{2}\biggr[k_1f_1(\lambda)+k_2f_2(\lambda)+k_3\left(f_3(\lambda)-\frac{1}{d^2}(1+\frac{d}{\bd})f_2(\lambda)+\frac{1}{d^2}f_1(\lambda)\right)\\
    &+k_4\left(f_3(\lambda)-\frac{1}{d\bar d}f_1(\lambda)\right)+k_5\left(f_2(\lambda)-f_1(\lambda)\right)+k_6\biggr],
\end{eqns}
where the parameters $k_1,\cdots, k_6$ are independent, non-negative functions of the matrix $W_R$ (defined in Eq.~\eqref{eq:recoverbd}).
All three functions $f_1(\lambda)$, $f_2(\lambda)$ and $f_3(\lambda)$ are monotonically decreasing functions of the bulk entanglement, while the combination $f_2(\lambda)-f_1(\lambda)$ is monotonically increasing.\\ 
In Appendix~\ref{app:monotonic-pure} we show that $S_{corr}(A)$ decreases monotonically with bulk entanglement whenever $k_5<k_3$.  For a typical random choice of $W$ from the GUE, the expected ratio between these two coefficients satisfies
\begin{eqns}
    \Big\langle{\frac{k_5}{k_3}}\Big\rangle_{\rm GUE}=\frac{1}{d^2},
\end{eqns}
and the probability of violating the inequality $k_5< k_3$ happens at probability 
\begin{eqns}
    P(k_5\geq k_3)\sim \mathcal{O}(e^{-d^2}).
\end{eqns}
\end{proof}
Similar to the case with mixed bulk state, the optimal recovery is obtained by maximizing the coherent information between reference system and the recovered bulk state. We prove the following lemma regarding optimization in Appendix~\ref{app:monotonic-pure}. 

\begin{lemma}\label{lemma:optimization-pure}
    The optimal recovery, defined in Eq.~\eqref{eq:Rstar}, is achieved when $k_1(W_R)=0$ in Eq.~\eqref{eq:Scorrpure}. 
\end{lemma}

The increase in the PA entropy can be similarly compared to the increase of QES area in holography \cite{Akers:2019wxj,Porrati_2004} where the extremal surface area can increase in the presence of bulk entanglement to reduce the amount of generalized entropy contribution from bulk entanglement.

\subsection{Perturbing general erasure correction code}\label{sec:generalize}
We further extend our results to general codes  with approximate subsystem complementary recovery. First we show that the skewing condition specified in Eq.~\eqref{eq:recoverbd} is equivalent to the approximate Knill-Laflamme (KL) condition.

\begin{theorem}\label{thm:KLcondition}
Let $\mH_P=\mH_A\otimes \mH_{\bar A}$ be the physical Hilbert space, and let
$\mathcal C\subset \mH_P$ be a code subspace isomorphic to a logical Hilbert space
$\mL=\mL_a\otimes \mL_{\bar a}$ of dimension $d_L:=\dim\mL$.
Fix an orthonormal basis $\{\ket{i_L}\}_{i=1}^{d_L}$ of $\mL$ and an isometry
$V:\mL\to \mH_P$.  Denote the corresponding codewords by
\[
\ket{\tilde i}:=V\ket{i_L}\in \mathcal C.
\]
Then the following two statements are equivalent.

\begin{enumerate}
\item \textbf{Approximate Knill--Laflamme (aKL) conditions.}
There exist completely positive linear maps
\[
\mathcal E_{\bar A}: \mathcal B(\mH_{\bar A})\to \mathcal B(\mL_{\bar a}),\qquad
\mathcal E_{A}: \mathcal B(\mH_{A})\to \mathcal B(\mL_{a}),
\]
and families of sesquilinear error functionals
$\{Y_{ij}(\,\cdot\,)\}_{i,j=1}^{d_L}$ and $\{\bar Y_{ij}(\,\cdot\,)\}_{i,j=1}^{d_L}$ such that for all
$i,j$ and all $X_{\bar A}\in\mathcal B(\mH_{\bar A})$, $X_{A}\in\mathcal B(\mH_A)$,
\begin{eqns}
\bra{\tilde i}(I_A\otimes X_{\bar A})\ket{\tilde j}
&=
\bra{i_L}\!\left(I_a\otimes \mathcal E_{\bar A}(X_{\bar A})\right)\!\ket{j_L}
\;+\;\epsilon\,Y_{ij}(X_{\bar A}), \\
\bra{\tilde i}(X_A\otimes I_{\bar A})\ket{\tilde j}
&=
\bra{i_L}\!\left(\mathcal E_{A}(X_A)\otimes I_{\bar a}\right)\!\ket{j_L}
\;+\;\epsilon\,\bar Y_{ij}(X_A). \label{eq:appKL}
\end{eqns}

\item \textbf{Approximate recovery up to fixed ancillas.}
Introduce auxiliary systems $E$ and $\bar E$ initialized in a fixed product state
$\ket{0}_{E\bar E}\in \mH_E\otimes\mH_{\bar E}$, and consider the enlarged physical space
\[
\mH_{P'}:=\mH_{AE}\otimes \mH_{\bar A\bar E},
\qquad
\mH_{AE}:=\mH_A\otimes\mH_E,\quad
\mH_{\bar A\bar E}:=\mH_{\bar A}\otimes\mH_{\bar E}.
\]
Assume that $\mH_{AE}$ and $\mH_{\bar A\bar E}$ admit the following decompositions
\[
\mH_{AE}\simeq (\mH_{A_1}\otimes \mH_{A_2})\oplus \mH_{A_3},\qquad 
\mH_{\bar A\bar E}\simeq (\mH_{\bar A_1}\otimes \mH_{\bar A_2})\oplus \mH_{\bar A_3},
\]
together with identifications $\mH_{A_1}\simeq \mL_a$ and $\mH_{\bar A_1}\simeq \mL_{\bar a}$.
Then there exist unitaries $R_{AE}$ on $\mH_{AE}$ and $R_{\bar A\bar E}$ on $\mH_{\bar A\bar E}$,
a fixed state $\ket{\chi}\in \mH_{A_2}\otimes \mH_{\bar A_2}$ independent of $i$, and a Hermitian operator
$W$ on $\mH_{P'}$ such that for every $i$,
\begin{equation}\label{eq:apprecovery}
\ket{\tilde i}\otimes \ket{0}_{E\bar E}
=
e^{\,i\epsilon' W}\,
(R_{AE}^\dagger\otimes R_{\bar A\bar E}^\dagger)
\Bigl(\ket{i}_{A_1\bar A_1}\otimes \ket{\chi}_{A_2\bar A_2}\Bigr).
\end{equation}
Here $\ket{i}_{A_1\bar A_1}$ is the image of $\ket{i_L}\in\mL_a\otimes\mL_{\bar a}$ under the fixed identifications
$\mL_a\simeq \mH_{A_1}$ and $\mL_{\bar a}\simeq \mH_{\bar A_1}$.
\end{enumerate}

\noindent\textbf{Parameter relations.}
For the implication \textup{(1)$\Rightarrow$(2)}, given $\epsilon$, $Y_{ij}$ and $\bar Y_{ij}$,
one can choose $W$ and $\epsilon'$ such that
\begin{equation}\label{eq:param-12}
\epsilon'\,\|W\|_2\;\le\; \pi\sqrt{\epsilon\,d_L}\!\left(\sqrt{d_{A}\,\|Y\|}+\sqrt{d_{\bar A}\,\|\bar Y\|}\right),
\end{equation}
where
\[
\|Y\|:=\sup_{\|X\|\le1}\max_{i,j}|Y_{ij}(X)|,\qquad
\|\bar Y\|:=\sup_{\|X\|\le1}\max_{i,j}|\bar Y_{ij}(X)|.
\]
Conversely, for \textup{(2)$\Rightarrow$(1)}, given $W$ and $\epsilon'$, the error functionals can be chosen so that
\begin{equation}\label{eq:param-21}
\epsilon\,\max\{\|Y\|,\|\bar Y\|\}\;\le\; 2\epsilon'\,\|W\|_2.
\end{equation}
\end{theorem}

 See Appendix~\ref{app:KLcondition} for proof.
 Starting from Eq.~\eqref{eq:apprecovery} with general fixed state $\ket{\chi}$, we show  that results analogous to  Theorem~\ref{thm:monotonic-mixed} and Theorem~\ref{thm:monotonic-pure} still hold.

\begin{theorem}\label{thm:general_mixed}
    For general state $\ket{\chi}$ with entanglement spectrum $\{\mu_n\}$, and the input bulk state is a mixed state with spectrum $\{\lambda_i\}$, the local-unitary-averaged PA entropy is $\langle S_{PA}\rangle=S(\chi)-\langle S_{corr}\rangle$, with $\langle S_{corr}\rangle$ taking the following form: 
    \begin{eqns}
        \langle S_{corr}\rangle=\epsilon c_0+\frac{\epsilon^2}{2}\sum_{mn}\left(c_1^{mn}f_1^{mn}(\lambda)+c_2^{mn}f_2^{mn}(\lambda)+c_3^{mn}\right). 
    \end{eqns}
where $c_0$, $c_1^{mn}$ and $c_2^{mn}$, $c_3^{mn}$ are non-negative function of $W_R$ and $\chi$. $f_1^{mn}(\lambda)$ and $f_2^{mn}(\lambda)$ depend on $\nu_n$ and $\nu_m$. They are all monotonic decreasing function of the bulk entropy. 
\end{theorem}

See Appendix~\ref{app:general_mixed} for details of the proof. This result implies that the averaged PA entropy is a monotonic increasing function of the bulk entropy when the input logical state is mixed. Similarly, for the case with pure logical state, we have the following generalization:

\begin{theorem}\label{thm:general_pure}
        For general $\ket{\chi}_{A_2\bar A_2}$ with entanglement spectrum $\{\mu_n\}$, and the input logical state is pure with entanglement spectrum $\{\lambda_i\}$ shared by the subregion $A_1$ and $\bar A_1$, the averaged PA entropy of subregion $A$ is $\langle S_{PA}\rangle=S(\chi)-\langle S_{corr}\rangle$, where $\langle S_{corr}\rangle$ takes the following form: 
\begin{eqns}
    \langle S_{corr}\rangle=&\epsilon k_0+ \frac{\epsilon^2}{2}\sum_{mn}\biggr[k_1^{mn}f^{mn}_1(\lambda)+k_2^{mn}f^{mn}_2(\lambda)+k_3^{mn}\left(f^{mn}_3(\lambda)-\frac{1}{d^2}(1+\frac{d}{\bd})f_2^{mn}(\lambda)\right.\\
    &\left.+\frac{1}{d^2}f_1^{mn}(\lambda)\right)+k_4^{mn}\left(f_3^{mn}(\lambda)-\frac{1}{d\bar d}f_1^{mn}(\lambda)\right)
    +k_5^{mn}\left(f_2^{mn}(\lambda)-f_1^{mn}(\lambda)\right)+k_6^{mn}\biggr] ,
\end{eqns}
        where $\dim \mathcal{L}_a= d$, $\dim \mathcal{L}_{\bar{a}}=\bar{d}$. The parameters $k_0, k_1^{mn},\cdots, k_6^{mn}$ are independent, non-negative functions of the matrix $W_R$ (defined in Eq.~\eqref{eq:recoverbd}).
The functions $f_1^{mn}(\lambda)$, $f_2^{mn}(\lambda)$ and $f_3^{mn}(\lambda)$ depend on $\mu_n$ and $\mu_m$, and are monotonically decreasing functions of the bulk entanglement, while the combinations $f_2^{mn}(\lambda)-f_1^{mn}(\lambda)$ are monotonically increasing. 
\end{theorem}
Proof can be found in Appendix~\ref{app:general_pure}. For $W_R$ drawn from the Gaussian random ensemble, one can analyze the typical behavior of PA entropy similarly as in Theorem~\ref{thm:monotonic-pure}. In the large $d$ and $\bar d$ limit, the \emph{typical} PA entropy correction in the non-flat case is dominated by the block-summed $f^{mn}_3$ sector,
\begin{equation}
\big\langle S_{\mathrm{corr}}\big\rangle
=
\frac{\epsilon^2}{2}\sum_{m,n}\big(k^{mn}_3+k^{mn}_4\big)\,f^{mn}_3(\lambda)
+\mathcal{O}\!\left(\frac{1}{d^2}\right)
+\mathcal{O}\!\left(\frac{1}{\bar d^{2}}\right)
\end{equation}
and therefore decreases monotonically as a function of bulk entropy.

\section{Non-local Magic in Skewed Stabilizer codes}
\label{sec:5}

The above analysis implies that in the limit of large logical Hilbert space, the typical behavior of the PA entropy correction will depend only the entanglement of the logical information, with the strength of the coupling dominated by $c_2$ if the logical state is mixed,  and by $k_2$--$k_5$ if the state is pure.
Now we show that these couplings are given by the amount of non-local tripartite magic in the system. Importantly, the relevant magic here has to be non-local to reproduce features of gravity \cite{Cao:2023mzo,Cao:2024nrx} --- any local or bipartite non-Clifford deformations can be absorbed into the recovery unitaries and therefore cannot correlate the recovered bulk degrees of freedom with the geometric entanglement.

Let us now be more precise by defining the magic of a quantum code. 
\begin{definition}[Magic of code]Let $\ket{V}$ be the Choi state of a code associated with an encoding map (usually an isometry) $V:\mathcal{H}_L\rightarrow\mathcal{H}_P$ where
\begin{eqns}
    \ket{V}=\frac{1}{\sqrt d_L}\sum_{i=1}^{d_L} \ket{i}_{r}\ket{\tilde i}=\frac{1}{\sqrt d_L}\sum_{i=1}^{d_L} \ket{i}_{r}V\ket{i}_L,
\end{eqns}
and $d_L:=\dim \mH_L$.
The magic of the code is defined as
\begin{eqns}
    \mathcal{M}(V):=\mathcal{M}(\ket{V}), 
\end{eqns}
for some magic measure $\mathcal{M}$. 
\end{definition}

The magic measure we will use is the \textit{Stabilizer Renyi Entropy} (SRE) \cite{Leone_2022}  because of its computability. The SRE of a  state $\ket{\phi}$ is defined as,
    \begin{eqns}
    \mathcal{M}_{\alpha}(\ket{\phi}):=\frac{1}{1-\alpha}\log\left(2^{-n}\sum_{t=1}^{4^n}|\bra{\phi}P_t\ket{\phi}|^{2\alpha}\right),
\end{eqns}
where $n$ is the total number of qudits in $\ket{\phi}$. It is a magic monotone when $\alpha\geq 2$ \cite{PhysRevA.110.L040403}.

In the general context of perturbing away from stabilizer codes, let us now define the perturbative tripartite non-local magic. 

\begin{definition}[perturbative tripartite non-local magic]
Let $|\phi\rangle$ be a quantum state over $n$ subsystems $A_1, A_2,\dots, A_n$. The perturbative tripartite nonlocal magic is defined as 
    \begin{eqns}
    \mathcal{M}^{NL}(|\phi\rangle)=\min_{\forall \sigma, T_{A_{i}A_{j}}}  \mathcal{M}\left(\prod_{ij} e^{i\epsilon T_{A_{\sigma(i)} A_{\sigma(j)}}}|\phi\rangle\right),
    \label{eq: tripartite magic}
\end{eqns}
where the minimization is done at the leading order of $\epsilon$ over any Hermitian operator $T_{A_i A_j}$ restricted to subsystems $A_i, A_j, i\ne j$ and any permutation of the indices $\sigma(i)=i'$. 

\label{def:tripartitemagic}
\end{definition}

For the following, we will restrict ourselves to the  case where the encoding unitary is close to Clifford, so that the base code we perturb from is a stabilizer code. To inject magic, now consider a skewed stabilizer code with encoding map $V^{(\epsilon)}$ satisfying the subsystem erasure-correction as in Eq.~\eqref{eq:exact-code} with perturbation $e^{i\epsilon W}$ (c.f. \cite{Cao_2021} for skewed codes). The skewing generally injects magic into the code, some of which is non-local, as $W$ has support over the entire system. This provides the necessary condition for emerging gravitational features \cite{Cao:2024nrx}.  

For a skewed stabilizer code, the Choi state $\ket{V^{(\epsilon)}}$ is locally Clifford-equivalent to
\begin{eqns}
    \ket{V_R^{(\epsilon)}}:=R_A^{(0)}R_{\bar A}^{(0)}\ket{V^{(\epsilon)}}=
     \frac 1 {\sqrt{d_L}}\sum_{i=1}^{d_L} \ket{i}_r\otimes e^{i\epsilon W_R}\ket{i}_{A_1\bar A_1}\ket{\chi}_{A_2\bar A_2}, 
\end{eqns}
where $R_A^{(0)}$ and $R_{\bar A}^{(0)}$ are the Clifford recovery unitaries for the unperturbed ($\epsilon=0$) stabilizer code, and
\begin{eqns}
W_R := \bigl(R_A^{(0)}\otimes R_{\bar A}^{(0)}\bigr)\,W\,\bigl(R_A^{(0)}\otimes R_{\bar A}^{(0)}\bigr)^\dagger.
\end{eqns}

We then define the perturbative tripartite nonlocal magic of the code by
\begin{eqns}
\mathcal{M}^{NL}\!\left(V^{(\epsilon)}\right)
:=\mathcal{M}^{NL}\!\left(\ket{V_R^{(\epsilon)}}\right),
\end{eqns}

Taking the case where $\bar a$ (or equivalently $\bar A_1$) is empty, the parameter $c_2$ in the PA entropy calculation (see Theorem~\ref{thm:monotonic-mixed}) is exactly given by the perturbative tripartite non-local SRE $\mathcal{M}_{\alpha}^{NL}(V^{(\epsilon)})$, up to constant prefactor. 

\begin{theorem}\label{thm:nlmagicmixed}
     Let $V^{(\epsilon)}$ be  a skewed stabilizer code, and subregion $A$ approximately recovers the entire bulk state.  In leading order of $\epsilon$, the PA entropy correction $\langle S_{corr}(V^{(\epsilon)}, \sigma^{(L)}, A)\rangle$ is proportional to $\mathcal{M}_{\alpha}^{NL}(V^{(\epsilon)})$. More specifically, 
    \begin{eqns}
        \langle S_{corr}(V^{(\epsilon)},\sigma^{(L)},A)\rangle =\frac{\alpha-1}{2\alpha}\lrp{1-\frac{1}{d^2}}^{-1}\mathcal{M}_{\alpha}^{NL}(V^{(\epsilon)})f_2(\lambda)+const,
    \end{eqns}
where $f_2(\lambda)$ is a function of the spectrum of the input mixed logical state. 
\end{theorem}

See proof in Appendix~\ref{app:magic}.

For the case where neither $a$ nor $\bar a$ is empty, the perturbative tripartite non-local magic $\mathcal{M}_{\alpha}^{NL}(V^{(\epsilon)})$ is directly related to the parameters $k_2, \cdots, k_5$ that appear in the PA entropy calculation when the encoded state is pure (see Theorem~\ref{thm:monotonic-pure}).

\begin{theorem}
    Let $V^{(\epsilon)}$ be  a skewed stabilizer code, encoding a pure bulk state. To leading order in $\epsilon$, the perturbative tripartite non-local magic of the code is
     \begin{eqns}
         \mathcal{M}^{NL}_\alpha\bigl(V^{(\epsilon)}\bigr)=& \frac{\alpha \epsilon^2}{(\alpha-1)}\lrp{(1-\frac{1}{d^2})(1-\frac{1}{\bar d^2})(k_3+k_4)+(1-\frac{1}{d^2})k_2+(1-\frac{1}{\bar d^2})k_5}+\mO(\epsilon^3),
     \end{eqns}
     in terms of the stabilizer $\alpha$-Renyi entropy.
\end{theorem}

\begin{corollary}
To leading order in $\epsilon$, and in the limit of large logical Hilbert-space dimensions $d, \bar d\gg 1$, the typical PA-entropy correction satisfies
\begin{eqns}
    \langle S_{corr}\rangle=\frac{\alpha-1}{2\alpha}(1-\frac{1}{d^2})^{-1}(1-\frac{1}{\bar d^2})^{-1}\mathcal{M}_{\alpha}^{NL}(V^{(\epsilon)})f_3(\lambda)+\mathcal{O}(\frac{1}{\bar d^2})+\mathcal{O}(\frac{1}{d^2}).
\end{eqns}

\end{corollary}

At finite $d, \bar d$, there is residual coupling between bulk entanglement and some function of $W$ which should be better understood. We will leave this for future work.

Although we have restricted ourselves to perturbations of stabilizer codes for technical convenience, we conjecture that the notion that non-local magic gives rise to PA variation is 
general. Indeed, even if there is local magic that is originally present in the code, it is known that they also do not contribute to a non-trivial area operator and yielding no PA dependence on the logical state \cite{Cao:2023mzo}. However, an extended version of Definition \ref{def:tripartitemagic} and the above theorems will be needed to account for the initial magic that did not originate from the non-local perturbation $W$.

\section{Discussion}\label{sec:6}
In this work we studied how deviations from exact quantum error-correcting
codes can produce entropic features reminiscent of gravity.
We introduced a modified RT-like entropy decomposition for approximate
subsystem erasure-correcting codes and showed that the resulting proto-area
entropy becomes state dependent. In particular, the proto-area typically
increases with the entropy or entanglement of the logical state, closely
resembling gravity-induced effects such as the response of extremal surfaces to bulk entropy in the quantum
extremal surface (QES) formula or changes of geometry in gravitational backreaction. We further showed that the strength of this
dependence is controlled by a tripartite form of non-local magic in the encoding.
These results apply to a broad class of skewed subsystem quantum codes and
provide toy models capable of reproducing gravity-like entropic behavior in
both hyperbolic and near-flat emergent geometries.

A key lesson is that exact subsystem erasure-correcting codes are too rigid to reproduce
gravitational backreaction. For such codes, matter and
geometric degrees of freedom are cleanly separated and the area term is
necessarily state independent. Approximate recovery relaxes this separation,
allowing correlations between recoverable bulk degrees of freedom and the
geometric entanglement structure. The absence of backreaction in stabilizer
codes can therefore be traced to the absence of the required non-local quantum
resources.

Our analysis also highlights an important distinction between operator
reconstruction and state recovery in approximate codes. While these notions
coincide in exact QECCs, they can diverge when recovery becomes approximate.
The framework developed here focuses on optimal state recovery and entropic
quantities rather than operator reconstruction. Clarifying the relation between
this entropic picture and operator-based formulations remains an important 
direction for further study.

From an algebraic perspective, extending holographic code models beyond the
exact-code regime might require new mathematical tools. Approximate
erasure correction breaks the exact algebraic structures underlying area
operators in operator algebra quantum error correction. Developing a theory of
approximate $C^*$ or von Neumann algebras may therefore be helpful for
defining geometric observables in approximate codes \cite{Kitaev:2024qak}. Furthermore, our work does not constrain the possibility of having state-dependent proto-area-like entropy in exact \textit{subalgebra} erasure correcting codes. Indeed, such codes are expected to be able to support non-trivial area operators. Concurrently, the emergence of the QES-like behavior in approximate codes should be examined in close connection with systems that can produce non-trivial area operators through algebraic constraints in the bulk \cite{Donnelly_2017,Cao_2021,Dolev_2022,Dong:2023kyr,Qi_2022} and random tensor networks \cite{Hayden_2016}.

An important open question is the geometric interpretation of the proto-area
entropy. Although its monotonic increase with bulk entropy resembles the
behavior of quantum extremal surfaces, it remains unclear whether the
proto-area corresponds to the area of a backreacted classical extremal surface, the quantum extremal surface area, or some other quantity entirely. A more refined and specialized model that more closely resembles holographic theories may provide stronger constraints or evidence.

We emphasize that the proto-area is an information-theoretic quantity defined through approximate bulk reconstruction. In the exact-recovery limit, it coincides with the fixed minimal-surface contribution of the reference code. Away from this limit, it acquires an additional state-dependent correction. This correction need not originate exclusively from gravitational backreaction; it may also reflect imperfect reconstruction or changes in the preferred quantum extremal surface. We therefore do not identify the proto-area universally with either the geometric QES area or a backreaction-induced change of area. Determining the conditions under which the proto-area admits a more precise geometric interpretation is left for future work. 

Moving beyond AdS/CFT, it is natural to ask whether specific subclasses of
approximate quantum codes can reproduce the entanglement equilibrium conditions
that lead to linearized Einstein equations in near-flat geometries
\cite{Cao_2018}. The perturbations considered in this work do not satisfy these
conditions, but localized bulk entanglement perturbations may produce relations
with the correct qualitative structure.

In summary, our results point toward a broad information-theoretic
perspective on emergent gravity, in which tripartite non-local
magic serves as the resource that enables correlations between bulk
matter and geometric entanglement. 
Because this behavior arises in a wide class of skewed quantum codes
generated by non-Clifford circuits, such phenomena may also be experimentally
realizable on near-term quantum devices. Exploring
these systems may therefore provide an informative route for probing
emergent spacetime dynamics in controllable quantum
platforms.

\section*{Acknowledgment}
We thank  Ning Bao, Aidan Chatwin-Davies, Alexander Jahn and Sreehari A. P for helpful discussions and comments.  C.C. acknowledges funding from the Commonwealth Cyber Initiative. J.P. acknowledges funding provided by the Institute for Quantum Information and Matter, an NSF Physics Frontiers Center (PHY-2317110), and the DOE Office of High Energy Physics (DE-SC0018407).

\appendix
\section{Proof of theorems}
\subsection{Theorem \ref{thm:relative}}\label{app:relative}

\begin{proof}
Let $\sigma^{(R^{(\epsilon)})}_{M}$ denote the recovered algebraic state corresponding to the skewed code, and $\sigma^{(R^{(0)})}_{M}$ the recovered state of the exact code. We define their difference by  
\begin{equation}
\delta\sigma^{(R^{(\epsilon)})}_{A_1A_2}
:=
\sigma^{(R^{(\epsilon)})}_{A_1A_2}-\sigma^{(R^{(0)})}_{A_1A_2},
\qquad
\delta\sigma^{(R^{(\epsilon)})}_{A_1}
:=
\sigma^{(R^{(\epsilon)})}_{A_1}-\sigma^{(R^{(0)})}_{A_1}.
\end{equation}

These variations satisfy
\begin{equation}
\Tr\!\left(\delta\sigma^{(R^{(\epsilon)})}_{A_1A_2}\right)=0,
\qquad
\delta\sigma^{(R^{(\epsilon)})}_{A_1}
=
\Tr_{A_2}\!\left(\delta\sigma^{(R^{(\epsilon)})}_{A_1A_2}\right).
\end{equation}

For notational simplicity, in the remainder of this section we abbreviate
\begin{eqns}
    \sigma_M^{(\epsilon)}:=\sigma^{(R^{(\epsilon)})}_{M}, \qquad \sigma_M^{(0)}:=\sigma^{(R^{(0)})}_{M}
\end{eqns}

First we show that the following combination vanishes to all orders in $\epsilon$:
\begin{equation}
\operatorname{Tr}\left(\delta\sigma^{(\epsilon)}_{A_1A_2}\,\ln\sigma^{(0)}_{A_1A_2}\right)
-
\operatorname{Tr}\left(\delta\sigma^{(\epsilon)}_{A_1}\,\ln\sigma^{(0)}_{A_1}\right)
-
\operatorname{Tr}\left(\delta\chi\,\ln \chi\right)
=0,
\label{eq:vanish-identity}
\end{equation}
where
\begin{equation}
\delta\chi
=
\operatorname{Tr}_{A_1}\left(\delta\sigma^{(\epsilon)}_{A_1A_2}\right).
\end{equation}

Moreover, for the undeformed code the recovered state factorizes as
\begin{equation}
\sigma^{(0)}_{A_1A_2}
=
\sigma^{(0)}_{A_1}\otimes \chi_{A_2} .
\label{eq:factorization-0}
\end{equation}
Using $\ln(X\otimes Y)=\ln X\otimes I + I\otimes \ln Y$, together with
\eqref{eq:factorization-0} and the trace relation defining $\delta\chi$, immediately yields
\eqref{eq:vanish-identity}.

With this identity, we compute
\begin{align}
S_{PA}
&=
-\operatorname{Tr}\left(\sigma^{(\epsilon)}_{A_1A_2}\,\ln\sigma^{(\epsilon)}_{A_1A_2}\right)
+\operatorname{Tr}\left(\sigma^{(\epsilon)}_{A_1}\,\ln\sigma^{(\epsilon)}_{A_1}\right)
\nonumber\\
&\quad
+\operatorname{Tr}\left(\delta\sigma^{(\epsilon)}_{A_1A_2}\,\ln\sigma^{(0)}_{A_1A_2}\right)
-\operatorname{Tr}\left(\delta\sigma^{(\epsilon)}_{A_1}\,\ln\sigma^{(0)}_{A_1}\right)
-\operatorname{Tr}\left(\delta\chi\,\ln\chi\right)
\nonumber\\
&=
-\operatorname{Tr}\left(\sigma^{(0)}_{A_1A_2}\,\ln\sigma^{(0)}_{A_1A_2}\right)
+\operatorname{Tr}\left(\sigma^{(0)}_{A_1}\,\ln\sigma^{(0)}_{A_1}\right)
\nonumber\\
&\quad
-\Big(
D\left(\sigma^{(\epsilon)}_{A_1A_2}\Vert\sigma^{(0)}_{A_1A_2}\right)
-
D\left(\sigma^{(\epsilon)}_{A_1}\Vert\sigma^{(0)}_{A_1}\right)
\Big)
-\operatorname{Tr}\left(\delta\chi\,\ln\chi\right).
\label{eq:PA-step}
\end{align}
By the factorization \eqref{eq:factorization-0}, the first line of \eqref{eq:PA-step} reduces to $S(\chi)$. In the special case where $\chi$ has a flat entanglement spectrum across $A_2\bar A_2$, we have $\delta\chi=0$, and hence
\begin{equation}
S_{PA}
=
S(\chi)
-
\Big(
D\left(\sigma^{(\epsilon)}_{A_1A_2}\Vert\sigma^{(0)}_{A_1A_2}\right)
-
D\left(\sigma^{(\epsilon)}_{A_1}\Vert\sigma^{(0)}_{A_1}\right)
\Big),
\label{eq:S_corr}
\end{equation}
which is the desired relation, with $S_{\mathrm{corr}}$ given by the difference of relative entropies.

\end{proof}

\subsection{Theorem~\ref{thm:monotonic-mixed}}\label{app:mixed}
\begin{proof}
Based on Theorem~\ref{thm:relative}, we expand the boundary and bulk relative entropies
$D\!\left(\sigma^{(\epsilon)}_{A_1A_2}\Vert\sigma^{(0)}_{A_1A_2}\right)$ and
$D\!\left(\sigma^{(\epsilon)}_{A_1}\Vert\sigma^{(0)}_{A_1}\right)$
to leading order in $\epsilon$.
Since relative entropy is non-negative and vanishes at $\epsilon=0$, the leading contribution appears at order $\mathcal{O}(\epsilon^2)$.
For any subsystem $M\subseteq A\bar A$, with complement $\bar M$, we define the corresponding reduced recovered state by
\begin{equation}
\sigma^{(\epsilon)}_{M}
=
\operatorname{Tr}_{\bar M}\!\left[
\sigma^{(\epsilon)}_{A\bar A}
\right].
\label{rhoM-general}
\end{equation}
For a general subsystem $M$, we obtain
\begin{align}
D\!\left(\sigma^{(\epsilon)}_{M}\Vert\sigma^{(0)}_{M}\right)
&= \Tr\!\left(\sigma^{(\epsilon)}_{M}\,\delta\ln\sigma^{(0)}_{M}\right)
\notag\\
&= \Tr\!\left(\sigma^{(0)}_{M}\,\delta\ln\sigma^{(0)}_{M}\right)
   + \Tr\!\left(\delta\sigma^{(\epsilon)}_{M}\,\delta\ln\sigma^{(0)}_{M}\right)
\notag\\
&= \epsilon^2\,
\Tr\!\left(
\sigma^{(0)}_{M}\,
D_{\ln}\!\bigl(\sigma^{(0)}_{M}\bigr)
\left[\delta^{(2)}\sigma_{M}\right]
\right)
\notag\\
&\quad
-\epsilon^2\,
\Tr\!\left(
\sigma^{(0)}_{M}\,
D_{\ln}^2\!\bigl(\sigma^{(0)}_{M}\bigr)
\left[\delta^{(1)}\sigma_{M},\delta^{(1)}\sigma_{M}\right]
\right)
\notag\\
&\quad
+\epsilon^2\,
\Tr\!\left(
\delta^{(1)}\sigma_{M}\,
D_{\ln}\!\bigl(\sigma^{(0)}_{M}\bigr)
\left[\delta^{(1)}\sigma_{M}\right]
\right)
+\mathcal{O}(\epsilon^3)
\notag\\
&=
\frac{\epsilon^2}{2}\,
\Tr\!\left(
\delta^{(1)}\sigma_{M}\,
D_{\ln}\!\bigl(\sigma^{(0)}_{M}\bigr)
\left[\delta^{(1)}\sigma_{M}\right]
\right)
+\mathcal{O}(\epsilon^3),
\label{rel-ent}
\end{align}
where we expand the perturbed state as
\begin{equation}\label{eq:expandstate}
\sigma^{(\epsilon)}_{M}
=
\sigma^{(0)}_{M}
+\delta\sigma^{(\epsilon)}_{M}
=
\sigma^{(0)}_{M}
+\epsilon\,\delta^{(1)}\sigma_{M}
+\frac{\epsilon^2}{2}\,\delta^{(2)}\sigma_{M}
+\mathcal{O}(\epsilon^{3}).
\end{equation}
where $\delta^{(1)}\sigma_{M}$ and $\delta^{(2)}\sigma_{M}$ denote the first- and second-order
corrections in the perturbative expansion in $\epsilon$, and $M$ denotes the
subsystem of interest. We will use this expansion and notation throughout the paper for any subsystem $M$ under consideration. For any positive operator \(A\), we define \(D_{\ln}(A)[X]\) as the Fréchet derivative of \(\ln A\) at \(A\) in the direction \(X\), namely
\begin{equation}
D_{\ln}(A)[X]
:=
\int_{0}^{\infty}
(A+sI)^{-1}X(A+sI)^{-1}\,ds.
\label{eq:log_expansion}
\end{equation}

Using Eq.~\eqref{eq:S_corr} and Eq.~\eqref{rel-ent}, for flat $\chi$ spectrum, we obtain
\begin{equation}
\label{eq:S_corr_gen_equation}
S_{\mathrm{corr}}
=
\frac{\epsilon^2}{2}\Bigg[
\Tr\Big(
\delta^{(1)}\sigma_{A_1A_2}\,
D_{\ln}\big(\sigma^{(0)}_{A_1A_2}\big)
\big[
\delta^{(1)}\sigma_{A_1A_2}
\big]
\Big)
-
\Tr\Big(
\delta^{(1)}\sigma_{A_1}\,
D_{\ln}\big(\sigma^{(0)}_{A_1}\big)
\big[
\delta^{(1)}\sigma_{A_1}
\big]
\Big)
\Bigg].
\end{equation}

In this setting, we have $|\bar A_1|=0$ and $\bar A_2=\bar A$. To allow for a mixed logical input, we purify $\sigma^{(L)}$ by introducing a reference system $r$ and a pure state $|\psi\rangle_{Lr}$ such that
\begin{equation}
\sigma^{(L)}=\Tr_{r}\!\left(|\psi\rangle\langle\psi|_{Lr}\right),
\qquad
|\psi\rangle_{Lr}
=
\sum_{i=1}^{d}\sqrt{\lambda_i}\,|i\rangle_{L}\,|i\rangle_{r},
\label{eq:purif_logical_mixed}
\end{equation}
where $\{\lambda_i\}$ is a probability distribution and $d=\dim(\mathcal{H}_L)=\dim(\mathcal{H}_r)$. After encoding and recovery, the joint state on $A\bar A$ is given by:
\begin{equation}
\sigma^{(\epsilon)}_{A\bar A}
=
\Tr_{r}\!\left[
e^{i\epsilon W_R}
\Big(|\psi\rangle\langle\psi|_{rA_1}\otimes |\chi\rangle\langle\chi|_{A_2\bar A}\Big)
e^{-i\epsilon W_R}
\right],
\label{eq:sigma_eps_AAbar_mixed}
\end{equation}
together with its reduced states on the boundary $A_1A_2$ and on the bulk factor $A_1$,
\begin{equation}
\begin{aligned}
\sigma^{(\epsilon)}_{A_1A_2}
&=
\Tr_{r\bar A}\!\left[
e^{i\epsilon W_R}
\Big(|\psi\rangle\langle\psi|_{rA_1}\otimes |\chi\rangle\langle\chi|_{A_2\bar A}\Big)
e^{-i\epsilon W_R}
\right],\\
\sigma^{(\epsilon)}_{A_1}
&=
\Tr_{rA_2\bar A}\!\left[
e^{i\epsilon W_R}
\Big(|\psi\rangle\langle\psi|_{rA_1}\otimes |\chi\rangle\langle\chi|_{A_2\bar A}\Big)
e^{-i\epsilon W_R}
\right].
\end{aligned}
\label{eq:reduced_states_eps_mixed}
\end{equation}

Here, $W_R$ acts trivially on the reference subsystem $r$. In the unperturbed limit $\epsilon\to 0$ the unitary drops out and the states factorize as
\begin{align}
\sigma^{(0)}_{A\bar A}
&=
\Tr_{r}\!\left(|\psi\rangle\langle\psi|_{rA_1}\right)\otimes |\chi\rangle\langle\chi|_{A_2\bar A},
\label{eq:unpert_AAbar_mixed}\\
\sigma^{(0)}_{A_1A_2}
&=
\Tr_{r}\!\left(|\psi\rangle\langle\psi|_{rA_1}\right)\otimes \chi_{A_2},
\label{eq:unpert_A_mixed}\\
\sigma^{(0)}_{A_1}
&=
\Tr_{r}\!\left(|\psi\rangle\langle\psi|_{rA_1}\right),
\label{eq:unpert_A1_mixed}
\end{align}
where we use the shorthand
\begin{equation}
\chi_{A_2\bar A}:=|\chi\rangle\langle\chi|_{A_2\bar A},
\qquad
\chi_{A_2}:=\Tr_{\bar A}\!\left(\chi_{A_2\bar A}\right).
\label{eq:chi_defs_mixed}
\end{equation}

Finally, in the Schmidt basis of the purification \eqref{eq:purif_logical_mixed} the reduced state on $A_1$ is diagonal with eigenvalues $\{\lambda_i\}$, and hence admits the spectral decomposition
\begin{equation}
\sigma^{(0)}_{A_1}
=
\sum_{i=1}^{d}\lambda_i\,|i\rangle\langle i|_{A_1},
\qquad d=\dim(\mathcal{H}_{A_1}).
\label{eq:sigmaA1_spectral_mixed}
\end{equation}

We now turn to the evaluation of the Haar-averaged correction, $\langle S_{\mathrm{corr}}\rangle$, obtained by averaging over local unitaries acting on the logical subsystem $\sigma_a^{(L)}$, or equivalently, on the recovered bulk subsystem $A_1$. Concretely, for any operator $X$ supported on $A_1$ (or on $A_1A_2$ with trivial action on $A_2$), we define its Haar-rotated version by conjugation with a unitary $U$ on $A_1$,
\begin{equation}
\sigma^{(0)}_{A_1}(U)
:=
U\,\sigma^{(0)}_{A_1}\,U^\dagger,
\qquad
\sigma^{(0)}_{A_1A_2}(U)
:=
(U\otimes I_{A_2})\,\sigma^{(0)}_{A_1A_2}\,(U^\dagger\otimes I_{A_2}),
\label{eq:haar_rot_defs_mixed}
\end{equation}
and similarly for the perturbative corrections $\delta^{(1)}\sigma_{A_1}(U)$ and
$\delta^{(1)}\sigma_{A_1A_2}(U)$. We then average over $U$ with respect to the
Haar measure $dU$ on $\mathrm{U}(d)$, where $d=\dim(\mathcal{H}_{A_1})$.

Using the second-order entropy expansion, the Haar-averaged correction
takes the form
\begin{align}
\big\langle S_{\mathrm{corr}}\big\rangle
=\frac{\epsilon^2}{2}\int dU \Bigg[
&\Tr\!\left(
\delta^{(1)}\sigma_{A_1A_2}(U)\,
D_{\ln}\!\big(\sigma^{(0)}_{A_1A_2}(U)\big)
\Big[\delta^{(1)}\sigma_{A_1A_2}(U)\Big]
\right)\nonumber\\
&\hspace{2.8em}-
\Tr\!\left(
\delta^{(1)}\sigma_{A_1}(U)\,
D_{\ln}\!\big(\sigma^{(0)}_{A_1}(U)\big)
\Big[\delta^{(1)}\sigma_{A_1}(U)\Big]
\right)
\Bigg].
\label{eq:Scorr_Haar_mixed}
\end{align}
Next, we expand the recovered reduced states perturbatively in $\epsilon$. In
the unperturbed limit the joint state on $A\bar A$ factorizes as
\begin{equation}
\sigma^{(0)}_{A\bar A}
=
\sigma^{(0)}_{A_1}\otimes \chi_{A_2\bar A},
\label{eq:sigma0_factor_mixed}
\end{equation}
The first-order corrections to the reduced states on $A_1A_2$
and $A_1$ are obtained by expanding the perturbation
$e^{i\epsilon W_R}$ to linear order and tracing out the appropriate
subsystems. This yields
\begin{align}
\delta^{(1)}\sigma_{A_1A_2}
&=
i\,\Tr_{\bar A}\!\Big(
\big[W_R,\ \sigma^{(0)}_{A_1}\otimes \chi_{A_2\bar A}\big]
\Big),
\label{eq:delta1_bdry_mixed}\\
\delta^{(1)}\sigma_{A_1}
&=
i\,\Tr_{A_2\bar A}\!\Big(
\big[W_R,\ \sigma^{(0)}_{A_1}\otimes \chi_{A_2\bar A}\big]
\Big).
\label{eq:delta1_bulk_mixed}
\end{align}

A technical point is that the Fr\'echet derivative $D_{\ln}(\sigma)$ is
well-defined only when $\sigma_{A_1}$ is strictly positive. If $\sigma_{A_1}$ has zero
eigenvalues, then $\ln\sigma$ (and hence $D_{\ln}(\sigma)$) is singular on the
corresponding subspace, which can lead to divergences in
\eqref{eq:Scorr_Haar_mixed}. To avoid this issue, we assume throughout that the
unperturbed reduced states are full rank,
\begin{equation}
\sigma^{(0)}_{A_1} > 0,
\qquad
\sigma^{(0)}_{A_1A_2} > 0.
\label{eq:fullrank_assumption_mixed}
\end{equation}
Equivalently, one may view this as working with an implicit regularization
$\sigma^{(0)}\mapsto (1-\eta)\sigma^{(0)}+\eta\,I/\dim$ and taking
$\eta\downarrow 0$ at the end. We keep the notation uncluttered and proceed
under the full-rank assumption \eqref{eq:fullrank_assumption_mixed}.

Substituting \eqref{eq:delta1_bdry_mixed}--\eqref{eq:delta1_bulk_mixed} into \eqref{eq:Scorr_Haar_mixed}, expanding the commutators and using the diagrammatic convention described in Appendix~\eqref{WGcalculus}. we obtain

\begin{align}
\big\langle S_{\mathrm{corr}}\big\rangle= \frac{\epsilon^2}{2d_{\chi}^2}\int ds \int dU 
 &\left(-\vcenter{\hbox{\includegraphics[height=7em]{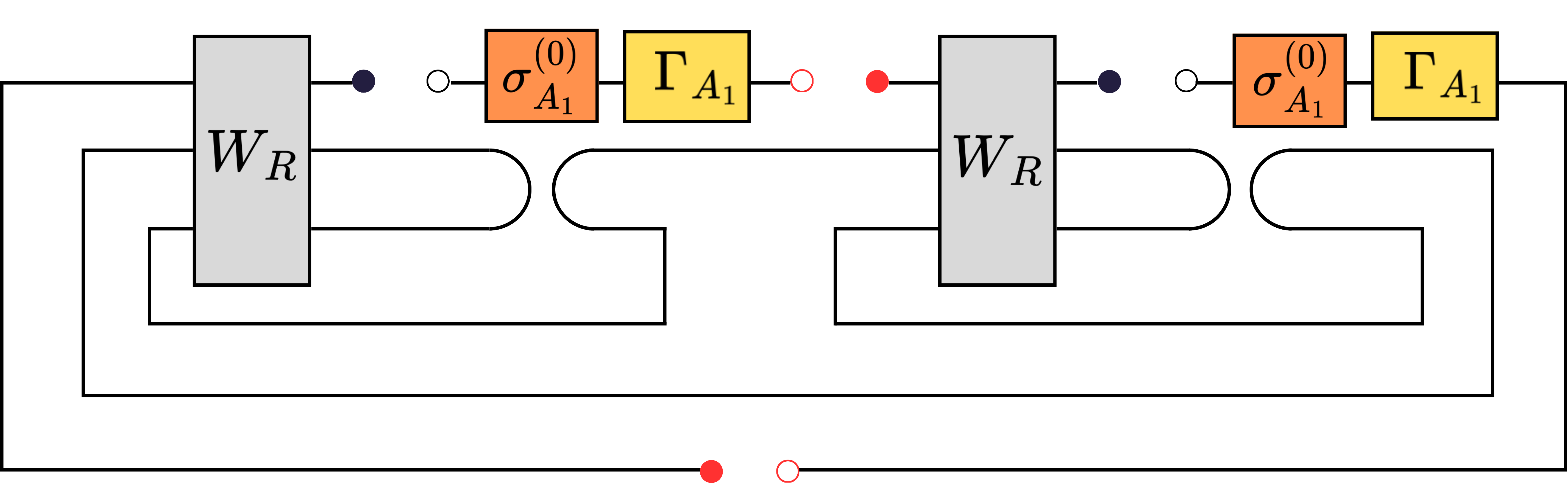}}}\right.\notag\\ + & \quad 2 \,  \vcenter{\hbox{\includegraphics[height=7em]{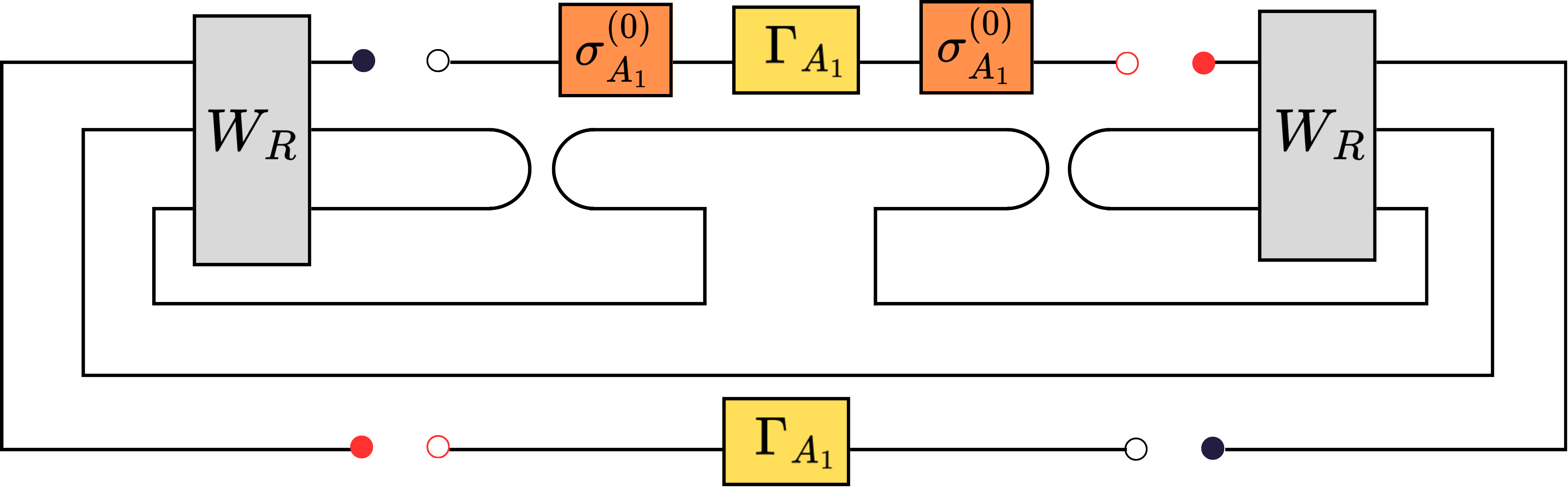}}} \notag\\
 &\quad - \vcenter{\hbox{\includegraphics[height=7em]{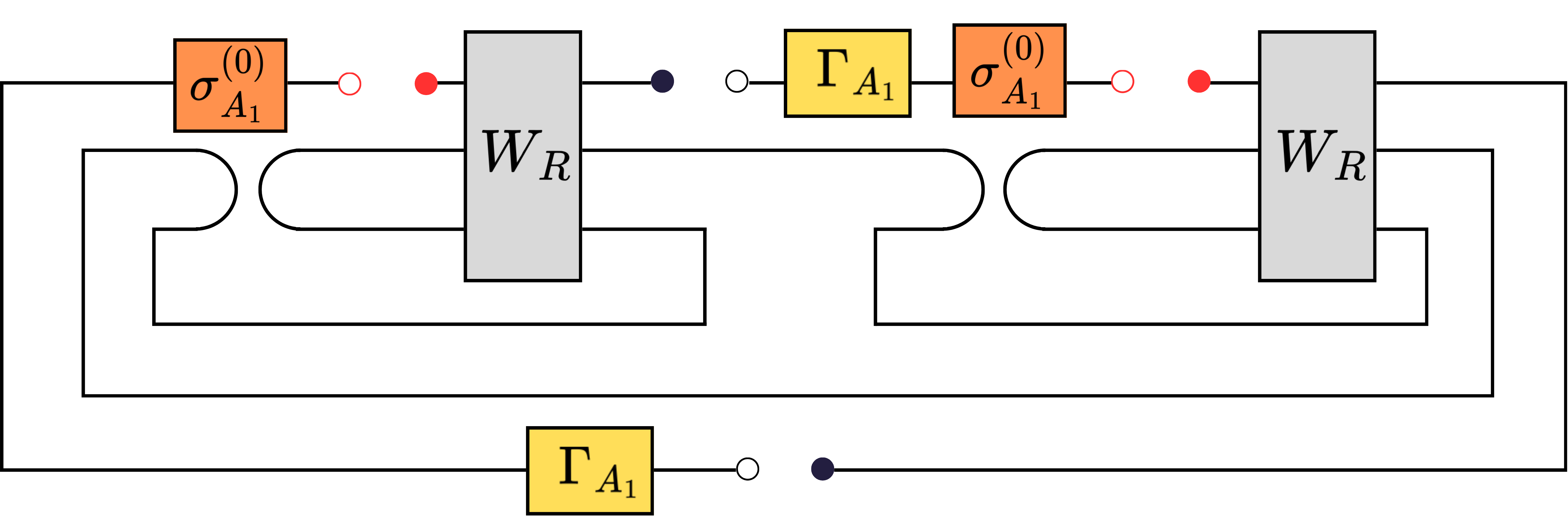}}} \notag\\
 +\frac{1}{d_\chi}  &\quad \vcenter{\hbox{\includegraphics[height=7em]{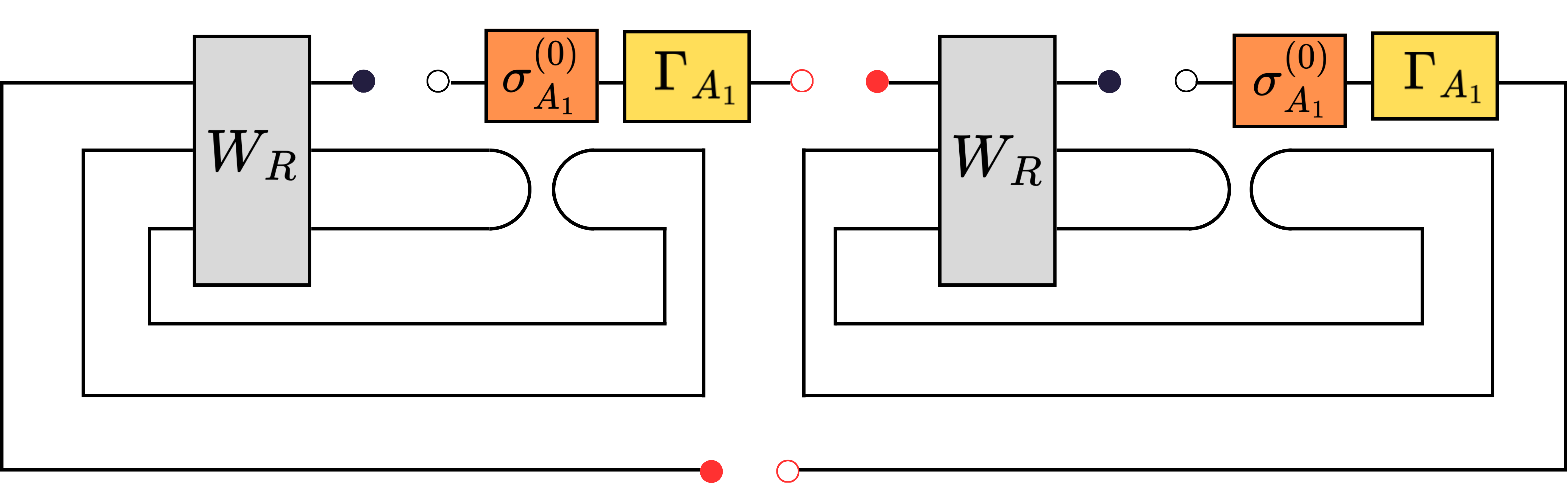}}}
 \notag \\ -\frac{2}{d_\chi}\,
  &\quad \vcenter{\hbox{\includegraphics[height=7em]{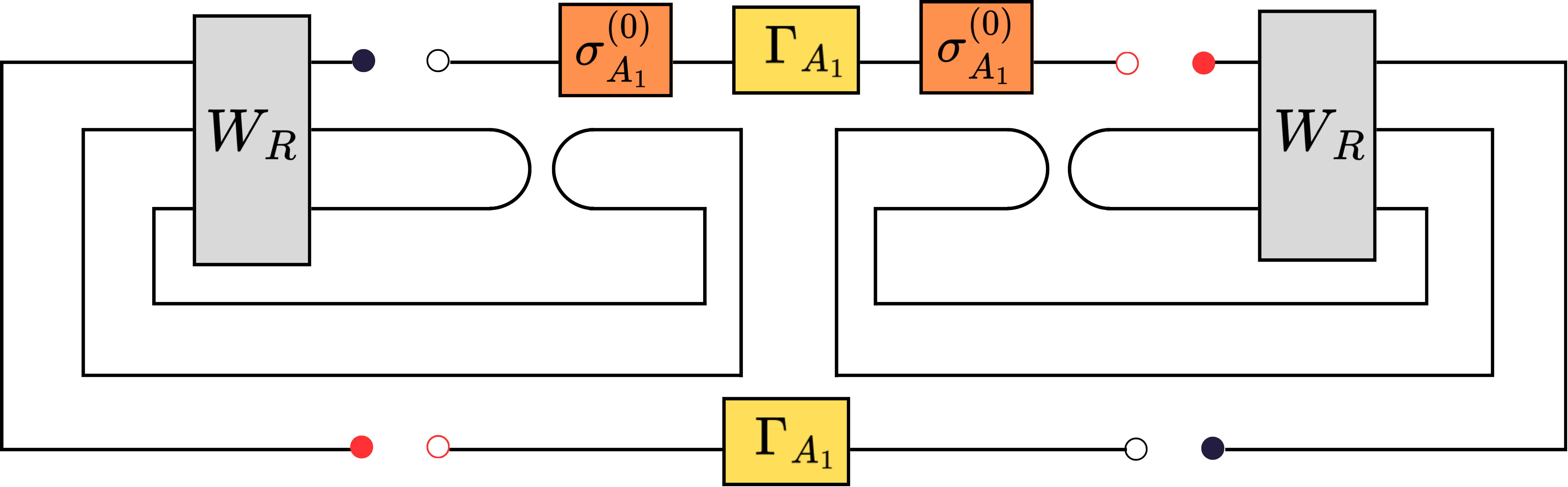}}}\notag\\
 +  \frac{1}{d_\chi} & \quad \left.\vcenter{\hbox{\includegraphics[height=7em]{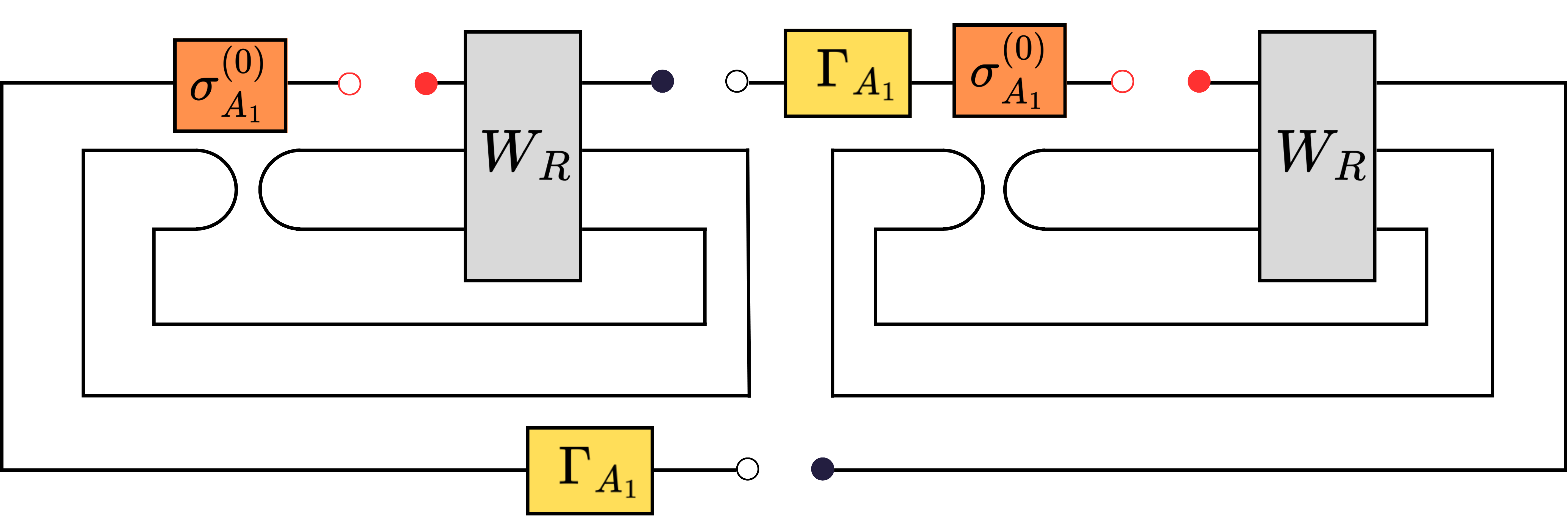}}}
\quad \right)
\end{align}

In the above diagrams, we use the following diagrammatic representation for the state $\chi_{A_2\bar A_2}$ (in this particular case $\bar A_2=\bar A$):
\begin{eqns}
\chi_{A_2\bar A_2} 
= \frac{1}{d_\chi} \ \ \vcenter{\hbox{\includegraphics[height=3em]{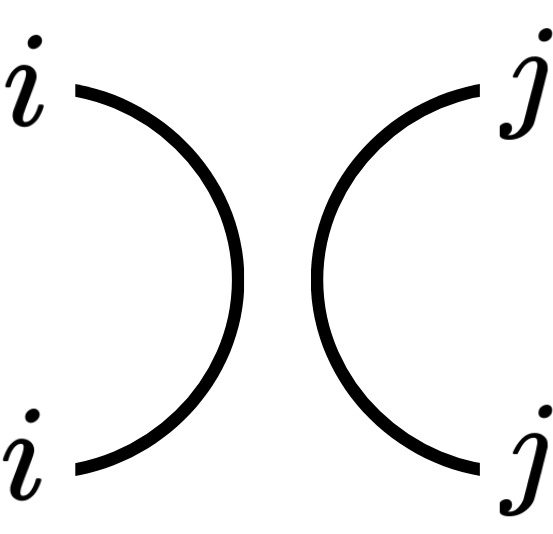}}}
\end{eqns}
where $d_{\chi}=\dim(\mathcal{H}_{A_2})=\dim(\mathcal{H}_{\bar A_2})$. This diagram will be used as a shorthand for $\chi_{A_2\bar A_2}$ in all subsequent figures and contractions.

For the maximally mixed resource on $A_2$ (so that $\chi_{A_2}=I_{A_2}/d_\chi$
with $d_\chi:=\dim(\mathcal{H}_{A_2})$), the corresponding resolvent appearing
in $D_{\ln}$ simplifies as
\begin{align}
\left(\sigma^{(0)}_{A_1}\otimes \chi_{A_2} + s\,I_{A_1A_2}\right)^{-1}
&=
\left(\sigma^{(0)}_{A_1}\otimes \frac{I_{A_2}}{d_\chi} + s\,I_{A_1}\otimes I_{A_2}\right)^{-1}\nonumber\\
&=
\left(\frac{\sigma^{(0)}_{A_1}}{d_\chi}+s\,I_{A_1}\right)^{-1}\otimes I_{A_2}
\;:=\;
\Gamma_{A_1}(s)\otimes I_{A_2}.
\label{eq:resolvent_factorization_mixed}
\end{align}
Here we have defined 
\begin{equation}
\Gamma_{A_1}(s)
:=
\left(\frac{\sigma^{(0)}_{A_1}}{d_\chi}+s\,I_{A_1}\right)^{-1},
\label{eq:Gamma_def_mixed}
\end{equation}
which is precisely the object denoted by $\Gamma(s)$ in our trace-diagram notation. For the bulk contributions, the corresponding expressions involve an additional overall factor of $1/d_\chi$ relative to the boundary case. Equivalently, at the level of the resolvent factors, this amounts to using a rescaled $\Gamma_{A_1}(s)$. The underlying reason is simple: in the bulk term we encounter the Fr\'echet derivative acting on $\sigma^{(0)}_{A_1}$ rather than on $\sigma^{(0)}_{A_1A_2}=\sigma^{(0)}_{A_1}\otimes \chi_{A_2}$. Using the integral representation \eqref{eq:log_expansion}, the bulk contribution contains factors of the form
\begin{equation}
D_{\ln}\!\big(\sigma^{(0)}_{A_1}\big)[Y]
=
\int_{0}^{\infty}
\left(\sigma^{(0)}_{A_1}+s I_{A_1}\right)^{-1}
\,Y\,
\left(\sigma^{(0)}_{A_1}+s I_{A_1}\right)^{-1}\,ds,
\label{eq:Dlog_bulk_integral_mixed}
\end{equation}
for an operator $Y$ supported on $A_1$ (in our application,
$Y=\delta^{(1)}\sigma_{A_1}(U)$).

It is convenient to express \eqref{eq:Dlog_bulk_integral_mixed} in terms of the
same resolvent variable that appears in the boundary diagrams. Introducing
\begin{equation}
\Gamma_{A_1}(s)
:=
\left(\frac{\sigma^{(0)}_{A_1}}{d_\chi}+s I_{A_1}\right)^{-1},
\qquad
d_\chi:=\dim(\mathcal{H}_{A_2})=\dim \chi_{A_2},
\label{eq:Gamma_def_repeat_mixed}
\end{equation}
we use the identity
\begin{equation}
\left(\sigma^{(0)}_{A_1}+s I_{A_1}\right)^{-1}
=
\frac{1}{d_\chi}\,
\Gamma_{A_1}\!\left(\frac{s}{d_\chi}\right).
\label{eq:rescaling_resolvent_mixed}
\end{equation}
Substituting \eqref{eq:rescaling_resolvent_mixed} into
\eqref{eq:Dlog_bulk_integral_mixed} and changing variables $s\mapsto s/d_\chi$
gives
\begin{align}
D_{\ln}\!\big(\sigma^{(0)}_{A_1}\big)[Y]
&=
\int_{0}^{\infty}
\left(\sigma^{(0)}_{A_1}+s I_{A_1}\right)^{-1}
\,Y\,
\left(\sigma^{(0)}_{A_1}+s I_{A_1}\right)^{-1}\,ds \nonumber\\
&=
\frac{1}{d_\chi}
\int_{0}^{\infty}
\Gamma_{A_1}(s)\,Y\,\Gamma_{A_1}(s)\,ds.
\label{eq:Dlog_bulk_in_terms_of_Gamma_mixed}
\end{align}
Equation \eqref{eq:Dlog_bulk_in_terms_of_Gamma_mixed} explains the extra factor
of $1/d_\chi$ in the bulk diagrams. For notational convenience, in the remainder of this section we write $\Gamma(s)$ in place of $\Gamma_{A_1}(s)$.

Using the Weingarten calculus and the diagrammatic Haar-averaging rules collected in Appendix~\eqref{WGcalculus-rules}, we can now perform the $U$-integrals and obtain the expressions below.

\begin{align}
\langle S_{corr}\rangle =\frac{\epsilon^2 }{2d_\chi^2}\int_0^{\infty} ds\, \Bigg[
&\frac{1}{d^2 - 1}\, \mathscr{R}_1(s)\,\Big(
  \tfrac{1}{d}\, W_1
  + \tfrac{1}{d}\, W_3
  - W_2
  - W_4
  - \tfrac{2}{d\, d_\chi}\, W_8
  + \tfrac{2}{d_\chi}\, W_7
\Big)
\nonumber \\ &+
\frac{2}{d^2 - 1}\, \mathscr{R}_3(s)\,\Big(
  W_5 - \tfrac{1}{d}\, W_6 -\tfrac{1}{d_\chi}\, W_7+\tfrac{1}{dd_\chi}\, W_8
\Big)\nonumber \\ &
+
\frac{2}{d^2 - 1}\, \mathscr{R}_2(s)\,\Big(
  \tfrac{1}{d}\, W_2
  + \tfrac{1}{d}\, W_4
  - \tfrac{2}{d}\, W_4
  - W_1
  - W_3
  + 2 W_6
\Big)
\Bigg].
\end{align}
where $W_i$ and $\mathscr{R}_i(s)$ are given in Table~\eqref{W-R-diagrams} in Appendix. A neat feature is that, once we average over the Haar random unitary, the contributions essentially separate: the $\mathscr{R}$ terms depend only on the bulk state $\sigma^{0}_{A_1}$, while the $W$ terms depend only on the skewing matrix $W_R$.

Now we explicitly perform the integral on s, and obtain the following  functions of the spectrum:
\begin{align}
   \frac{1}{d_\chi}\int_0^{\infty} ds\, \mathscr{R}_1(s, \lambda_i) =  \frac{1}{d_\chi}\int_0^{\infty} ds \Tr(\sigma^{0}_{A_1}\Gamma(s))^2=&\sum_{ij}\int_0^{\infty} ds \frac{\lambda_i\lambda_j}{(\lambda_i+s)(\lambda_j+s)}\nonumber\\
    =&\sum_{i\neq j}\frac{\lambda_i\lambda_j}{\lambda_i-\lambda_j}\ln\frac{\lambda_i}{\lambda_j}+\sum_{i}\lambda_i
    \end{align}
    \begin{align}
    \frac{1}{d_\chi}\int_0^{\infty} ds\, \mathscr{R}_2(s, \lambda_i)=\int_0^{\infty} ds \Tr((\sigma^{0}_{A_1}\Gamma(s))^2)=1
\end{align}
    \begin{align}
    \frac{1}{d_\chi}\int_0^{\infty} ds\, \mathscr{R}_3(s, \lambda_i)=\frac{1}{d_\chi}\int_0^{\infty} ds \Tr((\sigma^{0}_{A_1})^2 \, \Gamma(s))\Tr(\Gamma(s))=&\sum_{ij}\int_0^{\infty} ds \frac{\lambda_i^2}{(\lambda_i+s)(\lambda_j+s)}\nonumber\\
    =&\sum_{i\neq j}\frac1 2\frac{\lambda_i^2+\lambda_j^2}{\lambda_i-\lambda_j}\ln\frac{\lambda_i}{\lambda_j}+\sum_{i}\lambda_i.
      \end{align}

To simplify our notation, we use $f(\lambda)$ to denote functions that depend on the entanglement spectrum $\{\lambda_i\}$ and write $S(\lambda) = -\sum_i \lambda_i\log \lambda_i$ as the Shannon entropy. 
From these two expressions, define two functions $f_1(\lambda)$ and $f_2(\lambda)$: 
\begin{equation}\label{eq:f1f2}
    \begin{split}
        f_1(\lambda):= & \frac 1{dd_\chi}\int_0^{\infty} ds \, \left(\mathscr{R}_3(s, \lambda_i)-\mathscr{R}_1(s, \lambda_i)\right)\\
        =&\frac 1 {2d}\sum_{i\neq j}(\lambda_i-\lambda_j)\ln\frac{\lambda_i}{\lambda_j}\\
        =&-\sum_i \ln(\lambda_i)/d- S(\lambda)\\
        f_2(\lambda):=& \frac 1 {dd_\chi}\int_0^{\infty} ds \, \left(\mathscr{R}_3(s, \lambda_i)+\mathscr{R}_1(s, \lambda_i)-d \,
        \ \mathscr{R}_2(s, \lambda_i)\right)\\
        =&\left(\frac {2}{d}+ \sum_{i\neq j}\frac {1}{2d}\frac{(\lambda_i+\lambda_j)^2}{\lambda_i-\lambda_j}\ln{\frac{\lambda_i}{\lambda_j}}\right)-2.
    \end{split}
\end{equation}
We specialize to the case where bulk qubits are a register of 
$n = \log_2 d$ qubits, each entangled with a reference qubit in $r$ in a
parameterised Bell state. \begin{equation}
|\Phi(\theta)\rangle
=
\cos\theta\,|00\rangle
+
\sin\theta\,|11\rangle,
\qquad.
\end{equation}The joint bulk-reference pure state  is 
\begin{equation}
|\Psi(\boldsymbol{\theta})\rangle_{A_1 r}
=
\bigotimes_{k=1}^{n}
\left(
\cos\theta_k\,|0\rangle_{A_{1,k}}|0\rangle_{r_k}
+
\sin\theta_k\,|1\rangle_{A_{1,k}}|1\rangle_{r_k}
\right),
\label{eq:Psi_theta_vec}
\end{equation}
with angles $\boldsymbol{\theta} = (\theta_1,\dots,\theta_n)$. As discussed before, the resolvent representation of $D_{\ln}$ requires the
relevant reduced states to be full rank. For the ansatz
\eqref{eq:Psi_theta_vec}, the reduced state on each bulk qubit has eigenvalues
$\cos^2\theta_k$ and $\sin^2\theta_k$, so full rank holds provided both are
nonzero. We therefore restrict to
\[
0<\theta_k<\frac{\pi}{2}
\qquad \text{for all } k=1,\dots,n,
\label{eq:theta_fullrank_condition}
\]
which ensures $\sigma^{(0)}_{A_1}>0$ and avoids the divergences associated with
zero eigenvalues. Tracing out the reference system $r$ gives 
\begin{equation}\label{eq:bulkbell}
\sigma^{(L)}\equiv\sigma^{0}_{A_1}(\boldsymbol{\theta})
=
\operatorname{Tr}_r\!\left(|\Psi(\boldsymbol{\theta})\rangle\langle\Psi(\boldsymbol{\theta})|\right)
=
\bigotimes_{k=1}^{n}
\begin{pmatrix}
p_k & 0 \\
0   & 1-p_k
\end{pmatrix},
\qquad
p_k := \cos^2\theta_k .
\end{equation}
Labeling the computational basis of $A_1$ by bitstrings 
$b=(b_1,\dots,b_n)\in\{0,1\}^n$, the eigenvalues of $\sigma^{(0)}_{A_1} (\equiv \sigma^{(L)})$ factorises as
\begin{equation}
\lambda_b
= \prod_{k=1}^{n}
\bigl[p_k\bigr]^{1-b_k}
\bigl[1-p_k\bigr]^{b_k},
\qquad d=2^n .
\label{eig-Lambda}
\end{equation}
Now, let us define a pair of  Hermitian matrices $J$ and $D$ as 
\begin{eqns}
    J:=d_\chi\Tr_{\bar A_2}\left(\{W_R,\chi_{A_2\bar A}\}\right)\\
    D:=id_\chi\Tr_{\bar A_2}\left([W_R,\chi_{A_2\bar A}]\right)
    \label{eq:JD-mixed}
\end{eqns}
which are diagrammatically represented in Fig. \ref{ref-PQ}, 
\begin{eqns}
J_{bj;ai}&=\vcenter{\hbox{\includegraphics[height=5em]{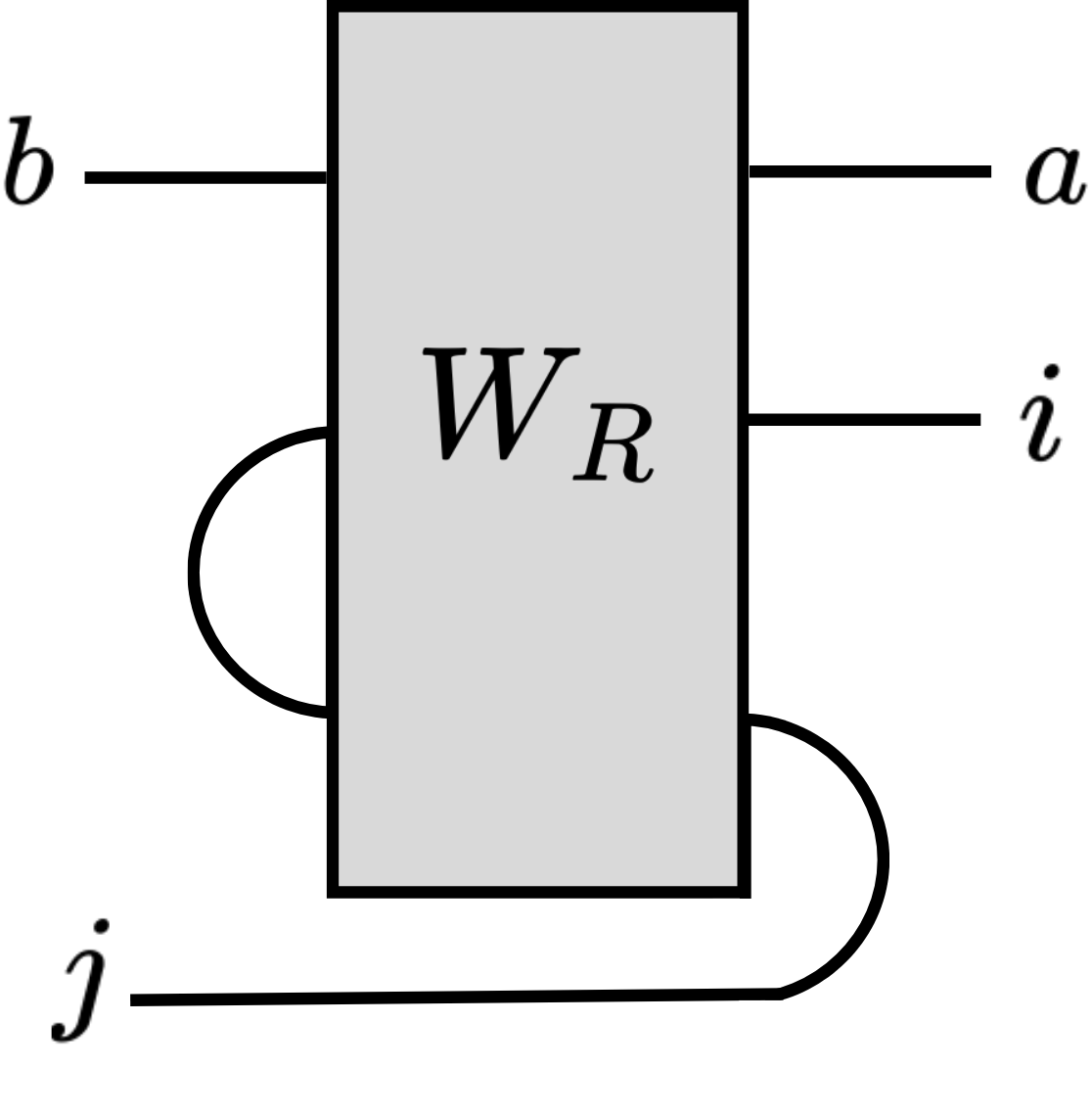}}}\,+\,\vcenter{\hbox{\includegraphics[height=5em]{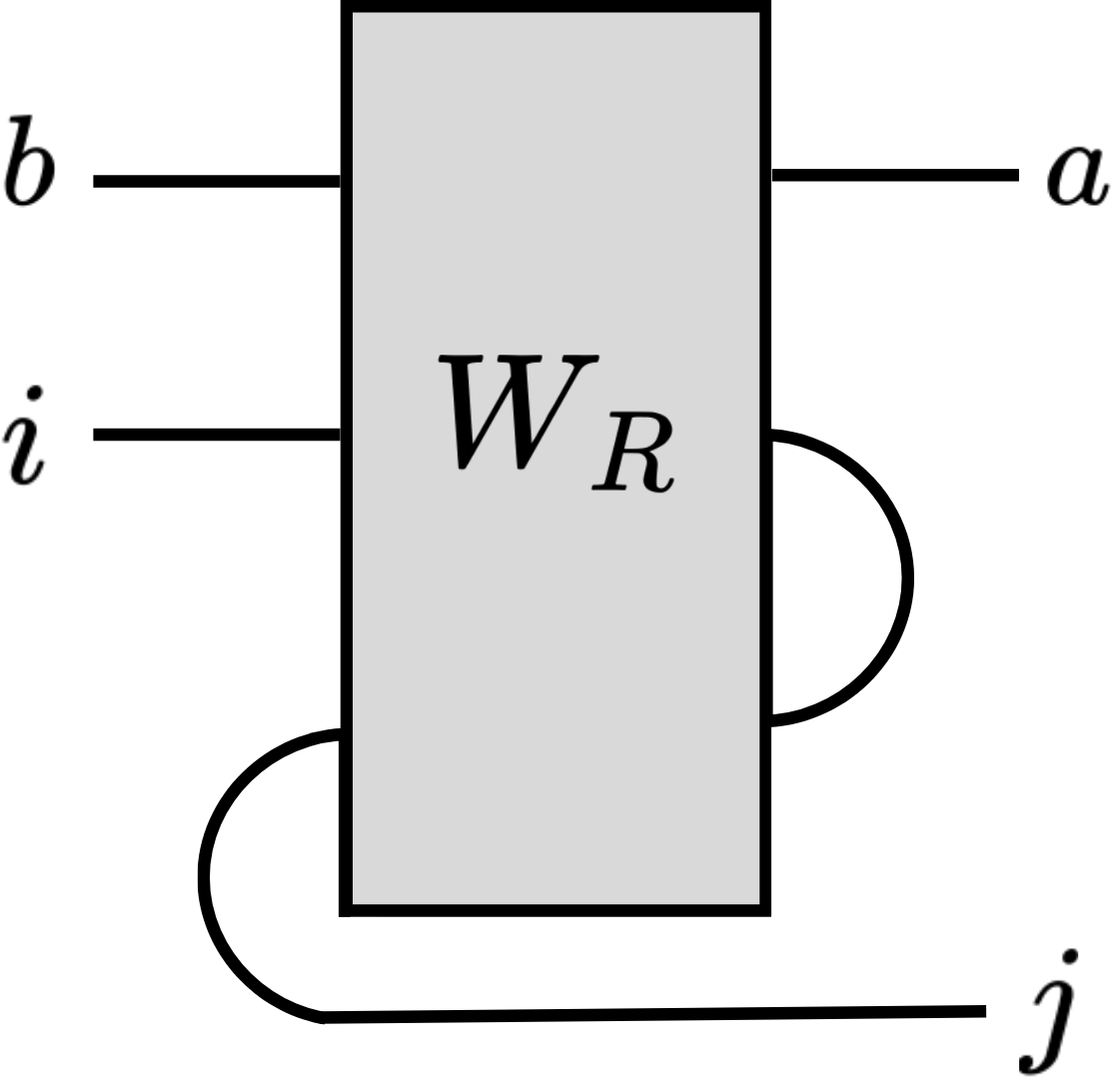}}}\equiv \mathcal{W}^{(1)}_{bj;ai}+\mathcal{W}^{(2)}_{bj;ai}\\
iD_{bj;ai}&=\vcenter{\hbox{\includegraphics[height=5em]{theory_draft_images/MatrixP.png}}}\, - \,\vcenter{\hbox{\includegraphics[height=5em]{theory_draft_images/MatrixQ.png}}}\equiv \mathcal{W}^{(1)}_{bj;ai}-\mathcal{W}^{(2)}_{bj;ai}
\label{ref-PQ}
\end{eqns}
Here, \(\mathcal{W}^{(1)}_{bj;ai}\) and \(\mathcal{W}^{(2)}_{bj;ai}\) denote the respective contributions from the first and second trace diagrams. Then $\langle S_{corr}\rangle$ can be reorganized into,
\begin{equation}
    \begin{split}
        \langle S_{corr}\rangle=\frac{1}{2}\epsilon^2\Big(c_1f_1(\lambda)+c_2f_2(\lambda)+c_3\Big)
    \end{split}
\end{equation}
where we have defined:

\begin{align}
\label{eq:c1c2c3}
c_1(J)
=&\frac{d}{2 d_\chi(d^{2}-1)}
\Big(
\Tr_{A_1A_2}(J^{2})
-\frac{1}{d}\,\Tr_{A_2}\!\left(\Tr_{A_1}(J)^{2}\right)
\\& \qquad -\frac{1}{d_\chi}\,\Tr_{A_1}\!\left(\Tr_{A_2}(J)^{2}\right)
+\frac{1}{d d_\chi}\,\Tr_{A_1A_2}(J)^{2}
\Big)
\\
c_2(D)
=&\frac{d}{2 d_\chi(d^{2}-1)}
\Big(
\Tr_{A_1A_2}((D)^{2})
-\frac{1}{d}\,\Tr_{A_2}\!\left(\Tr_{A_1}(D)^{2}\right)
\Big)
\\
c_3(D)
=&\frac{1}{d_\chi(d^{2}-1)}
\Big(
-\frac{1}{d}\,\Tr_{A_1A_2}((D)^{2})
+\Tr_{A_2}\!\left(\Tr_{A_1}(D)^{2}\right)
\Big)
+2c_2
\end{align}
 
It suffices to show that the three coefficients $c_1, c_2$, and $c_3$ are strictly positive and that the functions $f_1$ and $f_2$ are positive and monotonically decreasing in the bulk-entanglement parameter. Once these ingredients are established, the correction term satisfies $\langle S_{\text {corr }} \rangle \geq 0$ and decreases as the bulk entanglement increases. The monotonocity and positivity of $f_1$ and $f_2$ has been shown in Appendix \ref{app:f1f2f3}. To show that $c_1$, $c_2$, and $c_3$ are non–negative, we decompose $J$ and $D$ into orthonormal Pauli basis on $A_1$ and $A_2$,
\[
J=\sum_{a,b} p_{ab}\,P_a\otimes P_b,
\qquad
D=\sum_{a,b} q_{ab}\,P_a\otimes P_b,
\]
with $\Tr(P_a P_{a'}) = d\delta_{aa'}$ and $\Tr(P_b P_{b'}) = d_\chi\delta_{bb'}$, and all non–identity elements traceless. From this one directly obtains the trace combinations that appear in $c_1,c_2,c_3$:
\begin{align}\label{eq:tracePQ}
\Tr_{A_1A_2}(J^2) &= d d_\chi \sum_{a,b} p_{ab}^2, \qquad
\Tr_{A_2}\!\bigl(\Tr_{A_1}(J)^2\bigr) = d^2 d_\chi \sum_b p_{0b}^2,\notag \\
\Tr_{A_1}\!\bigl(\Tr_{A_2}(J)^2\bigr) &= d d_\chi^2 \sum_a p_{a0}^2,\qquad 
\Tr_{A_1A_2}(J) = d d_\chi p_{00}
\notag\\
\Tr_{A_1A_2}((D)^2) &= d d_\chi \sum_{a,b} q_{ab}^2, \qquad
\Tr_{A_2}\!\bigl(\Tr_{A_1}(D)^2\bigr) = d^2 d_\chi \sum_b q_{0b}^2.
\end{align}
Substituting these expressions into the definitions of $c_1$, $c_2$, and $c_3$ and simplifying yields
\begin{equation}
c_1=\frac{d^2}{2(d^2-1)}
       \sum_{a\neq 0,\,b\neq 0} p_{ab}^2,
\qquad
c_2=\frac{d^2}{2(d^2-1)}
       \sum_{a\neq 0,\,b} q_{ab}^2,
\qquad
c_3=\sum_{a,b} q _{ab}^2,   
\end{equation}
so each coefficient is a positive constant times a sum of squares of Pauli coefficients, and thus
\begin{equation}
c_1\ge 0,\qquad c_2\ge 0,\qquad c_3\ge 0.
\end{equation}
In Appendix~\ref{app:f1f2f3}, we further show that the functions $f_1(\lambda)$ and $f_2(\lambda)$ are both monotonic decreasing function as we increase the entanglement of each bulk bell pairs defined in Eq.~\eqref{eq:bulkbell}. 
\end{proof}

\subsection{Lemma \ref{lemma:optimization:mixed}}\label{app:optimization:mixed}
\begin{proof}

The recovery is optimized by choosing a channel $R^\ast$ that maximizes the coherent information of the logical-to-output map. For a fixed perturbation strength $\epsilon$, the encode--noise--recover procedure induces an effective quantum channel (cf.~Eq.~\eqref{eqn:recov_state})
\begin{equation}
\mathcal{N}_{R^{(\epsilon)}}:
\mathcal{L}(\mathcal{H}_L)\longrightarrow \mathcal{L}(\mathcal{H}_{A_1}),
\qquad
\mathcal{N}_{R^{(\epsilon)}}\!\left(\sigma^{(L)}\right)
=
\sigma^{(\epsilon)}_{A_1},
\label{eq:effective_channel_mixed}
\end{equation}
where $\sigma^{(\epsilon)}_{A_1}$ is the recovered output state on the subsystem
$A_1$. In our setup, Eq.~\eqref{eqn:recov_state} is specialized to $|\bar A_1|=0$ and $\bar A_2=\bar A$. To quantify the performance of $R^{(\epsilon)}$, we evaluate the coherent information of $\mathcal{N}_{R^{(\epsilon)}}$, defined in Eq.~\eqref{ref-coherentinfo} as:
\begin{equation}\label{eq:coherent}
I_c(\mathcal{N}_{R^{(\epsilon)}})
=
S\!\left(\Tr_r\!\left[(\mathcal{N}_{R^{(\epsilon)}}\otimes I_r)\bigl(|\Phi_d\rangle\langle\Phi_d|\bigr)\right]\right)
-
S\!\left[(\mathcal{N}_{R^{(\epsilon)}}\otimes I_r)\bigl(|\Phi_d\rangle\langle\Phi_d|\bigr)\right].
\end{equation}
\begin{figure}[H]
    \centering
    \includegraphics[width=0.6\linewidth]{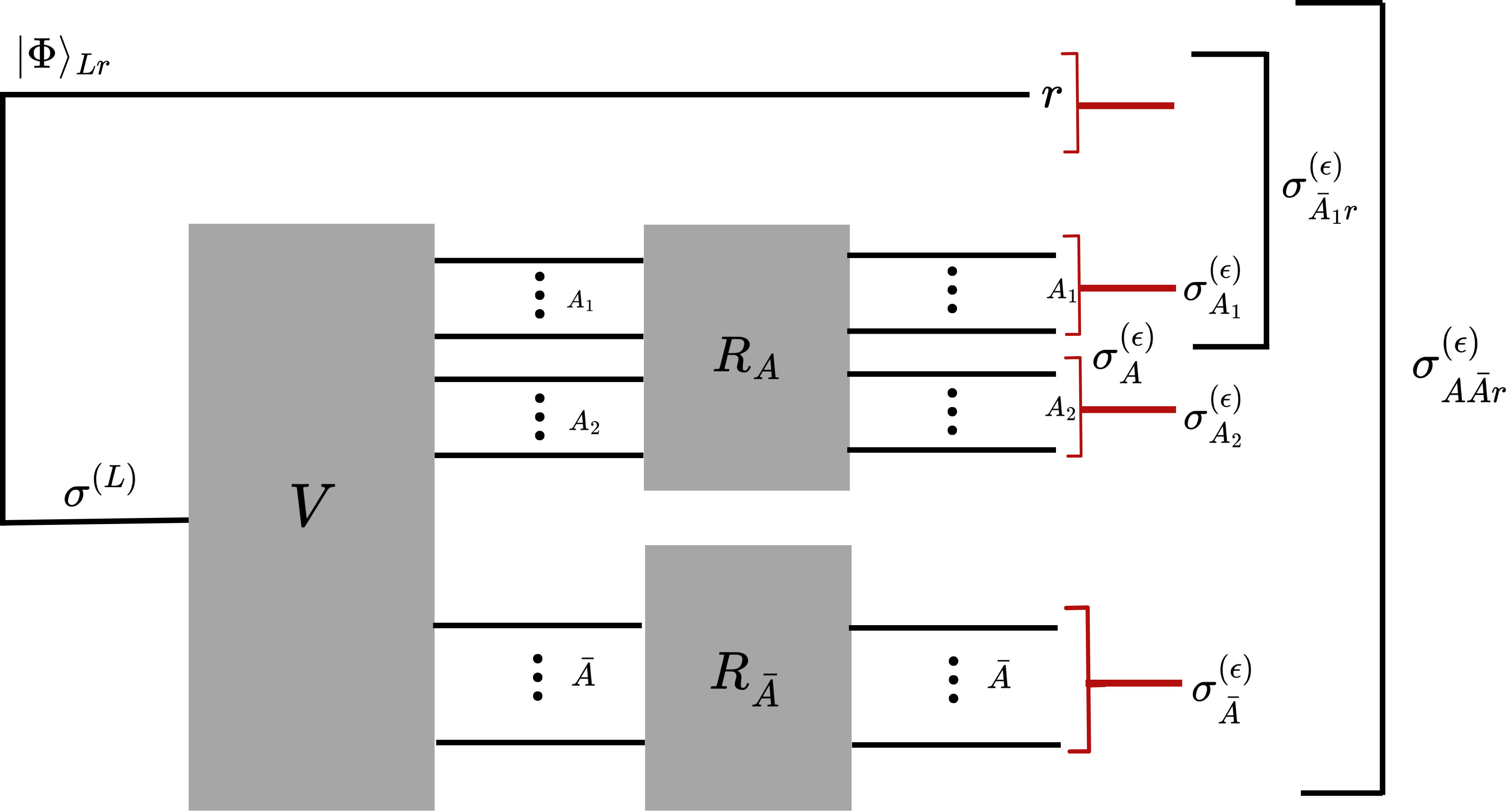}
    \caption{Circuit picture for optimizing recovery via coherent information. A maximally entangled state
$|\Phi\rangle_{Lr}=\frac{1}{\sqrt d}\sum_{i=1}^{d}|i\rangle_L\otimes|i\rangle_r$
purifies the logical input on $\mathcal{H}_L$, with $r$ an isomorphic reference system.
The logical system is encoded by the isometry $V$, and the boundary region $A$ is acted on by the recovery map $R_A$ (with ancillary outputs $A_1,A_2$), while $\bar A$ is processed by $R_{\bar A}$.
This induces the effective channel $\mathcal{N}_{R^{(\epsilon)}}:\mathcal{L}(\mathcal{H}_L)\to \mathcal{L}(\mathcal{H}_{A_1})$, with output state
$\sigma^{(\epsilon)}_{A_1}=\mathcal{N}_{R^{(\epsilon)}}(\sigma^{(L)})$.
The coherent information $I_c(\mathcal{N}_{R^{(\epsilon)}})$ is evaluated on the joint output
$\sigma^{(\epsilon)}_{A_1 r}$ and the full state $\sigma^{(\epsilon)}_{A\bar A r}$, and is maximized over recovery channels to obtain the optimal recovery $R^\ast$.
}
    \label{fig:placeholder}
\end{figure}
Here $|\Phi_d\rangle$ is a fixed maximally entangled state between the logical input system $\mathcal{H}_{L}$ and an isomorphic reference system $\mathcal{H}_{r}$,
\begin{equation}
|\Phi_d\rangle
=
\frac{1}{\sqrt d}\sum_{i=1}^{d}|i\rangle_L\otimes |i\rangle_r, \qquad d=\dim(\mathcal{H}_L)=\dim(\mathcal{H}_r),
\end{equation}
with $\{|i\rangle_L\}$ and $\{|i\rangle_r\}$ orthonormal bases of $\mathcal{H}_L$ and $\mathcal{H}_r$.

After encoding, perturbation, and recovery, the joint state on  $A_1r$ is
\begin{equation}
\sigma^{(\epsilon)}_{A_1 r}
=
\Tr_{A_2 \bar{A}}\!\left[
e^{i\epsilon W_R}
\bigl(\sigma^{(0)}_{A_1 r} \otimes \chi_{A_2 \bar{A}}\bigr)
e^{-i\epsilon W_R}
\right],
\end{equation}
where $\sigma^{(0)}_{A_1 r}\equiv |\Phi_d\rangle\!\langle\Phi_d|_{A_1 r}\in
\mathcal{L}(\mathcal{H}_{A_1}\otimes\mathcal{H}_{r})$. The perturbation acts trivially on the reference, and we will therefore write $W_R$ in place of $W_R\otimes I_r$ when no confusion can arise.

Tracing out the reference yields the reduced output state on $A_1$,
\begin{equation}
\sigma^{(\epsilon)}_{A_1}
=
\Tr_{r}\,\sigma^{(\epsilon)}_{A_1 r}.
\end{equation}
where $\sigma^{(0)}_{A_1}\equiv \Tr_r \sigma^{(0)}_{A_1 r}=I_d/d$. This admits the perturbative expansion
\begin{equation}
\sigma^{(\epsilon)}_{A_1}
=
\sigma^{(0)}_{A_1}
+\epsilon\,\delta^{(1)}\sigma_{A_1}
+\frac{\epsilon^2}{2}\,\delta^{(2)}\sigma_{A_1}
+O(\epsilon^3),
\label{eq:sigma-expansion_A1}
\end{equation}
with
\begin{align}
\delta^{(1)}\sigma_{A_1} 
&= i\,\Tr_{A_2\bar{A}}
\Big(
\big[W_R,\ \sigma^{(0)}_{A_1} \otimes \chi_{A_2 \bar A}\big]
\Big),\\[4pt]
\delta^{(2)}\sigma_{A_1} 
&= -\Tr_{A_2\bar A}
\Big(
\frac{1}{2}\big\{W_R^2,\ \sigma^{(0)}_{A_1} \otimes \chi_{A_2\bar A}\big\}
- W_R\,
\big(\sigma^{(0)}_{A_1} \otimes \chi_{A_2\bar A}\big)\,
W_R
\Big).
\end{align}

In this notation, the coherent information of $\mathcal{N}_{R^{(\epsilon)}}$ takes the standard form
\begin{equation}
I_c(\mathcal{N}_{R^{(\epsilon)}})
=
S\!\left(\sigma^{(\epsilon)}_{A_1}\right)
-
S\!\left(\sigma^{(\epsilon)}_{A_1 r}\right).
\end{equation}

We first expand $S(\sigma^{(\epsilon)}_{A_1})$. Using the logarithm expansion in Appendix~\eqref{logexpansion}, we find
\begin{align}
S\big(\sigma^{(\epsilon)}_{A_1} \big)
&= S\big(\sigma^{(0)}_{A_1} \big)
-\epsilon\,\Tr\Big(
\delta^{(1)}\sigma_{A_1} \,\ln\sigma^{(0)}_{A_1} 
+\sigma^{(0)}_{A_1} \,
D_{\ln}\big(\sigma^{(0)}_{A_1} \big)
\big[\delta^{(1)}\sigma_{A_1} \big]\Big)
\notag\\
&\quad
-\epsilon^2\,\Tr\Big(
\,\delta^{(2)}\sigma_{A_1} \,\ln\sigma^{(0)}_{A_1} 
+\frac{1}{2}\,
\delta^{(1)}\sigma_{A_1} \,
D_{\ln}\big(\sigma^{(0)}_{A_1} \big)
\big[\delta^{(1)}\sigma_{A_1} \big]
\notag\\
&\hspace{2.2cm}
+\,
\sigma^{(0)}_{A_1} \,
D_{\ln}\big(\sigma^{(0)}_{A_1} \big)
\big[\delta^{(2)}\sigma_{A_1} \big]
\Big)
+O(\epsilon^3).
\label{eq:SA1-expansion-A1-mixed}
\end{align}

Substituting the perturbations we obtain.
\begin{equation}
 S(\sigma^{(\epsilon)}_{A_1})
=
\ln d 
\label{eq:Lemma4.1:SrhoA1}
\end{equation}
and the correction vanishes. This is because $\sigma_{A_1}^{(0)}$ is already maximally mixed. 

We next evaluate $S(\sigma^{(\epsilon)}_{A_1 r})$. A technical subtlety is that $\sigma^{(0)}_{A_1 r}=|\Phi_d\rangle\langle\Phi_d|$ is pure, so a Taylor expansion of the entropy around $\sigma^{(0)}_{A_1 r}$ is ill-defined because $\sigma^{(0)}_{A_1 r}$ has zero eigenvalues. We therefore introduce a full-rank regulator family $\sigma^{(0)}_{A_1 r}(\Delta)$ on $\mathbb{C}^d\otimes\mathbb{C}^d$,
\begin{equation}
\sigma^{(0)}_{A_1 r}(\Delta)
=
(1-3\Delta)\,|\Phi_d\rangle\langle\Phi_d|
+
\frac{2\Delta}{d(d-1)}
\sum_{i \neq j} |ij\rangle\langle ij|
+
\frac{\Delta}{d-1}
\sum_{m=1}^{d-1}
|\Phi_d^{(m)}\rangle\langle\Phi_d^{(m)}|,
\qquad
 0< \Delta \ll 1.
\end{equation}
and average its entropy over local unitary acting on $\mH_{A_1}\simeq {\mL_a}$. Here $\Delta$ serves as a regulator. Sending $\Delta\to 0$ recovers the pure state $\sigma^{(0)}_{A_1r}$. The maximally entangled Fourier Bell basis is defined as 
\begin{equation}
\label{eq:fourier-bell}
\ket{\Phi_d^{(m)}} \;=\; \frac{1}{\sqrt{d}} \sum_{k=0}^{d-1} \omega^{mk}\ket{k\,k},
\qquad 
\omega = e^{2\pi i/d},
\qquad 
m = 0,1,\ldots,d-1 .
\end{equation}
The standard maximally entangled state corresponds to the $m=0$ element of this basis,
\begin{equation}
\label{eq:max-ent}
\ket{\Phi_d} \;\equiv\; \ket{\Phi_d^{(0)}} \;= \frac{1}{\sqrt{d}} \sum_{k=0}^{d-1} \ket{k\,k}.
\end{equation}

By construction, the reference extension is consistent with the original marginal on $A_1$, namely
\begin{equation}
\sigma^{(0)}_{A_1}
=
\Tr_{r}\!\big[\sigma^{(0)}_{A_1 r}(\Delta)\big]
=
\frac{I_{d}}{d}.
\end{equation}

Then we compute the Haar-averaged entropy of the perturbed state $\sigma^{(\epsilon)}_{A_1 r}(\Delta,U)$ obtained after conjugating by a local unitary $U$ on $A_1$.
Expanding $\sigma^{(\epsilon)}_{A_1 r}$ perturbatively as in Eq.~\eqref{eq:sigma-expansion_A1} (with $A_1\to A_1 r$), we obtain the entropy expansion
\begin{align}
S\big(\sigma^{(\epsilon)}_{A_1 r}(U)\big)
&=S(\sigma_{A_1 r}^{(0)}(U))+
\epsilon\,\operatorname{Tr}\Big(
\delta^{(1)}\sigma_{A_1 r}(U)\,
\ln\sigma^{(0)}_{A_1 r}(U)
+\,
\sigma^{(0)}_{A_1 r}(U)\,
D_{\ln}\big(\sigma^{(0)}_{A_1 r}(U)\big)
\big[\delta^{(1)}\sigma_{A_1 r}(U)\big]
\Big)\notag\\& \quad
-\epsilon^2\,\operatorname{Tr}\Big(
\delta^{(2)}\sigma_{A_1 r}(U)\,
\ln\sigma^{(0)}_{A_1 r}(U)
+\frac{1}{2}\,
\delta^{(1)}\sigma_{A_1 r}(U)\,
D_{\ln}\big(\sigma^{(0)}_{A_1 r}(U)\big)
\big[\delta^{(1)}\sigma_{A_1 r}(U)\big]
\notag\\
&\quad+\,
\sigma^{(0)}_{A_1 r}(U)\,
D_{\ln}\big(\sigma^{(0)}_{A_1 r}(U)\big)
\big[\delta^{(2)}\sigma_{A_1 r}(U)\big]
\Big)
\notag+\mO(\epsilon^3),
\label{eq:entropy-expansion-A1r}
\end{align}
with
\begin{align}
\sigma^{(0)}_{A_1 r}(U)
&= U\,\sigma^{(0)}_{A_1 r}\,U^\dagger,\\[4pt]
\delta^{(1)}\sigma_{A_1 r}(U)
&= i\,\operatorname{Tr}_{A_2\bar A}
\Big(
\big[W_R,\,
\sigma^{(0)}_{A_1 r}(U)\otimes\chi_{A_2\bar A}\big]
\Big),\\[4pt]
\delta^{(2)}\sigma_{A_1 r}(U)
&= -\,\operatorname{Tr}_{A_2\bar A}\Big(
\frac{1}{2}\big\{W_R^{2},\,
\sigma^{(0)}_{A_1 r}(U)\otimes\chi_{A_2\bar A}\big\}
- W_R\,
\big(\sigma^{(0)}_{A_1 r}(U)\otimes\chi_{A_2\bar A}\big)\,
W_R
\Big).
\end{align}

To ensure the entropy admits a controlled Taylor expansion, the regulator must not be taken parametrically smaller than the perturbation strength; we therefore impose $\Delta\gtrsim \epsilon$. At the same time, we work in the regime $\Delta\ll 1$, so that $\Delta$ remains small, justifying an expansion to leading order in $\Delta$.

Haar averaging over $U$ removes all terms linear in $\epsilon$ and the terms of the form $\Tr[\sigma^{(0)} D_{\ln}(\sigma^{(0)})[\cdot]]$, leaving
\begin{align}
\Big\langle S\big(\sigma^{(\epsilon)}_{A_1 r}\big)\Big\rangle
&=
\Big\langle S\big(\sigma^{(0)}_{A_1 r}\big)\Big\rangle
-\epsilon^2\int dU\;
\operatorname{Tr}\Big(
\delta^{(2)}\sigma_{A_1 r}(U)\,
\ln\sigma^{(0)}_{A_1 r}(U)
\notag\\
&\qquad\qquad
+\frac{1}{2}\,
\delta^{(1)}\sigma_{A_1 r}(U)\,
D_{\ln}\big(\sigma^{(0)}_{A_1 r}(U)\big)
\big[\delta^{(1)}\sigma_{A_1 r}(U)\big]
\Big)
+O(\epsilon^3).
\label{eq:Sbulk-mixed-expanded-A1r}
\end{align}

Diagrammatically, Eq~\eqref{eq:Sbulk-mixed-expanded-A1r} can be represented as 

\begin{align}
\big\langle S\!\left(\sigma^{(R)}_{A_1 r}\right)\big\rangle
&= \frac{\epsilon^2}{d_\chi}\int dU \left(
\frac{1}{2}\, \, \vcenter{\hbox{\includegraphics[height=9em]{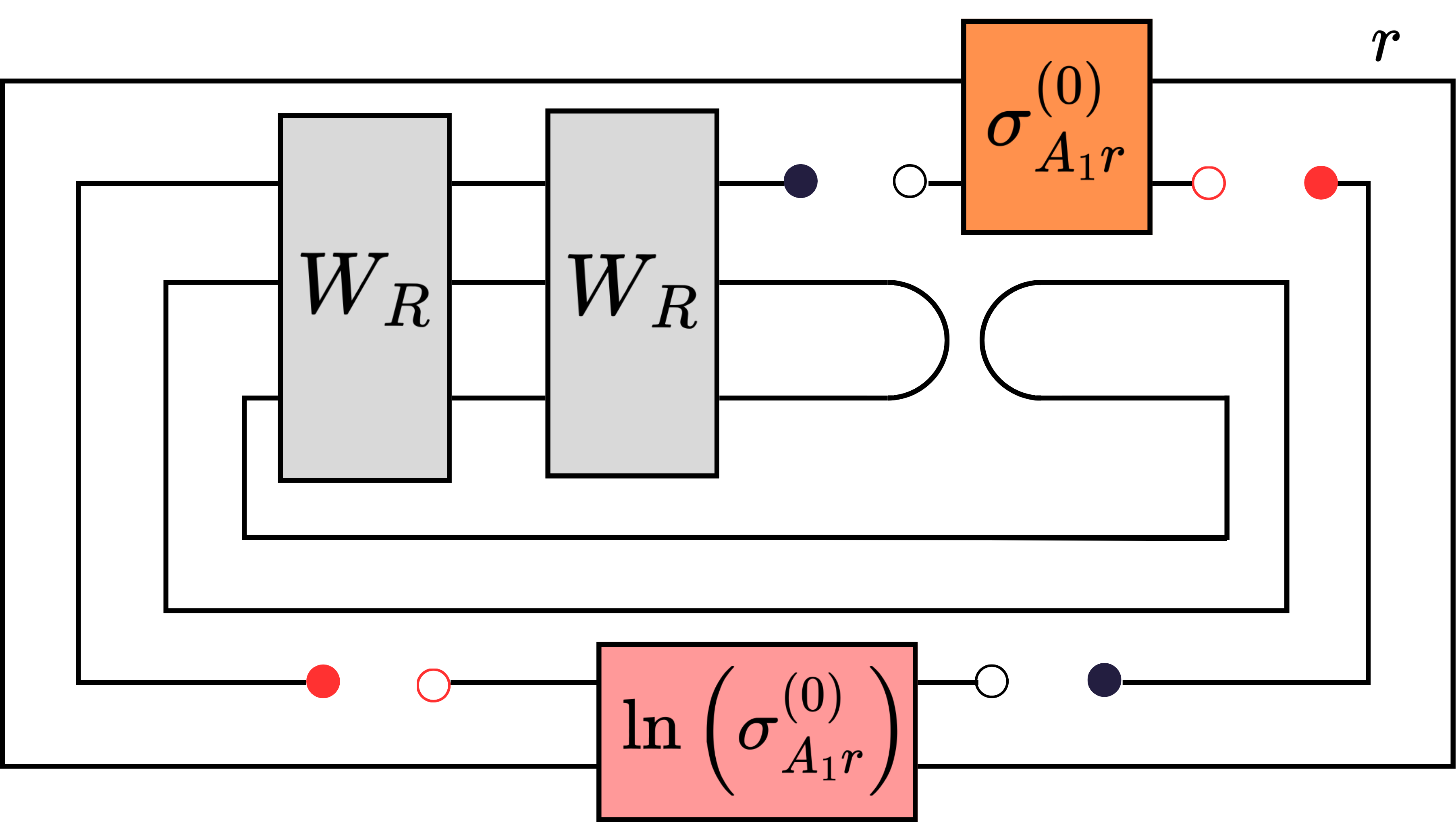}}}
\notag \right. \\  & \hspace{20mm} -\left.\vcenter{\hbox{\includegraphics[height=9em]{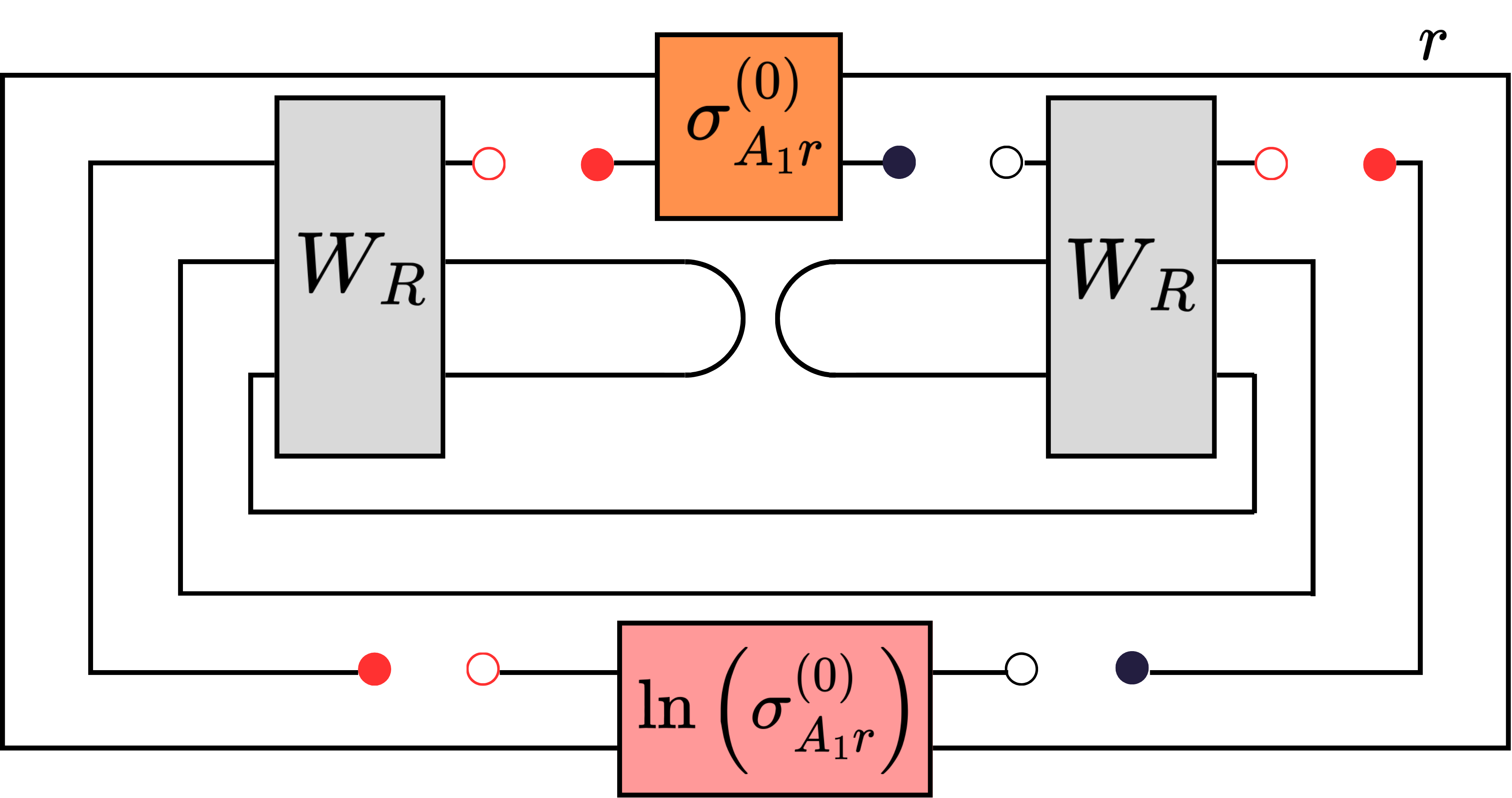}}}
\, \, \right)\nonumber\\
& \hspace{-5mm} + \frac{\epsilon^2}{d_\chi^3}\int dU \int_0^\infty ds \left(\, \,
\vcenter{\hbox{\includegraphics[height=9em]{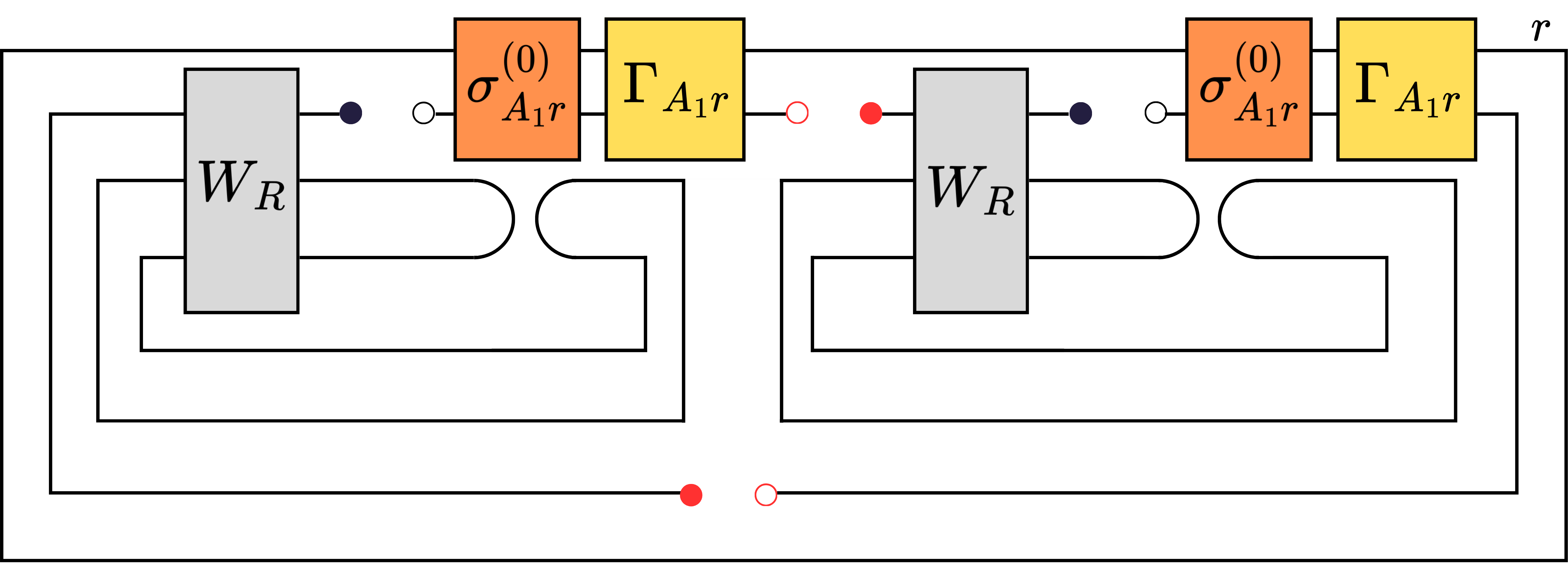}}}\right.
\notag\\ & \hspace{20mm} -2\,\quad \vcenter{\hbox{\includegraphics[height=9.7em]{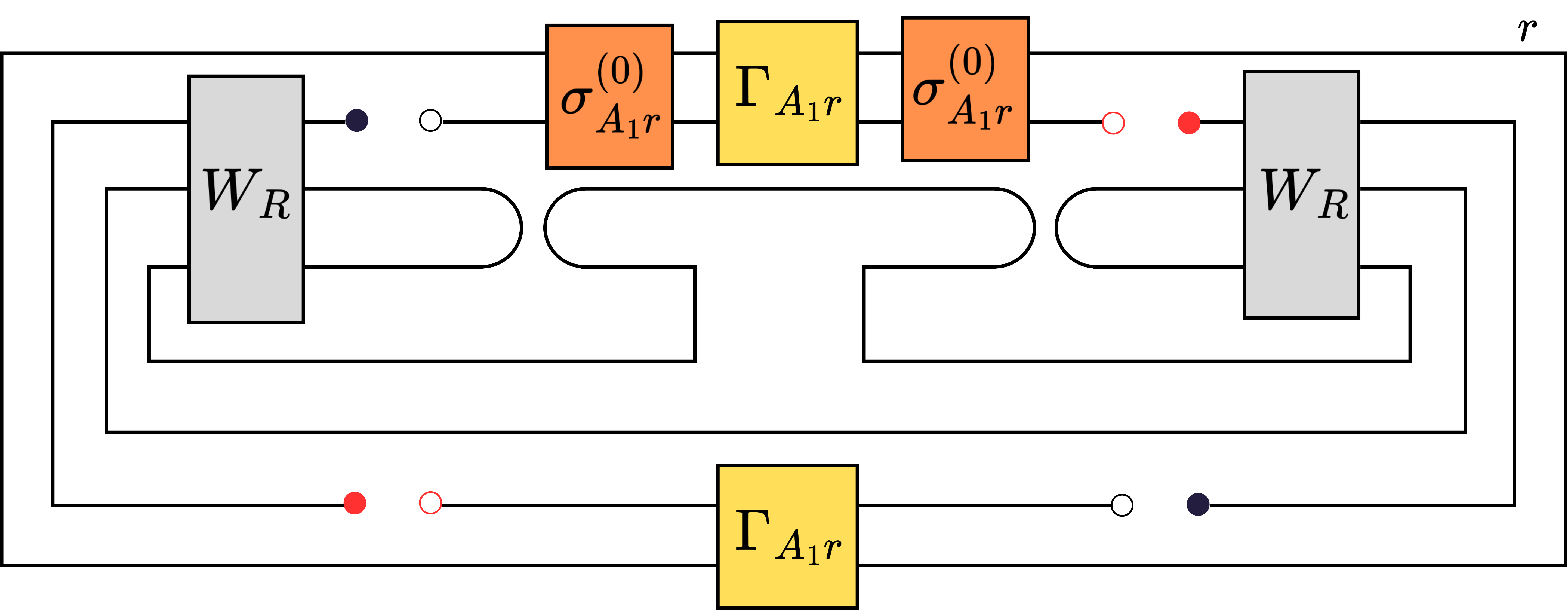}}}
 \notag\\ & \hspace{25mm} \left. + \vcenter{\hbox{\includegraphics[height=9.7em]{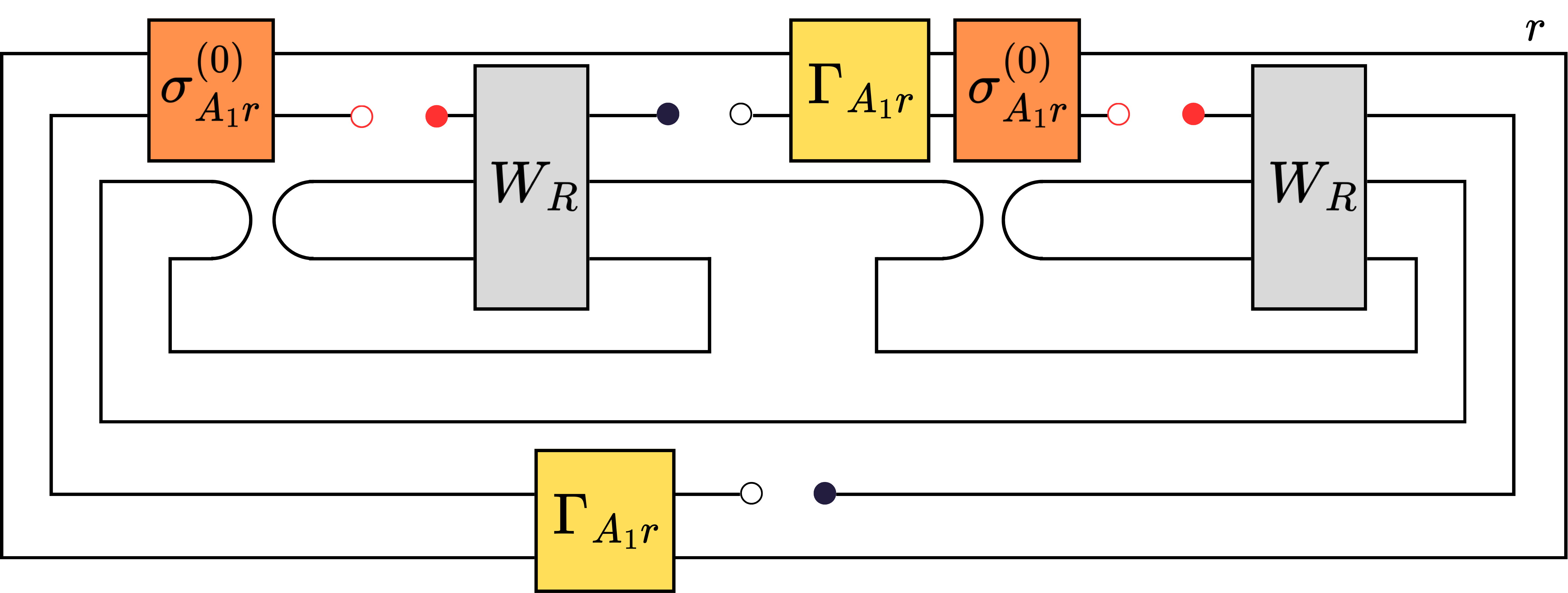}}}
\, \right).
\end{align}

In these diagrams, the outermost loop (labeled by $r$) represents the reference Hilbert space $\mathcal{H}_r$.
We define the resolvent as
\begin{equation}
\Gamma_{A_1 r}(s)
=
\left(\frac{\sigma^{(0)}_{A_1 r}}{d_{\chi}} + s\, I_{A_1 r}\right)^{-1}.
\end{equation}
Integrating over $U$, we obtain the following expression in terms of the $W$- and $\mathscr{R}$-diagrams:
\begin{equation}
\begin{aligned}
\Big\langle S\big(\sigma^{(\epsilon)}_{A_1 r}\big)\Big\rangle
&=
\epsilon^2\Bigg\{
\frac{d}{d_\chi(d^2-1)}
\left(\mathscr{R}_8-\frac{1}{d}\mathscr{R}_9\right)
\left(W_5-\frac{1}{d}W_6\right)
\\
&\qquad
-\frac{1}{d_\chi^3}
\int_0^{\infty} ds\;
\frac{1}{d^2-1}
\left(\mathscr{R}_6-\mathscr{R}_7\right)
\operatorname{Tr}\left(W_7-\frac{1}{d}W_8\right)
\Bigg\}.
\end{aligned}
\label{eq:Lemma4.1:SrhoA1r}
\end{equation}

For the regulated family $\sigma^{(0)}_{A_1 r}(\Delta)$, the $\mathscr{R}$-terms in
Eq.~\eqref{eq:Lemma4.1:SrhoA1r} evaluates to 
\begin{align}
\mathscr{R}_8 - \frac{1}{d}\mathscr{R}_9
=&
\Tr\left(\sigma^{(0)}_{A_1 r}\,\ln \sigma^{(0)}_{A_1 r}\right)
-
\frac{1}{d}\,
\Tr\left[
\Tr_{A_1}\left(\sigma^{(0)}_{A_1 r}\right)\,
\Tr_{A_1}\left(\ln \sigma^{(0)}_{A_1 r}\right)
\right]\notag\\
=& \left((1-3\Delta)-\frac{1}{d^2}\right)\ln(1-3\Delta)
+\left(\Delta-\frac{d-1}{d^2}\right)\ln\!\left(\frac{\Delta}{d-1}\right)
\notag\\ &+\left(2\Delta-\frac{d-1}{d}\right)\ln\!\left(\frac{2\Delta}{d(d-1)}\right)\notag\\
\int_0^\infty ds \Big(\mathscr{R}_6 - \frac{1}{d}\mathscr{R}_7\Big)
=&
\int_0^\infty ds\Tr\left[\Tr_{A_1}\left(\Gamma_{A_1 r}(s) \right)\Tr_{A_1}\left((\sigma^{(0)}_{A_1 r})^2 \ \Gamma_{A_1 r}(s)\right)-\left[\Tr_{A_1}(\sigma^{(0)}_{A_1 r} \Gamma_{A_1 r}(s))\right]^2\right]\notag \\ \
=& dd_\chi \left((1-3\Delta)-\frac{1}{d^2}\right)\ln(1-3\Delta)
+\left(\Delta-\frac{d-1}{d^2}\right)\ln\!\left(\frac{\Delta}{d-1}\right)
\notag\\ &+\left(2\Delta-\frac{d-1}{d}\right)\ln\!\left(\frac{2\Delta}{d(d-1)}\right)
\end{align}
Now, substituting this back in the Eq.\eqref{eq:Lemma4.1:SrhoA1r}, and writing the $W$-diagrams in terms of $J$ and $D$,  we get expression $\Big\langle S\big(\sigma^{(\epsilon)}_{A_1 r};\Delta\big)\Big\rangle$, in leading order of $\Delta$, as following.
\begin{align}
\Big\langle S\big(\sigma^{(\epsilon)}_{A_1 r};\Delta\big)\Big\rangle
&=
-\frac{\,\epsilon^{2}d}{d_\chi (d^2-1)}\,
\Bigg(\frac{d-1}{d^2}\left[d \log \left(\frac{2}{d}\right)+(d+1) \log \left(\frac{\Delta}{d-1}\right)\right]\Bigg)
\Bigg[\Big(
\Tr_{A_1A_2}(J^{2})
\notag\\&-\frac{1}{d}\,\Tr_{A_2}\!\left(\Tr_{A_1}(J)^{2}\right)
-\frac{1}{d_\chi}\,\Tr_{A_1}\!\left(\Tr_{A_2}(J)^{2}\right)
+\frac{1}{d d_\chi}\,\Tr_{A_1A_2}(J)^{2}
\Big)\notag\\&+\Big(
\Tr_{A_1A_2}((D)^{2})
-\frac{1}{d}\,\Tr_{A_2}\!\left(\Tr_{A_1}(D)^{2}\right)
\Big)\Bigg] + O(\epsilon^3, \Delta)\notag\\
 = &
-2\,\epsilon^{2}\,
\Bigg(\frac{d-1}{d^2}\left[d \log \left(\frac{2}{d}\right)+(d+1) \log \left(\frac{\Delta}{d-1}\right)\right]
\Bigg)
\Big[c_1(J)+c_2(D)\Big]\notag \\ & \qquad \qquad +O(\epsilon^3, \Delta)\notag.
\end{align}

The regulated coherent information is therefore
\begin{align}
\label{eq:regcoh}
I_c^{(\Delta)}\big(\mathcal{N}_{R^{(\epsilon)}}\big)
&=
\ln d
-
2\,\epsilon^{2}\,
\Bigg(\frac{d-1}{d^2}\left[d \log \left(\frac{d}{2}\right)+(d+1) \log \left(\frac{d-1}{\Delta}\right)\right]
\Bigg)
\Bigg[c_1(J)+c_2(D)\Bigg]\notag \\& \qquad \qquad \quad \qquad + O(\epsilon^3, \Delta).
\end{align}

The $\Delta$-dependence of the $O(\epsilon^2)$ term is dominated by $\log \Delta$; all omitted contributions are $O(1)$ or $O(\Delta)$ as $\Delta\to 0$.

Finally, we optimize over the recovery. The prefactor multiplying $c_1(J)+c_2(D)$ in Eq.~\eqref{eq:regcoh} is strictly positive for $0<\Delta\le 1/3$ and $d\ge 2$, and the diagrammatic coefficients $c_1(J)$ and $c_2(D)$ are nonnegative (see Eq.~\ref{ref-PQ}). Hence maximizing $I_c^{(\Delta)}$ is equivalent to minimizing $c_1(J)+c_2(D)$ over the allowed recovery operations.

Recall that the effective perturbation (conjugated by recovery unitaries) is 
\begin{equation}
e^{i\epsilon W_R}
=
R_A^{(\epsilon)}\,R_{\bar A}^{(\epsilon)}\,
e^{i\epsilon W}\,
R_A^{(0)\dagger}\,R_{\bar A}^{(0)\dagger},
\end{equation}
Composing the recovery with additional local unitaries,
\[
R_A^{(\epsilon)} \mapsto e^{i\epsilon O_A}\,R_A^{(\epsilon)},
\qquad
R_{\bar A}^{(\epsilon)} \mapsto e^{i\epsilon O'_{\bar A}}\,R_{\bar A}^{(\epsilon)},
\]
shifts $W_R$ (to leading order in $\epsilon$) as
\begin{equation}
W_R \longrightarrow W_R + O_A\otimes I_{\bar A} + I_A\otimes O'_{\bar A}.
\label{eq:WR-shift}
\end{equation}

By the definition of the operators $J$ and $D$ (see Fig.~\ref{ref-PQ}), this implies
\begin{equation}
J \longrightarrow J + 2O_A + 2I_{A_1}\otimes O'_{A_2},
\qquad
D \longrightarrow D .
\label{eq:PQ-shift}
\end{equation}
Thus local adjustments can change $J$ but leave $D$ invariant. In particular, $c_2(D)$ is unaffected, whereas $c_1(J)$ can be minimized by an appropriate choice of $O_A$ (and $O'_{A_2}$). For example, choosing $O_A=-\tfrac12 J$ and $O'_{A_2}=0$ sets the shifted $J$ to zero and yields the maximal value of $I_c^{(\Delta)}$ within this class of local variations.

\end{proof}

\subsection{Theorem~\ref{thm:monotonic-pure}}\label{app:monotonic-pure}
\begin{proof}
We take the bulk input state in the pure-state case (cf.~Eq.~\eqref{eq:bulk_state_pure_case}) and express it in Schmidt form:
\begin{equation}
|\psi\rangle_L
=
\sum_{i=1}^{d}\sqrt{\lambda_i}\,|i\rangle_{a}\,|i\rangle_{\bar a}, \qquad \lambda_i>0
\end{equation}
Here, we denote
$d=\dim(\mathcal{H}_{A_1})$ and $\bar d=\dim(\mathcal{H}_{\bar A_1})$.

In what follows, we assume $\bar d \ge d$ so that the Taylor expansion of the entropy is well defined. The case $d>\bar d$ can be recovered as a limiting case by taking some of the Schmidt coefficients to approach zero.  

After skewed encoding and decoding the effective perturbation, the recovered state on the full physical Hilbert space
$\mathcal{H}_{A\bar A}$ is
\begin{equation}
\sigma^{(\epsilon)}_{A\bar A}
=
e^{i\epsilon W_R}
\Big(
|\psi\rangle\langle\psi|_{A_1\bar A_1}
\otimes
|\chi\rangle\langle\chi|_{A_2\bar A_2}
\Big)
e^{-i\epsilon W_R}.
\end{equation}

In the unperturbed limit $\epsilon\to 0$, the recovered state factorizes as
\begin{align}
\sigma^{(0)}_{A\bar A}
&=
|\psi\rangle\langle\psi|_{A_1\bar A_1}
\otimes
\chi_{A_2\bar A_2},
\label{eq:unpert_bdAAbar}
\end{align}
where we define
\begin{equation}
\chi_{A_2\bar A_2}
:=
|\chi\rangle\langle\chi|_{A_2\bar A_2},
\qquad
\chi_{A_2}
:=
\operatorname{Tr}_{\bar A_2}\!\left[\chi_{A_2\bar A_2}\right].
\end{equation}

The corresponding reduced unperturbed states are
\begin{align}
\sigma^{(0)}_{A_1A_2}
&=
\operatorname{Tr}_{\bar A_1\bar A_2}\!\left[\sigma^{(0)}_{A\bar A}\right]
=
\operatorname{Tr}_{\bar A_1}\!\left[
|\psi\rangle\langle\psi|_{A_1\bar A_1}
\right]
\otimes
\chi_{A_2},
\label{eq:unpert_bd}\\
\sigma^{(0)}_{A_1\bar A_1}
&=
\operatorname{Tr}_{A_2\bar A_2}\!\left[\sigma^{(0)}_{A\bar A}\right]
=
|\psi\rangle\langle\psi|_{A_1\bar A_1},
\label{eq:unpert_A1A1bar}\\
\sigma^{(0)}_{A_1}
&=
\operatorname{Tr}_{A_2\bar A_1\bar A_2}\!\left[\sigma^{(0)}_{A\bar A}\right]
=
\operatorname{Tr}_{\bar A_1}\!\left[
|\psi\rangle\langle\psi|_{A_1\bar A_1}
\right].
\label{eq:unpert_bulk}
\end{align}

In the Schmidt basis of $|\psi\rangle_{A_1\bar A_1}$, the reduced state on $A_1$ is diagonal with eigenvalues $\{\lambda_i\}$, and hence admits the spectral decomposition
\begin{equation}
\sigma^{(0)}_{A_1}
=
\sum_{i=1}^{d} \lambda_i\, |i\rangle\langle i|_{A_1},
\end{equation}

We now evaluate the Haar-averaged correction $\langle S_{\mathrm{corr}}\rangle$.
Since the unperturbed bulk state is supported on $A_1\bar A_1$, we average independently over local unitaries acting on each factor.
Concretely, the logical state is conjugated by local unitaries, 
\begin{equation}
|\psi\rangle\langle\psi|_{A_1\bar A_1}
\;\longrightarrow\;
(U_{A_1}\otimes V_{\bar A_1})\,
|\psi\rangle\langle\psi|_{A_1\bar A_1}\,
(U_{A_1}^\dagger\otimes V_{\bar A_1}^\dagger),
\end{equation}
where $U$ acts on $A_1$ and $V$ acts on $\bar A_1$, and then average over $U$ and $V$ with respect to the Haar measure. We emphasize that we do not average with a single global Haar unitary on $A_1\bar A_1$, because such a transformation would generically modify the entanglement structure between $A_1$ and $\bar A_1$. By restricting to independent local Haar rotations, we randomize only local bases on each factor while preserving the intrinsic bipartite entanglement content of the bulk state. The resulting unperturbed rotated states are
\begin{align}
\sigma^{(0)}_{A\bar A}(U,V)
&=
(U\otimes V)\,
\sigma^{(0)}_{A_1\bar A_1}\,
(U^\dagger\otimes V^\dagger)
\otimes
\chi_{A_2\bar A_2},
\label{eq:sigma0-UV-full}\\
\sigma^{(0)}_{A_1A_2}(U,V)
&=
\operatorname{Tr}_{\bar A_1\bar A_2}\!\left[
\sigma^{(0)}_{A\bar A}(U,V)
\right],\\
\sigma^{(0)}_{A_1}(U,V)
&=
\operatorname{Tr}_{A_2\bar A_1\bar A_2}\!\left[
\sigma^{(0)}_{A\bar A}(U,V)
\right].
\end{align}
The corresponding skewed state on $A\bar A$ is
\begin{equation}
\sigma^{(\epsilon)}_{A\bar A}(U,V)
=
e^{i\epsilon W_R}\,
\sigma^{(0)}_{A\bar A}(U,V)\,
e^{-i\epsilon W_R},
\label{eq:sigmaeps-UV-full}
\end{equation}
with reduced states defined by partial tracing as in Eq.~\eqref{rhoM-general}. Taking the expansion in $\epsilon$ as Eq.~\eqref{eq:expandstate}, we obtain the Haar-averaged correction in the form
\begin{align}
\langle S_{\mathrm{corr}}\rangle
=
\frac{\epsilon^2}{2}
\int dU\,dV
\Bigg[
&\operatorname{Tr}\!\left(
\delta^{(1)}\sigma_{A_1A_2}(U,V)\,
D_{\ln}\big(\sigma^{(0)}_{A_1A_2}(U,V)\big)
\!\left[
\delta^{(1)}\sigma_{A_1A_2}(U,V)
\right]
\right)
\nonumber\\
&-
\operatorname{Tr}\!\left(
\delta^{(1)}\sigma_{A_1}(U,V)\,
D_{\ln}\big(\sigma^{(0)}_{A_1}(U,V)\big)
\!\left[
\delta^{(1)}\sigma_{A_1}(U,V)
\right]
\right)
\Bigg].
\label{eq:Scorr-Haar-pure}
\end{align}
while the first-order corrections are
\begin{align}
\delta^{(1)}\sigma_{A_1A_2}(U,V)
&=
i\,
\operatorname{Tr}_{\bar A_1\bar A_2}\!\left(
\big[
W_R,\,
\sigma^{(0)}_{A\bar A}(U,V)
\big]
\right),
\label{eq:delta1-bdry-pure}\\
\delta^{(1)}\sigma_{A_1}(U,V)
&=
i\,
\operatorname{Tr}_{A_2\bar A_1\bar A_2}\!\left(
\big[
W_R,\,
\sigma^{(0)}_{A\bar A}(U,V)
\big]
\right).
\label{eq:delta1-bulk-pure}
\end{align}

Substituting \eqref{eq:delta1-bdry-pure}--\eqref{eq:delta1-bulk-pure} into
\eqref{eq:Scorr-Haar-pure} and expanding the commutators into the resulting
trace terms using the diagrammatic convention and rules described in Appendix~\eqref{WGcalculus}, we obtain
\begin{align}
\big\langle S_{\mathrm{corr}}\big\rangle
= \frac{\epsilon^2}{2d_{\chi}^2}
\int ds \int dU &\left( \,
-\vcenter{\hbox{\includegraphics[height=8.7em]{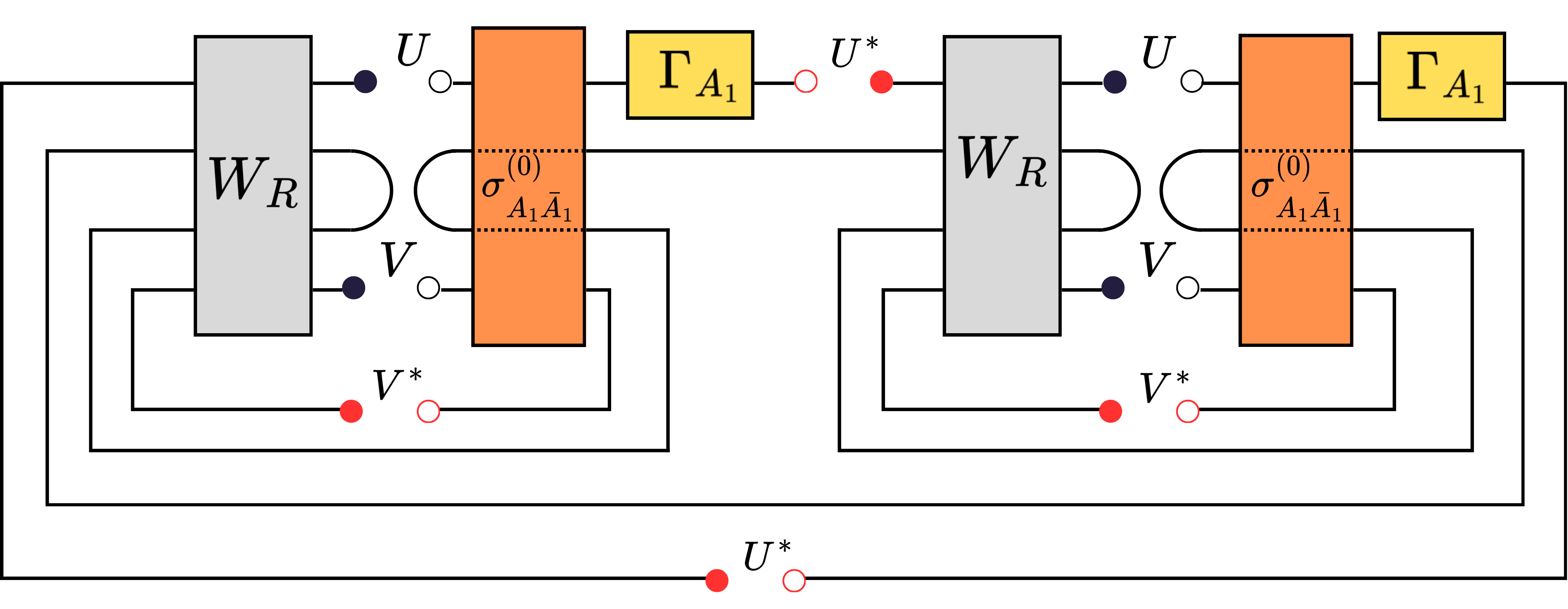}}}
\right.\notag\\
+ 2\,&\vcenter{\hbox{\includegraphics[height=8.7em]{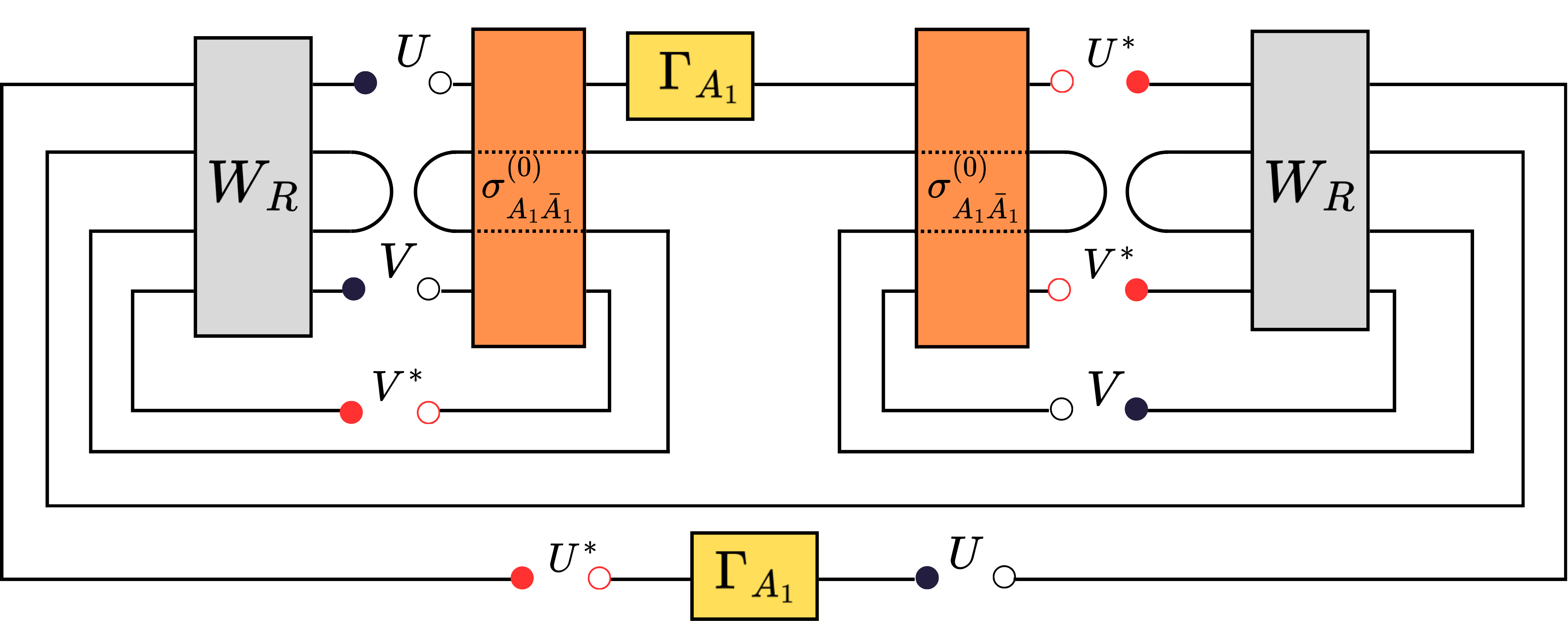}}}
\notag\\
-&\vcenter{\hbox{\includegraphics[height=8.7em]{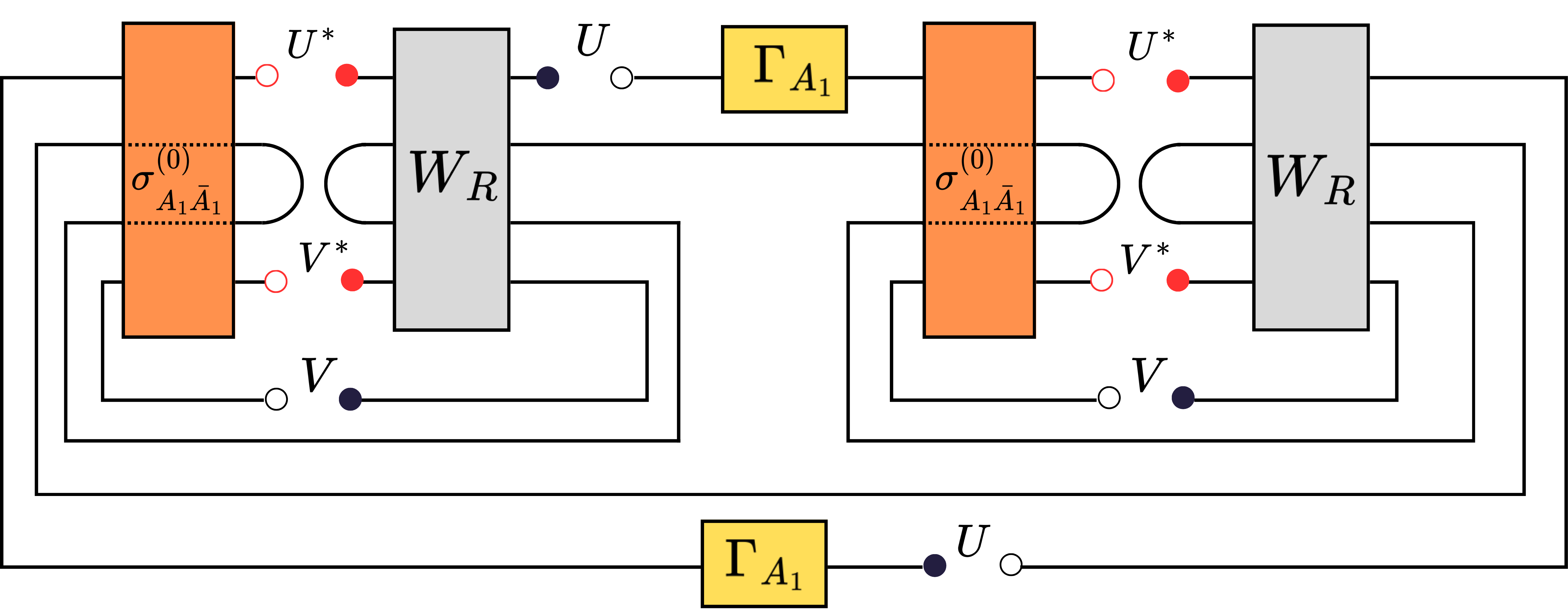}}}
\notag\\
+ \frac{1}{d_\chi}
  &\vcenter{\hbox{\includegraphics[height=8.7em]{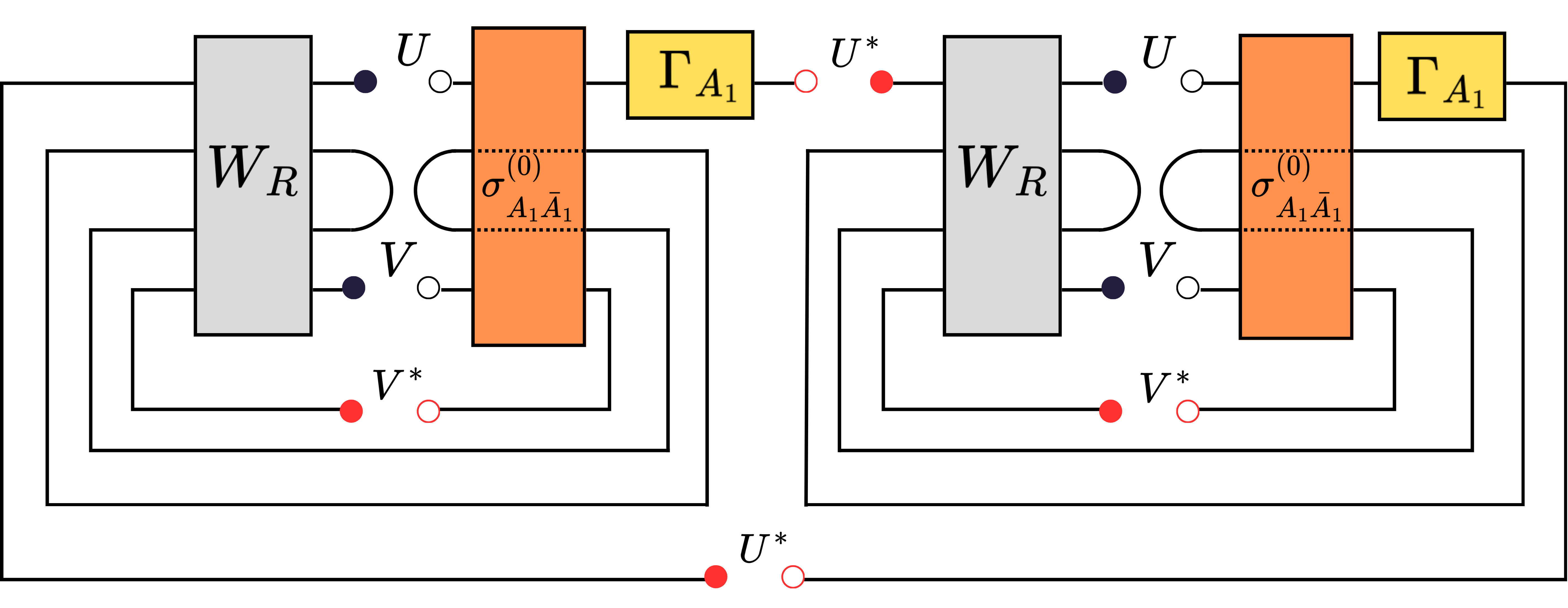}}}
\notag\\
-\frac{2}{d_\chi}
  &\vcenter{\hbox{\includegraphics[height=8.7em]{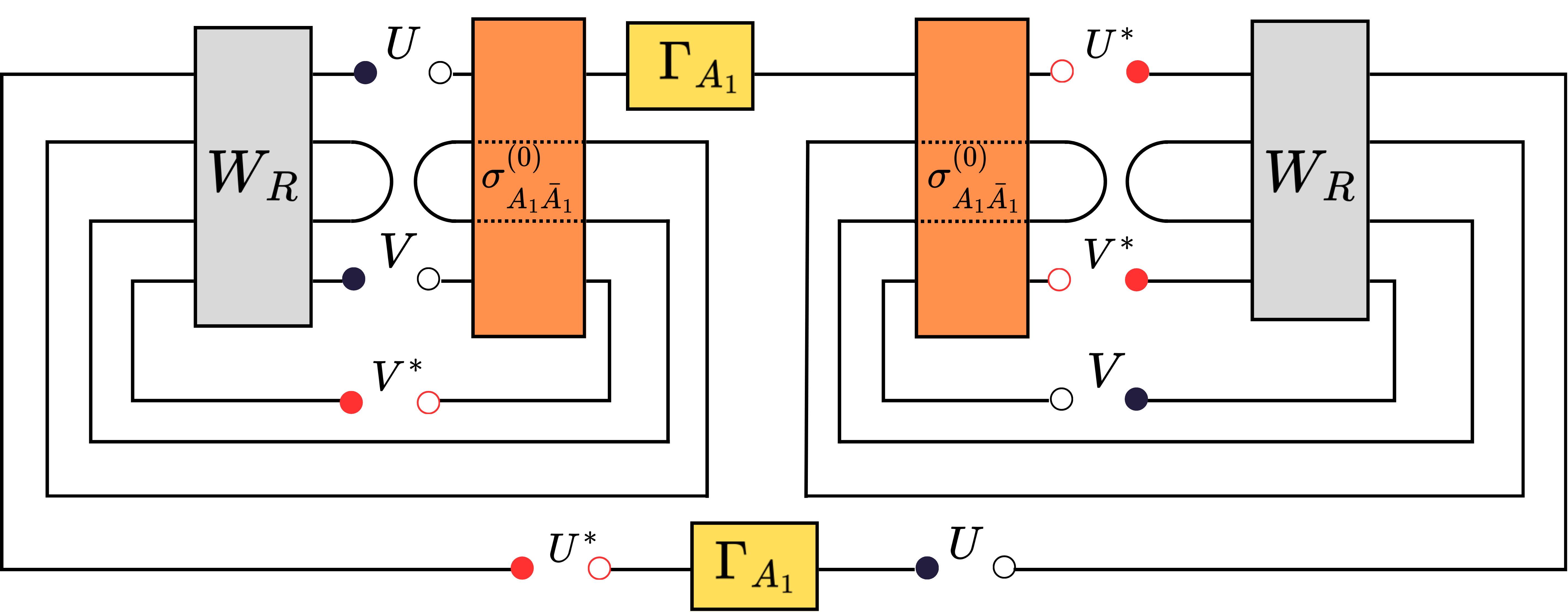}}}
\notag\\
+ \frac{1}{d_\chi}
  &\left.\vcenter{\hbox{\includegraphics[height=8.7em]{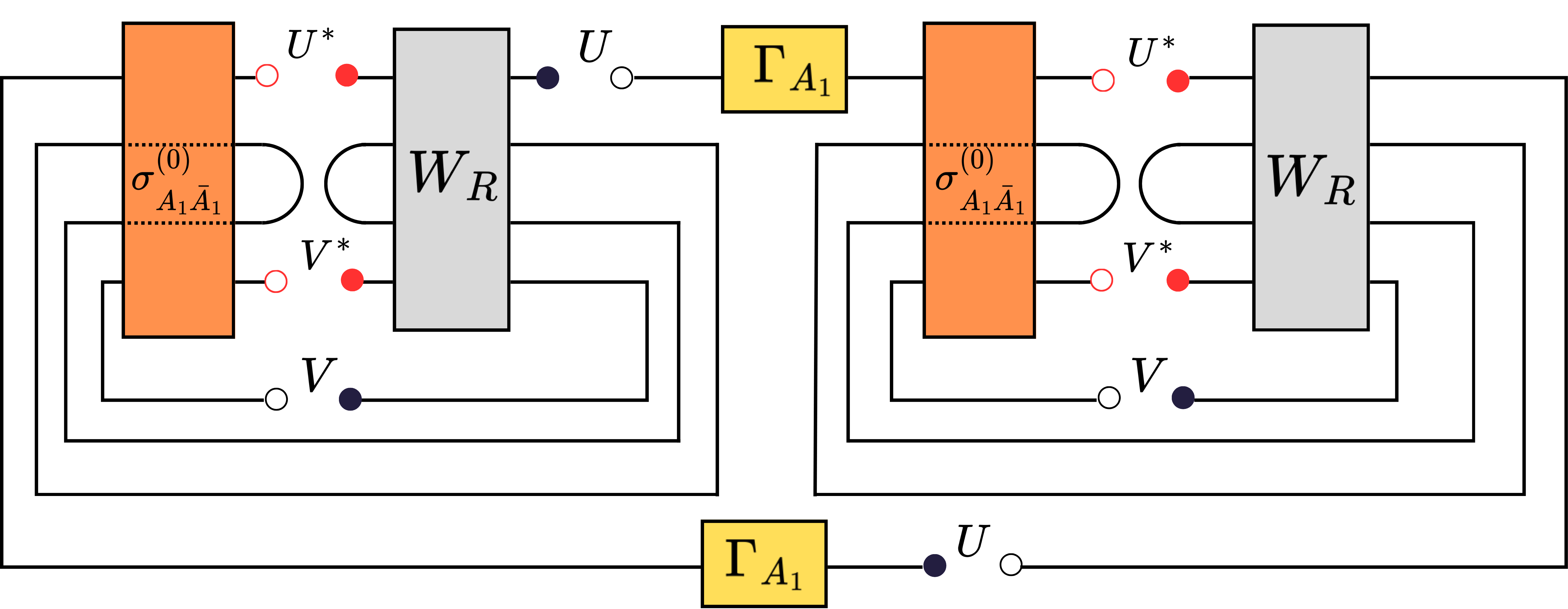}}}
\, \right)
\end{align}
Applying the Weingarten calculus together with the diagrammatic Haar-averaging rules summarized in Appendix~\eqref{WGcalculus-rules} yields
\begin{equation}
\begin{aligned}
\langle S_{corr}\rangle =&\frac{1}{2d_{\chi}d\bar d (d^2-1)(\bar d^2-1)}\epsilon^2 \int_0^{\infty} ds\, 
\Big[\, \\
&\widetilde{\mathscr{R}}_1(s)\,\Big(2\widetilde W_{10}- 2 \bar d\widetilde W_{9}-2d\widetilde W_{12}
+2d \bar d\widetilde W_{11}+2d\bar d \widetilde W_{8}-2d\widetilde W_{5}-2\bar d\widetilde W_{7}-2\widetilde W_{6} \notag \\ 
&\qquad-\left(1+ d \bar d\right)\widetilde W_{2} 
+\left(d+\bar d\right)\widetilde W_{1}+\left(d+\bar d\right)\widetilde W_{3}-\left(1+d\bar d\right)\widetilde W_{4}+\left(d+\bar d\right)\widetilde W_{15}\notag \\
&\qquad -\left(1+d\bar d\right)\widetilde W_{13}
-\left(1+d\bar d\right)\widetilde W_{14}+\left(d+\bar d\right)\widetilde W_{16}
\Big)\\
&+\widetilde{\mathscr{R}}_2(s)\,\Big( 
2\bar d\widetilde W_{12}-2\widetilde W_{11}-2d\bar d\widetilde W_{10}
+2d\widetilde W_{9}+2d\bar d\widetilde{W}_8-2d \widetilde{W}_5-2\bar d \widetilde{W}_7+2\widetilde{W}_6\Big)
\\
&+\widetilde{\mathscr{R}}_3(s)\,\Big(2d\ \widetilde W_{10}-2d\bar d \ \widetilde W_{9} -2 \ \widetilde W_{12}+2\bar d \ \widetilde W_{11}-2d \ \widetilde W_{8}+2d\bar d\ \widetilde W_{5}+2 \ \widetilde W_{7}-2 \bar d \ \widetilde W_{6}
\Big)\\
&+\widetilde{\mathscr{R}}_4(s)\,\Big(-(d+\bar d)\widetilde W_2+(1+d\bar d)\widetilde W_1+(1+d\bar d)\widetilde W_3-(d+\bar d)\widetilde W_4+2d\widetilde W_8\\
&\qquad -2d\bar d\widetilde W_5-2\widetilde W_7+2\bar d\widetilde W_6+2\bar d\widetilde W_{10}-2\widetilde W_9-2d\bar d\widetilde W_{12}+2d\widetilde W_{11}+(1+d\bar d)\widetilde W_{15}\\
&\qquad-(d+\bar d)\widetilde W_{13}-(d+\bar d)\widetilde W_{14}+(1+d\bar d)\widetilde W_{16}\Big)\Big].
\end{aligned}
\end{equation}
where $\widetilde W_i$ and $\widetilde{\mathscr{R}}_i(s)$ are given in Table. \eqref{W-tilde-R-tilde-diagrams}. They're functions of the state on $A_1\bar A_1$:
\begin{eqns}
    \sigma_{A_1\bar A_1}^{(0)}=\sum_{i,j=1}^d\sqrt{\lambda_i\lambda_j}\ket{ii}_{A_1\bar A_1}\bra{jj}. 
\end{eqns}

Substituting this into the definition of $\widetilde{\mathscr{R}}i(s)$ and using
\begin{eqns}
    \sum_{i=1}^d\sqrt{\lambda_i}\ket{i}\bra{i}=\sqrt{\sigma_{A_1}^{(0)}}, 
\end{eqns}
we obtain $\widetilde{\mathscr{R}}_1(s)=\widetilde{\mathscr{R}}_5(s)=\mathscr{R}_1(s)$ and $\widetilde{\mathscr{R}}_2(s)=\mathscr{R}_3(s)$ as functions of the eigenvalues ${\lambda_i}$. Equivalently, they are determined by the same spectral functions $f_1(\lambda)$ and $f_2(\lambda)$, defined (as linear combinations) in Eq.~\eqref{eq:f1f2}.
$\int ds\,\widetilde{\mathscr{R}}_4(s)=\int ds\, \widetilde{\mathscr{R}}_6(s)$ is a constant function independent of entanglmeent spectrum, while $\widetilde{\mathscr{R}}_3(s)$ gives rise to a new spectra function $f_3(\lambda)$,
\begin{eqns}
    f_3(\lambda):=&\frac{1}{d\bar dd_{\chi}}\int_0^{\infty} ds\,\widetilde{\mathscr{R}}_3(s)\\
    =&\frac{1}{d\bar dd_{\chi}}\int ds \Tr\lrp{\rho\Gamma(s)}\Tr(\rho)\Tr\lrp{\Gamma(s)}\\
    =&\frac{1}{d\bar d}\sum_{ij}\frac{\lambda_i+\lambda_j}{2(\lambda_i-\lambda_j)}\ln\frac{\lambda_i}{\lambda_j}.
\end{eqns} 
In analogy with Eq.~\eqref{eq:JD-mixed}, we define the following pair of Hermitian matrices:
\begin{eqns}\label{eq:PQ}
    J:=d_\chi\Tr_{\bar A_2}\left(\{W_R,\chi_{A_2\bar A_2}\}\right)\\
    D:=id_\chi\Tr_{\bar A_2}\left([W_R,\chi_{A_2\bar A_2}]\right)
\end{eqns}

Since $\bar A_2$ is traced out, both $J$ and $D$ are supported on $A_1$, $\bar A_1$, and $A_2$. Their diagrammatic representations are shown in Fig.~\ref{ref-JD-pure}.
\begin{eqns}
J_{bj \alpha;ai \beta}&=\vcenter{\hbox{\includegraphics[height=5em]{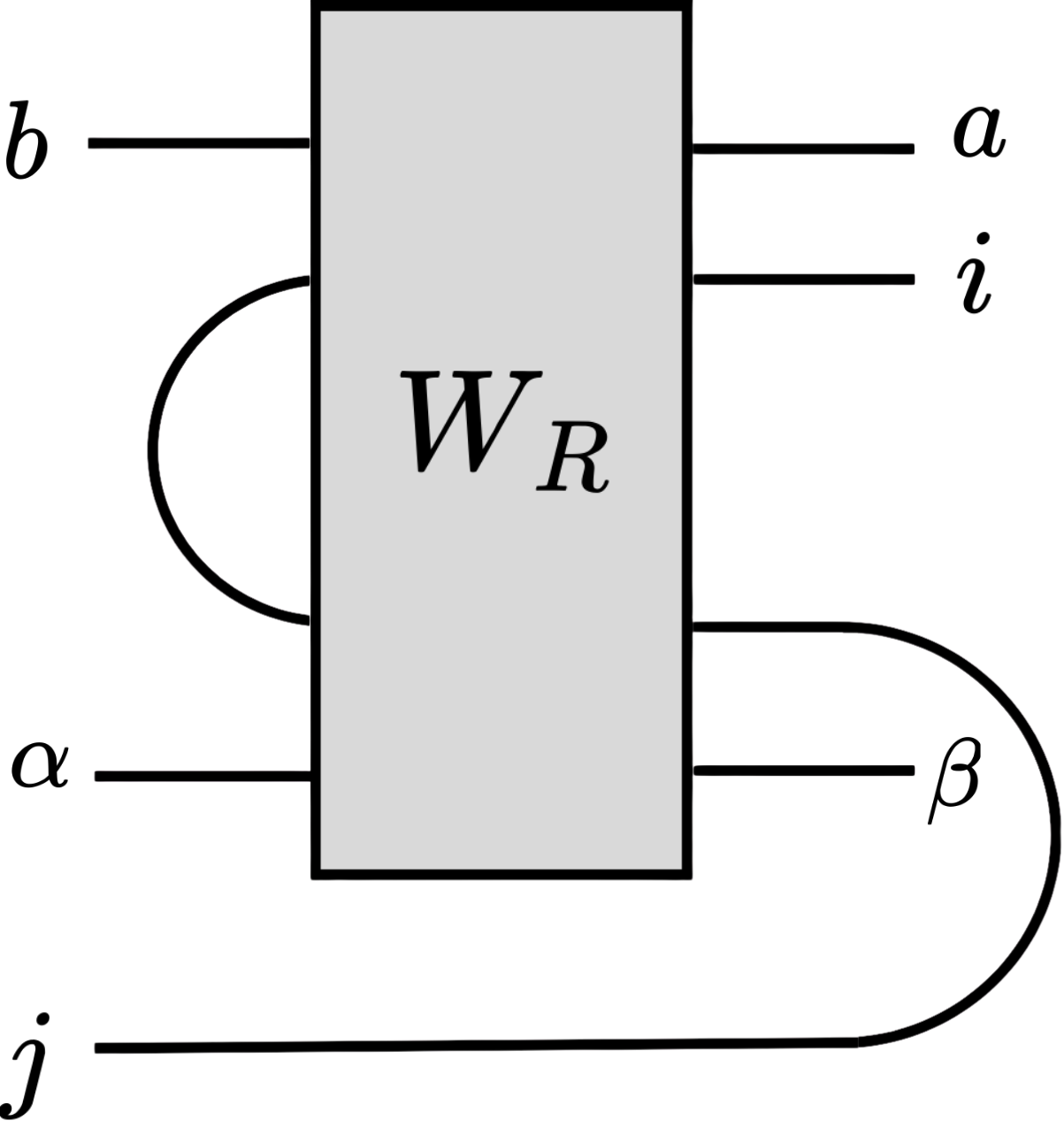}}}\,+\,\vcenter{\hbox{\includegraphics[height=5em]{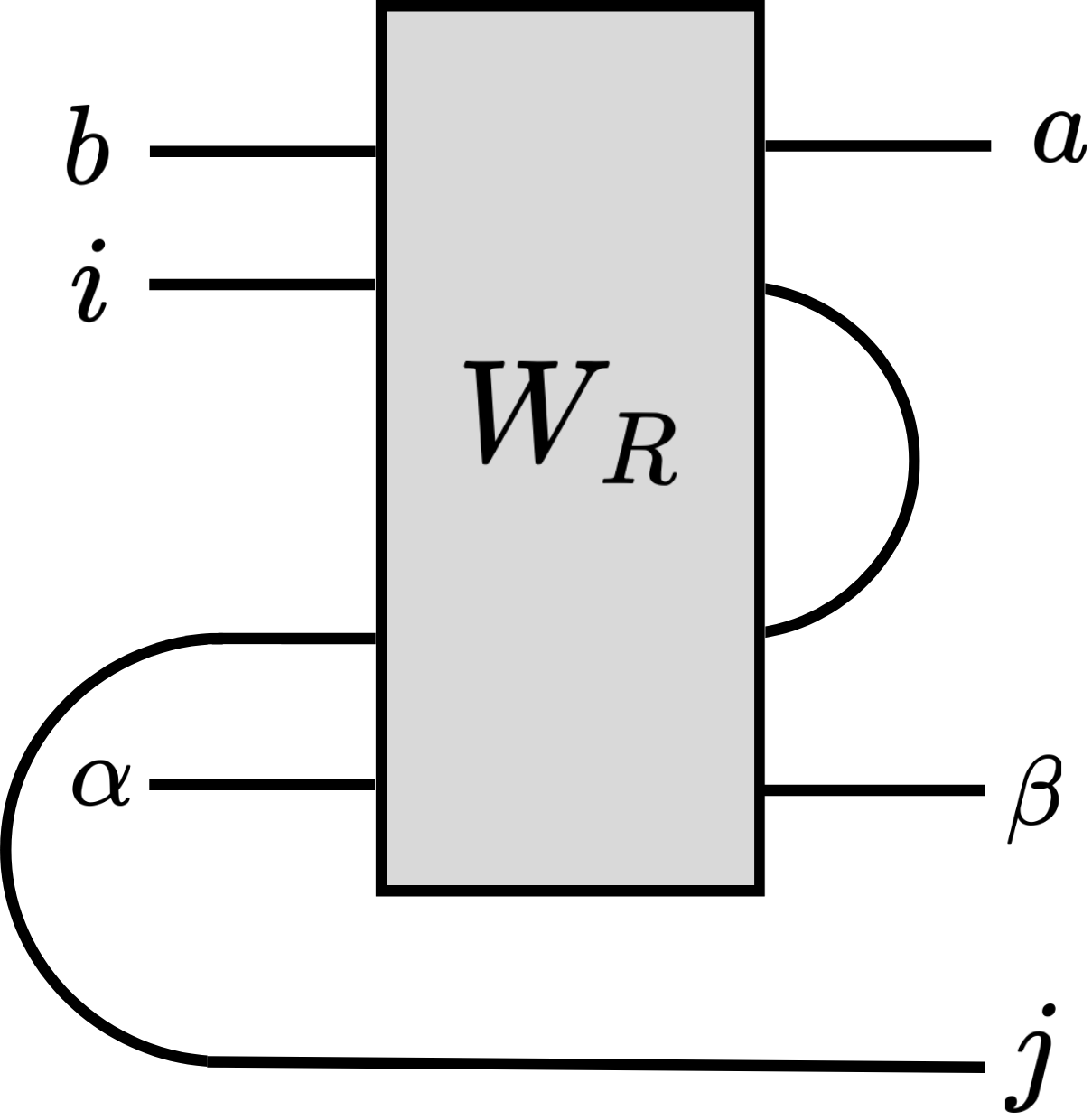}}}\equiv \widetilde {\mathcal{W}}^{(1)}_{b j \alpha;ai \beta}+\widetilde {\mathcal{W}}^{(2)}_{bj \alpha;ai \beta}\\
iD_{bj \alpha;ai \beta}&=\vcenter{\hbox{\includegraphics[height=5em]{theory_draft_images/J2-diag-mixed.png}}}\, - \,\vcenter{\hbox{\includegraphics[height=5em]{theory_draft_images/J1-diag-mixed.png}}}\equiv \widetilde {\mathcal{W}}^{(1)}_{bj;ai}-\widetilde {\mathcal{W}}^{(2)}_{bj;ai}
\label{ref-JD-pure}
\end{eqns}
Here, \(\widetilde {\mathcal{W}}^{(1)}_{bj \alpha;ai \beta }\) and \(\widetilde {\mathcal{W}}^{(2)}_{bj \alpha;ai \beta}\) denote the contributions from the first and second trace diagrams, respectively. Note that these quantities differ from $\mathcal{W}^{(1)}_{bj;ai}$ and $\mathcal{W}^{(2)}_{bj;ai}$ defined in Eq.~\eqref{eq:PQ}, because of the additional contribution from $\bar A_1$. 

Rearranging the terms and completing the squares, we obtain the following expression for the averaged $S_{corr}$:
\begin{eqns}
\langle S_{corr}\rangle=&\langle S_{corr}(J)\rangle+\langle S_{corr}(D)\rangle.
\end{eqns}
where
\begin{align}
    \langle S_{corr}(J)\rangle &=\frac{1}{4d_{\chi}d\bar d (d^2-1)(\bar d^2-1)}\epsilon^2\Bigg[\Big(d\bar d\Tr\big(Tr_{\bar A_1}(J)^2\big)-d\Tr\big(J^2\big)-\bar d\Tr\big(Tr_{A_1 \bar A_1}(J)^2\big)\notag\\&+\Tr\big(Tr_{A_1}(J)^2\big)-\frac{1}{d_\chi}\Big(d\bar d \Tr\big(Tr_{A_2 \bar A_1}(J)^2\big)-d\Tr\big(Tr_{A_2}(J)^2\big)-\bar d \Tr\big(J\big)^2\notag\\&+\Tr\big(Tr_{A_1A_2}(J)^2\big)\Big)
    \Big)df_1(\lambda)+\Bigg(d\Tr\big(Tr_{\bar A_1}(J)^2\big)-d\bar d\Tr\big(J^2\big)+\bar d\Tr\big(Tr_{A_1}(J)^2\big)\notag\\&-\Tr\big(Tr_{A_1 \bar A_1}(J)^2\big)+\frac{1}{d_\chi}\Big(d \Tr\big(Tr_{A_2 \bar A_1}(J)^2\big)-d\bar d\Tr\big(Tr_{A_2}(J)^2\big)- \Tr\big(J\big)^2 \notag\\&+\bar d\Tr\big(Tr_{A_1A_2}(J)^2\big)\Big)
    \Bigg)d\bar df_3(\lambda)
\Bigg]
\end{align}
\begin{align}
\langle S_{corr}(D)\rangle &= \frac{1}{4d_{\chi}d\bar d (d^2-1)(\bar d^2-1)}\epsilon^2\Bigg[\Big(\Tr\big(Tr_{\bar A_1}(D)^2\big)-\bar d \Tr\big(D^2\big)-d\Tr\big(\Tr_{A_1\bar A_1}(D)^2\big)\notag\\
&+d \bar d \Tr\big(Tr_{ A_1}(D^2)\big)\Big)df_1(\lambda)+\Big(-(1+d\bar d) \Tr\big(Tr_{\bar A_1}(D)^2\big)+(d+\bar d)\Tr\big(D^2\big)\notag\\&+(d+\bar d)\Tr\big(\Tr_{A_1\bar A_1}(D)^2\big)-(1+d\bar d) \Tr\big(Tr_{ A_1}(D^2)\big)\Big)f_2(\lambda)\  + \Big(-d\bar d\Tr\big(D^2\big)\notag\\&+d\Tr\big(Tr_{\bar A_1}(D)^2\big)-\Tr\big(\Tr_{A_1\bar A_1}(D)^2\big)+\bar d \Tr\big(Tr_{ A_1}(D^2)\big)\Big)d\bar df_3(\lambda)\Bigg]
\end{align}
We now expand $J$ and $D$ in a product Pauli basis adapted to the tensor-factor decomposition.
Let $\{P_r^{(A_1)}\}_{r=0}^{d^2-1}$ and $\{P_s^{(\bar A_1)}\}_{s=0}^{\bar d^{\,2}-1}$ be Pauli (or generalized-Pauli) bases
on $\mathcal{H}_{A_1}$ and $\mathcal{H}_{\bar A_1}$, respectively, and let
$\{P_t^{(A_2)}\}_{t=0}^{d_\chi^2-1}$ be a Pauli basis on $\mathcal{H}_{A_2}$.
We choose these bases to be orthonormal with respect to the Hilbert--Schmidt inner product,
\begin{equation}
\Tr\!\left(P_\alpha^{(X)}P_\beta^{(X)}\right)=\dim(\mathcal{H}_X)\,\delta_{\alpha\beta},
\qquad X\in\{A_1,\bar A_1,A_2\},
\end{equation}
With this notation,
\begin{eqns}
J&=\sum_{l=0}^{d^2-1}\sum_{m=0}^{\bar d^{\,2}-1}\sum_{n=0}^{d_\chi^2-1}
p_{lmn}(J)\,P_l^{(A_1)}\otimes P_m^{(\bar A_1)}\otimes P_n^{(A_2)},\\
D&=\sum_{l=0}^{d^2-1}\sum_{m=0}^{\bar d^{\,2}-1}\sum_{n=0}^{d_\chi^2-1}
p_{lmn}(D)\,P_l^{(A_1)}\otimes P_m^{(\bar A_1)}\otimes P_n^{(A_2)}.
\label{eq:JDexpansion_pure}
\end{eqns}
Since $\Tr_{A_2}(D)=0$, we have $p_{lm0}(D)=0$ for all $l,n$.

 Using the Pauli-basis coefficients of $D$, the trace expressions appearing in
$\langle S_{\mathrm{corr}}(D)\rangle$ can be written in terms of the following quadratic coefficient sums:
\begin{eqns}\label{eq:Dtermscoe}
A_D&:=\sum_{n>0}p_{00n}^2(D),\qquad
B_D:=\sum_{m,n>0}p_{0mn}^2(D),\qquad
C_D:=\sum_{l,n>0}p_{l0n}^2(D),\qquad
E_D:=\sum_{l,m,n>0}p_{lmn}^2(D).
\end{eqns}

In terms of these coefficients, the corresponding trace expressions become
\begin{align}
&\Tr\!\left(\Tr_{A_1\bar A_1}(D)^2\right)
=d^2\bar d^{\,2}d_{\chi}\sum_{n>0}p_{00n}^2(D)
:=d^2\bar d^{\,2}d_{\chi}\,A_D,\\
&\Tr\!\left(\Tr_{\bar A_1}(D)^2\right)
=d\bar d^{\,2}d_{\chi}\left(\sum_{n>0}p_{00n}^2(D)+\sum_{l,n>0}p_{l0n}^2(D)\right)
:=d\bar d^{\,2}d_{\chi}\left(A_D+C_D\right),\\
&\Tr\!\left(\Tr_{A_1}(D)^2\right)
=d^2\bar d\,d_{\chi}\left(\sum_{n>0}p_{00n}^2(D)+\sum_{m,n>0}p_{0mn}^2(D)\right)
:=d^2\bar d\,d_{\chi}\left(A_D+B_D\right),\\
&\Tr(D^2)
=d\bar d\,d_{\chi}\left(\sum_{n>0}p_{00n}^2(D)+\sum_{m,n>0}p_{0mn}^2(D)+\sum_{l,n>0}p_{l0n}^2(D)+\sum_{l,m,n>0}p_{lmn}^2(D)\right)\\
& \qquad \quad :=d\bar d\,d_{\chi}\left(A_D+B_D+C_D+E_D\right).
\end{align}

Similarly, we define the following quadratic coefficient sums of the Pauli-expansion coefficients $p_{lmn}(J)$:
\begin{eqns}\label{eq:Jtermscoe}
A_J&:=\sum_{n>0}p_{00n}^2(J),\qquad
B_J:=\sum_{m,n>0}p_{0mn}^2(J),\qquad
C_J:=\sum_{l,n>0}p_{l0n}^2(J),\\
E_J:&=\sum_{l,m,n>0}p_{lmn}^2(J),\qquad
F_J:=p_{000}^2(J),\qquad
G_J:=\sum_{m>0}p_{m00}^2(J),\\
H_J:&=\sum_{m>0}p_{0m0}^2(J),\qquad
I_J:=\sum_{l,m>0}p_{lm0}^2(J).
\end{eqns}

In complete analogy, using the Pauli-basis coefficients of $J$, the trace terms appearing in
$\langle S_{\mathrm{corr}}(J)\rangle$ can be written as
\begin{align}
\Tr\!\left(\Tr_{A_1 \bar{A}_1}(J)^2\right)
&=d^2 \bar{d}^2 d_\chi\left(A_J+F_J\right),
\label{eq:TrTrA1Abar1-J2}\\[6pt]
\Tr\!\left(\Tr_{\bar{A}_1}(J)^2\right)
&=d \bar{d}^2 d_\chi\left(A_J+C_J+F_J+G_J\right),
\label{eq:TrTrAbar1-J2}\\[6pt]
\Tr\!\left(\Tr_{A_1}(J)^2\right)
&=d^2 \bar{d} d_\chi\left(A_J+B_J+F_J+H_J\right),
\label{eq:TrTrA1-J2}\\[6pt]
\Tr\!\left(J^2\right)
&=d \bar{d} d_\chi\left(A_J+B_J+C_J+E_J+F_J+G_J+H_J+I_J\right),
\label{eq:Tr-J2}\\[8pt]
\Tr\!\left(\Tr_{A_2 \bar{A}_1}(J)^2\right)
&=d \bar{d}^2 d_\chi^{\,2}\left(F_J+G_J\right),
\label{eq:TrTrA2Abar1-J2}\\[6pt]
\Tr\!\left(\Tr_{A_2}(J)^2\right)
&=d \bar{d} d_\chi^{\,2}\left(F_J+G_J+H_J+I_J\right),
\label{eq:TrTrA2-J2}\\[6pt]
\Tr(J)^2
&=d^2 \bar d^{\,2} d_\chi^{\,2}\,F_J,
\label{eq:TrJ-sq}\\[6pt]
\Tr\!\left(\Tr_{A_1A_2}(J)^2\right)
&=d^2 \bar{d} d_\chi^{\,2}\left(F_J+H_J\right).
\label{eq:TrTrA1A2-J2}
\end{align}

Now we can expand $\langle S_{corr}(D)\rangle$ in terms of these coefficients in a concise form: 
\begin{eqns}
    \langle S_{corr}(D)\rangle=&\frac{\epsilon^2}{2}\left[\frac{1}{2} d^2 \left(\bd^2-1\right)C_Df_2(\lambda)+\frac{1}{2}\left(d^2-1\right)\bd^2 B_D\lrp{f_2(\lambda)-f_1(\lambda)}\right.
    \\
    &\left.+E_D\left(\frac{d^2 \bd^2}{2}f_3(\lambda)-\frac{d(d+\bd)}{2}f_2(\lambda)+\frac{d\bd}{2}f_1(\lambda)\right)+const\right]/\left[(d^2-1)(\bar d^2-1)\right]. \\
\end{eqns}

Similarly $\langle S_{corr}(J)\rangle$ in terms of $J$ is 
\begin{eqns}\label{eq:PAcorrectionpure}
    \langle S_{corr}(J)\rangle =\frac{\epsilon^2}{2}\left[\frac{1}{2}d^2(\bd^2-1)C_Jf_1(\lambda)+E_J\lrp{\frac{d^2\bd^2}{2}f_3(\lambda)-\frac{d^2}{2}f_1(\lambda)}\right]/\left[(d^2-1)(\bar d^2-1)\right].
\end{eqns}

Now define the following coefficients:
\begin{eqns}
    &k_1:=\frac{d^2}{2(d^2-1)}C_J, \qquad \qquad \quad \  k_2:=\frac{d^2}{2(d^2-1)}C_D\\
    &k_3:=\frac{d^2\bd^2}{2(d^2-1)(\bd^2-1)}E_D, \qquad k_4:=\frac{d^2\bd^2}{2(d^2-1)(\bd^2-1)}E_J\\
    &k_5:=\frac{\bd^2}{2(\bd^2-1)}B_D.
\end{eqns}

The full $\langle S_{corr}\rangle$ becomes 
\begin{eqns}
    \langle S_{corr}\rangle=&\frac{\epsilon^2}{2}\left[k_1f_1(\lambda)+k_2f_2(\lambda)+k_3\lrp{f_3(\lambda)-\frac{d+\bd}{d^2\bd}f_2(\lambda)+\frac{1}{d^2}f_1(\lambda)}\right.\\
    &\left.+k_4\lrp{f_3(\lambda)-\frac{1}{d\bd}f_1(\lambda)}+k_5\lrp{f_2(\lambda)-f_1(\lambda)}+const\right].
    \label{PA-entropy_pure_full}
\end{eqns}

Monotonicity of $f_1(\lambda)$, $f_2(\lambda)$ and $f_3(\lambda)$ are discussed in Appendix~\ref{app:f1f2f3}. 

Note that $f_2-f_1$ is not a decreasing function of entanglement. Making the averaged PA entropy not monotonic for generic perturbation $W_R$. However, we show that non-monotonicity only occurs with small probability. 

Assume that $W_R$ is drawn from a Gaussian ensemble. In the product-Pauli expansion, we therefore model the Pauli-basis coefficients as independent random variables with zero mean and a common variance. Under this model, each quadratic block (e.g.\ $B_D$, $C_D$, $E_D$, and the corresponding blocks for $J$) is a sum of many squared coefficients, so its typical magnitude is proportional to the number of Pauli strings contributing to that block. Consequently, typical ratios of such blocks are governed by the ratio of the corresponding term counts, with fluctuations suppressed when the sums involve many terms.

We first estimate the typical ratios $B_D/E_D$ and $C_D/E_D$ by simple counting. $B_D$ collects coefficients with the $A_1$-index fixed to the identity and with the remaining indices restricted to $m>0$ and $n>0$. Hence the sum defining $B_D$ contains $(\bar d^2-1)(d_\chi^2-1)$ terms. Similarly, $C_D$ fixes the $\bar A_1$-index to the identity and sums over $l>0$ and $n>0$, and therefore contains $(d^2-1)(d_\chi^2-1)$ terms. Finally, $E_D$ sums over $l>0$, $m>0$, and $n>0$, and thus contains $(d^2-1)(\bar d^2-1)(d_\chi^2-1)$ terms. Since each contributing squared coefficient has the same typical size, we obtain
\begin{eqns}
\Big\langle \frac{B_D}{E_D}\Big\rangle_{\mathrm{GUE}}
\approx
\frac{(\bar d^2-1)(d_\chi^2-1)}{(d^2-1)(\bar d^2-1)(d_\chi^2-1)}
&=
\frac{1}{d^2-1},
\\
\Big\langle \frac{C_D}{E_D}\Big\rangle_{\mathrm{GUE}}
\approx
\frac{(d^2-1)(d_\chi^2-1)}{(d^2-1)(\bar d^2-1)(d_\chi^2-1)}
&=
\frac{1}{\bar d^2-1}.
\end{eqns}

Moreover, since $B_D$ and $E_D$ are sums of squares of Gaussian random variables, their normalized ratio $B_D/E_D$ follows Fisher-Snedecor \(F\) distribution, 
\begin{equation}
\frac{(B_D/N_1)}{(E_D/N_2)} \;\sim\; F_{N_1,N_2},
\qquad
N_1=(\bar d^{\,2}-1)(d_\chi^{\,2}-1),\quad
N_2=(d^{2}-1)(\bar d^{\,2}-1)(d_\chi^{\,2}-1).
\end{equation}

The non-monotonicity of \(S_{PA}\) requires the atypical fluctuation \(B_D \gtrsim E_D\), i.e.\ \(B_D/E_D \ge 1\). In terms of the \(F\) variable this corresponds to \(F_{N_1,N_2}\ge N_2/N_1=d^{2}-1\), and hence
\begin{eqns}
    \mathbb P\!\left(\frac{B_D}{E_D}\ge 1\right)
=
\mathbb P\!\left(F_{N_1,N_2}\ge \frac{N_2}{N_1}\right)
=
\int_{N_2/N_1}^{\infty}\!dx\, f_F(x;N_1,N_2)
\sim \mathcal O\!\left(e^{-\, d^{2}}\right).
\end{eqns}

We carry out the same counting estimate for the remaining \(J\)-term coefficients. Here  $C_J$ sums over $(l,n)$ with $l>0$ and $n>0$, while $E_J$ sums over $(l,m,n)$ with $l>0$, $m>0$, and $n>0$. Thus $C_J$ contains $(d^2-1)(d_\chi^2-1)$ terms and $E_J$ contains $(d^2-1)(\bar d^2-1)(d_\chi^2-1)$ terms, and hence
\begin{equation}
\Big\langle \frac{C_J}{E_J}\Big\rangle_{\mathrm{GUE}}
\approx
\frac{(d^2-1)(d_\chi^2-1)}{(d^2-1)(\bar d^2-1)(d_\chi^2-1)}
=
\frac{1}{\bar d^2-1}.
\end{equation}

Substituting these typical ratios into the definitions of the coefficients $k_i$ yields the corresponding estimates for the typical relative sizes. Using
\begin{equation}
\frac{k_5}{k_3}=\frac{d^2-1}{d^2}\,\frac{B_D}{E_D},
\qquad
\frac{k_2}{k_3}=\frac{\bar d^2-1}{\bar d^2}\,\frac{C_D}{E_D},
\qquad
\frac{k_1}{k_4}=\frac{\bar d^2-1}{\bar d^2}\,\frac{C_J}{E_J},
\end{equation}
we obtain the corresponding typical GUE estimates by substituting the term-counting results for the block ratios.
\begin{equation}
\Big\langle \frac{k_5}{k_3}\Big\rangle_{\mathrm{GUE}}
=
\frac{d^2-1}{d^2}\Big\langle \frac{B_D}{E_D}\Big\rangle_{\mathrm{GUE}}
=
\frac{d^2-1}{d^2}\cdot\frac{1}{d^2-1}
=
\frac{1}{d^2},
\end{equation}
\begin{equation}
\Big\langle \frac{k_2}{k_3}\Big\rangle_{\mathrm{GUE}}
=
\frac{\bar d^2-1}{\bar d^2}\Big\langle \frac{C_D}{E_D}\Big\rangle_{\mathrm{GUE}}
=
\frac{\bar d^2-1}{\bar d^2}\cdot\frac{1}{\bar d^2-1}
=
\frac{1}{\bar d^2},
\end{equation}
and
\begin{equation}
\Big\langle \frac{k_1}{k_4}\Big\rangle_{\mathrm{GUE}}
=
\frac{\bar d^2-1}{\bar d^2}\Big\langle \frac{C_J}{E_J}\Big\rangle_{\mathrm{GUE}}
=
\frac{\bar d^2-1}{\bar d^2}\cdot\frac{1}{\bar d^2-1}
=
\frac{1}{\bar d^2}.
\end{equation}

So in the large $d$ and $\bar d$ limit,   we find that the typical PA entropy correction can be approximated as 
\begin{equation}
    \langle S_{corr}\rangle = \frac{\epsilon^2}{2} (k_3+k_4)f_3(\lambda)+\mathcal{O}(\frac{1}{d^2})+\mathcal{O}(\frac{1}{\bd^2}).
\end{equation}

\end{proof}
\subsection{Lemma~\ref{lemma:optimization-pure}}\label{app:optimization-pure}
\begin{proof}
The overall encode--noise--recover process defines a quantum channel
\begin{equation}
\mathcal{N}_{R^{(\epsilon)}}:
\mathcal{L}(\mathcal{H}_L)\longrightarrow \mathcal{L}(\mathcal{H}_{A_1\bar A_1}),
\qquad
\mathcal{N}_{R^{(\epsilon)}}(\sigma^{(L)})=\sigma^{(\epsilon)}_{A_1\bar A_1},
\label{pure:N-channel}
\end{equation}
where $\sigma^{(\epsilon)}_{A_1\bar A_1}$ denotes the recovered output state on $\mathcal{H}_{A_1\bar A_1}$. In the present setting, the logical Hilbert space factorizes as
\begin{equation}
\mathcal{H}_L=\mathcal{H}_a\otimes \mathcal{H}_{\bar a}.
\label{pure:HL-factor}
\end{equation}
Accordingly, we introduce two isomorphic reference systems $\mathcal{H}_r\simeq \mathcal{H}_a$ and $\mathcal{H}_{\bar r}\simeq \mathcal{H}_{\bar a}$.

To quantify the performance of the recovery channel, we consider the coherent information  defined in Eq.~\eqref{ref-coherentinfo} as
\begin{equation}\label{pure:eq:coherent}
I_c\!\left(\mathcal{N}_{R^{(\epsilon)}}\right)
=
S\!\left(
\Tr_{r\bar r}\!\left[
(\mathcal{N}_{R^{(\epsilon)}}\otimes I_{r\bar r})
\bigl(|\Phi\rangle\!\langle\Phi|_{a\bar a\, r\bar r}\bigr)
\right]
\right)
-
S\!\left[
(\mathcal{N}_{R^{(\epsilon)}}\otimes I_{r\bar r})
\bigl(|\Phi\rangle\!\langle\Phi|_{a\bar a\, r\bar r}\bigr)
\right].
\end{equation}
where the maximally entangled state in Eq.~\eqref{pure:eq:coherent} factorizes as
\begin{equation}
|\Phi\rangle_{a\bar a\, r\bar r}
=
|\Phi\rangle_{a r}\otimes|\Phi\rangle_{\bar a\,\bar r},
\label{pure:Phi-ar-abar-rbar}
\end{equation}
with
\begin{equation}
|\Phi\rangle_{a r}
=
\frac{1}{\sqrt{d_a}}\sum_{i=1}^{d_a}|i\rangle_{a}\otimes|i\rangle_{r},
\qquad
|\Phi\rangle_{\bar a\,\bar r}
=
\frac{1}{\sqrt{d_{\bar a}}}\sum_{j=1}^{d_{\bar a}}|j\rangle_{\bar a}\otimes|j\rangle_{\bar r},
\label{pure:Phi-factors-ar-abar}
\end{equation}
and $d_a=\dim(\mathcal{H}_a)$ and $d_{\bar a}=\dim(\mathcal{H}_{\bar a})$.
When the logical system consists of $n$ qubits, i.e. $\mathcal{H}_L\simeq (\mathbb{C}^2)^{\otimes n}$, one may view $|\Phi\rangle_{Lr}$ as tensor product of $n$ independent EPR pairs shared between the logical system and the reference:
\begin{equation}\label{eq:bell}
|\Phi\rangle_{Lr}
=
\bigotimes_{k=1}^{n}
\frac{1}{\sqrt{2}}
\left(
|0\rangle_{L_k}|0\rangle_{r_k}
+
|1\rangle_{L_k}|1\rangle_{r_k}
\right),
\end{equation}
so that $n_a$ Bell pairs are shared between $a$ and $r$ together and $n_{\bar a}$ Bell pairs are shared between $\bar a$ and $\bar r$ (with $n=n_a+n_{\bar a}$ for qubits).
\begin{figure}[H]
    \centering
    \includegraphics[width=0.6\linewidth]{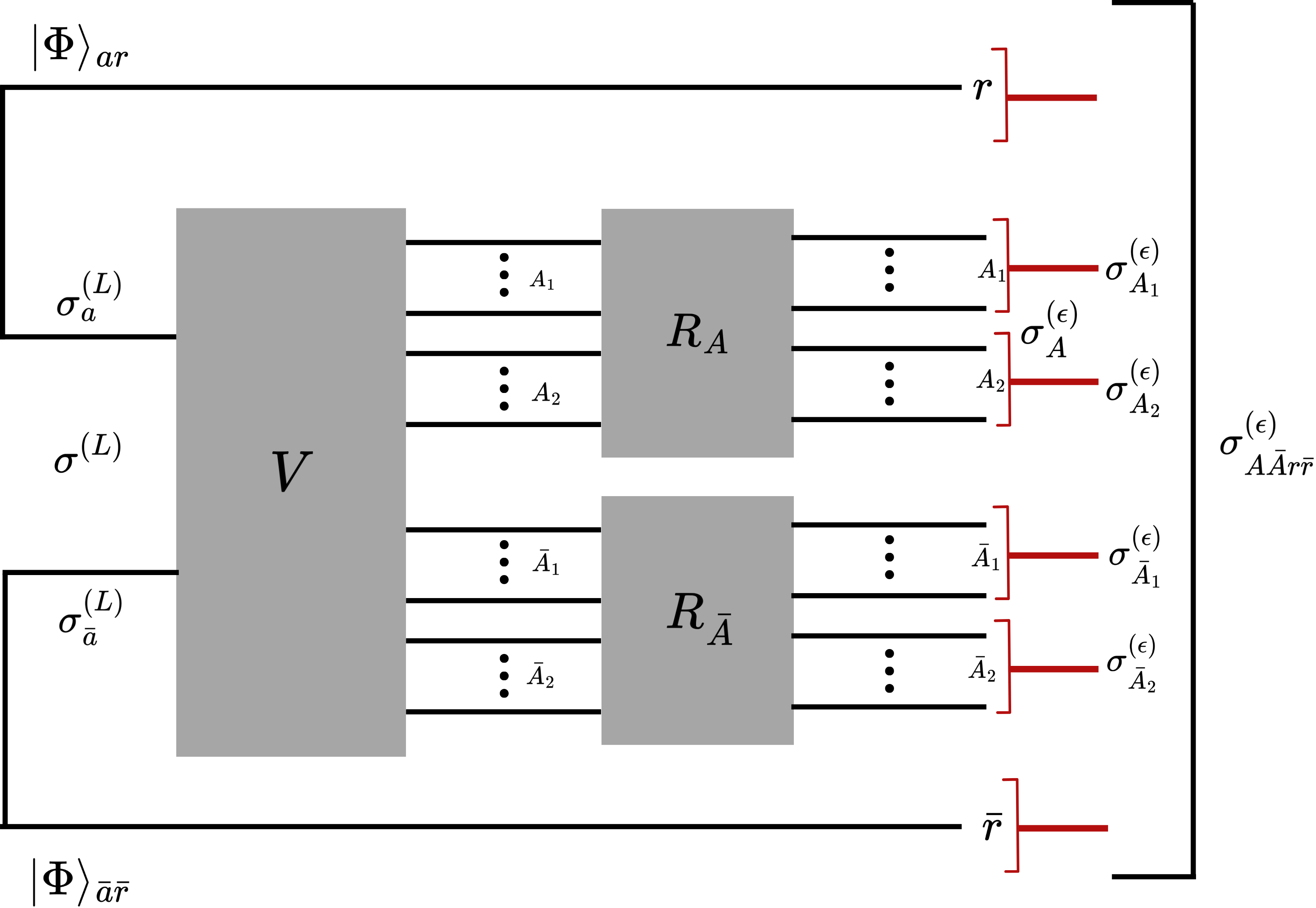}
    \caption{Channel picture for optimizing recovery via coherent information in the factorized logical setting $\mathcal{H}_L=\mathcal{H}_a\otimes \mathcal{H}_{\bar a}$. The logical input is purified by a product of maximally entangled states $|\Phi\rangle_{a r}\otimes|\Phi\rangle_{\bar a \bar r}$, with $r\simeq a$ and $\bar r\simeq \bar a$.The joint logical state $\sigma^{(L)}_{a\bar a}$ is encoded by the isometry $V$ into $\mathcal{H}_A\otimes\mathcal{H}_{\bar A}$, followed by recovery maps on each boundary region. This induces an effective channel $\mathcal{N}_{R^{(\epsilon)}}:\mathcal{L}(\mathcal{H}_L)\to \mathcal{L}(\mathcal{H}_{A_1\bar A_1})$, with output $\sigma^{(\epsilon)}_{A_1\bar A_1}=\mathcal{N}_{R^{(\epsilon)}}(\sigma^{(L)}_{a\bar a})$. The coherent information is evaluated on the joint outputs $\sigma^{(\epsilon)}_{A_1 r}$ and $\sigma^{(\epsilon)}_{\bar A_1 \bar r}$ (and the full state $\sigma^{(\epsilon)}_{A\bar A r\bar r}$), and maximized over recovery channels to obtain the optimal recovery $R^\ast$.}
    \label{fig:placeholder}
\end{figure}

After encoding, perturbation, and recovery, the joint state on the recovered subsystem $A_1\bar A_1$ and the reference $r\bar r$ is
\begin{equation}
\sigma^{(\epsilon)}_{A_1\bar A_1 r\bar r}
=
(\mathcal{N}_{R^{(\epsilon)}}\otimes I_{r\bar r})
\bigl(|\Phi\rangle\!\langle\Phi|_{a\bar a\, r\bar r}\bigr)
=
\Tr_{A_2 \bar A_2}\!\left[
e^{i\epsilon W_R}
\bigl(\sigma^{(0)}_{A_1\bar A_1 r\bar r}\otimes \chi_{A_2 \bar A_2}\bigr)
e^{-i\epsilon W_R}
\right],
\label{eq:sigmaAR}
\end{equation}
where $\sigma^{(0)}_{A_1\bar A_1 r\bar r}$ is the density matrix of the input maximally mixed state
$|\Phi\rangle_{a\bar a\, r\bar r}$. The perturbation acts trivially on the reference systems, so we adopt the shorthand
\begin{equation}
W_R \otimes I_r \otimes I_{\bar r}\;\equiv\; W_R .
\label{pure:WR-shorthand}
\end{equation}

Tracing out the reference in Eq. \eqref{eq:sigmaAR} gives the reduced output state on $A_1\bar A_1$,
\begin{equation}
\sigma^{(\epsilon)}_{A_1\bar A_1}
= \Tr_{r\bar r}\,\sigma^{(\epsilon)}_{A_1\bar A_1 r\bar r}.
\label{pure:sigma-reduced}
\end{equation}

In this notation, the coherent information is 
\begin{equation}
I_c\!\left(\mathcal{N}_{R^{(\epsilon)}}\right)
=
S\!\left(\sigma^{(\epsilon)}_{A_1\bar A_1}\right)
-
S\!\left(\sigma^{(\epsilon)}_{A_1\bar A_1 r\bar r}\right).
\label{pure:Ic}
\end{equation}
Maximizing $I_c(\mathcal{N}_{R^{(\epsilon)}})$ over the choice of recovery unitary defines the optimal recovery, denoted by $R^\ast$, which preserves the largest possible amount of recoverable bulk information in the decoded outputs $A_1$ and $\bar A_1$.

Let's first calculate entropy $S(\sigma^{(\epsilon)}_{A_1\bar A_1})$, using the perturbative expansion
\begin{equation}
\sigma^{(\epsilon)}_{A_1\bar A_1}
=
\sigma^{(0)}_{A_1\bar A_1}
+\delta\sigma^{(\epsilon)}_{A_1\bar A_1}
=
\sigma^{(0)}_{A_1\bar A_1}
+\epsilon\,\delta^{(1)}\sigma_{A_1\bar A_1}
+\frac{\epsilon^2}{2}\,\delta^{(2)}\sigma_{A_1\bar A_1}
+O(\epsilon^3).
\label{pure:eq:sigma-expansion}
\end{equation}
where
\begin{align}
    \sigma^{(0)}_{A_1\bar A_1}(U,V)
&\equiv (U \otimes V)\,\sigma^{(0)}_{A_1\bar A_1}\,(U^\dagger \otimes V^\dagger)\\
\delta^{(1)}\sigma_{A_1\bar A_1}(U,V)
&= i\,\Tr_{A_2\bar{A_2}r\bar r}
\Big(
\big[W_R,\ \sigma^{(0)}_{A_1\bar A_1r\bar r}(U,V)\otimes \chi_{A_2 \bar A_2}\big]
\Big),\\[4pt]
\delta^{(2)}\sigma_{A_1 \bar A_1}(U,V)
&= -\Tr_{A_2\bar A_2r\bar r}
\Big(
\frac{1}{2}\big\{W_R^2,\ \sigma^{(0)}_{A_1\bar A_1r\bar r}(U,V)\otimes \chi_{A_2\bar A}\big\} \notag
\\ & \qquad \qquad \qquad- W_R\,
\big(\sigma^{(0)}_{A_1\bar A_1r\bar r}(U,V)\otimes \chi_{A_2\bar A_2}\big)\,
W_R
\Big).
\end{align}

 Using the logarithm expansion derived in Appendix.~\ref{logexpansion}, we find
\begin{align}
S\big(\sigma^{(\epsilon)}_{A_1\bar A_1}(U,V)\big)
&=
S\big(\sigma^{(0)}_{A_1\bar A_1}(U,V)\big)
-\epsilon\,\operatorname{Tr}\Big(
\delta^{(1)}\sigma_{A_1\bar A_1}(U,V)\,
\ln\sigma^{(0)}_{A_1\bar A_1}(U,V)
\notag\\
&\qquad\qquad
+\,
\sigma^{(0)}_{A_1\bar A_1}(U,V)\,
D_{\ln}\big(\sigma^{(0)}_{A_1\bar A_1}(U,V)\big)
\big[\delta^{(1)}\sigma_{A_1\bar A_1}(U,V)\big]
\Big)
\notag\\
&\qquad \qquad
-\epsilon^2\,\operatorname{Tr}\Big(
\delta^{(2)}\sigma_{A_1\bar A_1}(U,V)\,
\ln\sigma^{(0)}_{A_1\bar A_1}(U,V)
\notag\\
&\qquad \qquad
+\frac{1}{2}\,
\delta^{(1)}\sigma_{A_1\bar A_1}(U,V)\,
D_{\ln}\big(\sigma^{(0)}_{A_1\bar A_1}(U,V)\big)
\big[\delta^{(1)}\sigma_{A_1\bar A_1}(U,V)\big]
\notag\\
&\qquad \qquad
+\,
\sigma^{(0)}_{A_1\bar A_1}(U,V)\,
D_{\ln}\big(\sigma^{(0)}_{A_1\bar A_1}(U,V)\big)
\big[\delta^{(2)}\sigma_{A_1\bar A_1}(U,V)\big]
\Big)+O(\epsilon^3).
\label{eq:Sbulk-pure-expanded}
\end{align}

Similar as in Sec.~\ref{app:optimization:mixed}, this entropy becomes 
\begin{equation}\label{eq:trivialbulkentropy}
S\big(\sigma^{(\epsilon)}_{A_1\bar A_1}\big)=S\big(\sigma^{(0)}_{A_1\bar A_1}\big)=\ln(d \bar d).
\end{equation}
and correction vanishes. 

We next compute the entropy of the joint recovered-and-reference state,
$S(\sigma^{(\epsilon)}_{A_1\bar A_1 r\bar r})$. $\sigma^{(\epsilon)}_{A_1\bar A_1 r\bar r}$ admits the perturbative expansion
\begin{equation}
\sigma^{(\epsilon)}_{A_1\bar A_1 r\bar r}
=
\sigma^{(0)}_{A_1\bar A_1 r\bar r}
+\delta\sigma^{(\epsilon)}_{A_1\bar A_1 r\bar r}
=
\sigma^{(0)}_{A_1\bar A_1 r\bar r}
+\epsilon\,\delta^{(1)}\sigma_{A_1\bar A_1 r\bar r}
+\frac{\epsilon^2}{2}\,\delta^{(2)}\sigma_{A_1\bar A_1 r\bar r}
+O(\epsilon^3).
\label{eq:sigma-expansion-joint}
\end{equation}
where
\begin{align}
\sigma^{(0)}_{A_1\bar A_1 r\bar r}(U,V)
&=
\operatorname{Tr}_{A_2\bar A_2}\!\left[
\sigma^{(0)}_{A\bar A\, r\bar r}(U,V)
\right]
\equiv
(U\otimes V)\,
\sigma^{(0)}_{A_1\bar A_1 r\bar r}\,
(U^\dagger\otimes V^\dagger),
\\
\delta^{(1)}\sigma_{A_1\bar A_1 r\bar r}(U,V)
&=
i\,\operatorname{Tr}_{A_2\bar A_2}\Big(
\big[W_R,\ \sigma^{(0)}_{A_1\bar A_1 r\bar r}(U,V)\otimes \chi_{A_2 \bar A_2}\big]
\Big),
\\[4pt]
\delta^{(2)}\sigma_{A_1\bar A_1 r\bar r}(U,V)
&=
-\operatorname{Tr}_{A_2\bar A_2}\Big(
\frac{1}{2}\big\{W_R^2,\ \sigma^{(0)}_{A_1\bar A_1 r\bar r}(U,V)\otimes \chi_{A_2\bar A_2}\big\}
\notag\\
&\qquad\qquad\qquad\qquad
- W_R\,
\big(\sigma^{(0)}_{A_1\bar A_1 r\bar r}(U,V)\otimes \chi_{A_2\bar A_2}\big)\,
W_R
\Big).
\end{align}
The derivation proceeds exactly as before, with the only change being that the entropy expansion is now carried out on the enlarged Hilbert space
$\mathcal{H}_{A_1}\otimes\mathcal{H}_{\bar A_1}\otimes\mathcal{H}_{r}\otimes\mathcal{H}_{\bar r}$.
Accordingly, every instance of the zeroth-order reduced state in the previous computation is replaced by
$\sigma^{(0)}_{A_1\bar A_1 r\bar r}$, and the perturbation generator is taken to act trivially on the reference systems,
$W_R \mapsto W_R\otimes I_{r}\otimes I_{\bar r}$.
With these replacements in place, we obtain
\begin{align}
S\big(\sigma^{(\epsilon)}_{A_1\bar A_1 r\bar r}(U,V)\big)
&=
S\big(\sigma^{(0)}_{A_1\bar A_1 r\bar r}(U,V)\big)
-\epsilon\,\operatorname{Tr}\Big(
\delta^{(1)}\sigma_{A_1\bar A_1 r\bar r}(U,V)\,
\ln\sigma^{(0)}_{A_1\bar A_1 r\bar r}(U,V)
\notag\\
&\qquad\qquad
+\,
\sigma^{(0)}_{A_1\bar A_1 r\bar r}(U,V)\,
D_{\ln}\big(\sigma^{(0)}_{A_1\bar A_1 r\bar r}(U,V)\big)
\big[\delta^{(1)}\sigma_{A_1\bar A_1 r\bar r}(U,V)\big]
\Big)
\notag\\
&\qquad \qquad
-\epsilon^2\,\operatorname{Tr}\Big(
\delta^{(2)}\sigma_{A_1\bar A_1 r\bar r}(U,V)\,
\ln\sigma^{(0)}_{A_1\bar A_1 r\bar r}(U,V)
\notag\\
&\qquad \qquad
+\frac{1}{2}\,
\delta^{(1)}\sigma_{A_1\bar A_1 r\bar r}(U,V)\,
D_{\ln}\big(\sigma^{(0)}_{A_1\bar A_1 r\bar r}(U,V)\big)
\big[\delta^{(1)}\sigma_{A_1\bar A_1 r\bar r}(U,V)\big]
\notag\\
&\qquad \qquad
+\,
\sigma^{(0)}_{A_1\bar A_1 r\bar r}(U,V)\,
D_{\ln}\big(\sigma^{(0)}_{A_1\bar A_1 r\bar r}(U,V)\big)
\big[\delta^{(2)}\sigma_{A_1\bar A_1 r\bar r}(U,V)\big]
\Big)
\notag\\
&\qquad \qquad
+O(\epsilon^3).
\label{eq:Sjoint-pure-expanded}
\end{align}

A technical point arises at this stage. The Fr\'echet-derivative term
\[
D_{\ln}\big(\sigma^{(0)}_{A_1\bar A_1 r\bar r}\big)\Big[\delta^{(1)}\sigma_{A_1\bar A_1 r\bar r}\Big]
\]
appearing in $\langle S(\sigma_{A_1 \tilde{A}_1 r \tilde{r}}^{(e)})\rangle$ given in Eq~ \eqref{eq:SRA-pure-expanded}is ill-defined because $\sigma^{(0)}_{A_1\bar A_1 r\bar r}$ is a pure state and hence not full rank. In particular, it has zero eigenvalues, so the resolvent
\((\sigma^{(0)}_{A_1\bar A_1 r\bar r}+sI)^{-1}\) diverges as \(s\to 0\).

We handle this by introducing a full-rank regulator at the level of the Bell-pair factorization \eqref{eq:bell}.
Specifically, for each Bell pair state, denoted as $\ket{\Phi_+}$, we replace the projector $|\Phi^+\rangle\langle\Phi^+|$ by the following state
\begin{equation}
\bigl(|\Phi^+\rangle\langle\Phi^+|\bigr)^{\mathrm{reg}}
=
(1-3\Delta)\,|\Phi^+\rangle\langle\Phi^+|
+
\Delta\Bigl(|01\rangle\langle 01|+|10\rangle\langle 10|\Bigr)
+
\Delta\,|\Phi^-\rangle\langle\Phi^-|,
\label{pure:bell-regulator}
\end{equation}
where $|\Phi^\pm\rangle=(|00\rangle\pm|11\rangle)/\sqrt{2}$.
The regulator is introduced at the level of the input purification. 
The regulated output state of the channel is then obtained by pushing this input state through the channel $\mathcal{N}_{R^{(\epsilon)}}$
\begin{equation}
\sigma^{(\epsilon)}_{A_1\bar A_1 r\bar r}(\Delta)
:=
\bigl(\mathcal{N}_{R^{(\epsilon)}}\otimes I_{r\bar r}\bigr)\!\left(|\Phi\rangle\langle\Phi|^{\mathrm{reg}}_{a\bar a\, r\bar r}\right).
\label{pure:sigmaeps-reg-defined}
\end{equation}
We evaluate $S\big(\sigma^{(\epsilon)}_{A_1\bar A_1 r\bar r}\big)$ by first performing the perturbative expansion at fixed $\Delta>0$ and then taking the limit $\Delta\to 0$ at the end of the calculation.

Note that Eq.~\eqref{eq:trivialbulkentropy} is still valid since
$\sigma^{(0)}_{A_1}(\Delta)=I_d/d$ and
$\sigma^{(0)}_{\bar A_1}(\Delta)=I_{\bar d}/\bar d$,
which remain maximally mixed even in the presence of the regulator.

Then we compute the Haar-averaged entropy $\big\langle S\big(\sigma^{(\epsilon)}_{A_1\bar A_1 r\bar r}(U,V)\big)\big\rangle$. Averaging the entropy expansion over $U$ and $V$ with respect to the Haar measure gives
\begin{align}
\Big\langle S\big(\sigma^{(\epsilon)}_{A_1\bar A_1 r\bar r}\big)\Big\rangle
&=
S\big(\sigma^{(0)}_{A_1\bar A_1 r\bar r}\big)
-\epsilon^2\int dU\,dV\;
\operatorname{Tr}\Big(
\delta^{(2)}\sigma_{A_1\bar A_1 r\bar r}(U,V)\,
\ln\sigma^{(0)}_{A_1\bar A_1 r\bar r}(U,V)
\notag\\
&\qquad\qquad\qquad
+\frac{1}{2}\,
\delta^{(1)}\sigma_{A_1\bar A_1 r\bar r}(U,V)\,
D_{\ln}\big(\sigma^{(0)}_{A_1\bar A_1 r\bar r}(U,V)\big)
\big[\delta^{(1)}\sigma_{A_1\bar A_1 r\bar r}(U,V)\big]
\Big)\notag\\
&\qquad\qquad\qquad
+O(\epsilon^3).
\label{eq:SRA-pure-expanded}
\end{align}
We now focus on the first $\epsilon^2$-order term in Eq.~\eqref{eq:SRA-pure-expanded}. This term alone is sufficient to determine the PA entropy when the channel-capacity optimization condition is imposed. The key point is that the PA entropy in Eq.~\eqref{eq:PAcorrectionpure} depends only on the Pauli coefficients $p_{rst}(J)$ (defined in Eq.~\eqref{eq:Jtermscoe}) with $t>0$. By contrast, after  carrying out the average over $U$ and $V$, one finds that the second $\epsilon^2$-order term in Eq.~\eqref{eq:SRA-pure-expanded} depends only on the sector $p_{rs0}(J)$ with $t=0$. Furthermore, the first term which we compute below (see Eq.~\eqref{eq:BEHIBE}),   receives decoupled contribution from the $t>0$ and $t=0$ sectors, so the optimization conditions for $t>0$ coefficients are determined independently of the $t=0$ ones. It is therefore enough to optimize the first term to fix the optimal $t>0$ components and hence the PA entropy.   
We denote this contribution by $\mathscr{S}_2^{(\epsilon)}$, defined as
\begin{equation}
\mathscr{S}_2^{(\epsilon)}
:=
-\epsilon^{2}\int dU\,dV\;
\operatorname{Tr}\left[
\delta^{(2)}\sigma_{A_1\bar A_1r \bar r}(U,V)\,
\ln\sigma^{(0)}_{A_1\bar A_1 r \bar r}(U,V)
\right].
\label{eq:ScrS1-eps-pure}
\end{equation}

Diagrammatically, $\mathscr{S}_2^{(\epsilon)}$ takes the following form
\begin{align}
\mathscr{S}_2^{(\epsilon)}
&= \frac{\epsilon^2}{d_\chi}\int dU \left( \quad
\vcenter{\hbox{\includegraphics[height=11.4em]{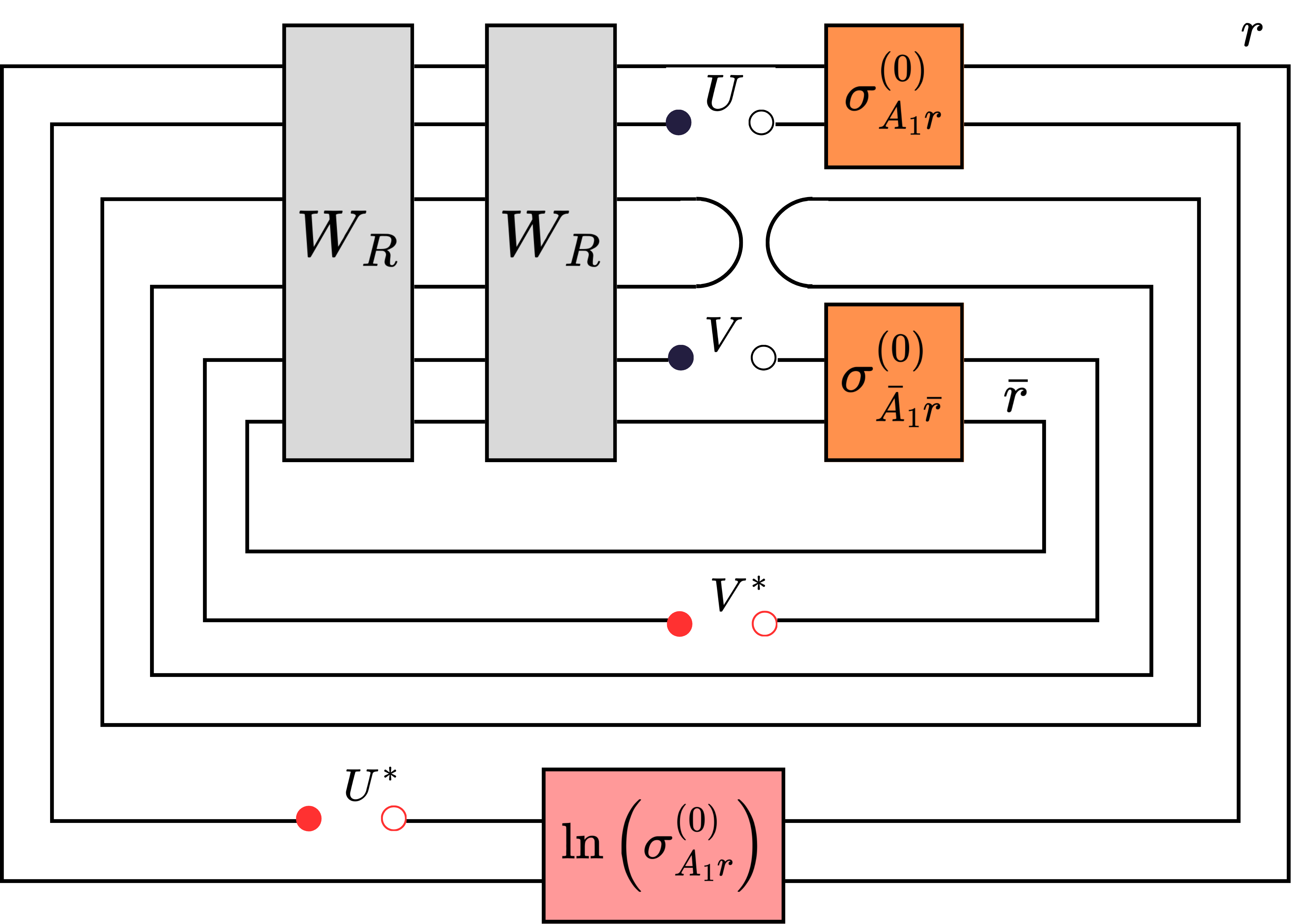}}}
\right.\notag\\
&\qquad\qquad \qquad
+ \vcenter{\hbox{\includegraphics[height=11.4em]{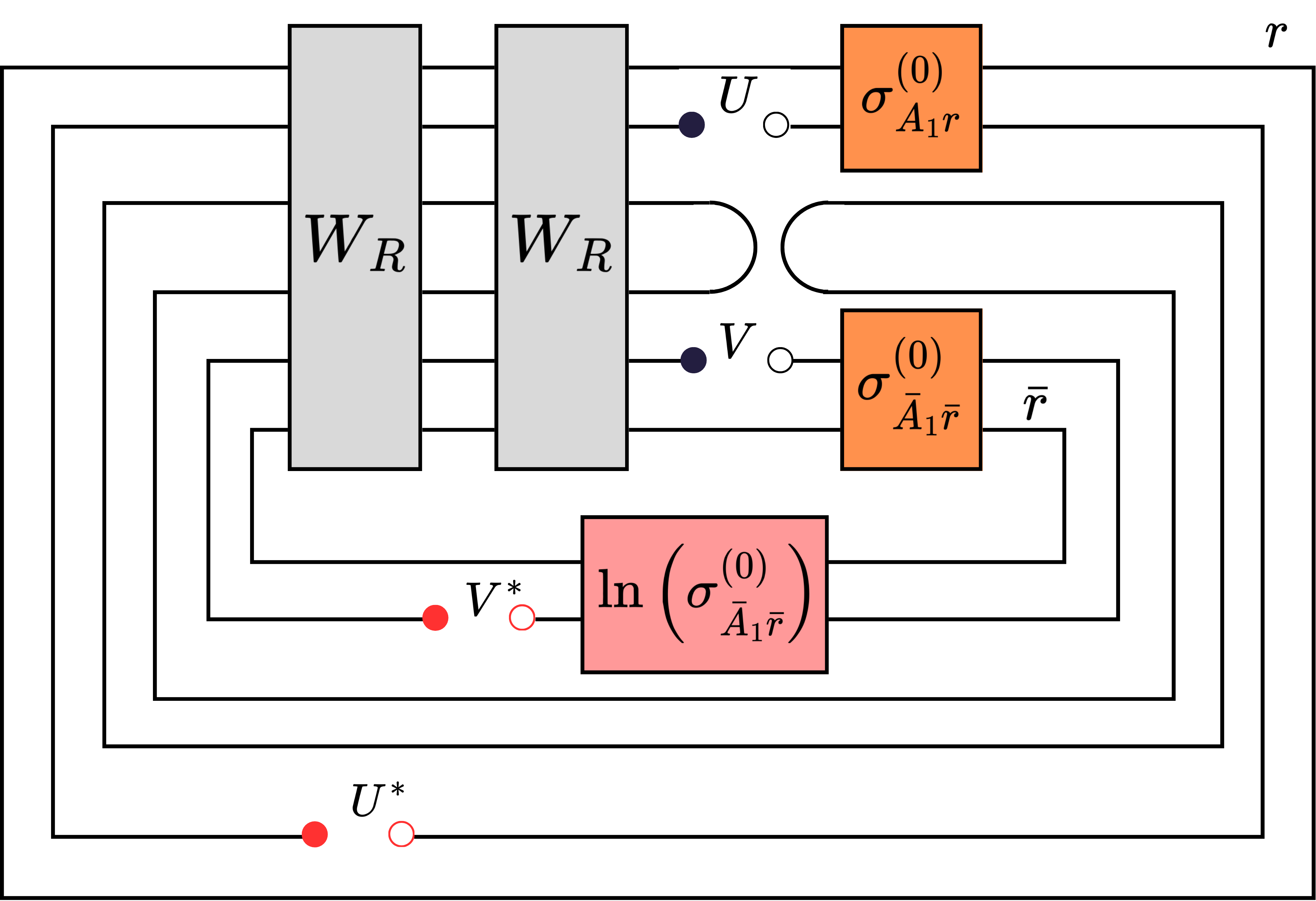}}}
\notag\\
&\qquad\qquad \qquad
- \vcenter{\hbox{\includegraphics[height=11.4em]{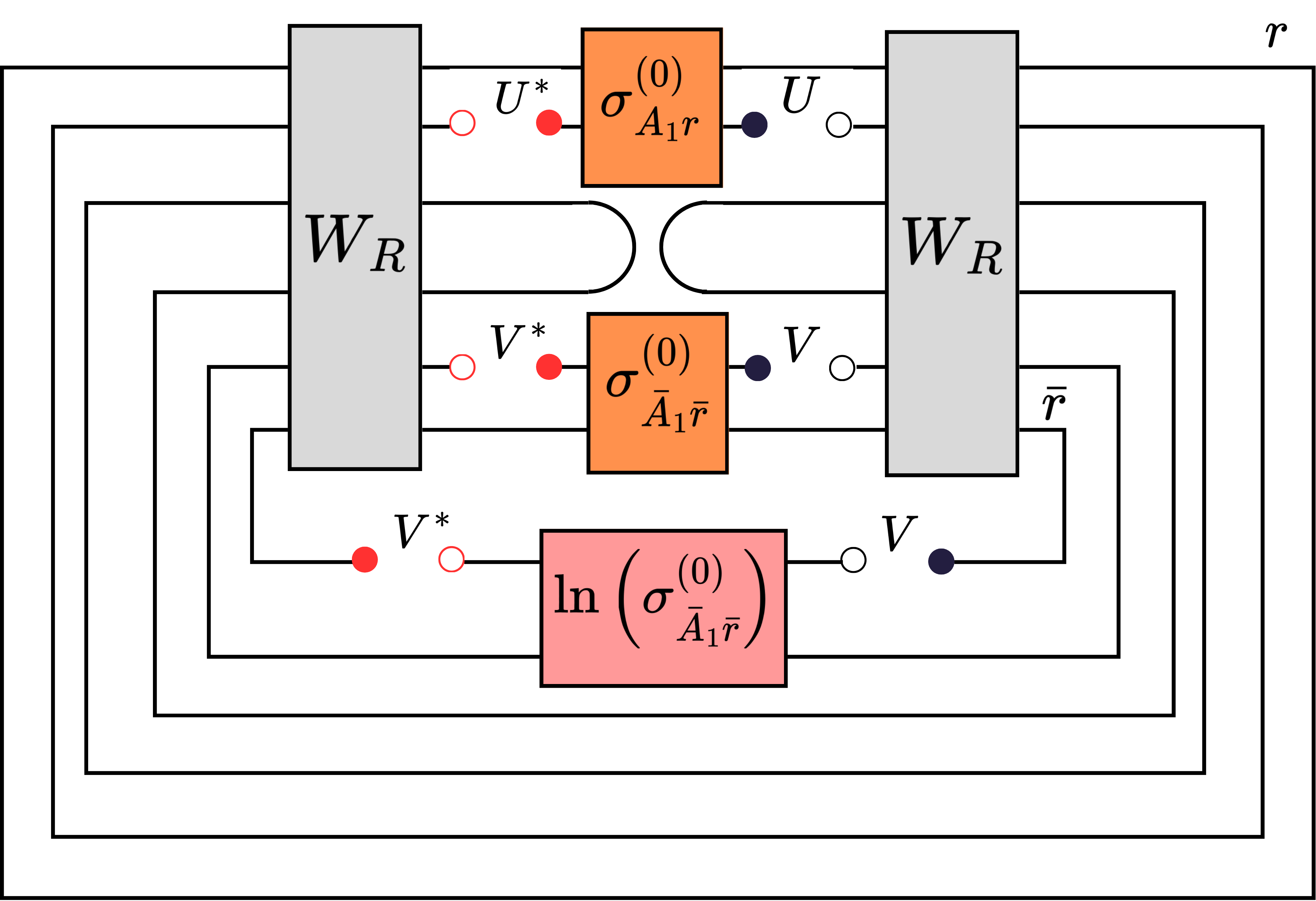}}}
\notag\\
&\left.\qquad\qquad \qquad
- \vcenter{\hbox{\includegraphics[height=11.4em]{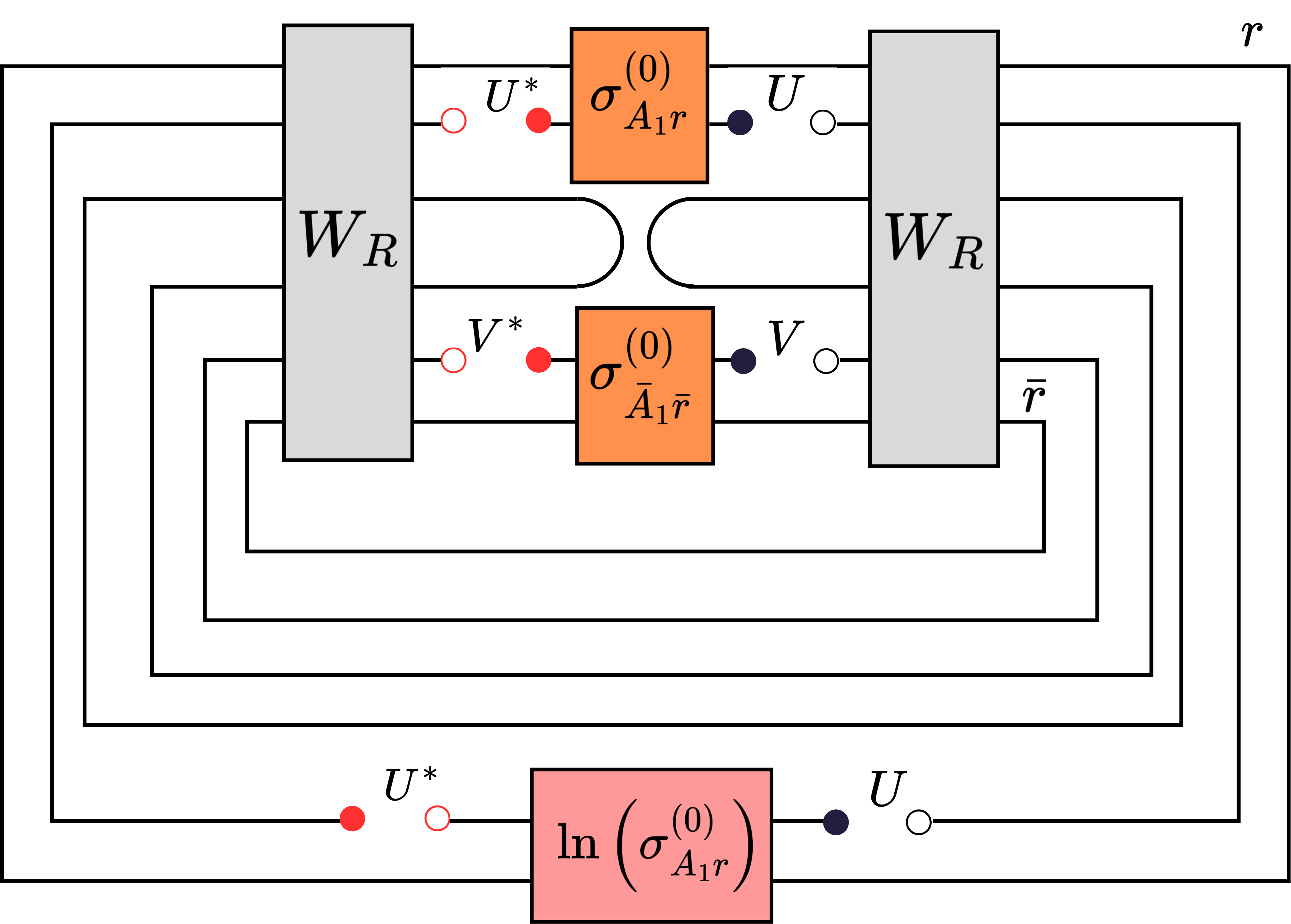}}}
\quad \right).
\end{align}
In the diagrams above, we have used the identity for the logarithm of a tensor-product state 
\begin{equation}
\ln\bigl(\sigma^{(0)}_{A_1}\otimes\sigma^{(0)}_{\bar A_1}\bigr)
=
\ln\sigma^{(0)}_{A_1}\otimes I_{\bar A_1}
+
I_{A_1}\otimes \ln\sigma^{(0)}_{\bar A_1}.
\label{eq:log-tensor-sum-0}
\end{equation}
After averaging over the Haar measures of $U$ and $V$ and substituting the regulated marginals
$\sigma^{(0)}_{A_1 r}(\Delta)$ and $\sigma^{(0)}_{\bar A_1 \bar r}(\Delta)$, the expression simplifies to
\begin{align}
\mathscr{S}_2^{(\epsilon)}&=\frac{3(1-4\Delta)\,\epsilon^{2}}{4 d_\chi}\,
\log\Bigg(\frac{1-3\Delta}{\Delta}\Bigg)
\Bigg[\frac{\widetilde{\mathcal W}_{9}(\log d+\log \bar{d})}{d \bar{d}}  +
\frac{\bigl(d\,\widetilde{\mathcal W}_{9}-\widetilde{\mathcal W}_{11}\bigr)\,\log d}
{\bar d\,(d^{2}-1)}
\notag \\ &\hspace{6cm} +
\frac{\bigl(\bar d\,\widetilde{\mathcal W}_{9}-\widetilde{\mathcal W}_{10}\bigr)\,\log \bar d}
{d\,(\bar d^{2}-1)}
\Bigg]
\end{align}
Using Table~\eqref{W-tilde-R-tilde-diagrams} together with
Eqs.~\eqref{eq:TrTrAbar1-J2}, \eqref{eq:TrTrA1-J2}, and \eqref{eq:Tr-J2}, we obtain
\begin{align}
    \widetilde{\mathcal W}_{9}&=\frac{1}{4}\Tr(J^2+D^2)=\frac{d\bar d d_\chi}{4}\Big(A_J+B_J+C_J+E_J+F_J+G_J+H_J+I_J \notag \\ & \qquad \qquad \qquad \qquad \qquad \qquad +A_D+B_D+C_D+E_D\Big)\\
    \widetilde{\mathcal W}_{10}&=\frac{1}{4}\Tr(\Tr_{\bar A_1}(J)^2+\Tr_{\bar A_1}(D)^2)=\frac{d{\bar d}^2 d_\chi}{4}\Big(A_J+C_J+F_J+G_J+A_D+C_D\Big)\\
     \widetilde{\mathcal W}_{11}&=\frac{1}{4}\Tr(\Tr_{A_1}(J)^2+\Tr_{A_1}(D)^2)
    =\frac{d^2\bar d d_\chi}{4}\Big(A_J+B_J+F_J+H_J+A_D+B_D\Big)
\end{align}
Now we get 
\begin{align}\label{eq:BEHIBE}
\mathscr{S}_2^{(\epsilon)}
&=
\frac{3\epsilon^{2}}{16}(1-4\Delta)\,
\log\Biggl(\frac{1-3\Delta}{\Delta}\Biggr)
\Biggl[
\frac{d^{2}\log d}{d^{2}-1}\,\bigl(C_J+E_J+G_J+I_J+C_D+E_D\bigr)
\notag\\  & +
\frac{\bar d^{2}\log \bar d}{\bar d^{2}-1}\,\bigl(B_J+E_J+H_J+I_J+B_D+E_D\bigr)+\Big(\log d+\log \bar{d}\Big)\Big(A_J+B_J+C_J\notag \\ & +E_J+F_J+G_J+H_J+I_J  +A_D+B_D+C_D+E_D\Big)
\Biggr].
\end{align}
Since correction of $S\big(\sigma^{(\epsilon)}_{A_1\bar A_1}\big)$ vanishes,  the relevant part of the coherent information needed to determine the PA entropy is given by
$-\mathscr{S}_2^{(\epsilon)}$. We therefore focus on analyzing and minimizing $\mathscr{S}_2^{(\epsilon)}$. 

Since $\log\!\big((1-3\Delta)/\Delta\big)$ diverges as $\Delta \to 0^+$, the coherent-information correction becomes singular in the regulator-removal limit. We therefore keep $\Delta$ fixed and use the coherent information as a regulated comparison functional for different recovery choices. Maximizing the coherent information over $R_A\otimes R_{\bar A}$ then provides a well-defined notion of optimal recovery at fixed $\Delta$, as discussed in Sec.~\eqref{app:optimization:mixed}.

In our expression for $\mathscr{S}_2^{(\epsilon)}$, the recovery dependence enters only through the  coefficients appearing in the square brackets in Eq.~\eqref{eq:BEHIBE}. All remaining prefactors are independent of the recovery choice.

Among the coefficients entering the PA entropy \eqref{PA-entropy_pure_full} through $\mathscr{S}_2^{(\epsilon)}$, namely $C_J$, $E_J$, $B_D$, $C_D$, and $E_D$, only $C_J$ is affected by the boundary-unitary freedom in the recovery.
To see this, let us examine how the boundary-unitary freedom acts on the operators
$J$ and $D$.  A redefinition of the recovery of the form
\begin{equation}
R_A^{(\epsilon)} \;\longrightarrow\; e^{i\epsilon O_A}\,R_A^{(\epsilon)},
\qquad
R_{\bar A}^{(\epsilon)} \;\longrightarrow\; e^{i\epsilon O'_{\bar A}}\,
R_{\bar A}^{(\epsilon)},
\label{eq:local-unitary-redef-recovery}
\end{equation}
with Hermitian $O_A$ on $A=A_1\otimes A_2$ and Hermitian $O'_{\bar A}$ on
$\bar A=\bar A_1\otimes \bar A_2$, induces, to first order in $\epsilon$, the shift
\begin{equation}
W_R \;\longrightarrow\; W_R + O_A\otimes I_{\bar A} + I_A\otimes O'_{\bar A}.
\label{eq:WR-shift-paper}
\end{equation}
Using Eq~\eqref{eq:PQ}, the corresponding variations are
\begin{align}
\delta J
&=
d_\chi\,\Tr_{\bar A_2}\!\Bigl(
\{O_A\otimes I_{\bar A},\chi_{A_2\bar A_2}\}
+
\{I_A\otimes O'_{\bar A},\chi_{A_2\bar A_2}\}
\Bigr),
\label{eq:deltaJ-paper}\\
\delta D
&=
i d_\chi\,\Tr_{\bar A_2}\!\Bigl(
[O_A\otimes I_{\bar A},\chi_{A_2\bar A_2}]
+
[I_A\otimes O'_{\bar A},\chi_{A_2\bar A_2}]
\Bigr).
\label{eq:deltaD-paper}
\end{align}
We now expand $J$ and $D$ in the same Pauli basis as in
Eq.~\eqref{eq:JDexpansion_pure},
\begin{equation}
X=\sum_{l=0}^{d^2-1}\sum_{m=0}^{\bar d^2-1}\sum_{n=0}^{d_\chi^2-1}
p_{lmn}(X)\,
P_l^{(A_1)}\otimes P_m^{(\bar A_1)}\otimes P_n^{(A_2)},
\qquad X\in\{J,D\}.
\label{eq:X-general-pauli-expansion}
\end{equation}
From \eqref{eq:deltaJ-paper} and \eqref{eq:deltaD-paper}, one sees that the
shift generated by $O_A$ is supported only in the sector with $m=0$, while
the shift generated by $O'_{\bar A}$ is supported only in the sector with
$l=0$.  Hence the coefficients with
\[
l>0,\qquad m>0
\label{eq:lm-positive-sector}
\]
are untouched by the boundary-unitary freedom.  Equivalently,
\[
\delta p_{lmn}(X)=0,
\qquad
X\in\{J,D\},\,  \text{for} \quad l>0,\quad m>0.
\label{eq:coeff-invariant-general}
\]
On the other hand, coefficients in the sectors with either $m=0$ or $l=0$ can be affected by the transformation \eqref{eq:local-unitary-redef-recovery}.
This immediately explains the behavior of the grouped coefficients appearing in
$\mathscr S_2^{(\epsilon)}$.Hence, the quantities
\begin{eqns}
E_J&=\sum_{l,m,n>0} p_{lmn}^2(J),
\qquad
B_D=\sum_{m,n>0} p_{0mn}^2(D),
\\
C_D&=\sum_{l,n>0} p_{l0n}^2(D),
\qquad
E_D=\sum_{l,m,n>0} p_{lmn}^2(D)
\label{eq:PA-coeffs-invariant-blocks}
\end{eqns}
are invariant under the local-unitary freedom relevant here, whereas
\begin{equation}
C_J=\sum_{l,n>0} p_{l0n}^2(J), \qquad \text{and}\qquad B_J=\sum_{m,n>0} p_{0mn}^2(J)
\label{eq:CJ-noninvariant-block}
\end{equation}
lies precisely in the $m=0$ and $l=0$ sector of $J$ and therefore can change under the
shift generated by $O_A$ and $O_{\bar{A}}'$. In fact, one can show that by proper choice of $O_A$ and $O_{\bar A}'$, the two coefficients $C_J$ and $B_J$ can be set to zero. Since both of them enter Eq.~\eqref{eq:BEHIBE} as a positive contribution. This choice optimizes the coherent information.

From the result in Sec.~\ref{app:monotonic-pure}, the PA entropy depends on $C_J$, but not on $B_J$, and all other contributing terms remain fixed uncer the local transformation. Therefore we conclude that the optimization condition sets 
\begin{equation}
C_J=0 \implies k_1=0,
\end{equation}

\end{proof}

\subsection{Theorem~\ref{thm:KLcondition}}\label{app:KLcondition}
\begin{proof}
First we show that approximate Knill-Laflamme (aKL) condition Eq.~\eqref{eq:appKL} implies the skewing condition  Eq.~\eqref{eq:apprecovery}.
    
To make the logical tensor-product structure explicit, choose a product basis
$\{\ket{i}_a\}_{i=1}^{d}$ of $\mL_a$ and $\{\ket{j}_{\bar a}\}_{j=1}^{\bar d}$ of $\mL_{\bar a}$,
and define the corresponding code basis
\[
\ket{\widetilde{ij}} \;:=\; V\bigl(\ket{i}_a\otimes \ket{j}_{\bar a}\bigr)\in\mathcal C .
\] 
 Now restrict to operators supported on $\bar A$.  Let $\{P_{\bar A}^{\gamma}\}_{\gamma=0}^{d_{\bar A}^2-1}$ be any operator basis of
$\mathcal B(\mH_{\bar A})$ (for concreteness one may take the Pauli basis when $\mH_{\bar A}$ is a qubit system).  Then the
$\bar A$-part of the aKL condition \eqref{eq:appKL} is equivalently the collection of matrix-element identities
\begin{equation}\label{eq:matrixelement}
\bra{\widetilde{ij}}\,(I_A\otimes P_{\bar A}^{\gamma})\,\ket{\widetilde{kl}}
\;=\;
\delta_{ik}\,\bra{j}\,\mathbb P_{\bar a}^{\gamma}\,\ket{l}
\;+\;
\epsilon\, \bar Y^{\gamma}_{ij;kl},
\end{equation}
where we have set
\[
\mathbb P_{\bar a}^{\gamma}:=\mathcal E_{\bar A}(P_{\bar A}^{\gamma})\in\mathcal B(\mL_{\bar a}),\qquad
\bar Y^{\gamma}_{ij;kl}:=\bar Y_{ij;kl}(P_{\bar A}^{\gamma}).
\]
We adopt the convention $P_{\bar A}^{0}=I_{\bar A}$, so that
$\mathbb P_{\bar a}^{0}=\mathcal E_{\bar A}(I_{\bar A})=I_{\bar a}$ and $\bar Y^{0}_{ij;kl}=0$.

Consider the  Choi state associated with the encoding isometry $V$,
\begin{equation}\label{eq:choiphi}
\ket{\phi}
\;:=\;
\frac{1}{\sqrt{d_L}}
\sum_{i=1}^{d}\sum_{j=1}^{\bar d}
\ket{ij}_{r\bar r}\,\ket{\widetilde{ij}}_{A\bar A},
\end{equation}
where we have written $d:=\dim\mL_a$, $\bar d:=\dim\mL_{\bar a}$ so that $d_L=d\,\bar d$, and where $r,\bar r$ are reference systems
isomorphic to $\mL_a,\mL_{\bar a}$, respectively.

Tracing out $A$ yields a state on $r\bar r\bar A$:
\begin{align}
\phi_{r\bar r\bar A}
&:=\Tr_A\!\left(\ket{\phi}\bra{\phi}\right)\notag\\
&=
\frac{1}{d}\,I_r\otimes \rho_{\bar r\bar A}
\;+\;
\epsilon\,\eta_{r\bar r\bar A}.
\label{eq:reducedphi}
\end{align}
The first term captures the ideal decoupling of $r$ from $\bar r\bar A$, while $\eta_{r\bar r\bar A}$ collects the deviation.

Using the matrix-element identity \eqref{eq:matrixelement}, we can express $\rho_{\bar r\bar A}$ and $\eta_{r\bar r\bar A}$
in an operator basis.  Let $\{P_{\bar A}^{\gamma}\}_{\gamma}$ be an orthonormal operator basis on $\bar A$
(e.g.\ Pauli operators for qubit systems), and let $\{\mathbb P^{\gamma}_{\bar r}\}_{\gamma}$ denote the corresponding operators on
$\bar r$ defined by the aKL map (i.e.\ $\mathbb P^{\gamma}_{\bar r}$ is the representation of $\mathbb P^{\gamma}_{\bar a}$ on $\bar r$
under the fixed identification $\bar r\simeq \mL_{\bar a}$).  Then
\begin{align}
\rho_{\bar r\bar A}
&=
\frac{1}{d_{\bar A}\,\bar d}\sum_{\gamma}
\mathbb P^{\gamma}_{\bar r}\otimes P^{\gamma}_{\bar A},
\label{eq:KLstate-rho}\\
\eta_{r\bar r\bar A}
&=
\frac{1}{d_L\,d_{\bar A}}
\sum_{\gamma} \bar Y^{\gamma}_{r\bar r}\otimes P^{\gamma}_{\bar A},
\qquad
\bar Y^{\gamma}_{r\bar r}
:=
\frac{1}{d_L}\sum_{\beta}\Tr\!\bigl(\bar Y^{\gamma} P^{\beta}\bigr)\,P^{\beta}_{r\bar r}.
\label{eq:KLstate-eta}
\end{align}
Here $\{P^{\beta}_{r\bar r}\}_{\beta}$ is an orthonormal operator basis on $r\bar r$ (again, Pauli strings in the qubit case),
and $\bar Y^{\gamma}$ denotes the matrix with entries $\bar Y^{\gamma}_{ij;kl}$ in the $\ket{ij}$ basis.

By Choi's theorem, a linear map $\mathcal E_{\bar A}:\mathcal B(\mH_{\bar A})\to\mathcal B(\mL_{\bar a})$ is completely positive
if and only if its Choi operator is positive semidefinite. Equivalently,
\begin{align}
J(\mathcal E_{\bar A})
&:=\frac{1}{d_{\bar A}}\sum_{m,n=1}^{d_{\bar A}}
\ket{m}\!\bra{n}_{\bar A}\otimes \mathcal E_{\bar A}\!\left(\ket{m}\!\bra{n}_{\bar A}\right)\succeq 0 \notag\\
&=\frac{1}{d_{\bar A}}\sum_{\gamma} P_{\bar A}^{\gamma}\otimes \mathbb P_{\bar a}^{\gamma}\succeq 0,
\label{eq:choi-positive}
\end{align}
where in the second line we expanded the maximally entangled operator
$\sum_{m,n}\ket{m}\!\bra{n}\otimes \ket{m}\!\bra{n}$ in the operator basis $\{P_{\bar A}^{\gamma}\}_{\gamma}$, and
used $\mathbb P_{\bar a}^{\gamma}:=\mathcal E_{\bar A}(P_{\bar A}^{\gamma})$.

Under the fixed identification $\bar r\simeq \mL_{\bar a}$, the operator $\rho_{\bar r\bar A}$ defined in
\eqref{eq:KLstate-rho} coincides (up to relabeling) with the normalized Choi operator of $\mathcal E_{\bar A}$.
Since $\mathcal E_{\bar A}$ is completely positive, its Choi operator is positive semi-definite, and hence
\[
\rho_{\bar r\bar A}\succeq 0.
\]
Moreover, with the convention $P_{\bar A}^0=I_{\bar A}$ and
$\mathbb P_{\bar a}^0=\mathcal E_{\bar A}(I_{\bar A})=I_{\bar a}$, we have
\begin{align}
\Tr(\rho_{\bar r\bar A})
&=\frac{1}{d_{\bar A}\,\bar d}\Tr\!\left(\mathbb P_{\bar r}^0\otimes P_{\bar A}^0\right)
=\frac{1}{d_{\bar A}\,\bar d}\,\Tr(I_{\bar r})\Tr(I_{\bar A})
=1,
\end{align}
so $\rho_{\bar r\bar A}$ is a valid density operator.

Any density operator admits a purification on an complement space of dimension at least its rank.
We therefore fix a Hilbert space $\mH_{A_2}$ with
\begin{equation}\label{eq:dimA2-rank}
\dim\mH_{A_2}\;\ge\;\rank(\rho_{\bar r\bar A}),
\end{equation}
and choose a purification $\ket{\xi}_{A_2\bar r\bar A}$ such that
\begin{equation}
\Tr_{A_2}\ket{\xi}\bra{\xi}=\rho_{\bar r\bar A}.
\end{equation}

Define the product state
\[
\rho_{r\bar r\bar A}:=\frac{1}{d}I_r\otimes \rho_{\bar r\bar A},
\]
which is the leading-order term in \eqref{eq:reducedphi}. A convenient purification of $\rho_{r\bar r\bar A}$ on
$rA_1A_2\bar r\bar A$ is
\begin{equation}\label{eq:phi0prime}
\ket{\phi_0'}
:=
\frac{1}{\sqrt d}\sum_{i=1}^{d}\ket{i}_r\otimes \ket{i}_{A_1}\otimes \ket{\xi}_{A_2\bar r\bar A},
\end{equation}
where we used $\mH_{A_1}\simeq r$ to purify $\frac{1}{d}I_r$ by the maximally entangled state
$\frac1{\sqrt d}\sum_i\ket{i}_r\ket{i}_{A_1}$.

If the physical subsystem $A$ is large enough to accommodate the purification $\ket{\phi_0'}$
i.e.\ if
\begin{equation}\label{eq:embed-condition-noanc}
d\,\rank \rho_{\bar r\bar A}\le d_A,
\end{equation}
then we choose $\mH_{A_2}$ such that the $\mH_A$ has the following decomposition,
\begin{equation}\label{eq:A-decomp}
\mH_A\simeq (\mH_{A_1}\otimes\mH_{A_2})\oplus \mH_{A_3}.
\end{equation}
Otherwise, if \eqref{eq:embed-condition-noanc} fails, we adjoin an auxiliary system $E$ initialized in $\ket{0}_E$ and work in
$\mH_{AE}=\mH_A\otimes\mH_E$, choosing $\mH_E$ large enough so that
\begin{equation}\label{eq:embed-condition-anc}
d\,\dim\mH_{A_2}\le d_A\,d_E.
\end{equation}
In this case we obtain the analogous decomposition
\begin{equation}\label{eq:AE-decomp}
\mH_{AE}\simeq (\mH_{A_1}\otimes\mH_{A_2})\oplus \mH_{A_3}.
\end{equation}
In what follows we write $AE$ for the purifying register, with the understanding that the special case $d_E=1$ corresponds to
no ancilla (for the case~\eqref{eq:embed-condition-noanc} holds).

Using the standard inequalities
\[
1-\sqrt{F(\sigma,\tau)}\le \frac12\|\sigma-\tau\|_1
\le \frac{\sqrt{D}}{2}\|\sigma-\tau\|_2,
\qquad D:=\dim(\mathrm{supp}),
\]
and noting that $\phi_{r\bar r\bar A}$ and $\rho_{r\bar r\bar A}$ act on a space of dimension $D=d_L\,d_{\bar A}$, we obtain
\begin{align}
1-\sqrt{F(\phi_{r\bar r\bar A},\rho_{r\bar r\bar A})}
&\le \frac{\sqrt{d_L d_{\bar A}}}{2}\,\|\phi_{r\bar r\bar A}-\rho_{r\bar r\bar A}\| \notag\\
&= \epsilon\,\frac{\sqrt{d_L d_{\bar A}}}{2}\,\sqrt{\Tr(\eta_{r\bar r\bar A}\eta_{r\bar r\bar A}^{\dagger})}.
\label{eq:fidelity-bound-start}
\end{align}
Substituting the expansion \eqref{eq:KLstate-eta} of $\eta_{r\bar r\bar A}$ and using orthonormality of the operator basis on $\bar A$
gives
\begin{align}
\Tr(\eta_{r\bar r\bar A}\eta_{r\bar r\bar A}^{\dagger})
&=
\frac{1}{d_L^2 d_{\bar A}}\sum_{\gamma=0}^{d_{\bar A}^2-1}\Tr\!\left(\bar Y^{\gamma}\bar Y^{\gamma\dagger}\right),
\end{align}
and therefore
\begin{equation}\label{eq:fidelity-bound}
1-\sqrt{F(\phi_{r\bar r\bar A},\rho_{r\bar r\bar A})}
\le
\frac{\epsilon}{2\sqrt{d_L}}\,
\Bigl(\sum_{\gamma}\Tr(\bar Y^{\gamma}\bar Y^{\gamma\dagger})\Bigr)^{1/2}
=
\frac{\epsilon}{2\sqrt{d_L}}\,\|\bar Y\|,
\end{equation}
where
\[
\|\bar Y\|:=\left(\sum_{\gamma=0}^{d_{\bar A}^2-1}\Tr(\bar Y^{\gamma}\bar Y^{\gamma\dagger})\right)^{1/2}.
\]

By Uhlmann's theorem, there exists a purification $\ket{\phi_0}$ of $\rho_{r\bar r\bar A}$ on the purifying register $AE$ such that
\begin{equation}
F(\phi_{r\bar r\bar A},\rho_{r\bar r\bar A})
=
\bigl|\bra{\phi, 0_E}\phi_0\rangle\bigr|^2,
\qquad
\Tr_{AE}\ket{\phi_0}\bra{\phi_0}=\rho_{r\bar r\bar A},
\end{equation}
and this overlap is maximal over all purifications of $\rho_{r\bar r\bar A}$ on $AE$.
Combining with \eqref{eq:fidelity-bound} yields
\begin{align}
\|\ket{\phi}\ket{0}_E-\ket{\phi_0}\|^2
&=2-2\Re\braket{\phi,0_E|\phi_0}
= 2\bigl(1-\sqrt{F(\phi_{r\bar r\bar A},\rho_{r\bar r\bar A})}\bigr)
\le \frac{\epsilon}{\sqrt{d_L}}\,\|\bar Y\|,
\end{align}
hence
\begin{equation}\label{eq:phi-close-purif}
\|\ket{\phi}\ket{0}_E-\ket{\phi_0}\|
\le
\left(\frac{\epsilon}{\sqrt{d_L}}\,\|\bar Y\|\right)^{1/2}.
\end{equation}

Finally, since $\ket{\phi_0}$ and $\ket{\phi_0'}$ are both purifications of $\rho_{r\bar r\bar A}$, they differ by a unitary on the
purifying register $AE$. That is, there exists a unitary $R_{AE}$ on $AE$ such that
\begin{equation}
\ket{\phi_0}=(R_{AE}^{\dagger}\otimes I_{r\bar r\bar A})\ket{\phi_0'}.
\end{equation}
Plugging this into \eqref{eq:phi-close-purif}, we may write
\begin{equation}\label{eq:phi0E}
\ket{\phi}\ket{0}_E
=
R_{AE}^{\dagger}\ket{\phi_0'}
+\mathcal N(\epsilon)\ket{\delta},
\end{equation}
for some unit vector $\ket{\delta}$, where the error amplitude obeys
\[
\mathcal N(\epsilon)
=
\bigl\|\ket{\phi}\ket{0}_E-R_{AE}^{\dagger}\ket{\phi_0'}\bigr\|
\le
\left(\frac{\epsilon}{\sqrt{d_L}}\,\|\bar Y\|\right)^{1/2}.
\]

Since $\rho_{\bar r}:=\Tr_{\bar A}(\rho_{\bar r\bar A})=\frac{1}{\bar d}I_{\bar r}$ has a flat spectrum, any purification
$\ket{\xi}_{A_2\bar r\bar A}$ of $\rho_{\bar r\bar A}$ has uniform Schmidt coefficients across the bipartition
$\bar r\,:\,(A_2\bar A)$. In particular, we may choose a Schmidt decomposition of the form
\begin{equation}
\ket{\xi}
=
\frac{1}{\sqrt{\bar d}}\sum_{j=1}^{\bar d}\ket{j}_{\bar r}\otimes \ket{\chi_j}_{A_2\bar A},
\label{eq:xi-schmidt}
\end{equation}
where the vectors $\{\ket{\chi_j}\}_{j=1}^{\bar d}$ are orthonormal in $\mH_{A_2}\otimes\mH_{\bar A}$.

Substituting \eqref{eq:xi-schmidt} into the purification $\ket{\phi_0'}$ and using
Eq.~\eqref{eq:phi0E}, we obtain
\begin{equation}\label{eq:fromA}
R_{AE}\ket{\phi}\ket{0}_E
=
\frac{1}{\sqrt{d_L}}\sum_{i=1}^{d}\sum_{j=1}^{\bar d}
\ket{ij}_{r\bar r}\otimes \ket{i}_{A_1}\otimes \ket{\chi_j}_{A_2\bar A}
\;+\;
\mathcal N(\epsilon)\,R_{AE}\ket{\delta}.
\end{equation}

By the same argument applied to the $A$-supported part of the aKL condition, there exists a unitary $R_{\bar A\bar E}$ on $\bar A$ and auxiliary $\bar E$, and
a purification $\ket{\xi'}_{\bar A_2 r A}$ of the corresponding reduced state such that
\begin{equation}\label{eq:fromAbar}
R_{\bar A\bar E}\ket{\phi}\ket{0}_{\bar E}
=
\frac{1}{\sqrt{d_L}}\sum_{i=1}^{d}\sum_{j=1}^{\bar d}
\ket{ij}_{r\bar r}\otimes \ket{j}_{\bar A_1}\otimes \ket{\chi'_i}_{\bar A_2 A}
\;+\;
\overline{\mathcal N}(\epsilon)\,R_{\bar A\bar E}\ket{\delta'},
\end{equation}
with $\overline{\mathcal N}(\epsilon)\le \sqrt{\epsilon\,\| Y\|/d_L}$.

Comparing \eqref{eq:fromA} and \eqref{eq:fromAbar} (they describe the same state $\ket{\phi}\ket{0}_{E\bar E}$ after adjoining with auxiliary systems) gives an averaged  bound:
\begin{align}
\frac{1}{d_L}\sum_{i=1}^{d}\sum_{j=1}^{\bar d}
\Bigl\|
\ket{i}_{A_1}\,R_{\bar A\bar E}\ket{\chi_j}\ket{0}_{\bar E}
-
\ket{j}_{\bar A_1}\,R_{AE}\ket{\chi'_i}\ket{0}_E
\Bigr\|^2
&=
\norm{\overline{\mathcal N}(\epsilon)\,\ket{\delta'}\ket{0}_{ E}
-\mathcal N(\epsilon)\,\ket{\delta}\ket{0}_{\bar E}}^2 \notag\\
&\le \bigl(\mathcal N(\epsilon)+\overline{\mathcal N}(\epsilon)\bigr)^2 .
\label{eq:avg-mismatch}
\end{align}
Hence there exists a fixed index $i_0\in\{1,\dots,d\}$ such that
\begin{equation}\label{eq:exists-i0}
\sum_{j=1}^{\bar d}
\Bigl\|
\ket{i_0}_{A_1}\,R_{\bar A\bar E}\ket{\chi_j}\ket{0}_{\bar E}
-
\ket{j}_{\bar A_1}\,R_{AE}\ket{\chi'_{i_0}}\ket{0}_E
\Bigr\|^2
\le
\bar d\, \bigl(\mathcal N(\epsilon)+\overline{\mathcal N}(\epsilon)\bigr)^2 .
\end{equation}

Taking the inner product of each term in \eqref{eq:exists-i0} with $\bra{i_0}_{A_1}$ produces a joint (unnormalized) state  on $A_2\bar A_2$, 
\[
\ket{\tilde\chi}_{A_2\bar A_2}
\;:=\;
(\bra{i_0}_{A_1}\otimes I_{A_2\bar A_2})\,R_{AE}\ket{\chi'_{i_0}}_{\bar A_2 A}\ket{0}_E.
\]
Let
\[
    \ket{\chi}_{A_2\bar A_2}:=\frac{\ket{\tilde \chi}_{A_2\bar A_2}}{\sqrt{\bra{\tilde\chi}\tilde{\chi}\rangle}}.
\]
Then \eqref{eq:exists-i0} implies the following bound
\begin{equation}\label{eq:chi-factor}
\sum_{j=1}^{\bar d}\norm{
R_{\bar A\bar E}\ket{\chi_j}_{A_2\bar A}\ket{0}_{\bar E}
-
\ket{j}_{\bar A_1}\otimes \ket{\chi}_{A_2\bar A_2}
}^2
\le
\bar d\,\lrp{\mathcal N(\epsilon)+\overline{\mathcal N}(\epsilon)}^2.
\end{equation}

Substituting \eqref{eq:chi-factor} back into \eqref{eq:fromA} gives the desired approximate factorization of the code basis.
Equivalently,
\begin{align}
&\sqrt{\sum_{i=1}^{d}\sum_{j=1}^{\bar d}
\Bigl\|
\ket{\widetilde{ij}}\ket{0}_{E\bar E}
-
(R_{AE}^\dagger\otimes R_{\bar A\bar E}^\dagger)
\bigl(\ket{i}_{A_1}\otimes\ket{j}_{\bar A_1}\otimes\ket{\chi}_{A_2\bar A_2}\bigr)
\Bigr\|^2}\\
=&\sqrt{d_L}\norm{R_{AE}\otimes R_{\bar A\bar E}\ket{\phi}\ket{0}_{E\bar E}-\frac{1}{\sqrt{d_L}}\sum_{i=1}^d\sum_{j=1}^{\bar d}\ket{ij}_{r\bar r}\otimes\ket{ij}_{A_1\bar A_1}\otimes\ket{\chi}_{A_2\bar A_2}}\\
\le& \sqrt{\sum_{i=1}^d\sum_{j=1}^{\bar d}\norm{
R_{\bar A\bar E}\ket{\chi_j}\ket{0}_{\bar E}
-
\ket{j}_{\bar A_1}\otimes \ket{\chi}_{A_2\bar A_2}
}^2}+\sqrt{d_L}\,\mathcal{N}(\epsilon)\\
\le&
2\sqrt{d_L}\,\bigl(\mathcal N(\epsilon)+\overline{\mathcal N}(\epsilon)\bigr)\\
\le&
2\sqrt{\epsilon}\,\bigl(\sqrt{\|Y\|}+\sqrt{\|\bar Y\|}\bigr).
\label{eq:RRphi}
\end{align}
where in the last step we used the bounds on $\mathcal N(\epsilon)$ and $\overline{\mathcal N}(\epsilon)$ established earlier.

Now define, for each logical basis label $k\in\{1,\dots,d_L\}$,
\begin{equation}
\ket{t_k}\;:=(R_{AE}^\dagger\otimes R_{\bar A\bar E}^\dagger)\;\ket{k}_{A_1\bar A_1}\otimes \ket{\chi}_{A_2\bar A_2}.
\end{equation}
The vectors $\{\ket{t_k}\}_{k=1}^{d_L}$ are orthonormal and span a subspace $\mH_T\subset \mH_{P'}:= \mH_{AE}\otimes\mH_{\bar A\bar E}$
 isomorphic to the unskewed (exact subsystem complementary) code subspace.
Extend $\{\ket{t_k}\}$ to an orthonormal basis of $\mH_{P'}$, and denote the orthogonal complement of $\mH_T$ by
$\mH_{T^\perp}$.

Define a unitary $U$ on $\mH_{P'}$ by specifying its action on $\mH_T$ and $\mH_{T^\perp}$:
\begin{equation}\label{eq:def-U}
U
\;:=\;
\Bigl(\ket{0}_{E\bar E}\otimes \sum_{k=1}^{d_L}\ket{\tilde k}\!\bra{t_k}\Bigr)\;\oplus\; I_{\mH_{T^\perp}}.
\end{equation}
(Here $\{\ket{\tilde k}\}$ is the corresponding orthonormal basis of the code $\mathcal C$, and the direct sum is with respect to
$\mH_{A\bar A}=\mH_T\oplus \mH_{T^\perp}$.)

We next bound $\|U-I\|$ on $\mH_T$.  Let $\ket{\psi}\in\mH_{P'}$ be arbitrary with expansion
$\ket{\psi}=\sum_{i=1}^{d_L}\alpha_i\ket{t_i}+\alpha_{\perp}\ket{t^{\perp}}$, where $\ket{t^{\perp}}\in \mH_{T^{\perp}}$ and  $\sum_{i=1}^{d_L}|\alpha_i|^2\le 1$. Then
\begin{align}
\|(U-I)\ket{\psi}\|
&=
\Bigl\|\sum_{i=1}^{d_L}\alpha_i\bigl(\ket{\tilde i}\ket{0}-\ket{t_i}\bigr)\Bigr\|
\le
\sqrt{\sum_{i=1}^{d_L}\|\ket{\tilde i}\ket{0}-\ket{t_i}\|^2}\\
&\leq 2\sqrt{\epsilon}\,\bigl(\sqrt{\|Y\|}+\sqrt{\|\bar Y\|}\bigr),
\label{eq:UminusI-bound}
\end{align}

where we used the bound in Eq.~\eqref{eq:RRphi} obtained previously.

Eq.~\eqref{eq:UminusI-bound} implies the spectral-norm (denoted as $\norm{\cdot}_2$) bound
\begin{equation}\label{eq:UminusI-opnorm}
\|U-I\|_2
\le
2\sqrt{\epsilon}\,\bigl(\sqrt{\|Y\|}+\sqrt{\|\bar Y\|}\bigr).
\end{equation}

Finally, using the standard estimate $\|\log U\|_2\le \frac{\pi}{2}\|U-I\|_2$ for unitaries with spectrum in the principal branch,
together with \eqref{eq:UminusI-opnorm}, we conclude that there exists a Hermitian operator $W$ on $\mH_{A\bar A}$ and a parameter $\epsilon'>0$ such that
$U=e^{i\epsilon' W}$ with
\[
\epsilon'\,\|W\|_2=\|\log U\|_2
\le
\pi\,\sqrt{\epsilon}\,\bigl(\sqrt{\|Y\|}+\sqrt{\|\bar Y\|}\bigr),
\]

Using the norm $|Y|$ and $|\bar Y|$, defined as 
\begin{eqns}
    |Y|:=\sup_{\|X\|\le1}\max_{i,j}|Y_{ij}(X)|,\qquad
|\bar Y|:=\sup_{\|X\|\le1}\max_{i,j}|\bar Y_{ij}(X)|,
\end{eqns}
we obtain 
\begin{eqns}
    \epsilon'\,\|W\|_2\le \pi\,\sqrt{\epsilon\,d_L}\,\bigl(\sqrt{d_{A}|Y|}+\sqrt{d_{\bar A}|\bar Y|}\bigr).
\end{eqns}

Since $U-I$ vanishes on $\mH_{T^\perp}$ by definition, this leads to a bound in Frobenius norm:
\begin{eqns}
    \epsilon'\,\|W\|_F\le \pi\,\sqrt{\epsilon}\,d_L\,\bigl(\sqrt{d_{A}|Y|}+\sqrt{d_{\bar A}|\bar Y|}\bigr)
\end{eqns}

We now prove the reverse implication \textup{(2)$\Rightarrow$(1)}.

In what follows, we incorporate the ancillas into the corresponding physical subsystems, i.e.\
we replace $A\leftarrow AE$ and $\bar A\leftarrow \bar A\bar E$, and similarly for the codewords $\ket{\tilde i}_P\leftarrow \ket{i}_P\ket{0}_{E\bar E}$.

Assume the approximate recovery form \eqref{eq:apprecovery}, and write each codeword as
\begin{equation}\label{eq:tilde-as-perturb}
\ket{\tilde i}=e^{i\epsilon' W}\ket{t_i},
\end{equation}
where
\begin{equation}\label{eq:def-ti}
\ket{t_i}
:=
(R_A^\dagger\otimes R_{\bar A}^\dagger)\Bigl(\ket{i_L}_{A_1\bar A_1}\otimes \ket{\chi}_{A_2\bar A_2}\Bigr).
\end{equation}
The family $\{\ket{t_i}\}$ therefore consists of codewords of an unskewed subsystem complementary code.
Equivalently, there exists an isometry $V_t:\mH_L\to \mH_P$ such that
\begin{equation}
\ket{t_i}=V_t\ket{i_L}\qquad (\forall\, i).
\end{equation}

For an exact subsystem complementary code, there is a linear map
$\mathcal E_{\bar A}:\mathcal B(\mH_{\bar A})\to \mathcal B(\mL_{\bar a})$ satisfying, for all $X_{\bar A}$,
\begin{equation}\label{eq:exact-KL-t}
\bra{t_i}(I_A\otimes X_{\bar A})\ket{t_j}
=
\bra{i_L}\bigl(I_a\otimes \mathcal E_{\bar A}(X_{\bar A})\bigr)\ket{j_L}.
\end{equation}
Moreover, $\mathcal E_{\bar A}$ is completely positive: the reduced state obtained from the Choi vector
$\frac{1}{\sqrt{d_L}}\sum_i \ket{i}_{r\bar r}\ket{t_i}$ is (up to normalization and the identification $\bar r\simeq \mL_{\bar a}$)
the Choi operator $J(\mathcal E_{\bar A})$, hence positive semidefinite, which is equivalent to complete positivity.

It remains to compare matrix elements in the perturbed code $\{\ket{\tilde i}\}$ to those in the exact code $\{\ket{t_i}\}$.
For any operator $O\in\mathcal B(\mH_P)$,
\begin{align}
\Bigl|\bra{\tilde i}O\ket{\tilde j}-\bra{t_i}O\ket{t_j}\Bigr|
&=
\Bigl|\bra{t_i}\bigl(e^{-i\epsilon' W} O e^{i\epsilon' W}-O\bigr)\ket{t_j}\Bigr|
\notag\\
&\le \bigl\|e^{-i\epsilon' W} O e^{i\epsilon' W}-O\bigr\|_2.
\label{eq:conj-diff-start}
\end{align}
Using the integral representation of the conjugation derivative,
\[
e^{-i\epsilon' W} O e^{i\epsilon' W}-O
=
\int_{0}^{\epsilon'} ds \,\frac{d}{ds}\Bigl(e^{-isW}Oe^{isW}\Bigr)
=
\int_{0}^{\epsilon'} ds\,\bigl(-i\bigr)\,e^{-isW}[W,O]e^{isW},
\]
and unitary invariance of the Hilbert--Schmidt norm, we obtain
\begin{align}
\bigl\|e^{-i\epsilon' W} O e^{i\epsilon' W}-O\bigr\|_2
&\le
\int_{0}^{\epsilon'} ds\,\|[W,O]\|_2
\le
\epsilon'\,\|[W,O]\|_2
\le
2\epsilon'\,\|W\|_2\,\|O\|_2,
\label{eq:conj-diff}
\end{align}
where the last inequality uses $\|AB\|_2\le \|A\|_2\|B\|_2$ and the triangle inequality.

In particular, for $O=I_A\otimes X_{\bar A}$ with $\|X_{\bar A}\|_2=1$,
combining \eqref{eq:exact-KL-t} with \eqref{eq:conj-diff-start}--\eqref{eq:conj-diff} yields
\begin{align}
\epsilon\,|Y_{ij}(X_{\bar A})|
&=
\Bigl|\bra{\tilde i}(I_A\otimes X_{\bar A})\ket{\tilde j}
-\bra{i_L}\bigl(I_a\otimes \mathcal E_{\bar A}(X_{\bar A})\bigr)\ket{j_L}\Bigr|\notag\\
&=
\Bigl|\bra{\tilde i}(I_A\otimes X_{\bar A})\ket{\tilde j}
-\bra{t_i}(I_A\otimes X_{\bar A})\ket{t_j}\Bigr|
\le
2\epsilon'\,\|W\|_2.
\end{align}
This is exactly the desired aKL bound for $\bar A$-supported operators (and the $A$-supported case is identical), completing the proof.
\end{proof}

\subsection{Theorem~\ref{thm:general_mixed}}\label{app:general_mixed}
We generalize the PA entropy calculation to the case where the resource state $\ket{\chi}$ has a non-flat entanglement spectrum.
 The proof follows the same steps as the flat-spectrum calculation in Appendix~\eqref{app:mixed}, with the only difference that the reduced state $\chi_{A_2}$ is no longer proportional to the identity. We therefore indicate only the required modifications to the Weingarten (WG) evaluation and to the spectral $s$-integrals. All other steps (the perturbative relative-entropy expansion \eqref{rel-ent}, the Haar averaging over $U, V$ on $A_1$ and $\bar A_1$, and the diagrammatic reduction to $W$--$\mathscr{R}$ building blocks) are unchanged.

 In this section, we focus on the case where  $|\bar A_1|$ is empty, generalizing Theorem~\ref{thm:monotonic-mixed}.
 Write the Schmidt decomposition of the resource state as
\begin{equation}
\ket{\chi}_{A_2\bar A_2}
=
\sum_{p}\sqrt{\mu_p}\,\ket{p}_{A_2}\ket{p}_{\bar A_2},
\qquad
\mu_p>0,\ \ \sum_p \mu_p=1,
\label{eq:chi-nonflat}
\end{equation}
so that the reduced state $\chi_{A_2}$ is diagonal with possible degeneracies,
\begin{equation}
\chi_{A_2}
=
\sum_m \mu_m\,\Pi_m
\equiv
\bigoplus_m \mu_m\, I_m,
\label{eq:chi-blockdiag}
\end{equation}
where $\Pi_m$ acting on $A_2$, is the projector onto the eigenspace of eigenvalue $\mu_m$, and $I_m$ denotes the identity on that
eigenspace.

In the unperturbed limit, the recovered boundary state of an exact code factorizes, so
$\sigma^{(0)}_{A_1A_2}=\sigma^{(0)}_{A_1}\otimes \chi_{A_2}$, and the resolvent appearing in
$D_{\ln}(\sigma^{(0)}_{A_1A_2})$ is
\begin{equation}
\Gamma_{A_1A_2}(s)
:=
(\sigma^{(0)}_{A_1}\otimes \chi_{A_2}+sI)^{-1}
=
\sum_m \Gamma_m(s)\otimes \Pi_m,
\qquad
\Gamma_m(s):=(\mu_m\sigma^{(0)}_{A_1}+sI)^{-1}.
\label{eq:Gamma-nonflat}
\end{equation}

For the boundary-WG diagrams, every occurrence of the flat-spectrum resolvent
$\Gamma(s)=(d_\chi^{-1}\sigma^{(0)}_{A_1}+sI)^{-1}$ is now replaced according to the the block
decomposition \eqref{eq:Gamma-nonflat}. For example, any boundary term containing 
$\Tr(\sigma^{(0)}_{A_1}\Gamma(s))\,\Tr(\sigma^{(0)}_{A_1}\Gamma(s))$ is replaced by
\begin{eqns}
\Tr(\sigma^{(0)}_{A_1}\Gamma_{A_1A_2}(s))\,
\Tr(\sigma^{(0)}_{A_1}\Gamma_{A_1A_2}(s))
\rightarrow
\Tr(\sigma^{(0)}_{A_1}\Gamma_m(s))\,
\Tr(\sigma^{(0)}_{A_1}\Gamma_n(s)).
\label{eq:block-splitting}
\end{eqns}
Similar block decomposition also occurs to operators whose support has overlap with $A_2$. To make this structure explicit, we adopt the convention that for any operator $X$ on $MA_2$ we write its
$\chi_{A_2}$-eigenbasis block components as
\begin{equation}
X_{mn}
:=
(I_{M}\otimes \Pi_m)\,X\,(I_{M}\otimes \Pi_n),
\label{eq:block-component-def}
\end{equation}
where $M$ denotes any subregions of $A_1\bar A$. We have  $X=\sum_{m,n}X_{mn}$ and each $X_{mn}$ is a linear map
$\mathcal{H}_{M}\otimes\mathrm{Ran}(\Pi_n)\to \mathcal{H}_{M}\otimes\mathrm{Ran}(\Pi_m)$. As a result, each boundary WG contraction splits into a sum over block labels. For example, the flat-spectrum combination
$\int ds\,\mathscr{R}_1(s)\,\mathcal{W}_2$ is replaced by
\begin{align}
\int_0^\infty ds\,\mathscr{R}_1(s)\,\mathcal{W}_2
\ &\rightarrow\
\sum_{m,n}\int_0^\infty ds\,
\Tr\!\big(\sigma^{(0)}_{A_1}\Gamma_m(s)\big)\,
\Tr\!\big(\sigma^{(0)}_{A_1}\Gamma_n(s)\big)\,
\times \\ & \qquad \qquad \Tr\!\Big[
\Tr_{\bar A}\!\big((W_R\chi)_{mn}\big)\,
\Tr_{\bar A}\!\big((W_R\chi)_{nm}\big)
\Big],
\label{eq:R1W2-nonflat}
\end{align}
where $(W_R\chi)_{mn}$ denotes the $(m,n)$ block defined in  \eqref{eq:block-component-def}. Each
block $(\cdot)_{mn}$ can be viewed as a linear map
$\mathcal{H}_{A_1\bar A}\otimes\mathcal{H}_{A_2^n}\to \mathcal{H}_{A_1\bar A}\otimes\mathcal{H}_{A_2^m}$.

On the bulk side, the WG evaluation is unchanged: all bulk diagrams trace out $A_2$ (as in the flat case, e.g. the analogues of $W_7$ and $W_8$), so no $A_2$-block labels survive in the resolvent. Equivalently, the bulk relative entropy still depends on $A_1$ only through the same operator \begin{equation} \Gamma(s)=(d_\chi^{-1}\sigma^{(0)}_{A_1}+sI)^{-1}, \label{eq:Gamma-bulk-same} \end{equation} while the dependence on $\chi$ enters only through the $A_2$-traced contractions of $W_R\chi$

We now express the resulting block-summed answer in a compact operator form.  Define the Hermitian operators
\begin{equation}
J:=d_\chi\,\Tr_{\bar A_2}\!\Big(\{W_R,\ket{\chi}\bra{\chi}\}\Big),
\qquad
D:=i d_\chi\,\Tr_{\bar A_2}\!\Big([W_R,\ket{\chi}\bra{\chi}]\Big),
\label{eq:JD-nonflat}
\end{equation}
 same as the definitions in \eqref{eq:JD-mixed} in the flat-spectrum case.  In the eigenbasis of
$\chi_{A_2}$, we decompose $J$ and $D$ into block components by inserting the projectors $\Pi_m$ on $A_2$:
\begin{equation}
J_{mn}:=(I_{A_1}\otimes \Pi_m)\,J\,(I_{A_1}\otimes \Pi_n),
\qquad
D_{mn}:=(I_{A_1}\otimes \Pi_m)\,D\,(I_{A_1}\otimes \Pi_n),
\label{eq:JD-blocks}
\end{equation}
so that $J=\sum_{m,n}J_{mn}$ and $D=\sum_{m,n}D_{mn}$,  
matching the block-index structure produced by resolving every $A_2$ wire as explained above.

After performing the Haar average over $U$ on $A_1$, the same Weingarten rules as in Appendix~\eqref{app:mixed}
again separate the answer into a product of (i) spectral functions coming from the $s$-integrals involving
$\Gamma_m(s)$ and $\Gamma_n(s)$, and (ii) purely algebraic contractions of $W_R$ and $\chi$, which
can be written in terms of the blocks $J_{mn}$ and $D_{mn}$.  We record the resulting replacements in the same
order as the flat-spectrum decomposition \eqref{eq:Scorr}, starting with the $J$-sector.

Concretely, the flat contribution $c_1 f_1(\lambda)$ is promoted to a block-sum
\begin{equation}
c_1 f_1(\lambda)
\ \longrightarrow\
\sum_{m,n} c_1^{mn}\, f_1^{mn}(\lambda),
\label{eq:c1f1-promotion}
\end{equation}
where $c_1^{mn}$ is obtained from the $J$-dependent trace combinations in \eqref{eq:c1c2c3} by resolving every
$A_2$-contraction into block components. 

Equivalently, one applies the rule
\[
J\ \mapsto\ J_{mn},\qquad J^2\ \mapsto\ J_{mn}J_{nm},\qquad
\Tr_{A_1}(J)\ \mapsto\ \Tr_{A_1}(J_{mn}),
\]
and similarly for the remaining traced quantities.
In particular, each occurrence of $\Tr(J^2)$ is replaced by $\Tr(J_{mn}J_{nm})$, and each occurrence of
$\Tr(\Tr_{M}(J)^2)$ is replaced, in the same manner, by
$\Tr\!\big(\Tr_{M}(J_{mn})\,\Tr_{M}(J_{nm})\big)$, where M is any subsystem of $A\bar A$

Hence, we have:
\begin{equation}
\begin{aligned}
c_1^{mn}(J)
:=
\frac{d}{2 d_\chi\left(d^2-1\right)}\Bigg[
&\Tr_{A_1 A_2}\!\left(J_{mn}J_{nm}\right)
-\frac{1}{d}\,\Tr_{A_2}\!\left(\Tr_{A_1}(J_{mn})\,\Tr_{A_1}(J_{nm})\right) \\
&-\frac{\delta_{mn}}{d_\chi}\,\Tr_{A_1}\!\left(\Tr_{A_2}(J_{mm})^2\right)
+\frac{\delta_{mn}}{d\,d_\chi}\,\Tr_{A_1 A_2}\!\left(J_{mm}\right)^2
\Bigg].
\end{aligned}
\end{equation}

The terms in $c_1$ that originate from tracing out $A_2$ remain diagonal in the block labels, since a complete trace over $A_2$ collapses $\Pi_m\Pi_n=\delta_{mn}\Pi_m$.  The corresponding spectral factor $f_1^{mn}(\lambda)$ is obtained from the same linear combination of $s$-integrals that defines $f_1(\lambda)$ in \eqref{eq:f1f2}, with the boundary resolvents replaced as $\Gamma(s)\mapsto \Gamma_m(s)$ and $\Gamma(s)\mapsto \Gamma_n(s)$.  Writing
$\sigma^{(0)}_{A_1}=\sum_i \lambda_i |i\rangle\langle i|$, one convenient closed form is 
\begin{equation}
f_1^{mn}(\lambda)
=\frac{1}{2d}\sum_{i j}\frac{(\lambda_i-\lambda_j)^2}{(\mu_m\lambda_i-\mu_n\lambda_j)}\ln\frac{\mu_m\lambda_i}{\mu_n\lambda_j}.
\label{eq:f1mn-def}
\end{equation}
with the continuous extension at $\mu_m\lambda_i=\mu_n\lambda_j$ understood.  When $\mu_m=\mu_n=1/d_\chi$ this reduces to the flat-spectrum function $f_1(\lambda)$. 

Next, the contribution that in the flat case recombined into $c_2 f_2(\lambda)$ is promoted to
\begin{equation}
c_2 f_2(\lambda)
\ \longrightarrow\
\sum_{m,n} c_2^{mn}\, f_2^{mn}(\lambda),
\label{eq:c2f2-promotion}
\end{equation}
with coefficients expressed by the same $D$-trace combinations as in \eqref{eq:c1c2c3}, but evaluated on the block products $D_{mn}D_{nm}$:
\begin{equation}
c_2^{mn}(D)
:=
\frac{d}{2d_\chi(d^2-1)}
\left[
\Tr_{A_1A_2}\!\big(D_{mn}D_{nm}\big)
-\frac{1}{d}\,
\Tr_{A_2}\!\Big(\Tr_{A_1}(D_{mn})\,\Tr_{A_1}(D_{nm})\Big)
\right].
\label{eq:c2mn-def}
\end{equation}
The associated spectral function $f_2^{mn}(\lambda)$ is obtained from the same linear combination of $s$-integrals that defines $f_2(\lambda)$ in \eqref{eq:f1f2}, with the replacement $\Gamma(s)\mapsto \Gamma_m(s)$ and
$\Gamma(s)\mapsto \Gamma_n(s)$.  In the eigenbasis of $\sigma^{(0)}_{A_1}$, a convenient closed form is
\begin{equation}
f_2^{mn}(\lambda)
=
\sum_{i,j}\frac{1}{2d}\,
\frac{(\lambda_i+\lambda_j)^2}{\mu_m\lambda_i-\mu_n\lambda_j}\,
\ln\!\frac{\mu_m\lambda_i}{\mu_n\lambda_j}
\;,
\label{eq:f2mn-def}
\end{equation}
where again the continuous extension at $\mu_m\lambda_i=\mu_n\lambda_j$ is understood, and the constant terms arise from the same normalization pieces as in the flat-spectrum evaluation.

Collecting these block-resolved contributions together with the similarly promoted constant term
$c_3\to\sum_{m,n}c_3^{mn}$ yields the stated structure
\[
\langle S_{\mathrm{corr}}\rangle
=
\epsilon\,c_0
+
\frac{\epsilon^2}{2}
\sum_{m,n}\Big(c_1^{mn}f_1^{mn}(\lambda)+c_2^{mn}f_2^{mn}(\lambda)+c_3^{mn}\Big),
\]
where $c_0$ arises solely from the term $-\Tr\!\big(\delta\chi\,\ln\chi\big)$ in the general proto-area identity
Eq.~\eqref{eq:PA-step}. For a flat entanglement spectrum of $\chi_{A_2}$ one has $\delta\chi=0$ and this contribution
vanishes, whereas for a non-flat spectrum it is generically nonzero and produces an $\mathcal{O}(\epsilon)$ correction
after Haar averaging.

The monotonicity statements follow from Appendix~\ref{app:f1f2f3} after the block replacement. For each fixed pair $(m,n)$, we have shown in Appendix~\ref{app:f1f2f3} that the functions $f^{mn}_1(\lambda)$ and $f^{mn}_2(\lambda)$ are each monotonically decreasing with bulk entanglement (even though their $s$-integral representations differ from those of $f_1,f_2$). Since the corresponding coefficients $c_i^{mn}$ are nonnegative by construction, the block-summed correction inherits these monotonicity properties.
 \qed

\subsection{Theorem~\ref{thm:general_pure}}\label{app:general_pure}
We now extend Theorem~\eqref{thm:monotonic-pure} to a general resource state
$\ket{\chi}_{A_2\bar A_2}$ with non-flat entanglement spectrum. The derivation proceeds exactly as in the flat-spectrum pure-state analysis of Appendix~\eqref{app:monotonic-pure}, with the only change being the block-resolution procedure introduced in the mixed-state generalization of Appendix~\eqref{app:general_mixed}.

Specifically, every boundary and bulk trace diagram containing an $A_2$ wire is modified by the same replacement rule: $\chi_{A_2}$ is promoted to the direct sum of spectral projectors (see Eq.~\eqref{eq:chi-blockdiag}), while the corresponding $A_1A_2$ resolvents are replaced by their block-diagonal versions (see Eq.~\eqref{eq:Gamma-nonflat}). Accordingly, we record only the resulting promoted coefficients and spectral functions.

Define $J$ and $D$ exactly as in the flat-spectrum pure case, \begin{equation} J:=d_\chi\,\Tr_{\bar A_2}\!\left(\{W_R,\ket{\chi}\bra{\chi}\}\right), \qquad D:=i d_\chi\,\Tr_{\bar A_2}\!\left([W_R,\ket{\chi}\bra{\chi}]\right), \end{equation} now regarded as operators acting on $A_1\bar A_1A_2$. Decomposing $A_2$ into the eigenspaces of $\chi_{A_2}$ induces corresponding block components $J_{mn}$ and $D_{mn}$, viewed as maps from the \(n\)-th eigenspace of \(\chi_{A_2}\) to the \(m\)-th eigenspace. 

Upon carrying out the same independent local Haar averages over \(U\) on \(A_1\) and \(V\) on \(\bar A_1\), the Weingarten calculus produces the same diagrammatic structure as in Appendix~\eqref{app:monotonic-pure}, with the only change that each flat-spectrum coefficient is promoted to a double sum over block labels \((m,n)\). Accordingly, the flat-spectrum decomposition in Eq.~\eqref{PA-entropy_pure_full} is replaced by
\begin{align}
\big\langle S_{\mathrm{corr}}\big\rangle
=&\epsilon k_0+\frac{\epsilon^2}{2}\sum_{m,n}\bigg[
k_1^{mn}f^{mn}_1(\lambda)+k_2^{mn}f^{mn}_2(\lambda)
+k_3^{mn}\!\left(f^{mn}_3(\lambda)-\frac{1}{d^2}\Bigl(1+\frac{d}{\bar d}\Bigr)f_2^{mn}(\lambda)\right.\notag\\
&\left.+\frac{1}{d^2}f_1^{mn}(\lambda)\right)
+k_4^{mn}\!\left(f^{mn}_3(\lambda)-\frac{1}{d\bar d}f_1^{mn}(\lambda)\right)
+k_5^{mn}\!\left(f_2^{mn}(\lambda)-f_1^{mn}(\lambda)\right)
+k_6^{mn}
\bigg],
\label{eq:Scorr-pure-nonflat-structure-app}
\end{align}
where $\dim\mathcal{L}_a=d$ and $\dim\mathcal{L}_{\bar a}=\bar d$ as before, and $f_3^{m n}$ is defined as:
\begin{equation}
f_3^{m n}(\lambda)=\frac{1}{2 d} \sum_{i, j} \frac{\lambda_i+\lambda_j}{\left(\mu_m \lambda_i-\mu_n \lambda_j\right)} \ln \left(\frac{\mu_m \lambda_i}{\mu_n \lambda_j}\right).
\end{equation}
Each coefficient \(k_\ell^{mn}\) is obtained from its flat-spectrum counterpart \(k_\ell\) by promoting \(J\to J_{mn}\) and \(D\to D_{mn}\) within the same trace polynomials. In particular, every occurrence of \(\Tr(J^2)\) and \(\Tr(D^2)\) is replaced by \(\Tr(J_{mn}J_{nm})\) and \(\Tr(D_{mn}D_{nm})\), respectively, while every squared partial trace is replaced by the corresponding product of blockwise partial traces. For example,
\[
\Tr\!\big(\Tr_M(J)^2\big)\ \mapsto\ \Tr\!\big(\Tr_M(J_{mn})\,\Tr_M(J_{nm})\big),\]
\[
\Tr\!\big(\Tr_M(D)^2\big)\ \mapsto\ \Tr\!\big(\Tr_M(D_{mn})\,\Tr_M(D_{nm})\big),
\]
for any subsystem $M$ appearing in $k_\ell$.

The spectral functions \(f_\ell^{mn}(\lambda)\) are defined by evaluating the same \(s\)-integral combinations that define \(f_\ell(\lambda)\), but with the \(A_1A_2\) resolvent factors on the two boundary legs replaced by \((\Gamma_m(s),\Gamma_n(s))\). Equivalently, after diagonalizing
\(\sigma^{(0)}_{A_1}=\sum_i \lambda_i\,|i\rangle\langle i|\),
they may be obtained from the flat-spectrum expressions by replacing
\(\lambda_i\mapsto \nu_m\lambda_i\) and \(\lambda_j\mapsto \nu_n\lambda_j\)
wherever these eigenvalues enter through the resolvents. In the special case \(\nu_m=\nu_n=1/d_\chi\), these definitions reduce to the flat-spectrum ones. Their monotonicity then follows pointwise in \((m,n)\) from the same integrand-level argument used in Appendix~\ref{app:f1f2f3}.

Finally, because \(\chi_{A_2}\) is not maximally mixed, the proto-area identity acquires an additional first-order contribution,
\(
\epsilon\,k_0,
\)
where \(k_0\) arises from the term \(\Tr(\delta\chi\,\ln\chi)\) in Eq.~\eqref{eq:PA-step}. Combining all block-resolved contributions yields the statement of the theorem, with
\(\langle S_{PA}\rangle=S(\chi)-\langle S_{\mathrm{corr}}\rangle\)
and \(\langle S_{\mathrm{corr}}\rangle\) given by Eq.~\eqref{eq:Scorr-pure-nonflat-structure-app}.

For each fixed block pair $(m,n)$, the functions $f^{mn}_1(\lambda)$, $f^{mn}_2(\lambda)$, and $f^{mn}_3(\lambda)$ are monotonically decreasing in $\lambda$ by the same integrand-level argument as in Appendix~\ref{app:f1f2f3}. Thus, in \eqref{eq:Scorr-pure-nonflat-structure-app} the only potentially non-monotone structure is the difference $f^{mn}_2(\lambda)-f^{mn}_1(\lambda)$, which appears with coefficient $k^{mn}_5$. The typicality argument from the flat-spectrum pure case carries over directly: under the same Gaussian-ensemble model for the product-basis coefficients of $W_R$, each promoted coefficient $k^{mn}_\ell$ concentrates around its term-counting estimate, and the coefficient sector multiplying $f^{mn}_2-f^{mn}_1$ is typically suppressed relative to the dominant monotone sector. Since \eqref{eq:Scorr-pure-nonflat-structure-app} is a finite sum over $(m,n)$, a union bound implies that the probability that any block produces a monotonicity-violating fluctuation remains exponentially small (up to a prefactor polynomial in $d_\chi$), exactly as in the flat-spectrum analysis. So in the large $d$ and $\bar d$ limit, we find that the \emph{typical} PA entropy correction in the non-flat case
is dominated by the block-summed $f^{mn}_3$ sector,
\begin{equation}
\big\langle S_{\mathrm{corr}}\big\rangle
=
\frac{\epsilon^2}{2}\sum_{m,n}\big(k^{mn}_3+k^{mn}_4\big)\,f^{mn}_3(\lambda)
+\mathcal{O}\!\left(\frac{1}{d^2}\right)
+\mathcal{O}\!\left(\frac{1}{\bar d^{2}}\right).
\end{equation}
\qed

\section{Magic of the code}\label{app:magic}
In this section, we compute the magic  of the skewed code. In the main text, we have defined the Choi state of the skewed code to be
\begin{equation}
    \ket{V^{(\epsilon)}}:=\frac 1 {\sqrt{d_L}}\sum_{i=1}^{d_L} \ket{i}_R\otimes e^{i\epsilon W_R}\ket{i}_{A_1\bar{A}_1}\ket{\chi}_{A_2\bar A_2}. 
\end{equation}

The magic of the state is quantified using stabilizer Renyi entropy (SRE), defined as 
\begin{equation}
    \mathcal{M}_{\alpha}(\ket{V^{(\epsilon)}})=\frac{1}{1-\alpha}\log\left(2^{-n}\sum_{s=1}^{4^n}|\bra{V^{(\epsilon)}}P_s\ket{V^{(\epsilon)}}|^{2\alpha}\right),
\end{equation}
where $P_s$'s are  the Pauli string operators on $R \cup A_1 \cup \bar{A}_1 \cup A_2 \cup \bar{A}_2$.
We decompose the Pauli string operator $P_s$ into tensor product of four parts, $P_a, P_b, P_c,P_d$, supported on $R, A_1\bar{A}_1, A_2,$ and $\bar A_2$ respectively. Up to order $\epsilon^2$, expectation value of the Pauli string is as:
\begin{align}
& |\langle V^{(\epsilon)}|P_a\otimes P_b\otimes P_c \otimes P_d|V^{(\epsilon)}\rangle| \\ \nonumber
&=\frac{1}{d_Ld_{\chi}}\left[\vcenter{\hbox{\includegraphics[height = 8 em ]{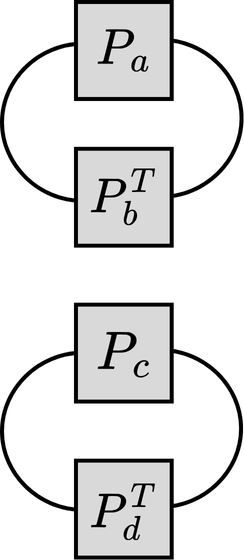}}}
+ i \epsilon \left(
\vcenter{\hbox{\includegraphics[height=8 em]{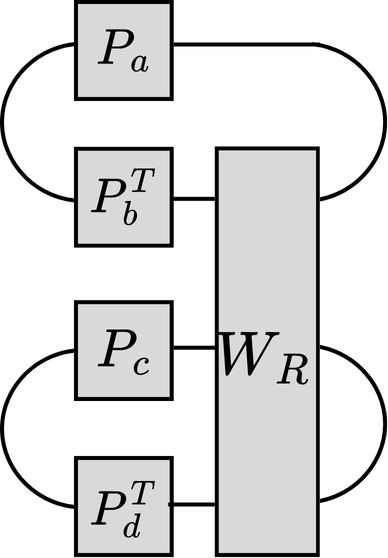}}} - 
\vcenter{\hbox{\includegraphics[height=8 em]{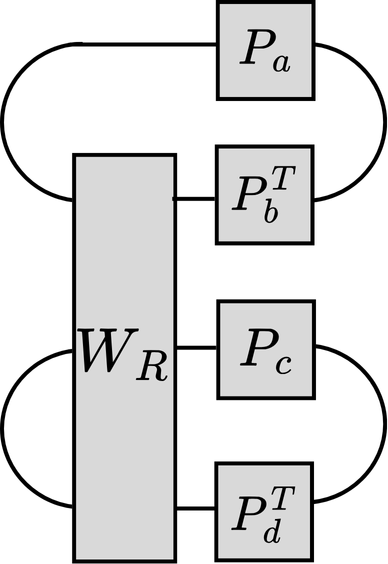}}} \right)\right. \\ \nonumber
& \left. \hspace{1cm}
 - \frac{\epsilon^2}{2 } \left( \vcenter{\hbox{\includegraphics[height=8 em]{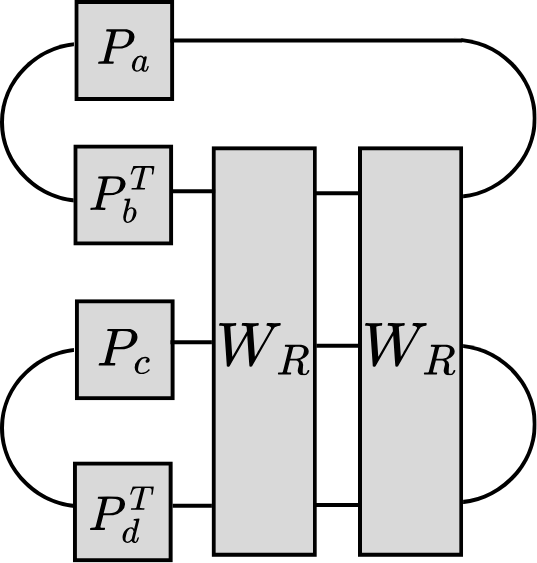}}} + 
 \vcenter{\hbox{\includegraphics[height=8 em]{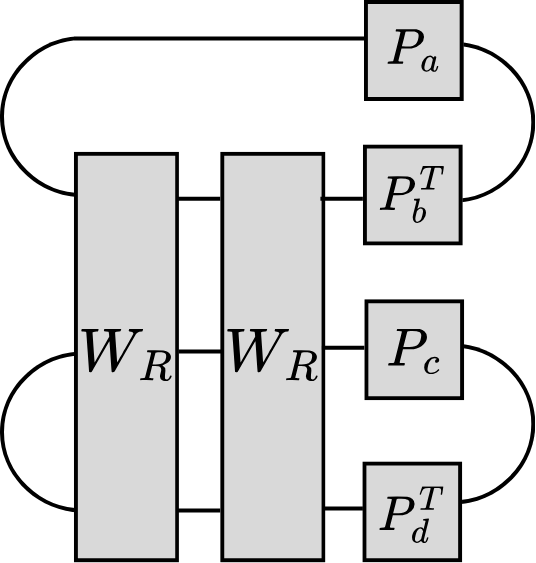}}} -2
 \vcenter{\hbox{\includegraphics[height=8 em]{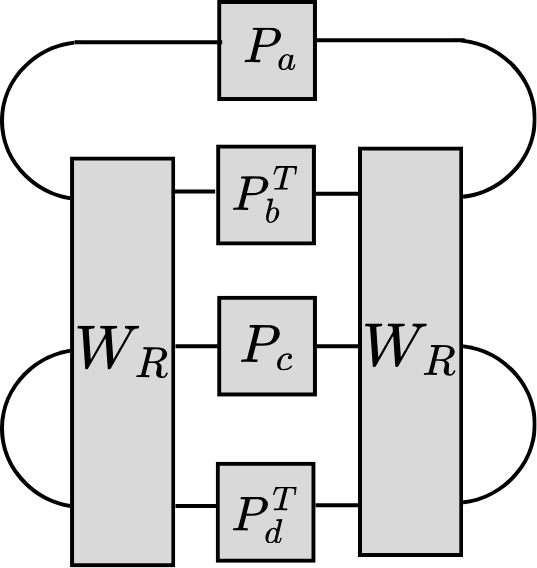}}}\right)\right]\\ \nonumber
 & \equiv\delta_{ab}\delta_{cd}+\epsilon f_{abcd}-\frac{\epsilon^2}{2}g_{abcd}+O(\epsilon^3)
\end{align}
where  $f_{abcd}$ and $g_{abcd}$ are defined in the equation.

Using this expansion, we derive the SRE up to second order of $\epsilon$ as:
\begin{align}
\mathcal{M}_{\alpha}\!\left(\ket{V^{(\epsilon)}}\right)
&=\frac{1}{1-\alpha}\log\!\left[
\frac{1}{(d_L d_{\chi})^{2}}
\sum_{a b c d}
\left|\delta_{ab}\delta_{cd}
+\epsilon f_{abcd}
-\frac{\epsilon^{2}}{2}g_{abcd}\right|^{2\alpha}
\right]
+O(\epsilon^{3})
\notag\\[4pt]
&=\frac{1}{1-\alpha}\log\!\left[
\frac{1}{(d_L d_{\chi})^{2}}
\sum_{a b c d}
\Bigl(
\delta_{ab}\delta_{cd}
+2\alpha\,\epsilon\,\delta_{ab}\delta_{cd} f_{abcd}
-\alpha\,\epsilon^{2}\,\delta_{ab}\delta_{cd}\bigl(g_{abcd}+(1-2\alpha)f_{abcd}^{2}\bigr)
\Bigr)
\right]
+O(\epsilon^{3})
\notag\\[4pt]
&=\frac{1}{1-\alpha}\log\!\left[
\frac{1}{(d_L d_{\chi})^{2}}
\sum_{a b c d}
\Bigl(
\delta_{ab}\delta_{cd}
-\alpha\,\epsilon^{2}\,\delta_{ab}\delta_{cd}\, g_{abcd}
\Bigr)
\right]
+O(\epsilon^{3})
\notag\\[4pt]
&=\frac{1}{1-\alpha}\log\!\left[
1-\frac{\alpha\,\epsilon^{2}}{(d_L d_{\chi})^{2}}
\sum_{a,c} g_{aacc}
\right]
+O(\epsilon^{3})
\notag\\[4pt]
&=\frac{\alpha}{\alpha-1}\,
\frac{\epsilon^{2}}{(d_L d_{\chi})^{2}}
\sum_{a,c} g_{aacc}
+O(\epsilon^{3}).
\label{eq:Missumg}
\end{align}

where we have used the fact that $\delta_{ab}\delta_{cd}f_{abcd}=0$, since $P_a^2=I$.

Using the identity $\frac{1}{d}\sum_a P_a O P_a=\Tr(O)I$, we can further simplify the sum to:
\begin{align}
    \sum_{ac} g_{aacc} &= \frac{1}{d_Ld_{\chi}}\sum_{abcd}\delta_{ab}\delta_{cd}\left( \vcenter{\hbox{\includegraphics[height=8 em]{theory_draft_images/SRE_diagrams/PsW_RW_R.png}}} + \vcenter{\hbox{\includegraphics[height=8 em]{theory_draft_images/SRE_diagrams/W_RW_RPs.png}}} - 2 \vcenter{\hbox{\includegraphics[height=8 em]{theory_draft_images/SRE_diagrams/W_RPsW_R.png}}}\right) \\ \nonumber
    & = 2 (d_L d_\chi ) \left( \vcenter{\hbox{\includegraphics[height=8 em]{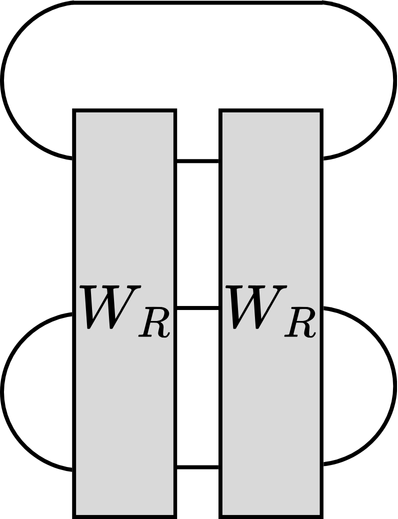 }}} - \frac{1}{d_L d_\chi} \vcenter{\hbox{\includegraphics[height=8 em]{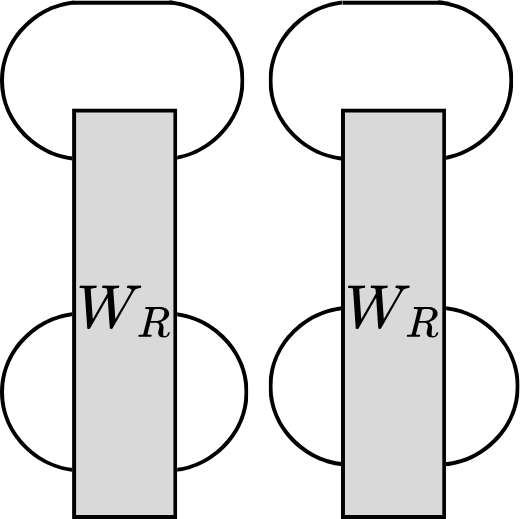}}}\right).
\end{align}

Substitute it into Eq.~\eqref{eq:Missumg} we obtain
\begin{equation}
    \mathcal{M}_{\alpha}(\ket{V^{(\epsilon)}})=\frac{2\alpha}{\alpha-1}\frac{\epsilon^2}{d_L}\left(\Tr(\langle W_R^2\rangle_{\chi})-\frac{1}{d_Ld_{\chi}}\Tr(\langle W_R\rangle_{\chi})^2\right)+\mO(\epsilon^3),
\end{equation}
where $\langle W_R\rangle_{\chi}:=\bra{\chi}W_R\ket{\chi}$, with $\ket{\chi}_{A_2\bar A}:=\frac{1}{\sqrt{d_{\chi}}}\sum_i\ket{i}_{A_2}\ket{i}_{\bar A_2}$.

Finally, we rewrite the SRE using $J$ and $D$ tensors defined in Eq.~\eqref{eq:PQ}:
\begin{equation}
    \mathcal{M}_{\alpha}(\ket{V^{(\epsilon)}})=\frac{\alpha}{2(\alpha-1)}\frac{\epsilon^2}{d_Ld_{\chi}}\left(\Tr(J^2)-\frac{1}{d_Ld_{\chi}}\Tr(J)^2+\Tr(D^2)\right)+\mO(\epsilon^3).
\end{equation}
Next, we use this formula to compute the amount of perturbative tripartite non-local magic in the state $\ket{V^{(\epsilon)}}$.

\subsection{Case I: $|\bar A_1|=0$}
In this case, we have $\dim \mH_L=\dim \mL_a=d$, and $\bar d=\dim \mL_{\bar a}=1$. The perturbative tripartite non-local magic is defined as the minimal magic that cannot be removed by the 2-local unitaries.
In the case when $\bar A_2$ is empty, the 2-local unitary is parametrized as 
\begin{eqns}
    \exp\left\{i\epsilon \left(T_{A_1A_2}\otimes I_{\bar A}+T_{A_1\bar A}\otimes I_{A_2}+T_{A_2\bar A}\otimes I_{A_1}\right)\right\}
\end{eqns}

The $J$ term can be set to 0 by shifting $W_R\rightarrow W_R+T_{A_1A_2}\otimes I_{\bar A}+T_{A_1\bar A}\otimes I_{A_2}$, corresponding to the local unitary transformation on $A_1A_2$ and $A_1\bar A$,  as discussed in the last section; $D$ is unchanged under this shift. The only 2-local transformation left is the shift $I_{A_1}\otimes T_{A_2\bar A}$, which changes $D$ by $D\rightarrow D+I_{A_1}\otimes T^D_{A_2 \bar{A}}$, with $T^D$ defined as 
\begin{eqns}
    T^D:=id_{\chi}\Tr_{A_2}\left([ T_{},\chi]\right).
\end{eqns}

After setting $J=0$ and include this shift to SRE, we find 

\begin{equation}\label{eq:3pmagic}
    \mathcal{M}^{(3)}_\alpha (\ket{V^{(\epsilon)}})=\frac{\alpha}{2(\alpha-1)}\frac{\epsilon^2}{dd_{\chi}}\min_{T^D}\Tr\left((D+I\otimes T^D)^2\right)+\mO(\epsilon^3).
\end{equation}

To find the $T^D$ that minimizes the nonlocal magic, we expand $D$ on $\mH_{A_1} \otimes \mH_{A_2 \bar{A}}$
\begin{equation}
    D=I\otimes \delta_2+ \delta_1\otimes I+\sum_a \delta_1^a\otimes \delta_2^a
\end{equation}
where each $\delta^a_i$ are traceless. The expression of the trace becomes
\begin{equation}
    \Tr\left((D+I\otimes T^D)^2\right)=d\Tr\left((\delta_2+T^D)^2\right)+d_{\chi}\Tr(\delta_1^2)+\sum_a \Tr((\delta_1^a)^2)\Tr((\delta_2^a)^2).
\end{equation}

Clearly, picking 
\[
T^D = -\delta_2 = -\frac{1}{d}\operatorname{Tr}_{A_2 \bar{A}}(D)
\]
minimizes Eq.~\eqref{eq:3pmagic}. Substituting it back, we find
\begin{equation}
    \mathcal{M}^{(3)}\bigl(\ket{V^{(\epsilon)}}\bigr)
    = \frac{\alpha}{2(\alpha - 1)}\frac{\epsilon^2}{dd_{\chi}}
    \left[
        \operatorname{Tr_{A_1A_2}}\!\big( D^{2} \big)
        - \frac{1}{d}\,\operatorname{Tr}_{A_2}\!\left( \operatorname{Tr}_{A_1}\!\big(D \big)^{2} \right)
    \right]+\mO(\epsilon^3).
\end{equation}
This is exactly \( c_{2} \) (see Eq.~\eqref{eq:c1c2c3}) up to constant prefactors. This proves Theorem~\ref{thm:nlmagicmixed}. 

\subsection{Case II: $|\bar A_1|>0$}
The case with nontrivial $\bar A_1$ is analogous by treating $A_1$ and $\bar A_1$ equivalent with their respective reference systems. Now we have $d_L=d\bar d$.  The 2-local perturbation can be on any pair in $\{A_1, A_2, \bar A_1 , \bar A_2\}$. To the $\mathcal{O}(\epsilon^2)$:
\begin{align}
    \mathcal{M}_\alpha\bigl(\ket{V^{(\epsilon)}}\bigr)= \frac{\alpha}{2(\alpha-1)} \frac{\epsilon^2 }{d\bar dd_{\chi} } \left(\operatorname{Tr}(J^2) - \frac{1}{d \bar{d} d_\chi} \operatorname{Tr} (J)^2 + \operatorname{Tr} (D^2) \right),
\end{align}
where $d$, $\bar{d}$ and $d_\chi$ are the dimensions of $\mathcal{L}_{a}$, $\mathcal{L}_{\bar a}$ and $\mathcal{H}_{A_2}$ (same size as  $\mH_{\bar A_2}$). To see how $\mathcal{M}$ is minimized in the equivalent class under 2-site local unitaries, we decompose $J$ and $D$ in terms of the Pauli operators on $A_1$, $\bar A_1$ and $A_2 \cup \bar A_2$: 
\begin{align}
    J &= \sum_{r=0}^{d^2 -1} \sum_{s=0}^{\bar{d}^2 -1} \sum_{t=0}^{d_\chi^2 -1} p_{r s t} (J) P_r P_s P_t \\
    D &= \sum_{r=0}^{d^2 -1} \sum_{s=0}^{\bar{d}^2 -1} \sum_{t=0}^{d_\chi^2 -1} p_{r s t} (D) P_r P_s P_t
\end{align}

Plugging these expressions back in gives:
\begin{align}
    \mathcal{M}_\alpha\bigl(\ket{V^{(\epsilon)}}\bigr)= \frac{\alpha}{2(\alpha-1)} \epsilon^2\sum_{r=0}^{d^2-1}\sum_{s=0}^{\bar d^2-1}\sum_{t=0}^{d_{\chi}^2-1}\left( p_{rst}^2(J) + p_{000}^2(J) + p_{rst}^2(D) \right).
    \end{align}
This is a sum of quadratic functions. Now we show that some of the terms in this sum can be removed by applying 2-local unitaries, thereby minimizing SRE of the code.  For example, a carefully chosen 2-site unitary on $A_1 \bar A_1$, in the form of $e^{i\epsilon T_{A_1\bar A_1}}$, can dial value of $c_{rs0}(J)$ with $r,s>0$ to zero. Because this transformation shift $W_R$ by 
\begin{eqns}
    W_R\rightarrow W_R+T_{A_1\bar A_1}\otimes I_{A_2\bar A_2}. 
\end{eqns}

This transform $J$ and $D$ by 
\begin{eqns}
    J\rightarrow & J+2T_{A_1\bar A_1}\otimes I_{A_2}\\
    D\rightarrow & D. 
\end{eqns}

This transformation only changes the Pauli coefficients $c_{rs0}(J)$ and can change it to any value. Therefore to minimize SRE we set this coefficient to zero. 

More generally we have the following list of two local unitary generators and the corresponding Pauli coefficients that they could affect: 
\begin{eqns}
    &T_{A_1A_2}: p_{r0t}(J), \qquad  T_{A_1\bar A_1}: p_{rs0}(J), \qquad 
T_{A_1\bar A_2}: p_{r0t}(J), \qquad T_{A_2\bar A_1}: p_{0st}(J), \\
&T_{\bar A_1\bar A_2}: p_{0st}(J), \qquad T_{A_2\bar A_2}: p_{00t}(J),\ p_{00t}(D).
\end{eqns}

By this token, after setting the corresponding coefficients to zero we obtain 
\begin{align}
    \mathcal{M}^{(3)}_\alpha\bigl(\ket{V^{(\epsilon)}}\bigr)=& \frac{\alpha \epsilon^2}{2(\alpha-1)} \sum_{r>0,s>0,t>0} \left( p_{r0t}^2(D) + p_{0st}^2(D) + p_{rst}^2(D) + p_{rst}^2(J) \right)\\
    =& \frac{\alpha\epsilon^2}{2(\alpha-1)}\lrp{C_D+B_D+E_D+E_J},
\end{align}
according to the definition in Eq.~\eqref{eq:Dtermscoe} and Eq.~\eqref{eq:Jtermscoe}.  Since these coefficients are still positive after all possible 2-site unitary perturbations, we conclude that there is non-trivial non-local magic in a generically perturbed $\ket{V^{(\epsilon)}}$.
In terms of the coefficients $k_i$'s, this is 
\begin{eqns}
    \mathcal{M}^{(3)}_\alpha\bigl(\ket{V^{(\epsilon)}}\bigr)=& \frac{\alpha \epsilon^2}{(\alpha-1)}\lrp{(1-\frac{1}{d^2})(1-\frac{1}{\bar d^2})(k_3+k_4)+(1-\frac{1}{d^2})k_2+(1-\frac{1}{\bar d^2})k_5}
\end{eqns}

Here we provide some intuition for why local or bipartite forms of magic do not produce a state-dependent proto-area. First, any local non-Clifford deformation of the code of the form
\begin{eqns} 
V\rightarrow O_A O_{A}^{\dagger} VO_A^{\dagger} O_{\bar A}^{\dagger}, 
\end{eqns} 
can be absorbed into the recovery unitaries $R_A$ and $R_{\bar A}$. Such deformations therefore leave both the boundary entropy and the bulk entropy unchanged.

Moreover, deformations that act separately on $A_1\bar A_1$ and on $A_2\bar A_2$ transform the bulk matter degrees of freedom $\ket{\psi}_L$ and the geometric auxiliary state $\ket{\chi}$ independently, without generating correlations between them. As discussed in Sec.~\ref{sec:2}, such matter--geometry correlations are necessary ingredients for nontrivial state dependence of the (proto-)area. We do not have a simple argument for why other bipartite unitaries do not produce state-dependence except by referring to the technical proof above.

\section{Monotonicity of functions}

\label{app:f1f2f3}
\subsection{Monotonicity of $f_1$}
\subsubsection*{Flat $\chi$ spectrum}
We will show that $f_1$ is a monotonic function of $\theta_k \in (0,\pi/4)$, where $\theta_k$ is an arbitrary angle defined in \eqref{eq:Psi_theta_vec}. We have defined $f_1$ in \eqref{eq:f1f2}, which can be rewritten in terms of the spectrum of $\sigma^{(L)}$ \eqref{eig-Lambda} as
\begin{align}
f_1(\{\lambda\}) &= - \frac{1}{2^n} \Tr({\ln \sigma^{(L)}}) + S(\sigma^{(L)})\\
& = \sum_{k=1}^n -\frac{1}{2}\bigl(\log p_k + \log(1-p_k)\bigr) - \bigl(p_k\log p_k + (1-p_k)\log(1-p_k)\bigr)\\
& = \sum_{k=1}^n (p_k - \frac{1}{2}) \log\frac{p_k}{1-p_k} :=  \sum_{k=1}^n g(p_k),
\end{align} 
where $S$ is the classical Shannon entropy $\sigma^{(L)}$ and we used the fact that the product structure \eqref{eq:bulkbell} implies that the two terms are the sum of the function of the single qubit spectrum. Now we only need to show the monotonicity of $g(p (\theta)).$

A direct differentiation gives
\begin{equation}
g'(p)
=
\ln\frac{p}{1-p}
+
\frac{2p-1}{2p(1-p)} >0,
\end{equation}
when $p\in(1/2,1)$ and $g'(1/2)=0$. On the other hand, we have $p(\theta)=\cos^2\theta\in[1/2,1)$ when $\theta\in(0,\pi/4]$, which gives 
\begin{equation}
\frac{dp}{d\theta}=-\sin(2\theta)<0 .
\label{eq: p_theta_der}
\end{equation}
By the chain rule, the following inequality holds:
\begin{equation}
\frac{d}{d\theta}\,g(p(\theta))
=
g'(p(\theta))\,\frac{dp}{d\theta}
<
0
\qquad
\text{for }\theta\in(0,\pi/4],
\end{equation}
with strict inequality for $\theta\in(0,\pi/4)$. \qed

\subsubsection*{Non-flat $\chi$ spectrum}
In the general setting of Theorems~\eqref{thm:general_mixed} and \eqref{thm:general_pure}, we defined the quantity $f_1^{mn}$ in Eq.~\eqref{eq:f1mn-def} by
\begin{equation}
f_1^{m n}(\lambda)=\frac{1}{2 d} \sum_{i, j} \frac{\left(\lambda_i-\lambda_j\right)^2}{\left(\mu_m \lambda_i-\mu_n \lambda_j\right)} \ln \frac{\mu_m \lambda_i}{\mu_n \lambda_j}\, .
\end{equation}
The expression is continuous, with the apparent singularity at $\mu_m \lambda_i=\mu_n \lambda_j$ being removable by taking the corresponding limit. We will show that $f_1^{mn}$ is a monotonic function of $\theta_k \in (0,\pi/4)$, where $\theta_k$ is an arbitrary angle introduced in \eqref{eq:Psi_theta_vec}. We also use the same parametrization of the eigenvalues introduced in Eq. \eqref{eig-Lambda}
Since $\mu_m\lambda_i>0$ and $\mu_n\lambda_j>0$, use
\begin{equation}
\ln\frac{x}{y}=\int_0^\infty ds\left(\frac{1}{s+y}-\frac{1}{s+x}\right),
\label{eq:lnxbyy}
\end{equation}
Then
\begin{equation}
f_1^{mn}(\lambda)
=\frac{1}{2d}\sum_{i, j}\int_0^\infty ds\, g^{mn}_s(\lambda_i,\lambda_j), \qquad g^{mn}_s(\lambda_i,\lambda_j):=\frac{(\lambda_i-\lambda_j)^2}{(s+\mu_m\lambda_i)(s+\mu_n\lambda_j)}.
\end{equation}
where
It suffices to prove that, for each fixed $s$, the integrand is pointwise monotonic in $\theta_k\in(0,\pi/4]$. To this end, rewrite the spectrum $\{\lambda_i\}$ as an explicit function of $p:=p_k$, keeping all remaining factors $(\ell\neq k)$ in~\eqref{eig-Lambda} fixed. The spectrum then decomposes into two families,
\[
\{\,p\,\omega_\beta,\ (1-p)\,\omega_\beta\,\}_{\beta\in\{0,1\}^{n-1}},
\]
where the weight associated with the $(n-1)$-bit string $\beta$ is
\begin{equation}
\omega_\beta=\prod_{\ell\neq k} p_\ell^{\,1-\beta_\ell}\bigl(1-p_\ell\bigr)^{\beta_\ell}.
\label{pl-weight_re}
\end{equation}
where $\beta_i$ denotes the $i$-th bit (component) of the string $\beta$.

Since each argument of $g^{mn}_s$ may come from either eigenvalue family, we can split the double sum into the four possible group pairings:
\begin{align}
\sum_{i, j} g^{mn}_s(\lambda_i,\lambda_j)
&=\sum_{\beta,\tau\in\{0,1\}^{n-1}}\Big[
g^{mn}_s(p\omega_\beta,\,p\omega_\tau)
+g^{mn}_s\big((1-p)\omega_\beta,\,p\omega_\tau\big) \notag\\
&\hspace{2cm}
+g^{mn}_s\big(p\omega_\beta,\,(1-p)\omega_\tau\big)
+g^{mn}_s\big((1-p)\omega_\beta,\,(1-p)\omega_\tau\big)
\Big],
\label{eq:hs-mn-4split}
\end{align}
Differentiating the integrand with respect to $p$ yields
\begin{align*}
\frac{\partial}{\partial p}\sum_{i,j} g_s(\lambda_i,\lambda_j)
=(2p-1)\,\mathcal{} \Upsilon_s(p),
\end{align*}
where we have introduced the shorthand
\begin{equation}
\Upsilon_s(p)
=\sum_{\beta,\tau}
\frac{X^{mn}_{\beta,\tau}(p;s)}
{\Big[\bigl(\mu_m\,\omega_\sigma\,p+s\bigr)
       \bigl(\mu_n\,\omega_\tau\,p+s\bigr)
       \bigl(\mu_m\,\omega_\sigma(1-p)+s\bigr)
       \bigl(\mu_n\,\omega_\tau(1-p)+s\bigr)\Big]^2}.
\end{equation}
with
\begin{align}
X^{mn}_{\beta,\tau}(p;s)
={}&\ \mu_m^{3}\,\omega_{\sigma}^{3}\,\omega_{\tau}^{2}\Big(
\mu_n^{3}(p-1)^{2}p^{2}\,\omega_{\tau}\big(\omega_{\sigma}^{2}+\omega_{\tau}^{2}\big)
-2\mu_n^{2}(p-1)p\,s\Big(-2(p-1)p\,\omega_{\sigma}\omega_{\tau}\nonumber\\
& \hspace{-10mm}-(p-1)p\,\omega_{\sigma}^{2}+\big((p-1)p+1\big)\omega_{\tau}^{2}\Big)
+\mu_n s^{2}\big(8(p-1)^{2}p^{2}\,\omega_{\sigma}+\omega_{\tau}\big)
+s^{3}\Big)
\nonumber\\
&\hspace{-10mm}+\mu_m^{2}s\,\omega_{\sigma}^{2}\omega_{\tau}\Big(
-2\mu_n^{3}(p-1)p\,\omega_{\tau}^{2}\Big(-2(p-1)p\,\omega_{\sigma}\omega_{\tau}+\big((p-1)p+1\big)\omega_{\sigma}^{2}-(p-1)p\,\omega_{\tau}^{2}\Big)
\nonumber\\
&\hspace{-10mm}
-4\mu_n^{2}(p-1)p\,s\,\omega_{\tau}\Big(\big((p-1)p+1\big)\omega_{\sigma}^{2}+\big((p-1)p+1\big)\omega_{\tau}^{2}+3\omega_{\sigma}\omega_{\tau}\Big)
+\mu_n s^{2}\omega_{\tau}\Big(2\big(1-\nonumber\\
&\hspace{-10mm}8(p-1)p\big)\omega_{\sigma}+3\omega_{\tau}\Big)
+2s^{3}\big(\omega_{\sigma}+2\omega_{\tau}\big)\Big)
+\mu_m s^{2}\omega_{\sigma}\Big(
\mu_n^{3}\omega_{\sigma}\omega_{\tau}^{3}\big(8(p-1)^{2}p^{2}\omega_{\tau}+\omega_{\sigma}\big)
\nonumber\\
&\hspace{-12mm}+\mu_n^{2}s\,\omega_{\sigma}\omega_{\tau}^{2}\Big(2\big(1-8(p-1)p\big)\omega_{\tau}+3\omega_{\sigma}\Big)
+\mu_n s^{2}\omega_{\tau}\Big(8\big(1-2(p-1)p\big)\omega_{\sigma}\omega_{\tau}+3\omega_{\sigma}^{2}+3\omega_{\tau}^{2}\Big)
\nonumber\\
\hspace{-10mm}&\hspace{-10mm}
+2s^{3}\Big(2\omega_{\sigma}\omega_{\tau}+\omega_{\sigma}^{2}+3\omega_{\tau}^{2}\Big)\Big)
+s^{3}\big(\mu_n\omega_{\tau}+2s\big)\Big(
\mu_n^{2}\omega_{\sigma}^{2}\omega_{\tau}^{2}
+2\mu_n s\,\omega_{\sigma}\omega_{\tau}(\omega_{\sigma}+\omega_{\tau})
\nonumber\\
\hspace{-10mm}&
\hspace{-10mm}+2s^{2}\big(\omega_{\sigma}^{2}+\omega_{\tau}^{2}\big)\Big).
\end{align}
A direct symbolic check in \textsc{Mathematica} shows that, for parameters in the range $0<\omega_\beta, \omega_\tau<1$, $s\ge 0$, $\tfrac12<p<1$, $0<\mu_m,\mu_n<1$, the numerator $X^{mn}_{\beta,\tau}(p;s)$ is strictly positive. Consequently the summand in $\Upsilon_s(p)$ is positive term-by-term, and therefore
\[
\Upsilon_s(p)>0 \qquad \text{for all } p\in\Big(\tfrac12,1\Big).
\]
Therefore, using Eq.~\eqref{eq: p_theta_der} and the chain rule, for $\theta_k\in(0,\pi/4]$ we obtain
\begin{equation}
\frac{\partial f_1^{mn}}{\partial \theta_k}
=\frac{\partial f_1^{mn}}{\partial p}\,\frac{\partial p}{\partial \theta_k}
<0.
\end{equation} \qed
\subsection{Monotonicity of $f_2$}
\subsubsection*{Flat $\chi$ spectrum}
Recall that we have defined~\eqref{eq:f1f2}:
\begin{equation}
   f_2(\{ \lambda_i \})=\left(\frac {1}{2^{n-1}}+ \sum_{i\neq j}\frac {1}{2^{n+1}}\frac{(\lambda_i+\lambda_j)^2}{\lambda_i-\lambda_j}\ln{\frac{\lambda_i}{\lambda_j}}\right)-2
   \label{eq:f2_re}
\end{equation}
and we want to know that it is a monotonic function $\theta_k \in(0,\pi/4]$ for all $k$. 
To study the $\theta_k$–dependence of $f_2$, we use the same parametrization of the eigenvalues introduced in Eq. \eqref{eig-Lambda} and differentiate with respect to fixed $p_k$. Since $\lambda_i>0$, we substitute Eq~\eqref{eq:lnxbyy} in Eq~\eqref{eq:f2_re} and the  part of $f_2$ in the parenthesis can be written as:
\begin{align}
    F(\{\lambda\}) = \int_0^\infty ds \sum_{i, j} \frac{(\lambda_i + \lambda_j)^2}{(s+\lambda_i) (s + \lambda_j)} := \sum_{i, j} \int_0^\infty ds h_s(\lambda_i,\lambda_j).
\end{align}
where we have defined
\begin{equation}
    h_s(\lambda_i,\lambda_j)= \frac{(\lambda_i + \lambda_j)^2}{(s+\lambda_i) (s + \lambda_j)} 
\end{equation}
It suffices to show that the integrand is a point-wise monotonic function of $\theta_k \in (0, \pi/4]$ in $s$. Rewrite the spectrum $\lambda$ as an explicit function of $p:=p_k$ for a fixed $k$ and all the other factors $(i \neq k)$ in~\eqref{eig-Lambda}. Then we can split the spectrum into two groups: $\{p \omega_\beta$, $(1-p) \omega_\beta\}_{\beta \in \{0, 1\}^{n-1}}$ where the $(n-1)$-bit string dependent $\omega_\beta$ is defined as Eq~\eqref{pl-weight_re}. Since the two arguments of $h_s$ can come from either group of eigenvalues, the sum can be split up into 4 terms that contains all pairs of group inputs, as follows
\begin{align}
\sum_{i, j} h_s(\lambda_i,\lambda_j) 
    & =\sum_{\beta, \tau \in \{0,1\}^{n-1}}\Bigg[ h_s(p \omega_\beta,  p \omega_\tau) + h_s( (1-p)\omega_\beta, p \omega_\tau)) + h_s(p \omega_\beta, (1-p) \omega_\tau) \\
    &+h_s((1-p) \omega_\beta, (1-p) \omega_\tau) \Bigg] 
\end{align}
Differentiate the integrand with respect to $p $ gives
\begin{align*}
    \frac{\partial}{\partial p}\sum_{i, j} h_s(\lambda_i,\lambda_j) 
 = (2p-1)\,\Phi_s(p),
\end{align*}
where 
\begin{equation}
\Phi_s(p,s)=
\sum_{\beta,\tau}\frac{K_{\beta,\tau}(p;s)}{\Big(\left(-p w_{\beta }+s+w_{\beta }\right) \left(p w_{\beta }+s\right) \left(-p w_{\tau }+s+w_{\tau }\right) \left(p w_{\tau }+s\right)\Big)^2} 
\end{equation}
and 
\begin{eqns}
    K_{\beta,\tau}(p;s)=&(p-1)^2 p^2 w_{\beta }^5 w_{\tau }^2 \left(2 s+w_{\tau }\right)+2 s w_{\beta }^2 \left(\left(8 p^2-8 p+1\right) s^2 w_{\tau }^3+(p-1)^2 p^2 w_{\tau }^5\right. \\ & \left. +2 p \left(-3 p^3+6 p^2-4 p+1\right) s w_{\tau }^4+8 (p-1) p s^3 w_{\tau }^2+2 s^5+s^4 w_{\tau }\right) \\ & +w_{\beta }^3 \left(2 \left(8 p^2-8 p+1\right) s^3 w_{\tau }^2+2 \left(6 p^2-6 p+1\right) s^2 w_{\tau }^3+(p-1)^2 p^2 w_{\tau }^5\right. \\ & \left. +2 p \left(-3 p^3+6 p^2-4 p+1\right) s w_{\tau }^4+2 s^5+s^4 w_{\tau }\right)+s^4 w_{\beta } w_{\tau }^2 \left(2 s+w_{\tau }\right) \\ &  -2 p \left(3 p^3-6 p^2+4 p-1\right) s w_{\beta }^4 w_{\tau }^2 \left(2 s+w_{\tau}\right)+2 s^5 w_{\tau }^2 \left(2 s+w_{\tau }\right).
\end{eqns}
One can verify that $K_{\beta,\tau}(p;s)$ is a strictly positive function for all
$0<w_\beta,w_\tau<1$, $s\ge0$, and $1/2<p<1$, using \textsc{Mathematica} and $\Phi_s(p)>0$ for $p\in(1/2,1)$. 
Hence, by chain rule, when $p\in(1/2,1)$ or $\theta \in (0, \pi/4]$,
\begin{equation}
    \frac{\partial f_2}{\partial \theta} = \frac{\partial f_2}{\partial p}\frac{\partial p}{\partial \theta} <0
\end{equation}
where we have applied equation~\eqref{eq: p_theta_der}. \qed

\subsubsection*{Non flat $\chi$ spectrum}
We define
\begin{equation}
f_2^{mn}(\lambda)
=\frac{1}{2d}\sum_{i j}\frac{(\lambda_i+\lambda_j)^2}{\mu_m\lambda_i-\mu_n\lambda_j}
\ln\frac{\mu_m\lambda_i}{\mu_n\lambda_j}
\end{equation}
We understand the potentially singular terms with $\mu_m\lambda_i=\mu_n\lambda_j$ by continuous extension, which is finite and well-defined. Substituting the following
\begin{equation}
\ln\frac{x}{y}=\int_{0}^{1}\frac{x-y}{s\,x+(1-s)\,y}\,ds,
\qquad x,y>0.
\label{eq:lnxbyy}
\end{equation}
into \eqref{eq:f2mn-def}, we get
\begin{align}
f_2^{mn}(p)
&=\frac{1}{2d}
\int_{0}^{1}ds\sum_{i,j}\Bigg(\frac{(\lambda_i+\lambda_j)^2}{s\,\mu_m\lambda_i+(1-s)\,\mu_n\lambda_j}\Bigg)
=\frac{1}{2d}\int_{0}^{1} ds\;\sum_{i,j} h^{mn}_s(\lambda_i,\lambda_j),
\label{eq:f2mn-s01}
\end{align}
where we define
\begin{equation}
h^{mn}_s(x,y):=\frac{(x+y)^2}{s\,\mu_m x+(1-s)\,\mu_n y}.
\label{eq:hs-def-alt}
\end{equation}
Next, as in the flat-spectrum case, we partition the index set for ${\lambda_i}$ into the two subsets corresponding to $\{p\,\omega_\beta\}_{\beta\in\{0,1\}^{n-1}}$ and $\{(1-p)\,\omega_\beta\}_{\beta\in\{0,1\}^{n-1}}$. Therefore,
\begin{align}
f_2^{mn}(p)
&=\frac{1}{2d}\int_{0}^{1} ds\;\sum_{\beta,\gamma}\Big[
h^{mn}_s\!\big(p\omega_\beta,\,p\omega_\gamma\big)
+h^{mn}_s\!\big(p\omega_\beta,\,(1-p)\omega_\gamma\big)
+h^{mn}_s\!\big((1-p)\omega_\beta,\,p\omega_\gamma\big)
\notag \\& \qquad \qquad+h^{mn}_s\!\big((1-p)\omega_\beta,\,(1-p)\omega_\gamma\big)
\Big].
\end{align}
For each fixed $s\in(0,1)$ and $(\beta,\gamma)$, define the four-term block
\begin{align}
F^{mn}_s(p;\beta,\gamma)
&:=h^{mn}_s\!\big(p\omega_\beta,\,p\omega_\gamma\big)
+h^{mn}_s\!\big(p\omega_\beta,\,(1-p)\omega_\gamma\big)
+h^{mn}_s\!\big((1-p)\omega_\beta,\,p\omega_\gamma\big)
\\& \qquad+h^{mn}_s\!\big((1-p)\omega_\beta,\,(1-p)\omega_\gamma\big).
\label{eq:F-fourterm-block}
\end{align}
By inspection, $F^{mn}_s(p;\beta,\gamma)=F^{mn}_s(1-p;\beta,\gamma)$, so $F^{mn}_s$ is symmetric about $p=\tfrac12$.
We now show that $F^{mn}_s(p;\beta,\gamma)$ is convex in $p$. First note that the first and last terms in
\eqref{eq:F-fourterm-block} are affine in $p$:
Hence
\begin{equation}
\frac{\partial^2}{\partial p^2}h^{mn}_s\!\big(p\omega_\beta,\,p\omega_\gamma\big)=0,
\qquad
\frac{\partial^2}{\partial p^2}h^{mn}_s\!\big((1-p)\omega_\beta,\,(1-p)\omega_\gamma\big)=0.
\label{eq:affine-secondzero}
\end{equation}
It remains to treat the two cross terms. Let
\begin{equation}
g(p):=h^{mn}_s\!\big(p\omega_\beta,\,(1-p)\omega_\gamma\big)
=\frac{\big((\omega_\beta-\omega_\gamma)p+\omega_\gamma\big)^2}
{\big(s\mu_m\omega_\beta-(1-s)\mu_n\omega_\gamma\big)p+(1-s)\mu_n\omega_\gamma}.
\end{equation}
This has the form $g(p)=\dfrac{(a p+b)^2}{c p+d}$ with $cp+d>0$, and therefore
\begin{equation}
g''(p)=\frac{2(ad-bc)^2}{(c p+d)^3}\ge 0,
\qquad p\in[1/2,1).
\end{equation}
The other cross term $h^{mn}_s\!\big((1-p)\omega_\beta,\,p\omega_\gamma\big)$ is obtained by swapping
$\beta\leftrightarrow\gamma$ and $p\leftrightarrow(1-p)$, hence it is convex as well.
 Consequently,
$F^{mn}_s(p;\beta,\gamma)$ is convex in $p$ for each fixed $(s,\beta,\gamma)$. Hence for each fixed $(s,\beta,\gamma)$, and sums/integrals preserve convexity, it follows from
\[
f_2^{mn}(p)=\frac{1}{2d}\int_0^1 ds\;\sum_{\beta,\gamma} F^{mn}_s(p;\beta,\gamma)
\]
that $f_2^{mn}(p)$ is convex in $p$. Moreover, $f_2^{mn}(p)$ is symmetric under $p\leftrightarrow(1-p)$. Hence
\begin{equation}
f_2^{mn}(p)=f_2^{mn}(1-p).
\label{eq:f2-sym}
\end{equation}
For a convex function symmetric about $p=\frac{1}{2}$, the minimum is attained at $p=\frac{1}{2}$, and the function is nondecreasing on $[\frac{1}{2},1]$. Therefore \footnote{In particular, taking $\mu_m=\mu_n=1$ recovers the function $f_2$ considered in the flat case, so the above provides an independent second proof of convexity/monotonicity for $f_2$},
\begin{equation}
\frac{d}{dp}f_2^{mn}(p)\ge 0,
\qquad p\in\Big[\frac{1}{2},1\Big).
\label{eq:f2-monotone}
\end{equation}
\qed

\subsection{Monotonicity of $f_3$}
\subsubsection*{Flat $\chi$ spectrum}
We will employ the same strategy by singling out $p:=p_k$, dividing the spectrum into two groups
$\{p\,\omega_\beta,(1-p)\omega_\beta\}_{\beta\in\{0,1\}^{n-1}}$, and differentiating $f_3$ with respect to $p$.

Notice that
\begin{equation}
   f_3(\{\lambda_i(p)\})
   := \sum_{i,j}\frac{\lambda_i+\lambda_j}{2(\lambda_i-\lambda_j)}\ln\frac{\lambda_i}{\lambda_j}
\end{equation}
depends only on the ratios between the eigenvalues $\{\lambda_i\}$. Let
\(
t_{ij}:=\lambda_i/\lambda_j
\)
and define the function on the positive axis
\begin{equation}
g(t):=\frac{t+1}{2(t-1)}\ln t,
\qquad g(1)=1
\end{equation}
by continuity. Note that $g(t)=g(t^{-1})$, and
\begin{equation}
g'(t)=\frac{t^2-1-2t\log t}{2(t-1)^2 t}.
\end{equation}
Now the monotonicity of $f_3$ can be repackaged in terms of all possible ratios between pairs of eigenvalues:
\begin{align}
    f_3(\lambda)
    &= \sum_{\sigma,\tau\in\{0,1\}^{n-1}}
    g\!\left(\frac{\omega_\sigma}{\omega_\tau}\right)
    +g\!\left(\frac{\omega_\sigma p}{\omega_\tau(1-p)}\right)
    +g\!\left(\frac{\omega_\sigma(1-p)}{\omega_\tau p}\right)
    +g\!\left(\frac{\omega_\tau}{\omega_\sigma}\right) \nonumber\\
    &= \sum_{\sigma,\tau\in\{0,1\}^{n-1}}
    2\,g\!\left(\frac{\omega_\sigma}{\omega_\tau}\right)
    +g\!\left(\frac{\omega_\sigma p}{\omega_\tau(1-p)}\right)
    +g\!\left(\frac{\omega_\sigma(1-p)}{\omega_\tau p}\right).
\end{align}

Let
\begin{equation}
t:=\frac{p}{1-p}.
\end{equation}
Then $t\ge 1$ for $p\in[1/2,1)$. We will show that the derivative of $f_3$ with respect to $p$,
\begin{align}
    \frac{\partial f_3}{\partial p}
    &=\sum_{\sigma,\tau\in\{0,1\}^{n-1}}
    \frac{1}{p(1-p)}
    \left[
    \frac{\omega_\sigma}{\omega_\tau}\,t\, g'\!\left(\frac{\omega_\sigma}{\omega_\tau}\,t\right)
    -\frac{\omega_\sigma}{\omega_\tau}\,\frac{1}{t}\, g'\!\left(\frac{\omega_\sigma}{\omega_\tau}\,\frac{1}{t}\right)
    \right] \nonumber\\
    &=\frac{1}{p(1-p)}
    \sum_{\sigma,\tau\in\{0,1\}^{n-1}}
    \left[
    h\!\left(\frac{\omega_\sigma}{\omega_\tau}\,t\right)
    -h\!\left(\frac{\omega_\sigma}{\omega_\tau}\,\frac{1}{t}\right)
    \right],
    \label{eq:df3dp}
\end{align}
takes a definite sign on $p\in(1/2,1]$ by showing that $h(u):=u\,g'(u)$
is monotone on $u\in(0,\infty)$. Since $\frac{\omega_\sigma}{\omega_\tau}>0$ and $t\ge 1$, we have
\[
\frac{\omega_\sigma}{\omega_\tau}\,t
\;\ge\;
\frac{\omega_\sigma}{\omega_\tau}\,\frac{1}{t},
\]
so if $h$ is increasing, then each bracket term in \eqref{eq:df3dp} is nonnegative. By definition,
\begin{equation}
h(t)=\frac{t^2-1-2t\log t}{2(t-1)^2},
\end{equation}
and its derivative is
\begin{equation}
h'(t)=\frac{t+1}{(t-1)^3}\left(\log t-\frac{2(t-1)}{t+1}\right)
=: \frac{t+1}{(t-1)^3}\,q(t).
\end{equation}
Now
\begin{equation}
q'(t)=\frac{(t-1)^2}{t(t+1)^2}\ge 0,
\qquad q(1)=0.
\end{equation}
Hence $q(t)\le 0$ for $0<t<1$ and $q(t)\ge 0$ for $t>1$. Since $q(t)$ has the same sign as $(t-1)$, it follows that
\begin{equation}
h'(t)\ge 0
\qquad\text{for all }t>0.
\end{equation}
Therefore $h$ is increasing on $(0,\infty)$, and thus every bracket term in \eqref{eq:df3dp} is nonnegative when
$p\in(1/2,1]$. This shows that
\begin{equation}
\frac{\partial f_3}{\partial p}\ge 0
\qquad\text{for }p\in(1/2,1].
\end{equation}

By the chain rule,
\begin{equation}
    \frac{d f_3}{d\theta_k}
=
\frac{d f_3}{dp_k}\,\frac{dp_k}{d\theta_k}
\le 0
\end{equation}
when $\theta_k\in(0,\pi/4]$ as before. \qed

\subsubsection*{Non-flat $\chi$ spectrum}
We define
\begin{equation}
f_3^{mn}(\lambda)
=\frac{1}{2d}\sum_{i j}\mathcal P^{mn}(\lambda_i,\lambda_j),
\qquad
\mathcal P^{mn}(x,y):=\frac{x+y}{2(\mu_m x-\mu_n y)}\ln\frac{\mu_m x}{\mu_n y},
\label{eq:def-f3mn}
\end{equation}
with $\mathcal P^{mn}(x,x)$ understood by continuous extension. Fix $(m,n)$ and set $\gamma^{mn}:=\mu_n/\mu_m$. Since $\mathcal P^{mn}(cx,cy)=\mathcal P^{mn}(x,y)$ for all $c>0$,
it depends only on the ratio $t:=\frac{x}{y}>0$.

Define the associated one-variable function
\begin{equation}
g^{mn}(t):=\frac{t+1}{2(t-\gamma^{mn})}\,\ln\!\Big(\frac{t}{\gamma^{mn}}\Big),
\qquad t>0,
\qquad 
g^{mn}(\gamma^{mn}):=\frac{\gamma^{mn}+1}{2\gamma^{mn}},
\label{eq:def-gmn-xy}
\end{equation}
so that
\begin{equation}
\mathcal P^{mn}(x,y)=\frac{1}{\mu_m}\,g^{mn}\!\Big(\frac{x}{y}\Big),
\qquad
f_3^{mn}(\lambda)=\frac{1}{2d\,\mu_m}\sum_{i,j} g^{mn}\!\Big(\frac{\lambda_i}{\lambda_j}\Big).
\label{eq:f3mn-sumg-xy}
\end{equation}

Fix an index $k$ and vary $p:=p_k\in[1/2,1)$, holding all other $p_\ell$ fixed. As in the flat-spectrum case, the
spectrum can be written as
\[
\{\,w_\sigma p,\; w_\sigma(1-p)\,\}_{\sigma\in\{0,1\}^{n-1}},
\]
where $w_\sigma>0$ depends only on $\{p_\ell\}_{\ell\neq k}$. For each ordered pair $(\sigma,\tau)$ define the ratio
\(
n_{\sigma\tau}:=\frac{w_\sigma}{w_\tau}>0.
\)
Then the contribution of the four ratios formed by $\{w_\sigma p,w_\sigma(1-p)\}$ against $\{w_\tau p,w_\tau(1-p)\}$
can be packaged into the two-level quantity
\begin{equation}
\mathcal Q^{mn}_n(p):=\sum_{a,b\in\{0,1\}} g^{mn}\!\Big(n\,\frac{q_a}{q_b}\Big),
\qquad (q_0=p,\ q_1=1-p),
\label{eq:def-Qmn-xy}
\end{equation}
and hence
\begin{equation}
f_3^{mn}(p)=\frac{1}{2d\,\mu_m}\sum_{\sigma,\tau}\mathcal Q^{mn}_{n_{\sigma\tau}}(p).
\label{eq:f3mn-sumQ}
\end{equation}

Expanding \eqref{eq:def-Qmn-xy} gives
\begin{equation}
\mathcal Q^{mn}_n(p)
=
2\,g^{mn}(n)
+g^{mn}\!\Big(n\,\frac{p}{1-p}\Big)
+g^{mn}\!\Big(n\,\frac{1-p}{p}\Big).
\label{eq:Qmn-expanded-xy}
\end{equation}
The first term is independent of $p$, so it suffices to control the last two terms. Let
\begin{equation}
t:=\frac{p}{1-p}\ge 1\quad\text{for }p\in\Big[\frac12,1\Big),
\qquad\text{so that}\qquad
\frac{dt}{dp}=\frac{t}{p(1-p)}.
\label{eq:t-odds}
\end{equation}
Differentiating \eqref{eq:Qmn-expanded-xy} using \eqref{eq:t-odds} yields
\begin{align}
\frac{d}{dp}\mathcal Q^{mn}_n(p)
&=\frac{t}{p(1-p)}\left[
n\,(g^{mn})'(nt)-n\,t^{-2}(g^{mn})'(n/t)
\right]
\nonumber\\
&=\frac{1}{p(1-p)}\left[
h^{mn}(nt)-h^{mn}(n/t)
\right],
\label{eq:dQmn-dp-xy}
\end{align}
where we define
\begin{equation}
h^{mn}(x):=x\,(g^{mn})'(x).
\label{eq:def-hmn-xy}
\end{equation}
Thus $\frac{d}{dp}\mathcal Q^{mn}_n(p)$ has a definite sign on $p\in[1/2,1)$ once we know that $h^{mn}$ is monotone
increasing on $(0,\infty)$, since $t\ge 1$ implies $nt\ge n/t$.

A direct computation  gives
\begin{equation}
(h^{mn})'(x)
=\frac{1+\gamma^{mn}}{2(\gamma^{mn})^2}\,
\frac{(u+1)\ln u-2(u-1)}{(u-1)^3}.
\label{eq:hmn-prime-u}
\end{equation}
where we defined $u:=x/\gamma^{mn}>0$. Define $q(u):=\ln u-\frac{2(u-1)}{u+1}$. One checks
\[
q'(u)=\frac{(u-1)^2}{u(u+1)^2}\ge 0,
\qquad q(1)=0,
\]
so $q(u)\le 0$ for $0<u<1$ and $q(u)\ge 0$ for $u>1$. Since
\(
(u+1)\ln u-2(u-1)=(u+1)q(u)
\)
has the same sign as $(u-1)$, the fraction in \eqref{eq:hmn-prime-u} is nonnegative for all $u>0$. Hence
\begin{equation}
(h^{mn})'(x)\ge 0\quad\text{for all }x>0,
\label{eq:hmn-increasing}
\end{equation}
so $h^{mn}$ is increasing on $(0,\infty)$. Therefore, for $p\in[1/2,1)$ we have $t\ge 1$ and thus $nt\ge n/t$, implying from \eqref{eq:dQmn-dp-xy} that
\[
\frac{d}{dp}\mathcal Q^{mn}_n(p)\ge 0.
\]
Summing over $(\sigma,\tau)$ preserves monotonicity, and \eqref{eq:f3mn-sumQ} gives
\begin{equation}
\frac{d}{dp}f_3^{mn}(p)\ge 0,
\qquad p\in\Big[\frac12,1\Big).
\label{eq:f3mn-monotone}
\end{equation}
Finally, with $p_k=\cos^2\theta_k$ and $\theta_k\in(0,\pi/4]$, we have $\frac{dp_k}{d\theta_k}=-\sin(2\theta_k)\le 0$, and thus
\[
\frac{d f_3^{mn}}{d\theta_k}
=
\frac{d f_3^{mn}}{dp_k}\,\frac{dp_k}{d\theta_k}
\le 0,
\]
as claimed. \qed
\newpage
\section{Trace diagram dictionary}
\subsection{W, $\mathscr{R}$ diagrams for mixed case}

\label{red-WRdiag}

\begin{longtable}{C{2.5cm} C{12cm}}
\caption{Trace diagram key for \(W_i\) and \(\mathscr{R}_i\) for the mixed bulk case.
The first, second, and third legs (from outermost to innermost) represent the 
Hilbert spaces \(\mathcal{H}_{A_1}\), \(\mathcal{H}_{A_2}\), and \(\mathcal{H}_{\bar A}\), 
respectively. If a reference leg is present, it is denoted \(r\) and corresponds 
to the Hilbert space \(\mathcal{H}_{r}\). In that case, the ordering becomes 
\(\mathcal{H}_{r}\), \(\mathcal{H}_{A_1}\), \(\mathcal{H}_{A_2}\), 
\(\mathcal{H}_{\bar A}\) from outermost to innermost. 
}
\label{W-R-diagrams} \\
\toprule
\textbf{Label} & \textbf{Trace diagram} \\
\midrule
\endfirsthead

\toprule
\textbf{Label} & \textbf{Trace diagram} \\
\midrule
\endhead

\midrule
\endfoot

\bottomrule
\endlastfoot

$W_1$ &
\includegraphics[height=20mm]{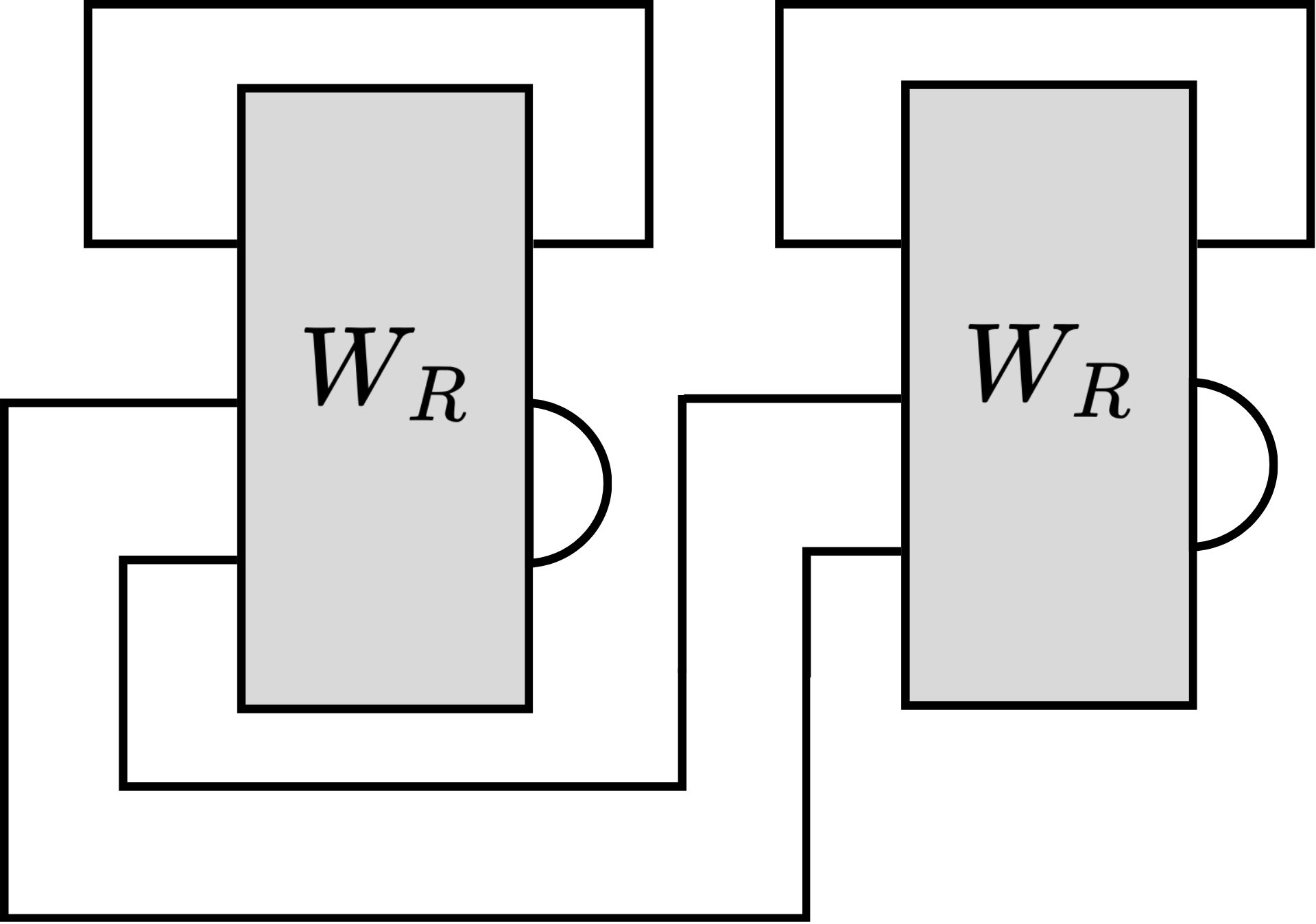} \\[6pt]

$W_2$ &
\includegraphics[height=20mm]{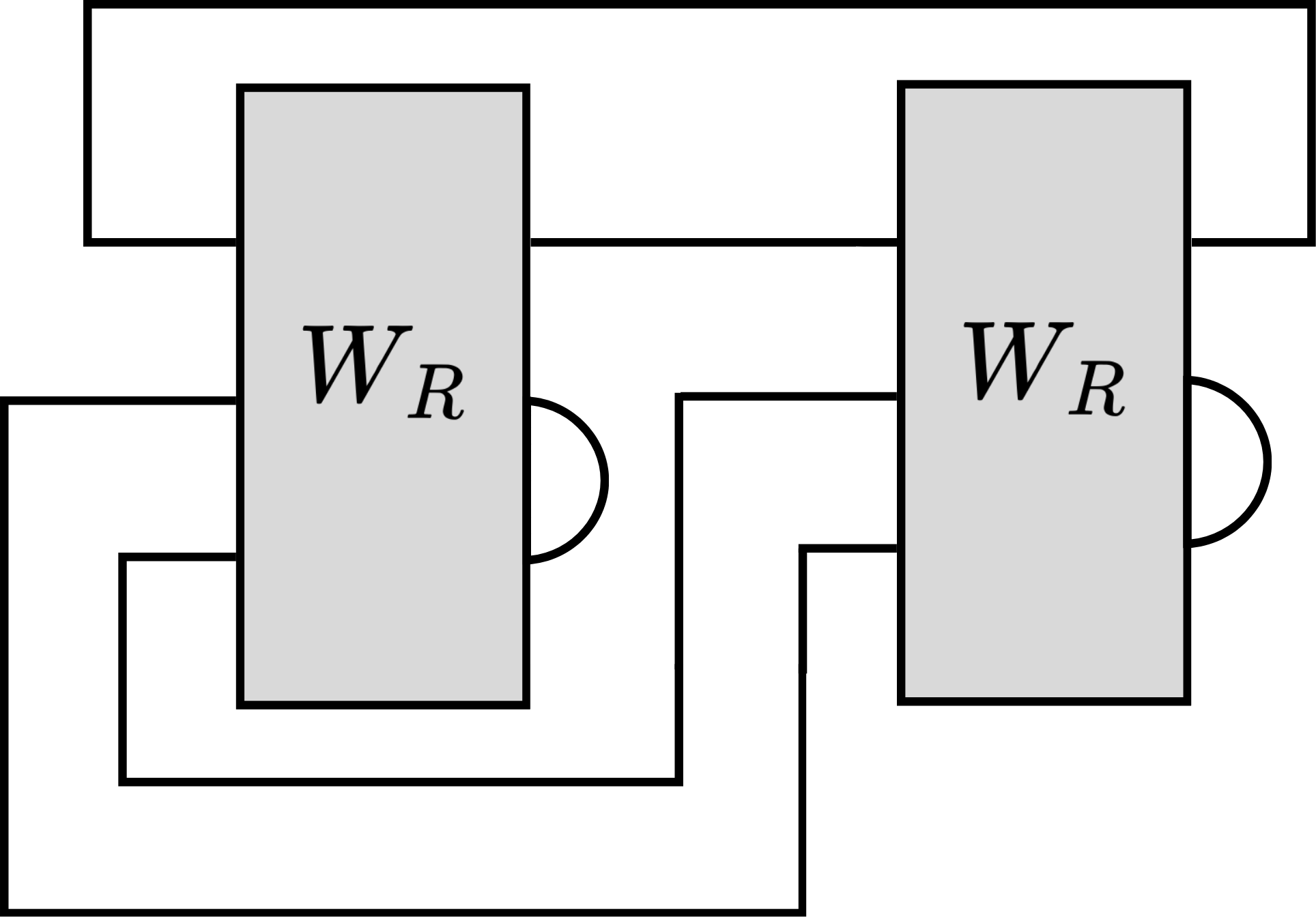} \\[6pt]

$W_3$ &
\includegraphics[height=20mm]{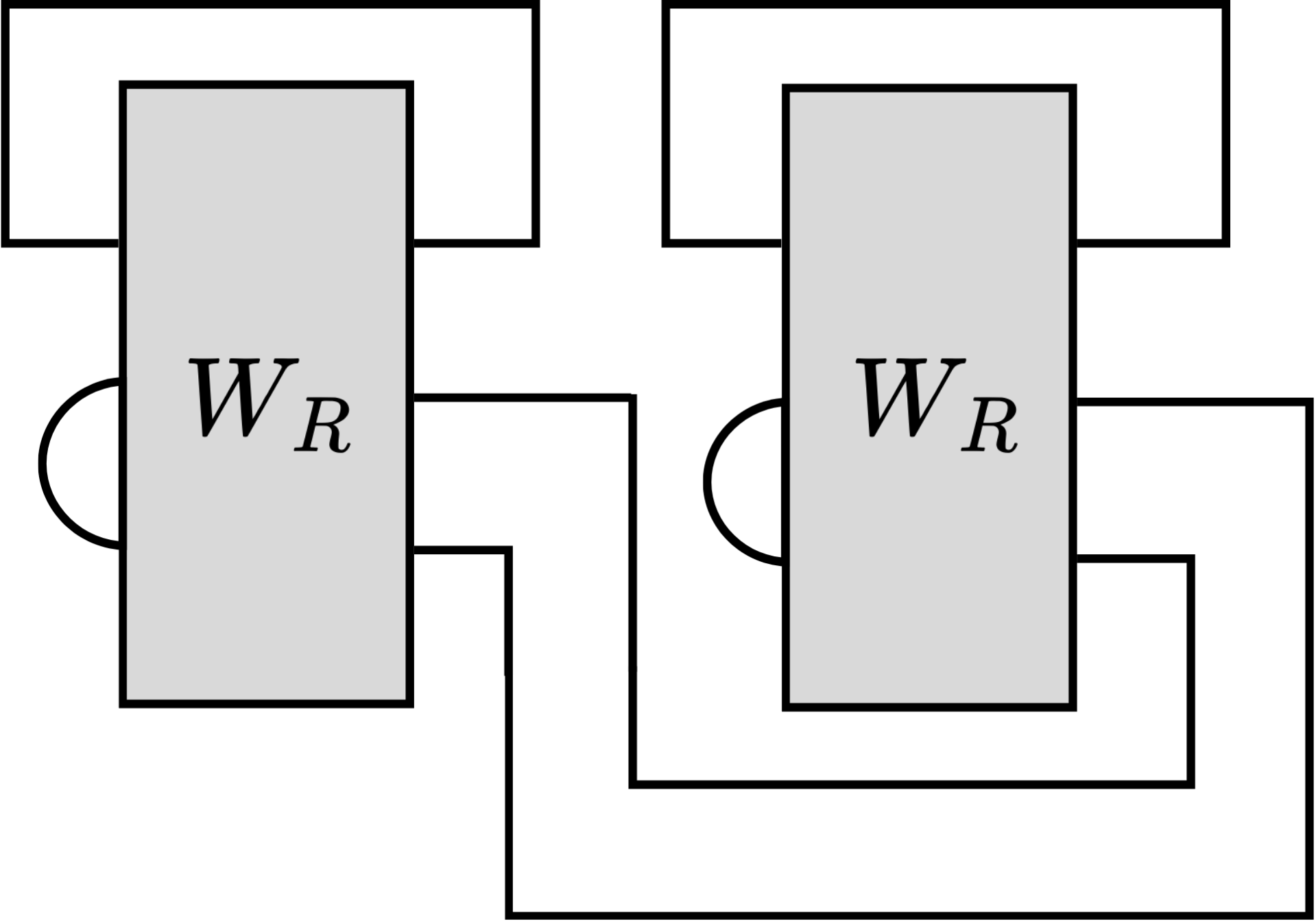} \\[6pt]

$W_4$ &
\includegraphics[height=21mm]{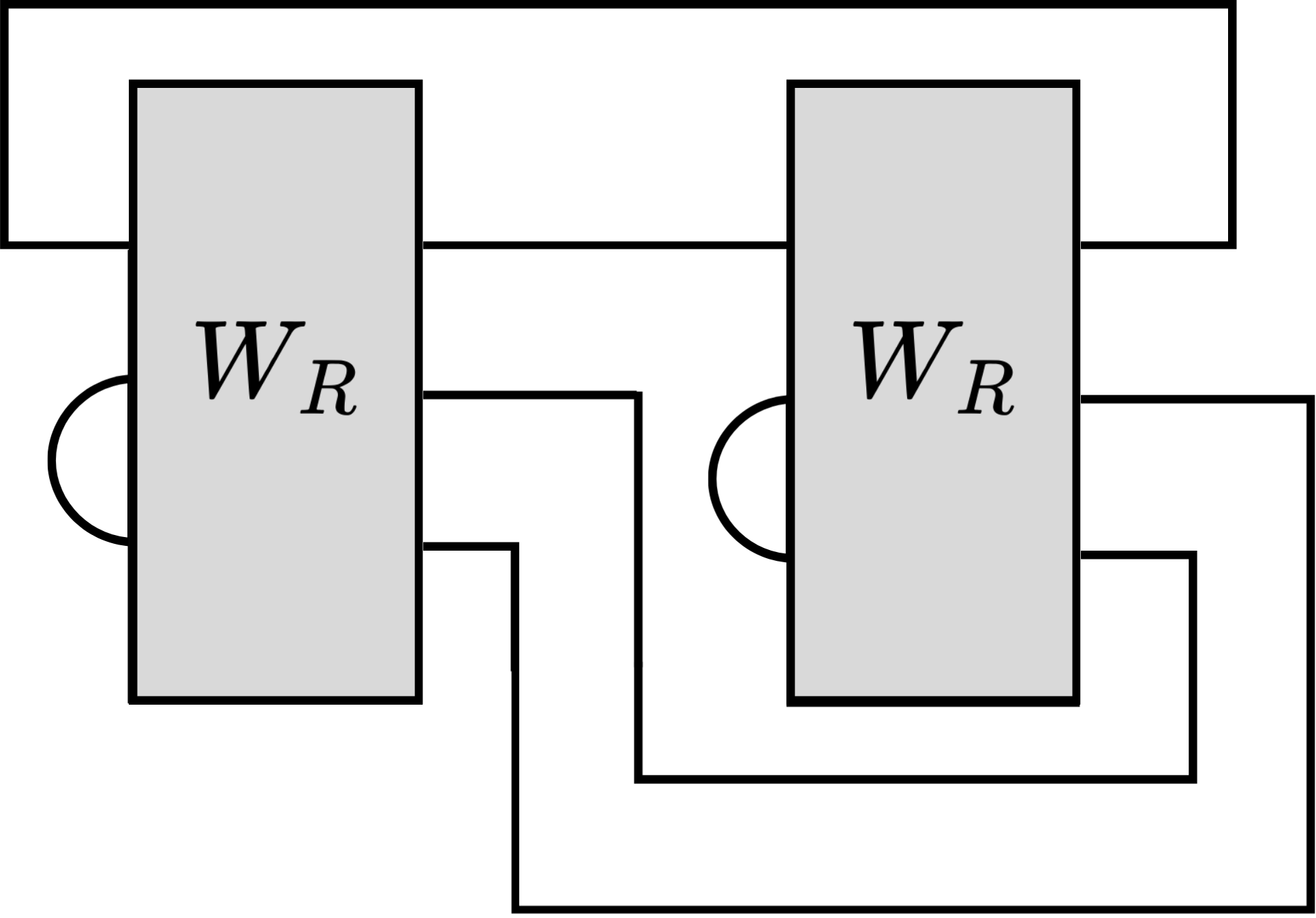} \\[6pt]

$W_5$ &
\includegraphics[height=22mm]{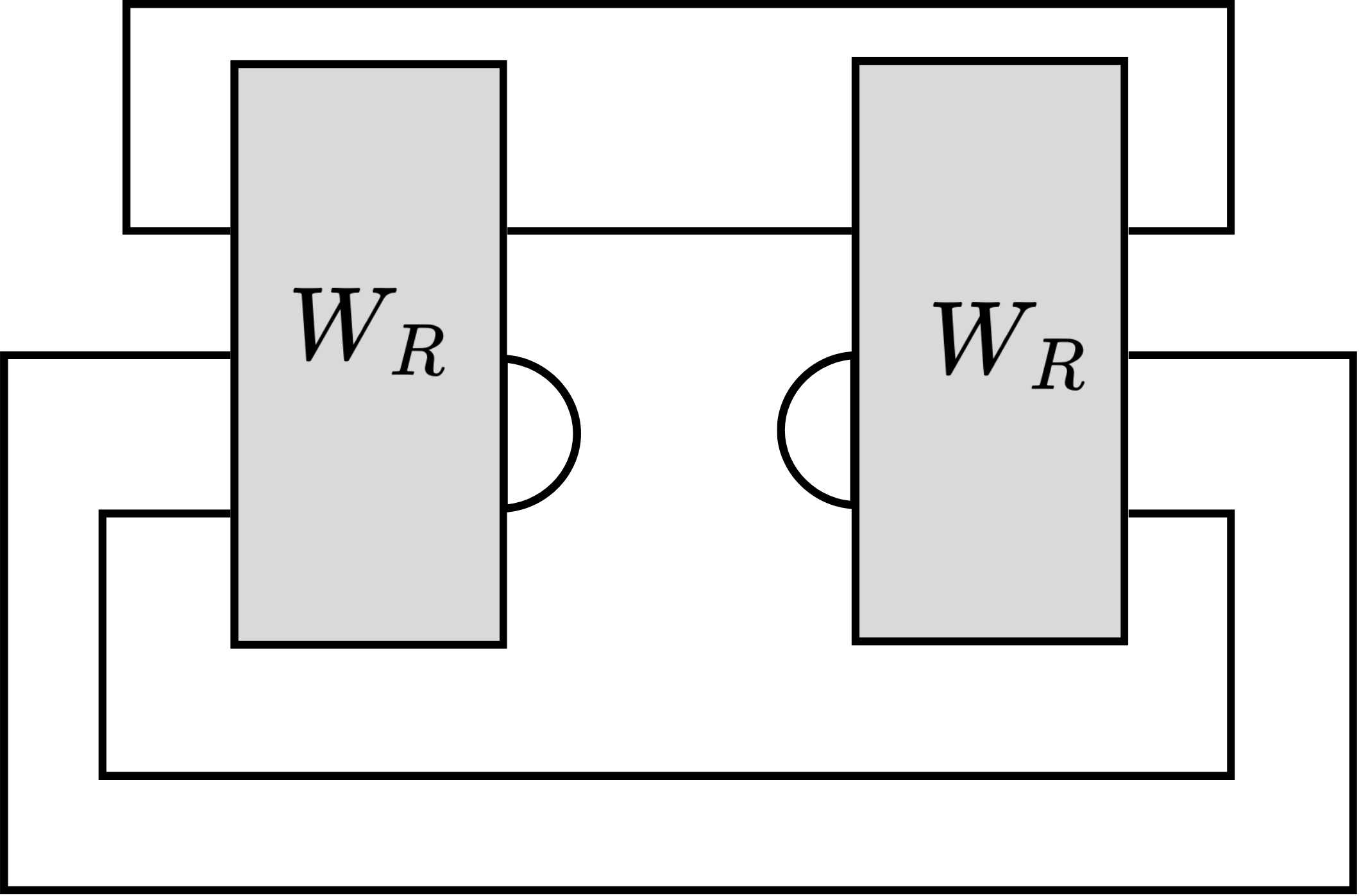} \\[6pt]

$W_6$ &
\includegraphics[height=23mm]{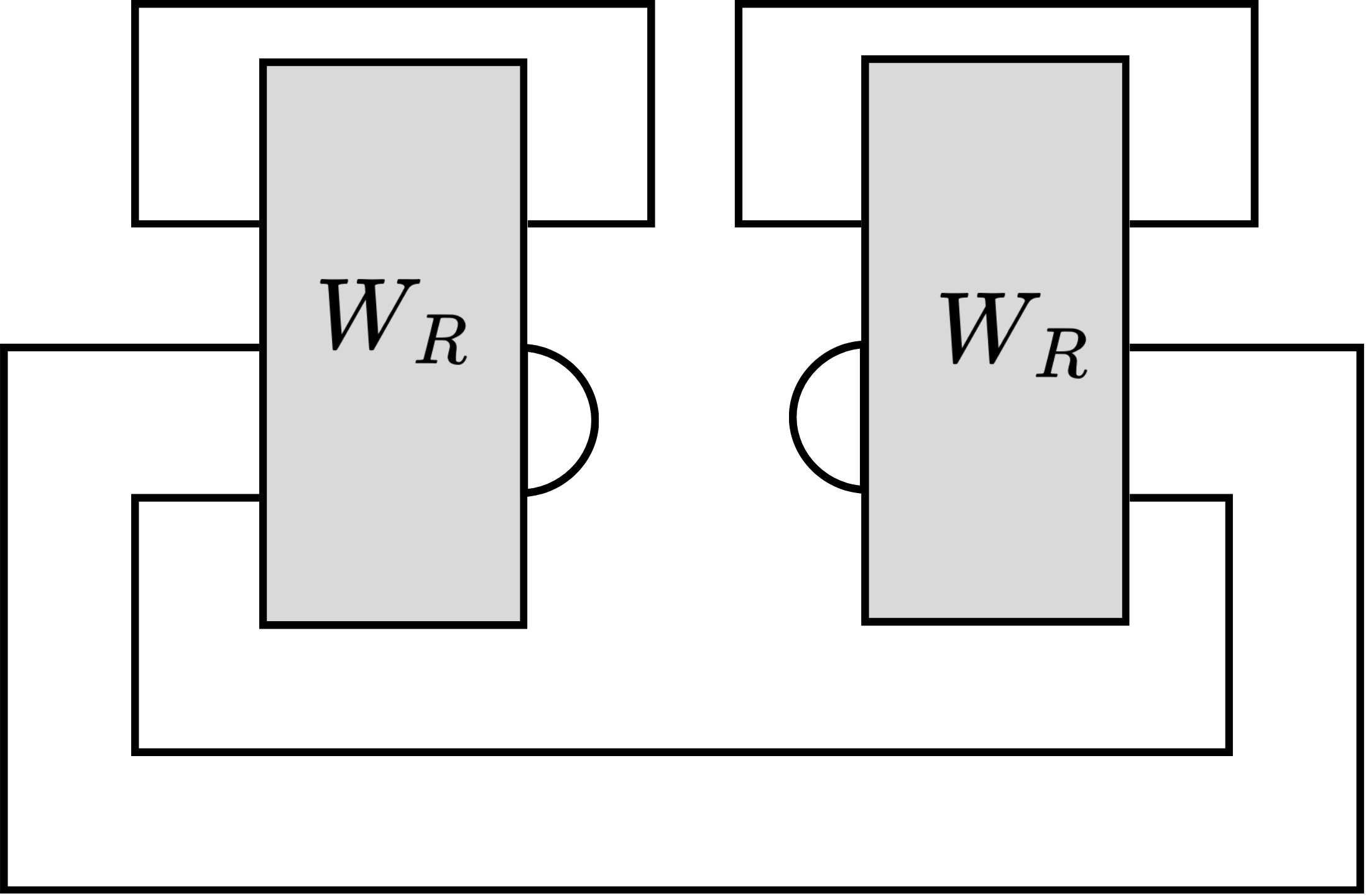} \\[6pt]

$W_7$ &
\includegraphics[height=16.5mm]{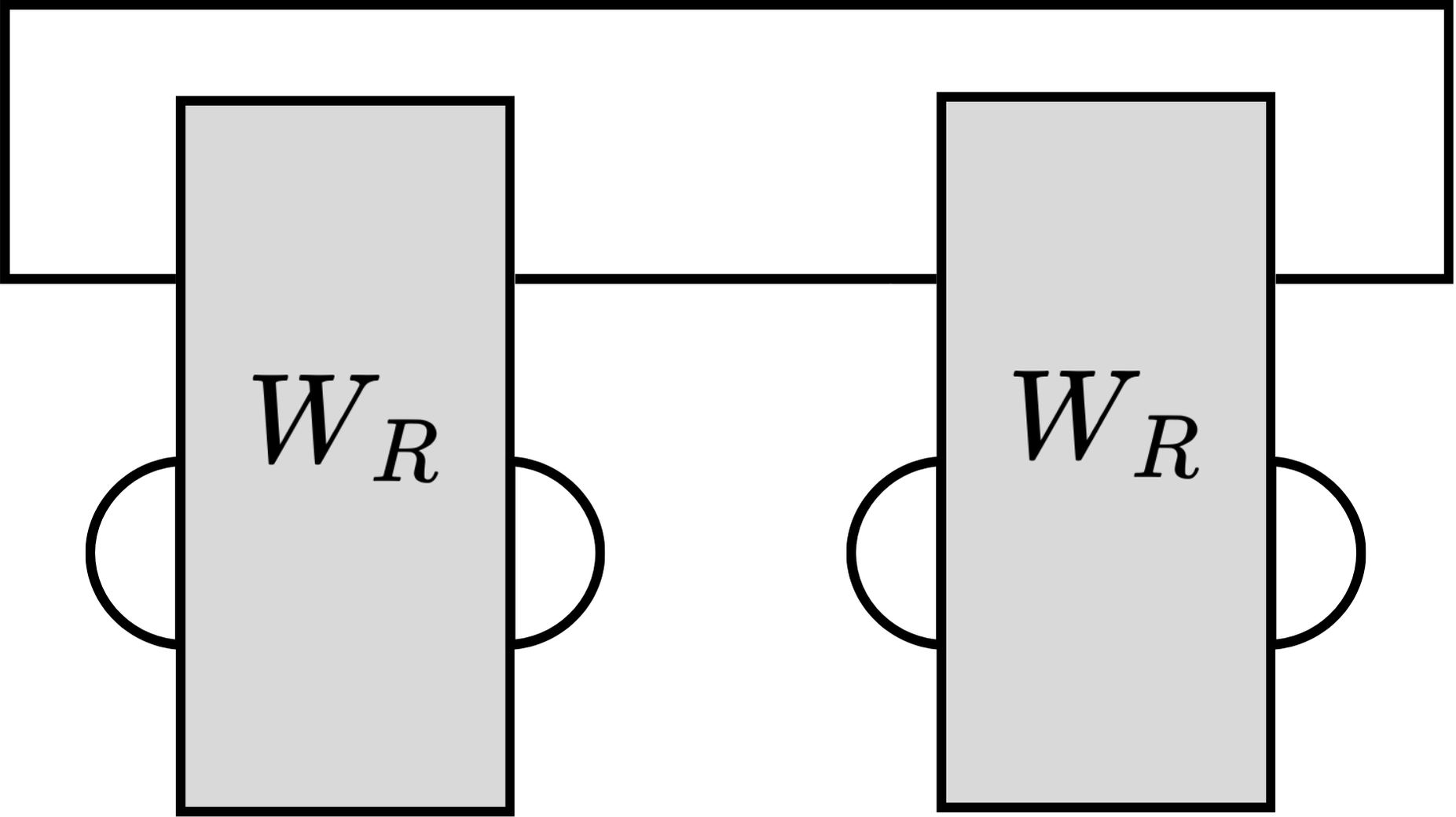} \\

$W_8$ &
\includegraphics[height=16.5mm]{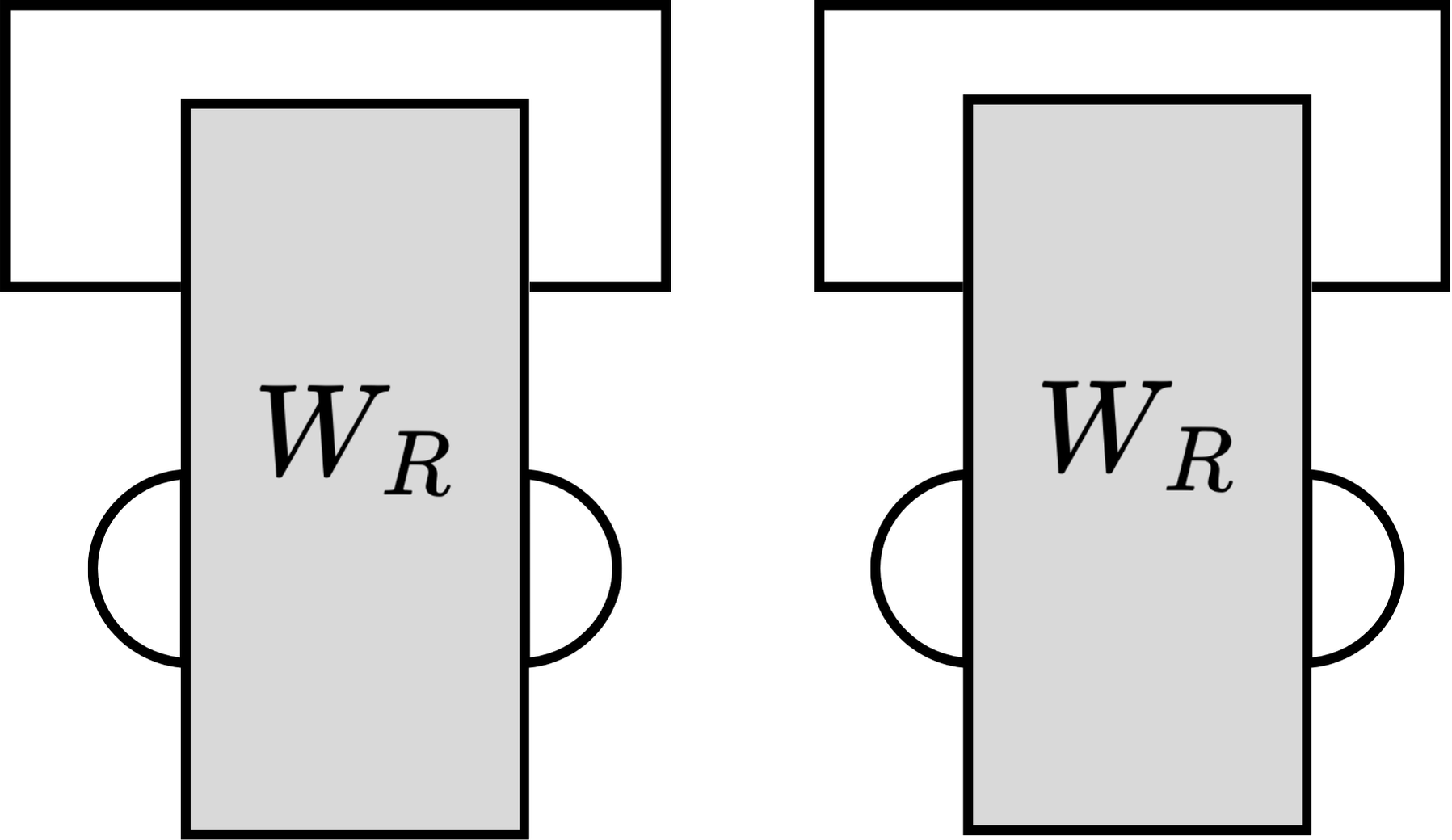} \\ \addlinespace[15pt]

$\mathscr{R}_1(s,\lambda_i)$ &
\includegraphics[height=23mm]{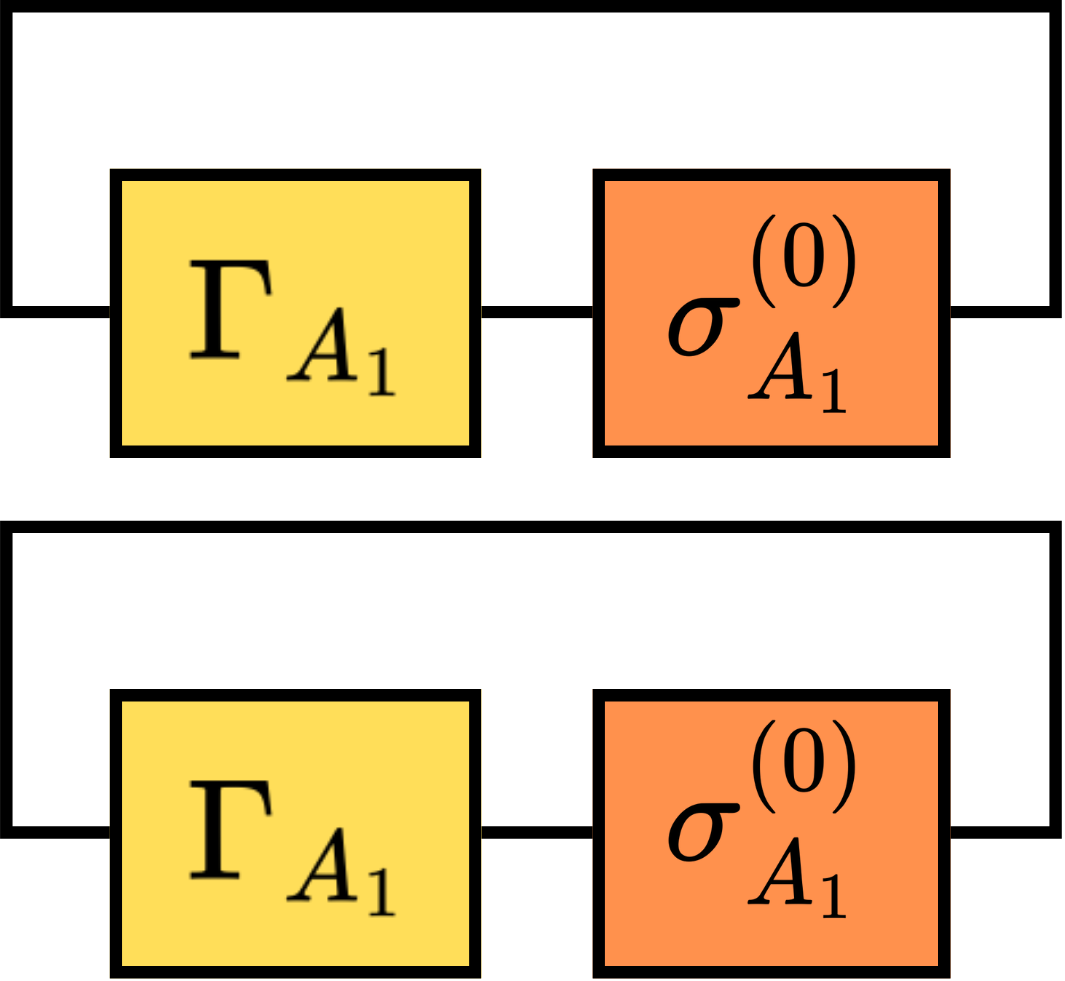}
 \\ \addlinespace[15pt]

$\mathscr{R}_2(s,\lambda_i)$ &
\includegraphics[height=16.5mm]{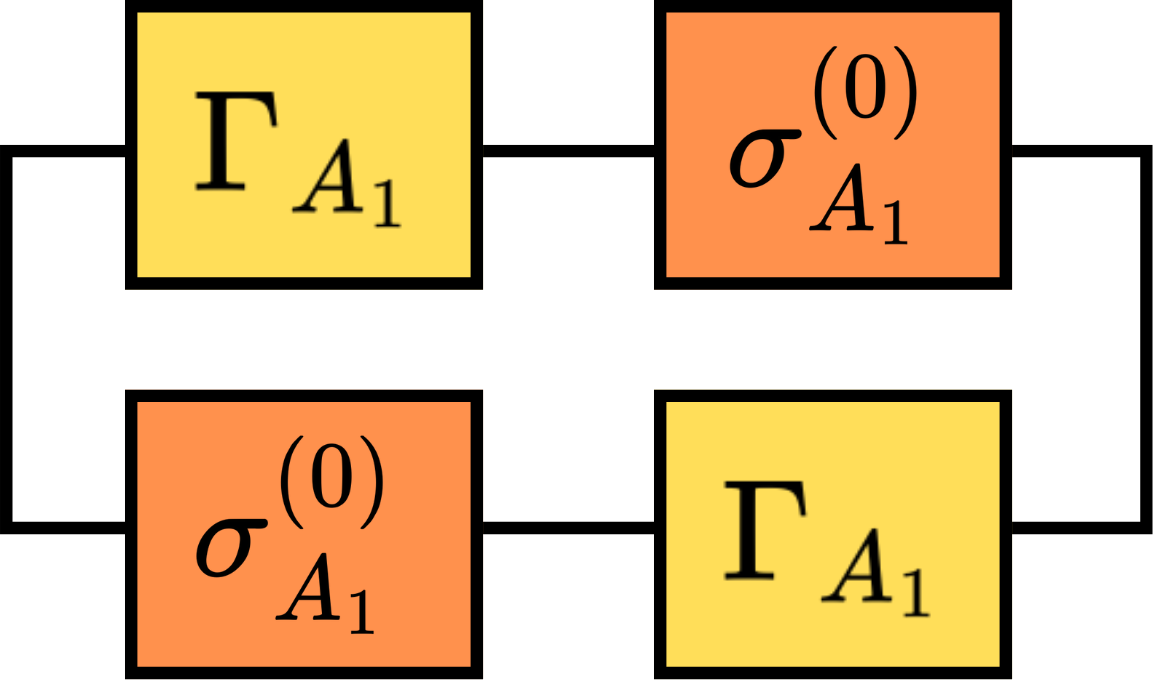} \\ \addlinespace[15pt]

$\mathscr{R}_3(s,\lambda_i)$ &
\includegraphics[height=20mm]{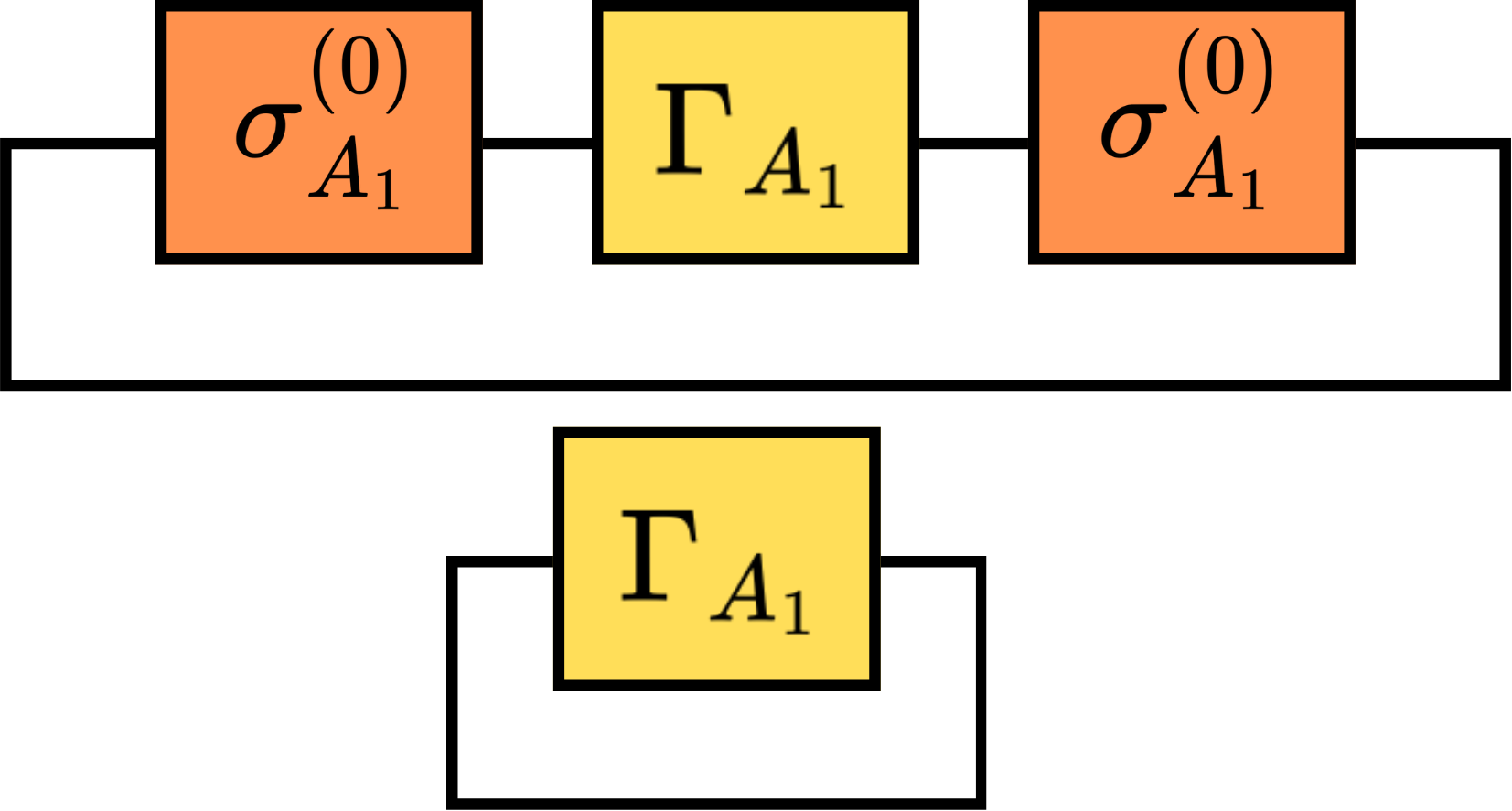} \\ \addlinespace[15pt]

$\mathscr{R}_4(\lambda_i)$ &
\includegraphics[height=10.1mm]{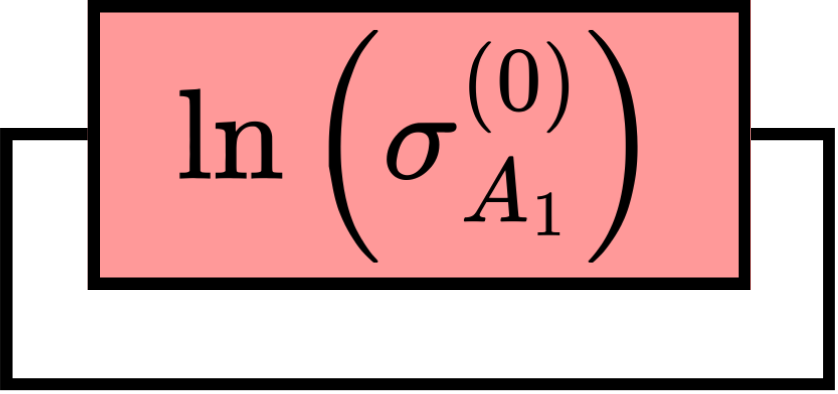} \\ \addlinespace[15pt]

$\mathscr{R}_5(\lambda_i)$ &
\includegraphics[height=10.1mm]{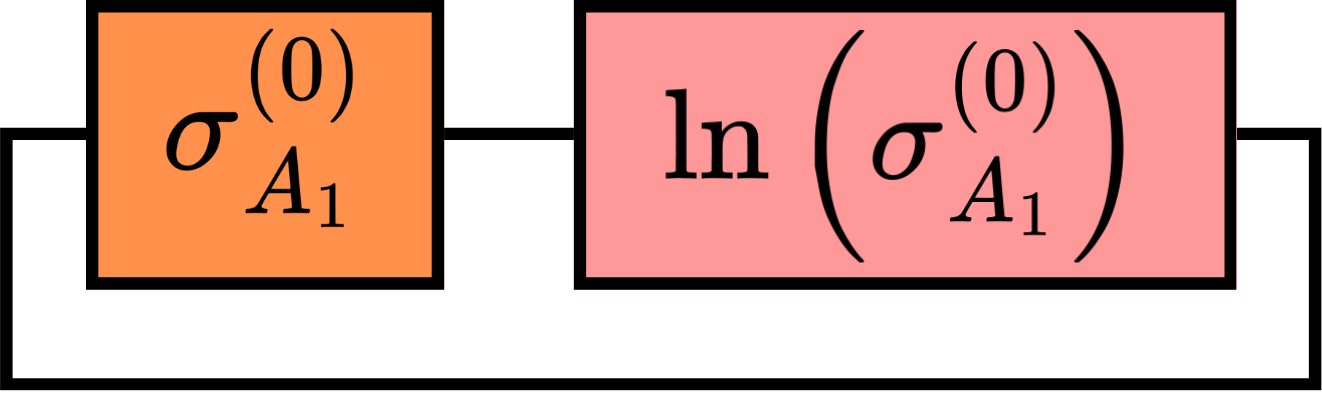}  \\ \addlinespace[15pt]

$\mathscr{R}_6(\lambda_i)$ &
\includegraphics[height=20mm]{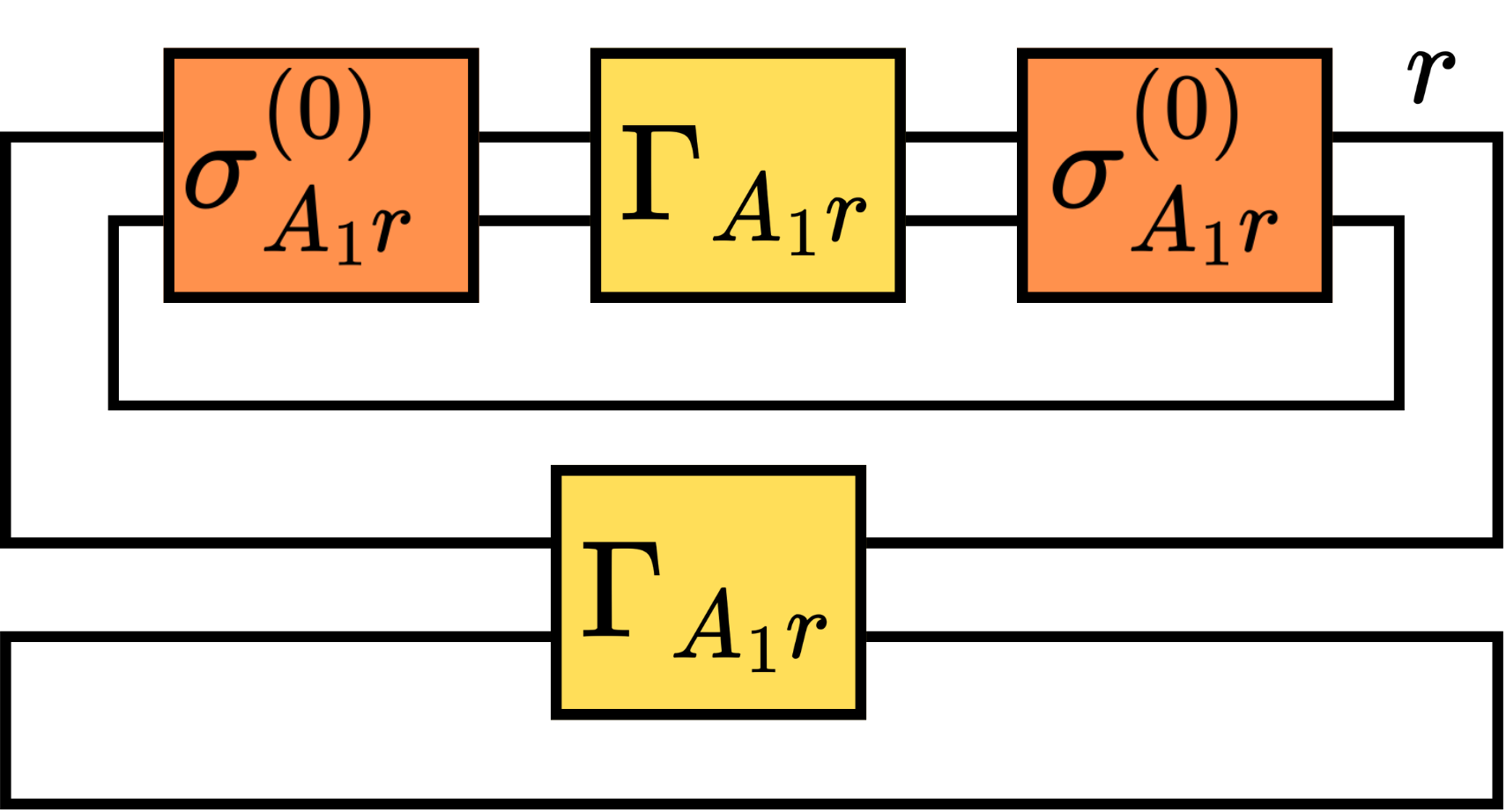} \\\addlinespace[15pt]

$\mathscr{R}_7(s,\lambda_i)$ &
\includegraphics[height=20mm]{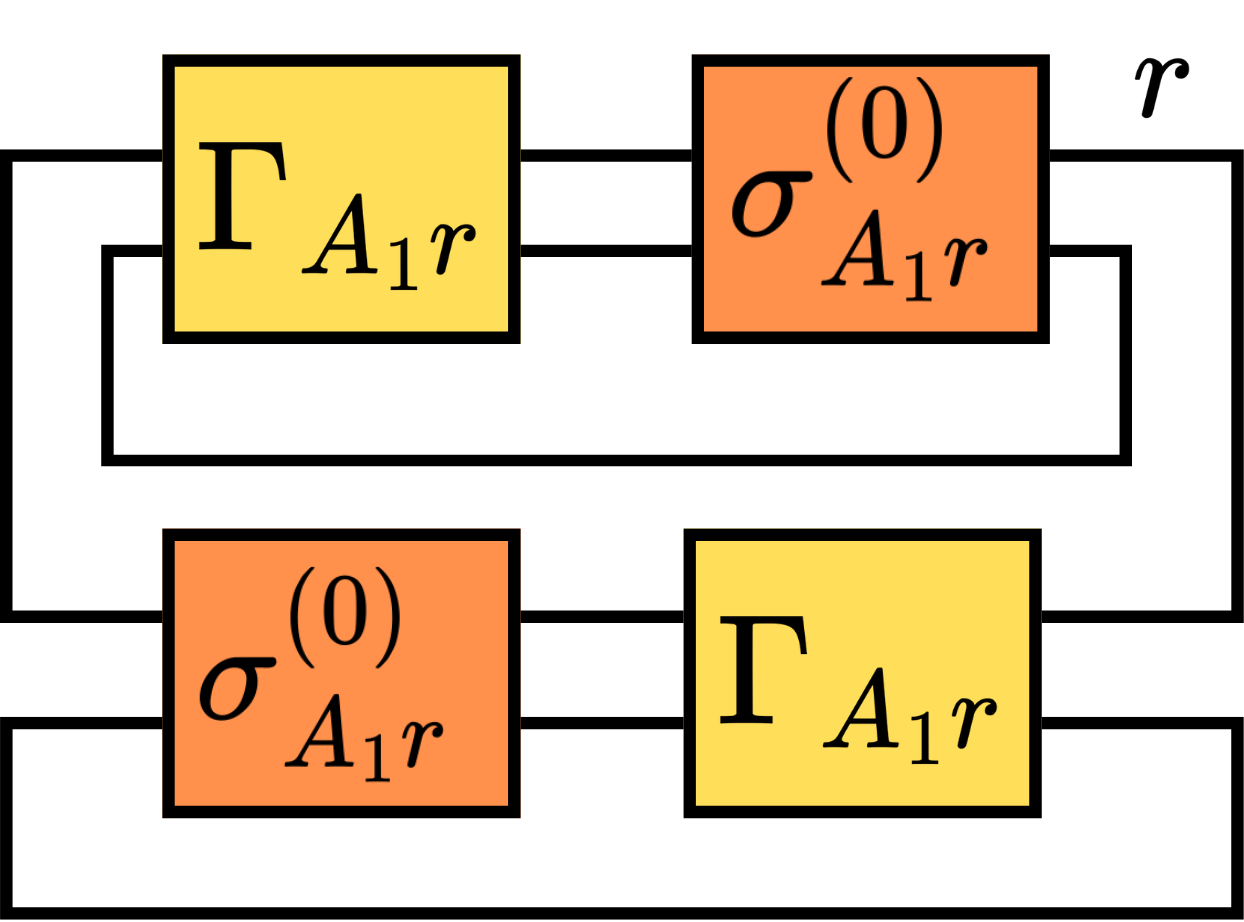} \\\addlinespace[15pt]

$\mathscr{R}_8(\lambda_i)$ &
\includegraphics[height=13mm]{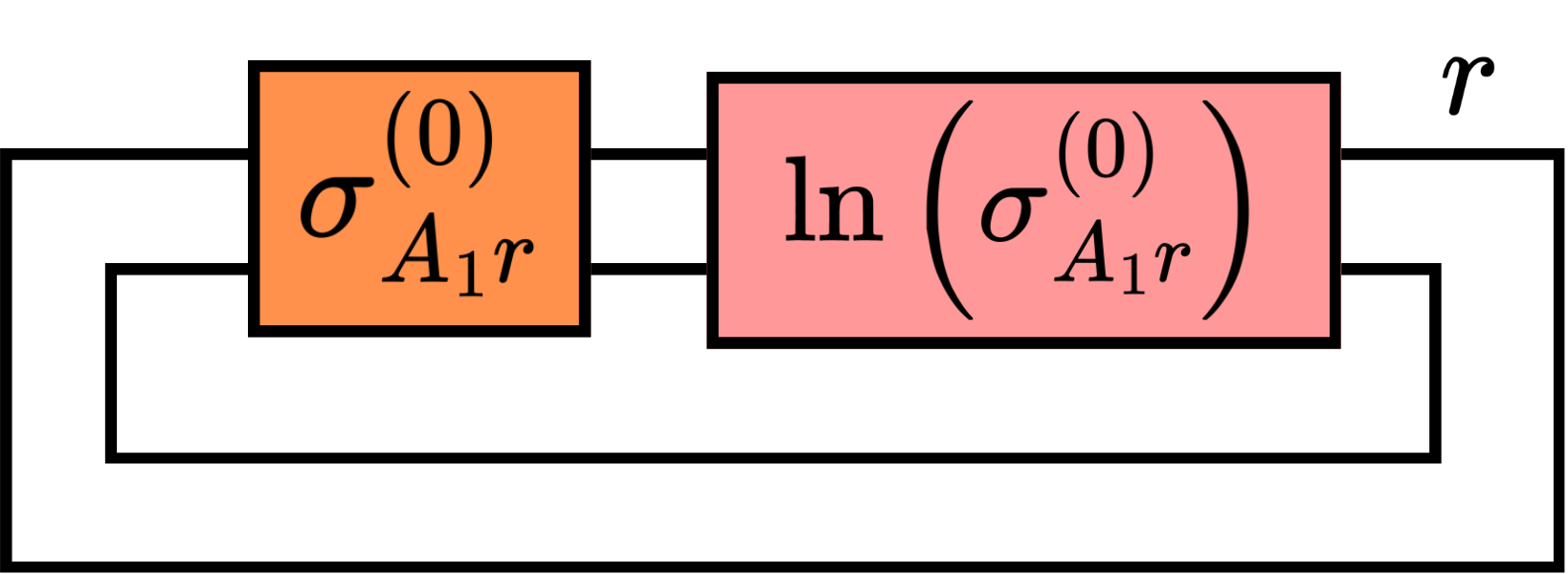}   \\\addlinespace[15pt]

$\mathscr{R}_9(\lambda_i)$ &
\includegraphics[height=13mm]{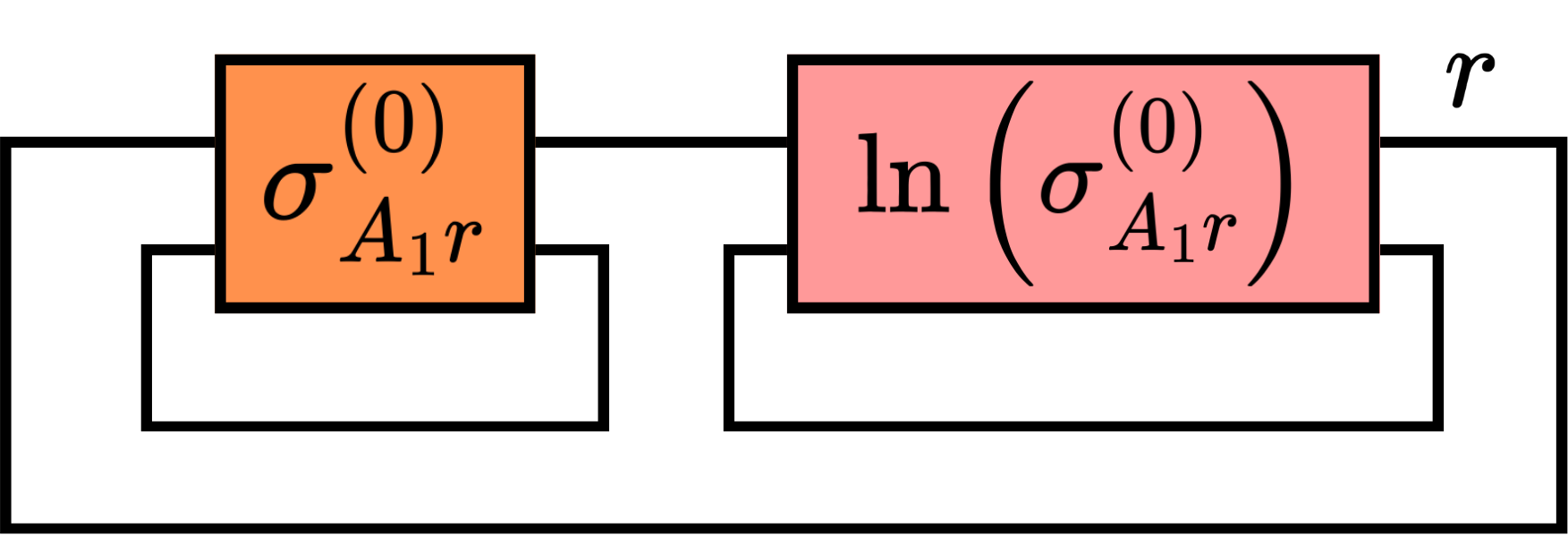} \\

\end{longtable}

\subsection{$\widetilde W$, $\widetilde{\mathscr{R}}$ diagrams for pure case}
\label{red-WRdiag-pure}

\begin{longtable}{C{2.5cm} C{12cm}}
\caption{Trace diagram key for \(\widetilde W_i\) for the pure bulk case. The first, second, third, and fourth legs (from outermost to innermost) represent the Hilbert spaces \(\mathcal{H}_{A_1}\), \(\mathcal{H}_{A_2}\), \(\mathcal{H}_{\bar A}\), and
\(\mathcal{H}_{\bar A_1}\), respectively.
}
\label{W-tilde-R-tilde-diagrams} \\
\toprule
\textbf{Label} & \textbf{Trace diagram} \\
\midrule
\endfirsthead

\toprule
\textbf{Label} & \textbf{Trace diagram} \\
\midrule
\endhead

\midrule
\endfoot

\bottomrule
\endlastfoot

$\widetilde W_1$ &
\includegraphics[height=20mm]{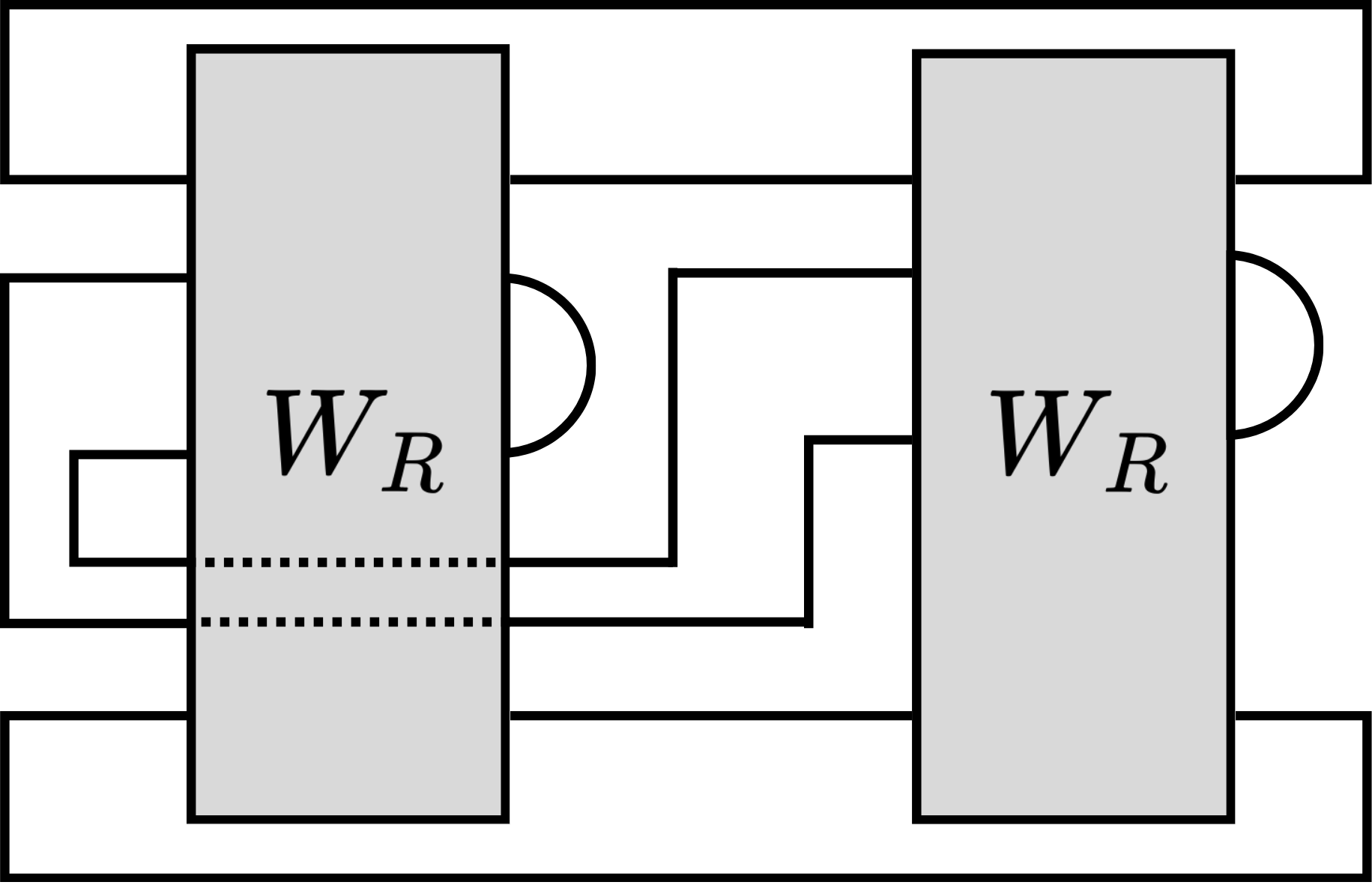} \\[6pt]

$\widetilde W_2$ &
\includegraphics[height=20mm]{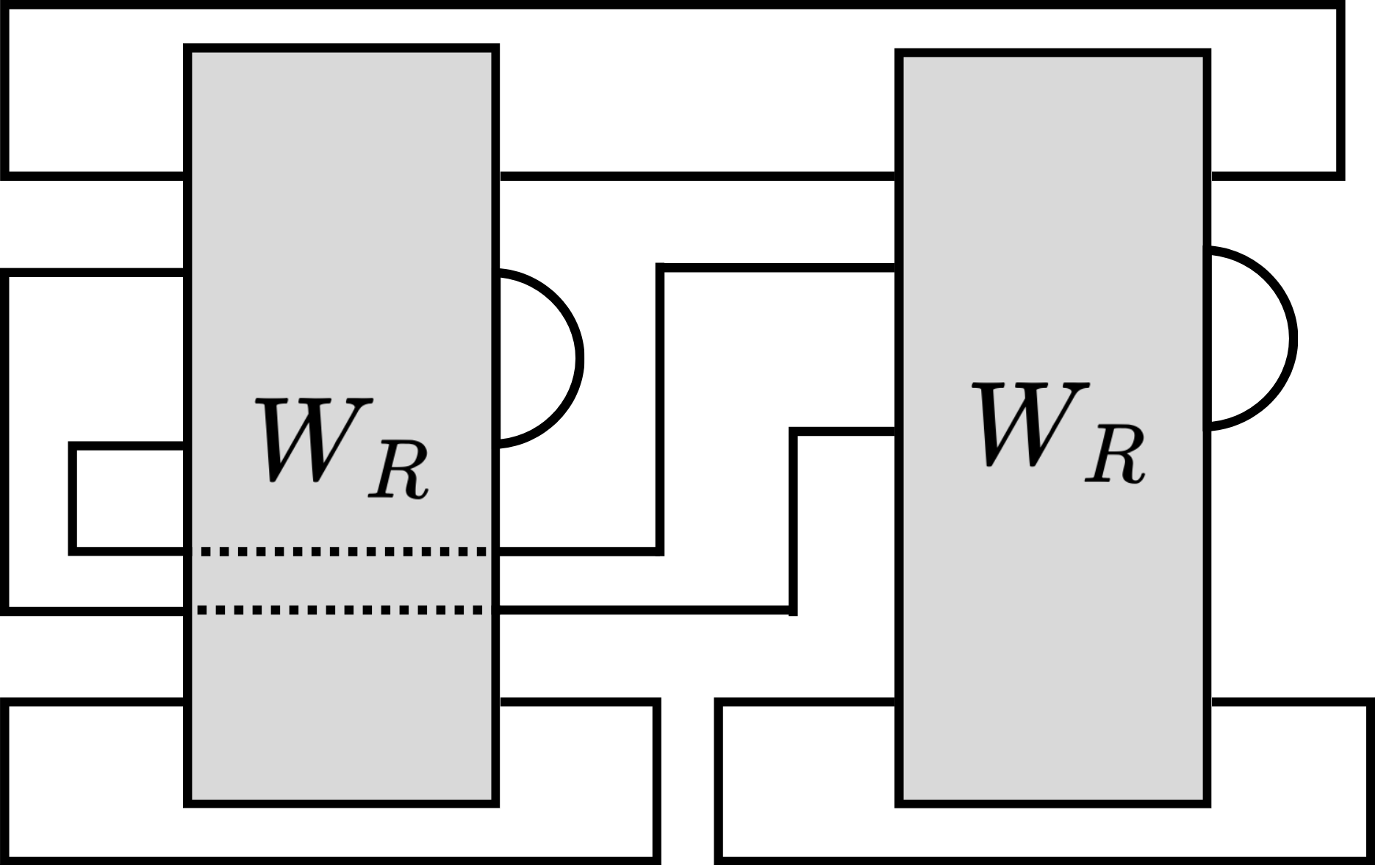} \\[6pt]

$\widetilde W_3$ &
\includegraphics[height=20mm]{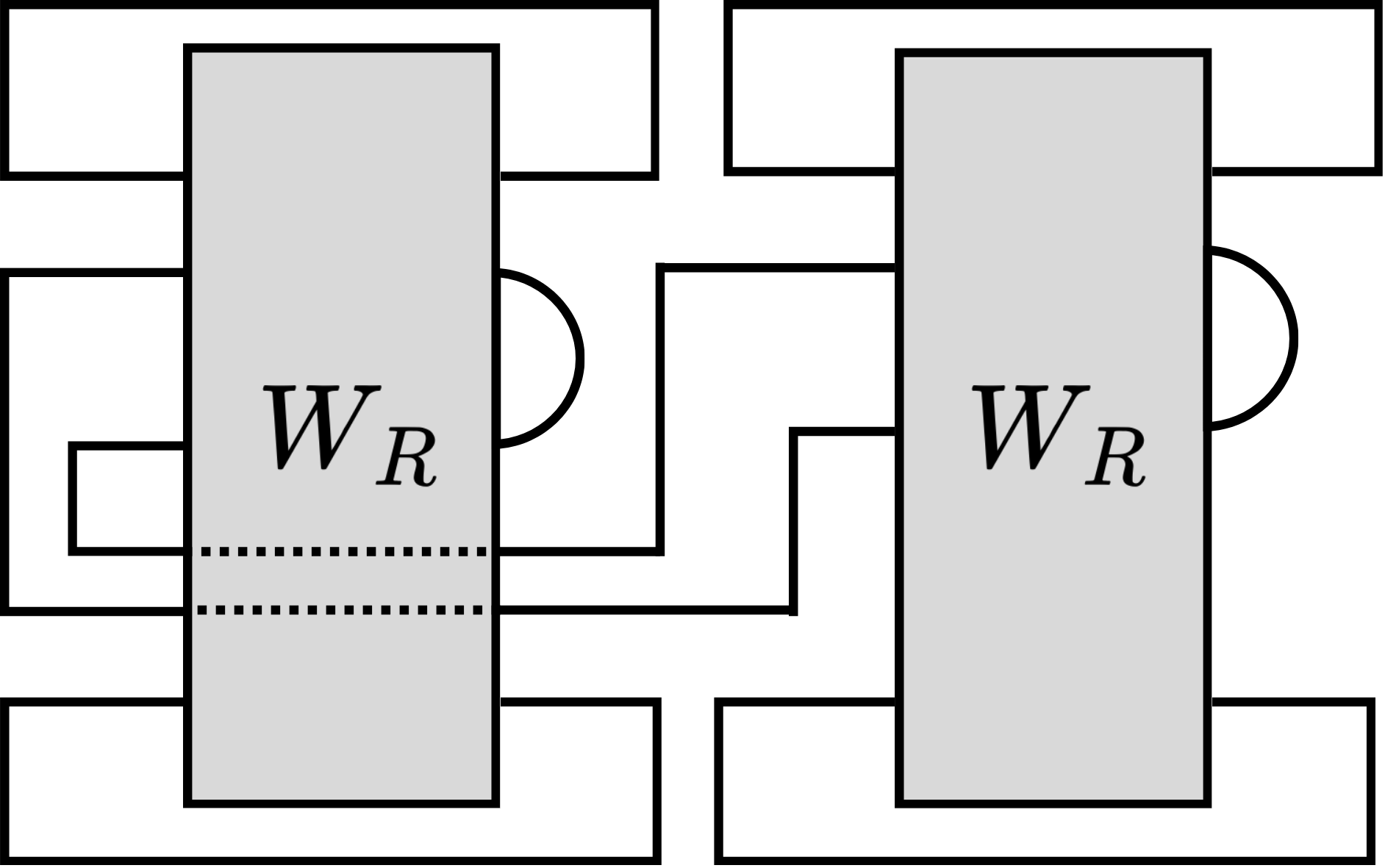} \\[6pt]

$\widetilde W_4$ &
\includegraphics[height=20mm]{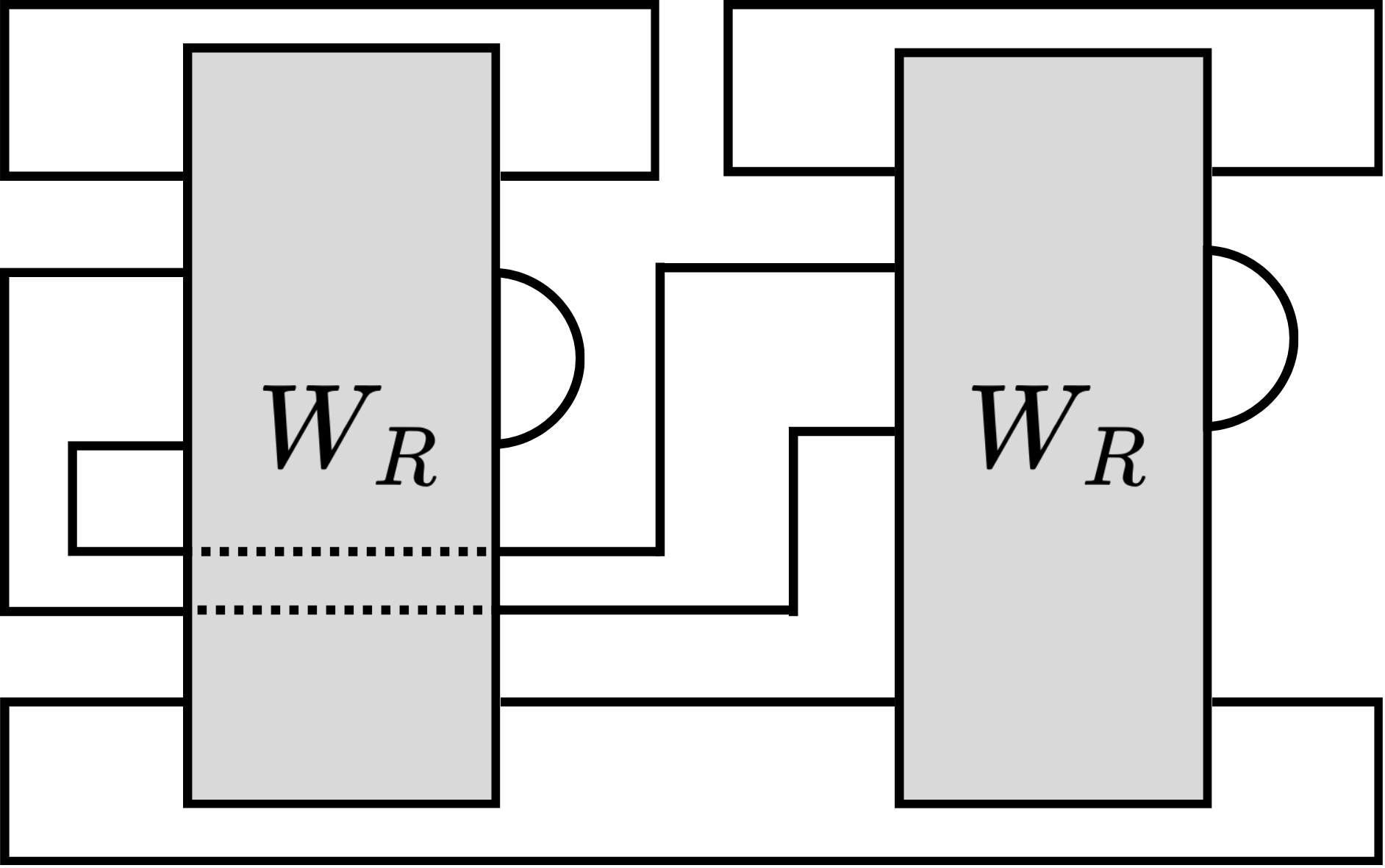} \\[6pt]

$\widetilde W_5$ &
\includegraphics[height=20mm]{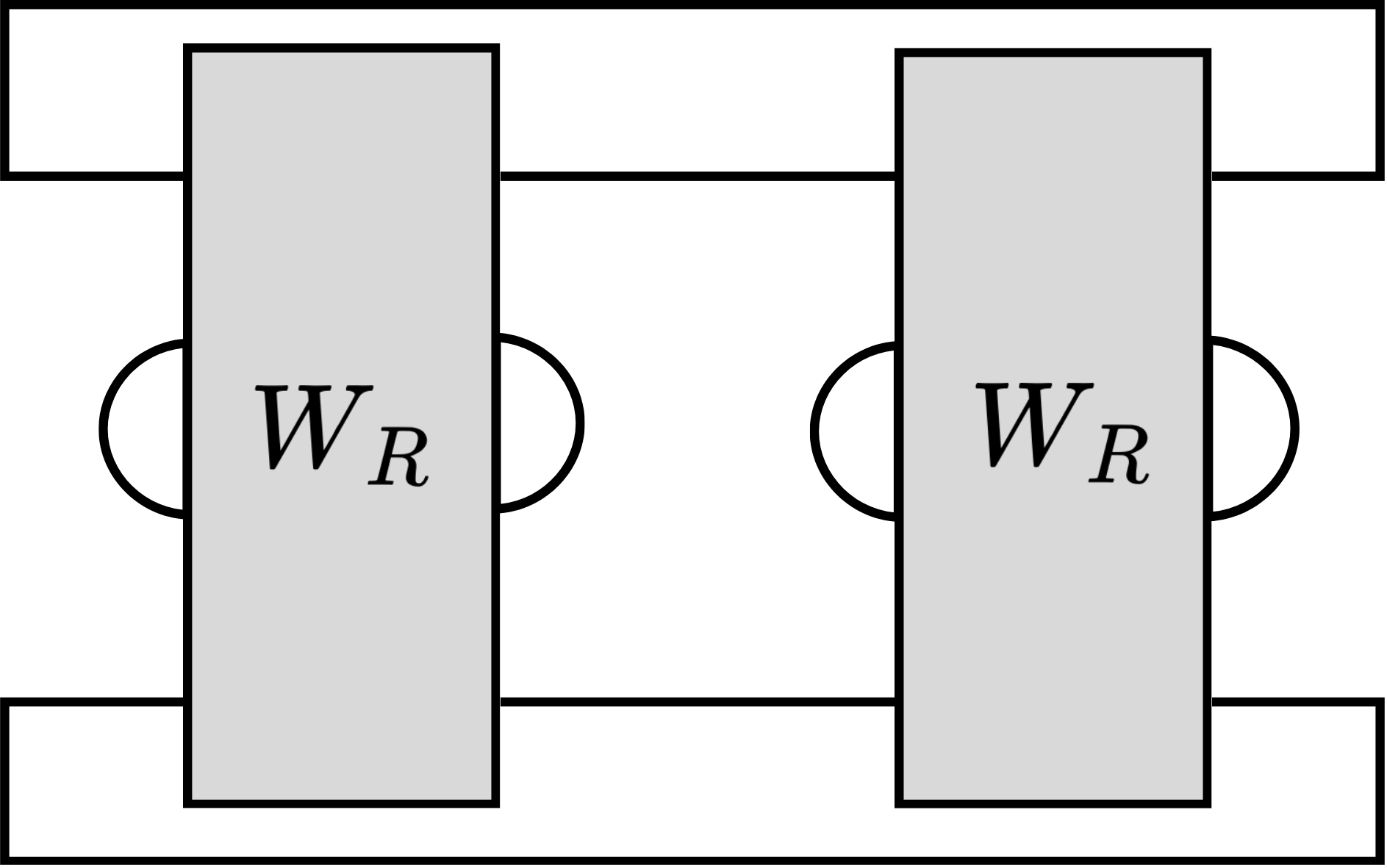} \\[6pt]

$\widetilde W_6$ &
\includegraphics[height=20mm]{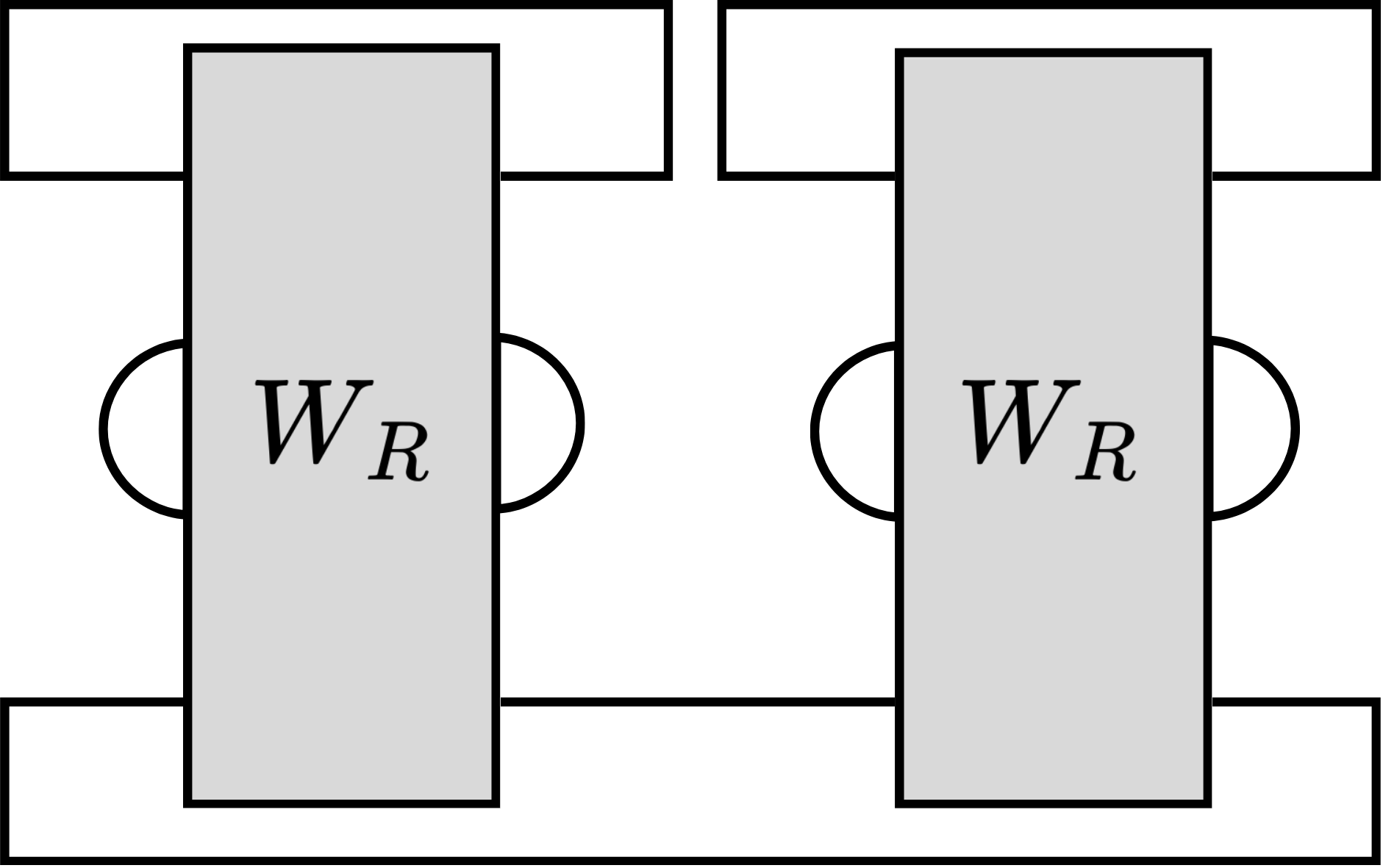} \\[6pt]

$\widetilde W_7$ &
\includegraphics[height=22mm]{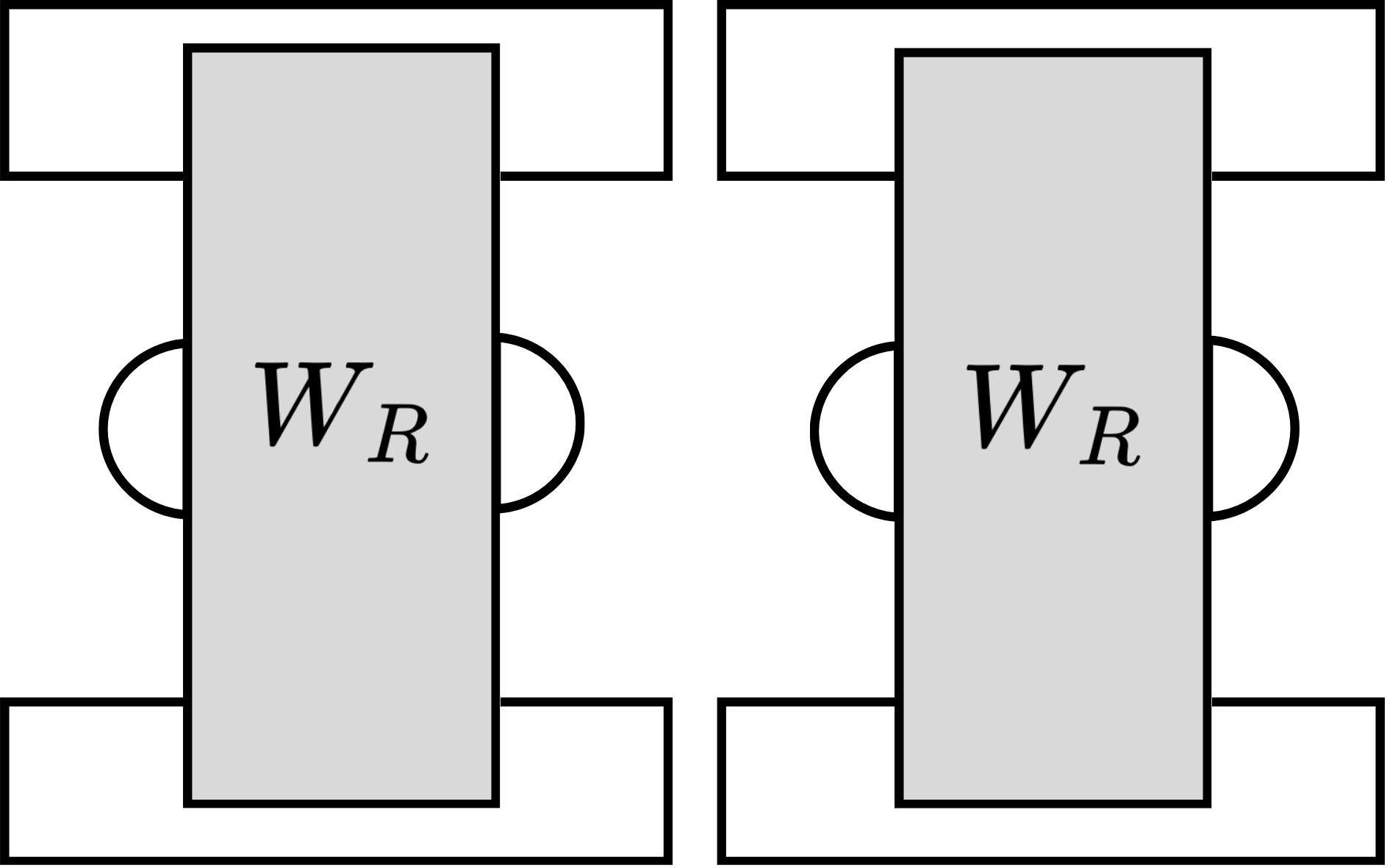} \\[6pt]

$\widetilde W_8$ &
\includegraphics[height=22mm]{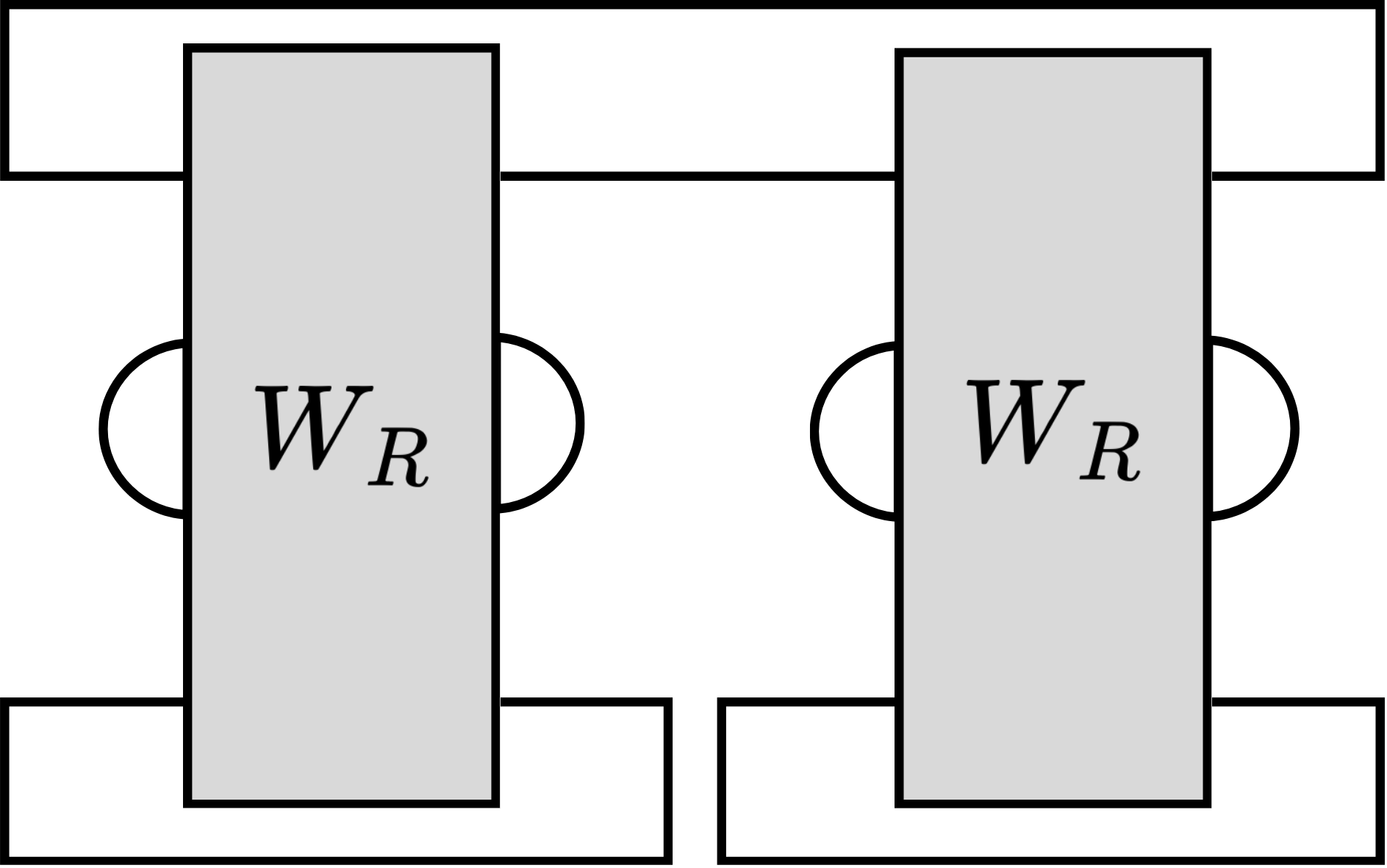} \\[6pt]

$\widetilde W_9$ &
\includegraphics[height=26mm]{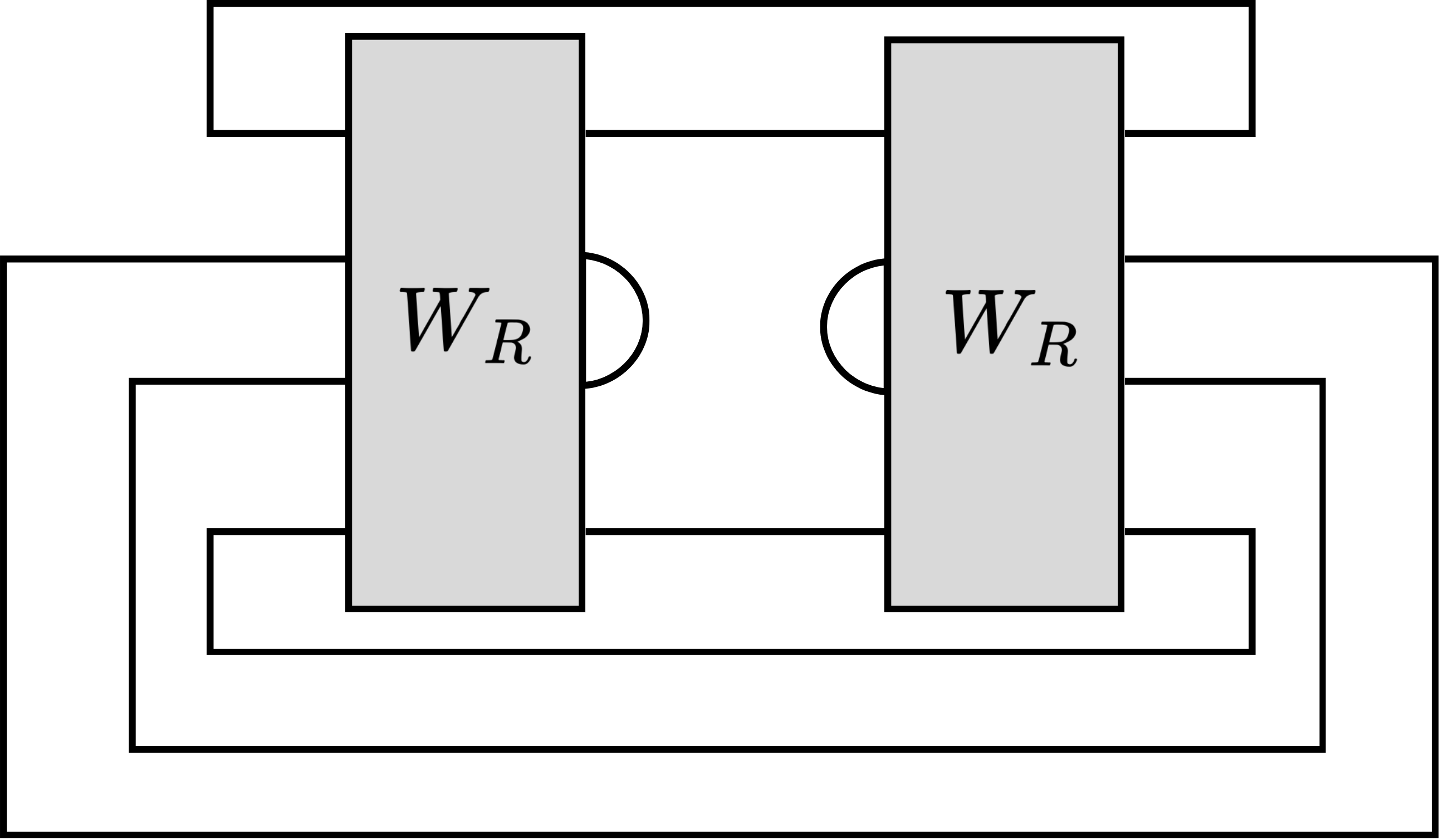} \\[6pt]

$\widetilde W_{10}$ &
\includegraphics[height=26mm]{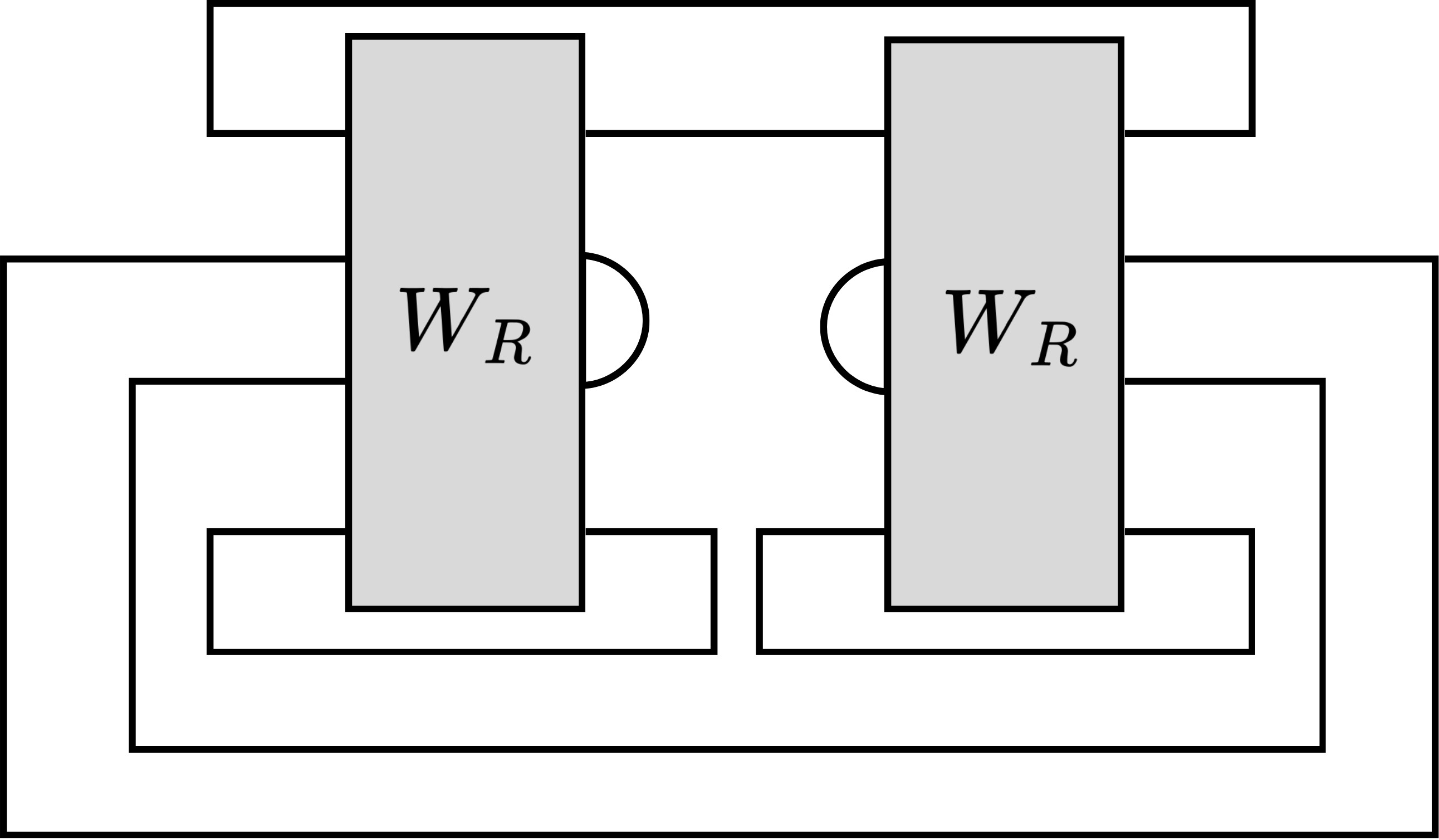} \\[6pt]

$\widetilde W_{11}$ &
\includegraphics[height=26mm]{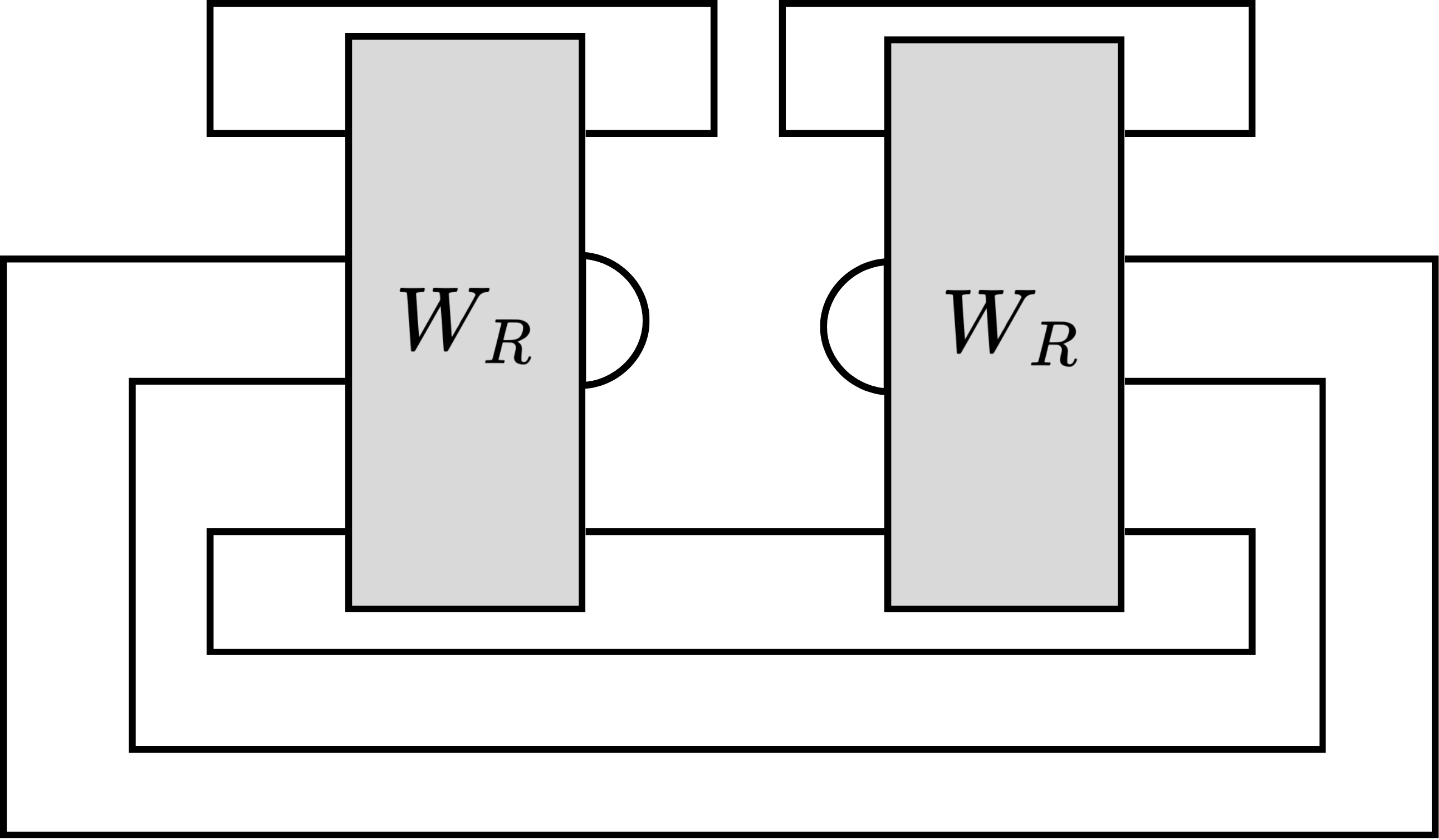} \\[6pt]

$\widetilde W_{12}$ &
\includegraphics[height=26mm]{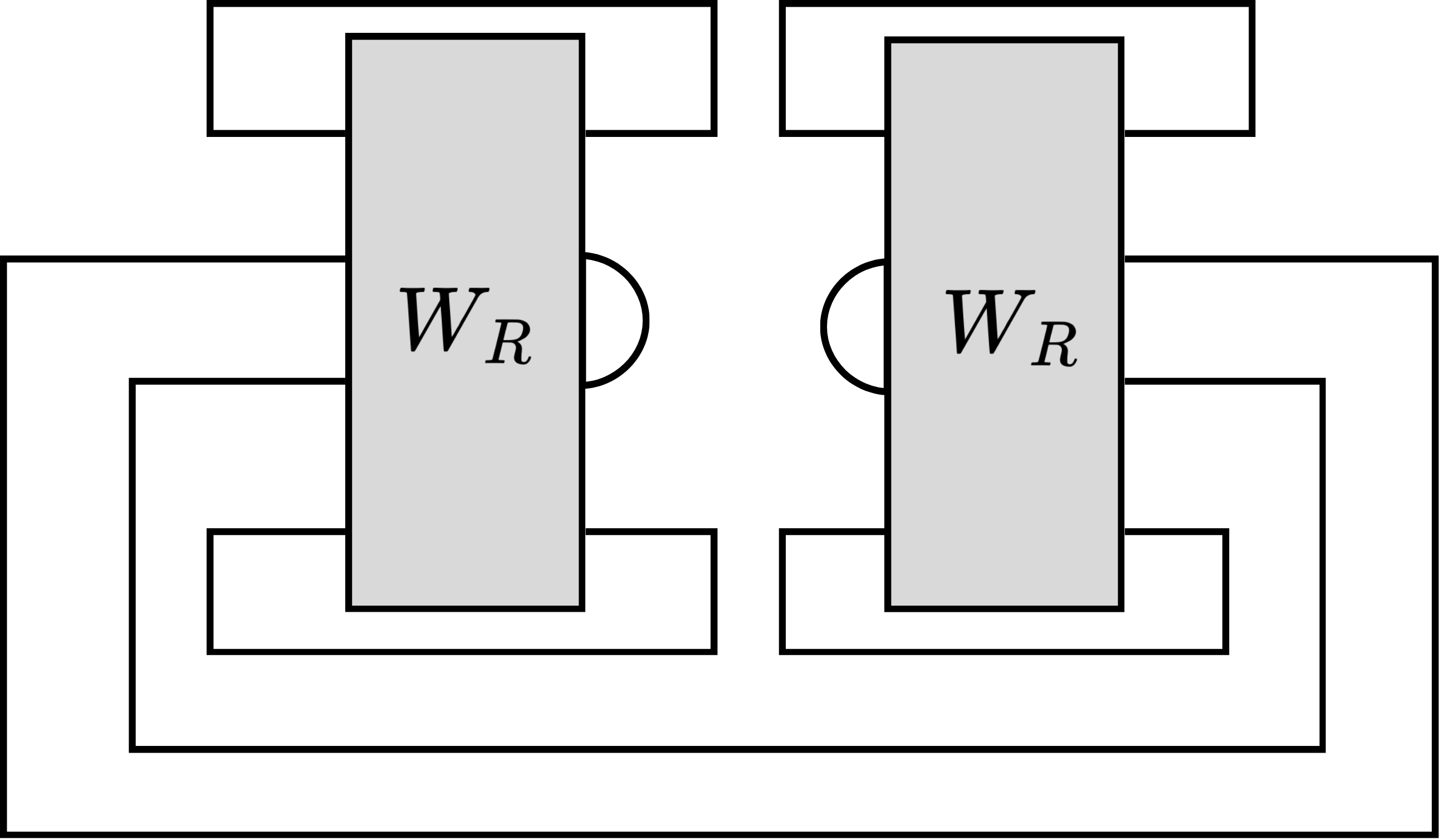} \\[6pt]

$\widetilde W_{13}$ &
\includegraphics[height=20mm]{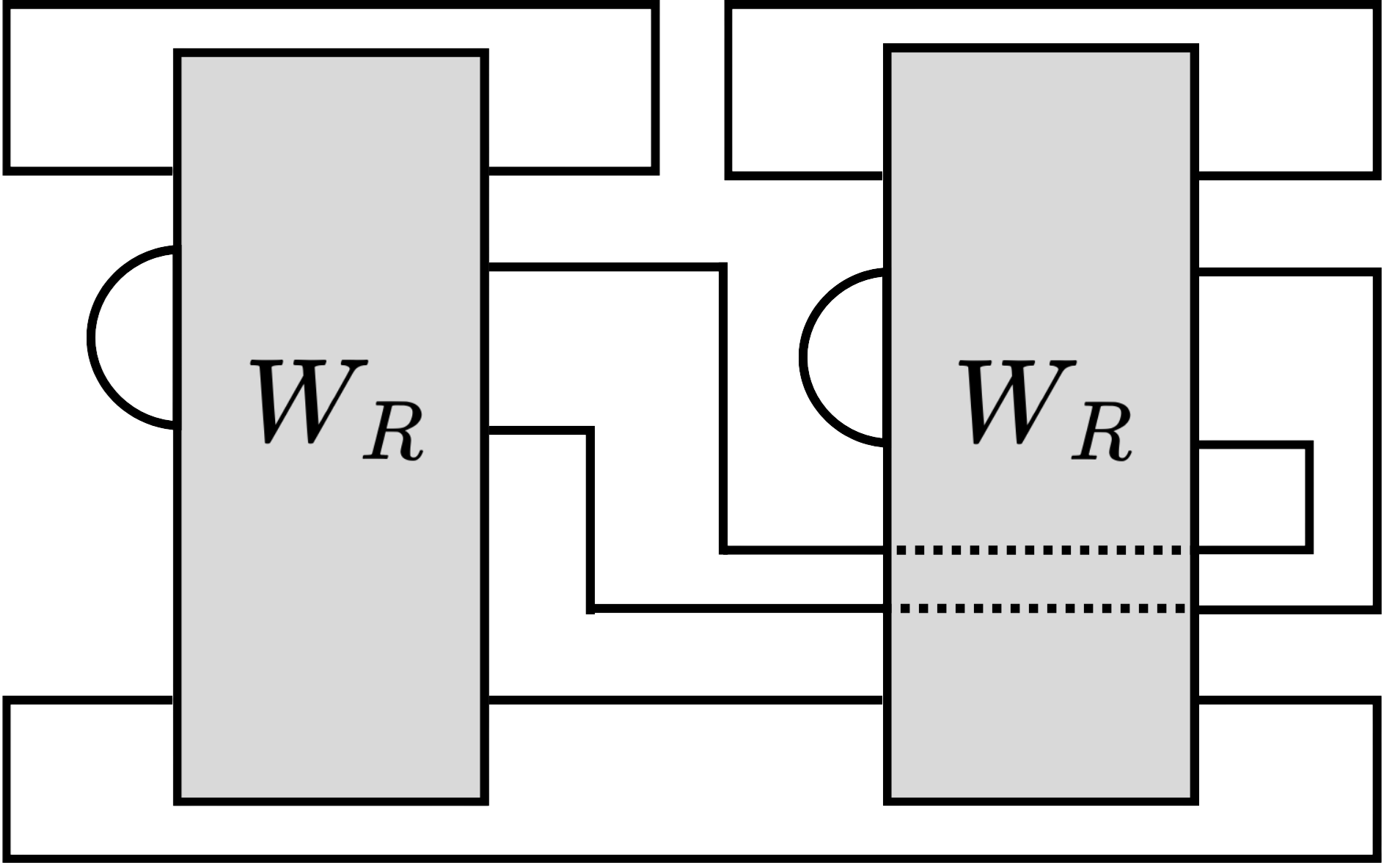} \\[6pt]

$\widetilde W_{14}$ &
\includegraphics[height=20mm]{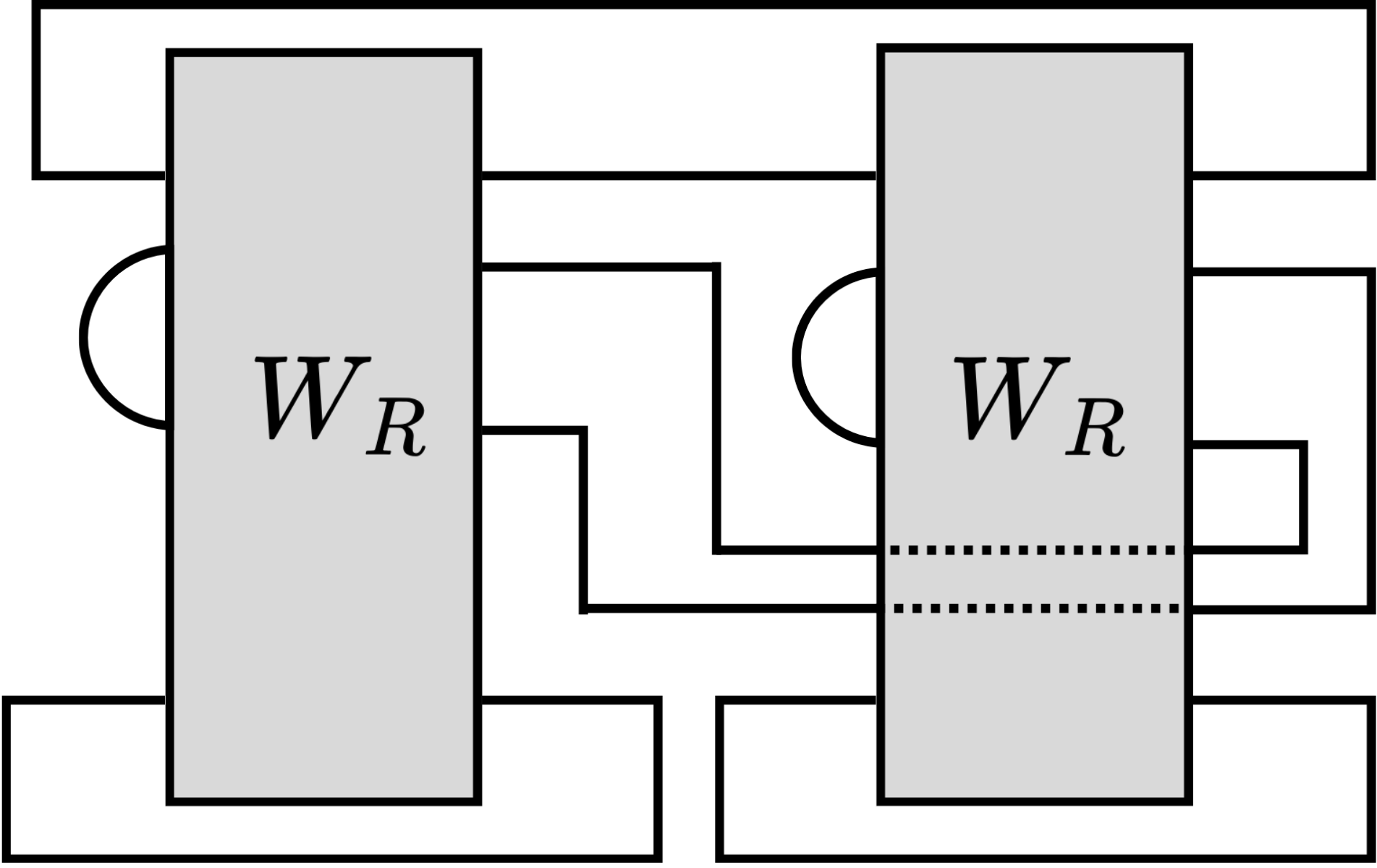} \\[6pt]

$\widetilde W_{15}$ &
\includegraphics[height=20mm]{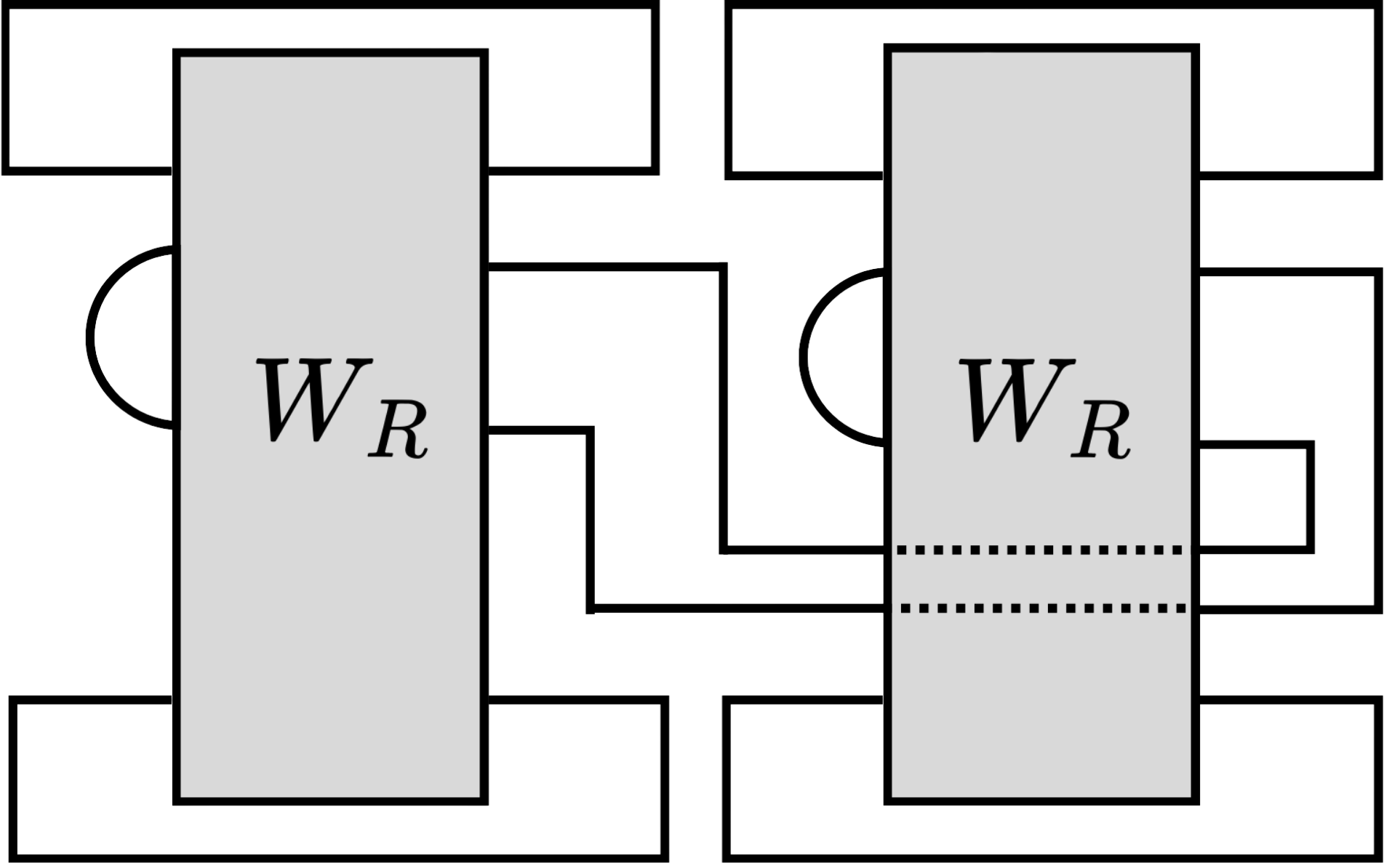} \\[6pt]

$\widetilde W_{16}$ &
\includegraphics[height=20mm]{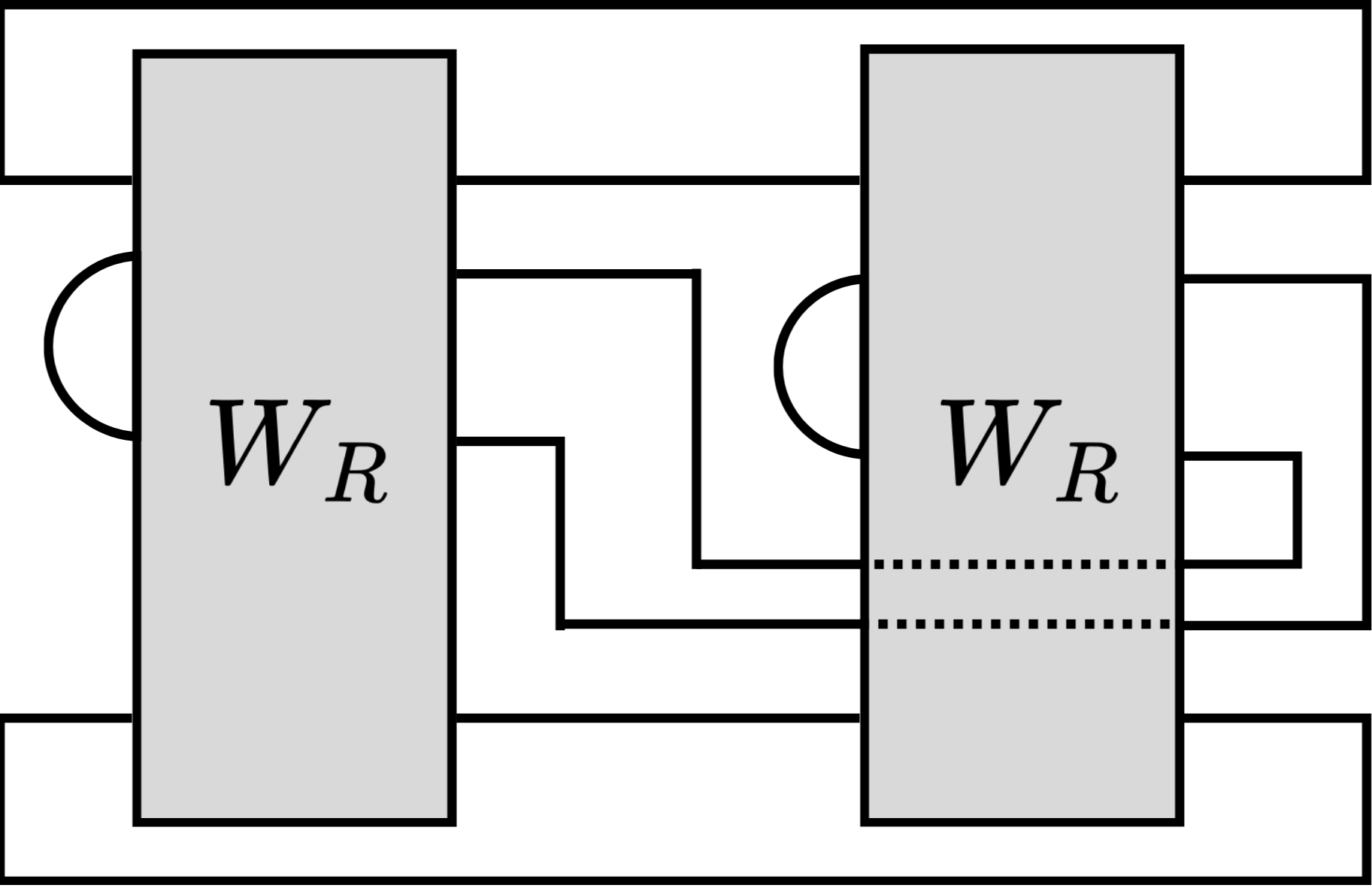} \\
$\widetilde{\mathscr{R}}_{1}$ &
\includegraphics[height=22mm]{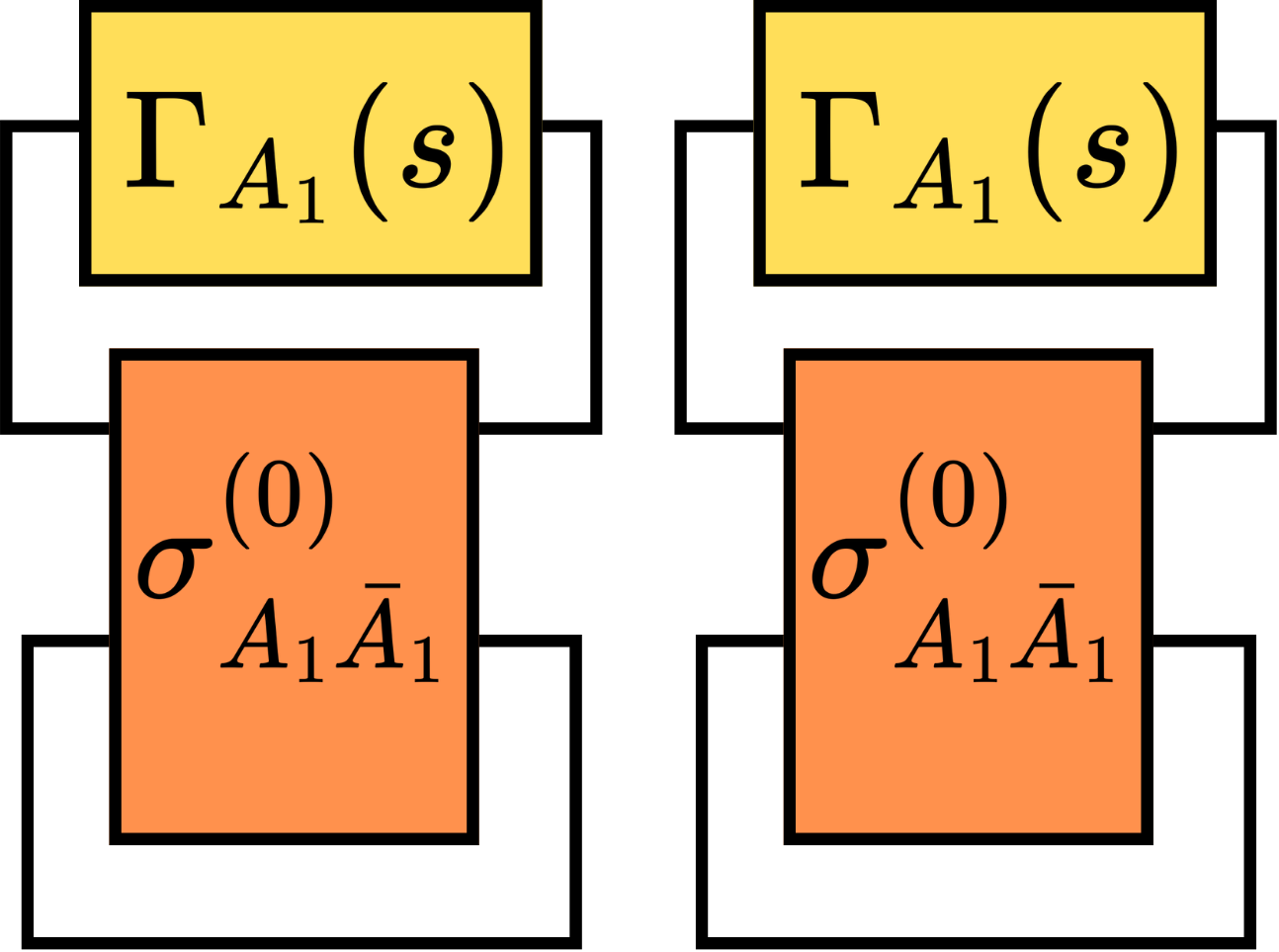} \\ 
$\widetilde{\mathscr{R}}_{2}$ &
\includegraphics[height=22mm]{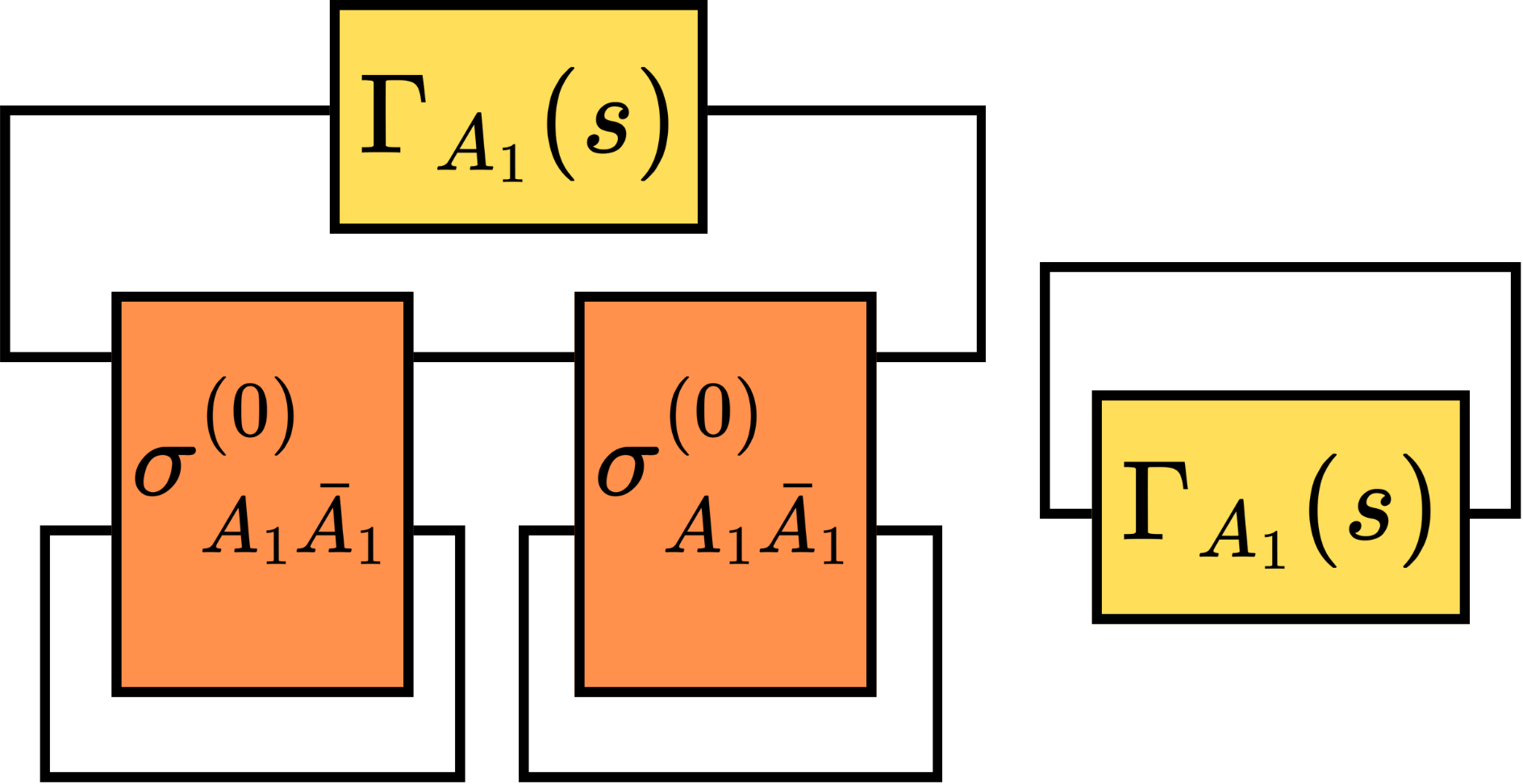} \\
$\widetilde{\mathscr{R}}_{3}$ &
\includegraphics[height=22mm]{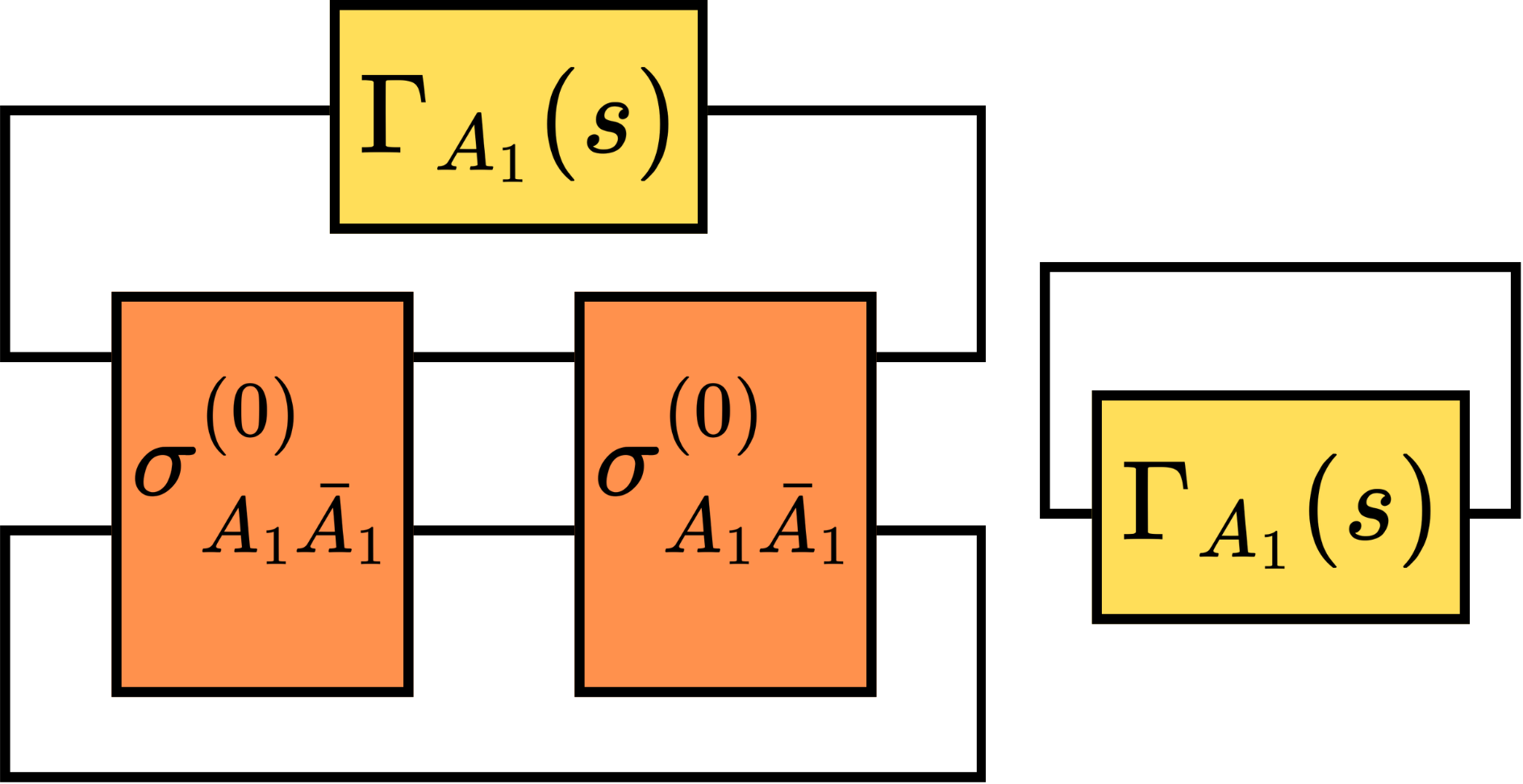} \\
$\widetilde{\mathscr{R}}_{4}$ &
\includegraphics[height=22mm]{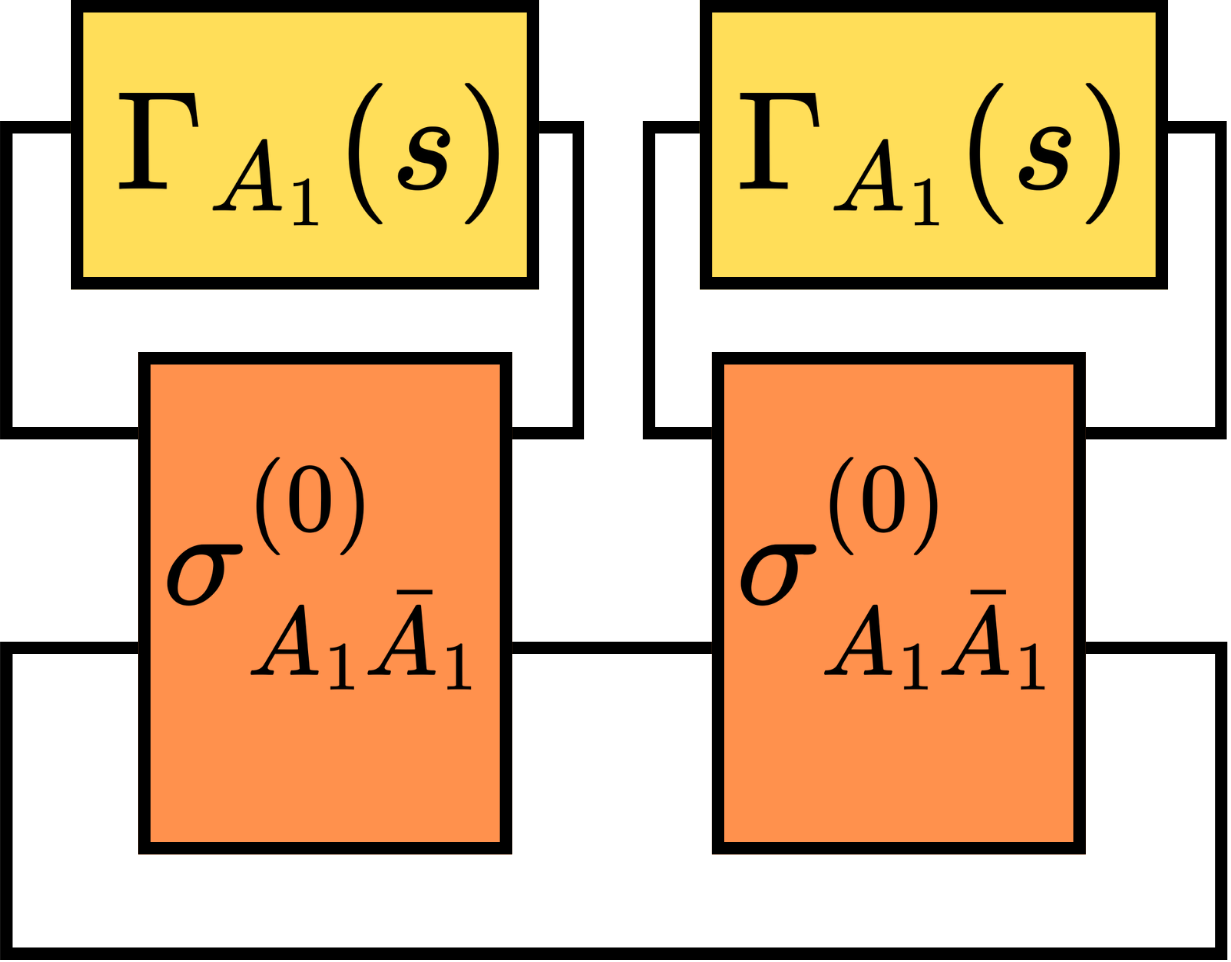} \\
$\widetilde{\mathscr{R}}_{5}$ &
\includegraphics[height=22mm]{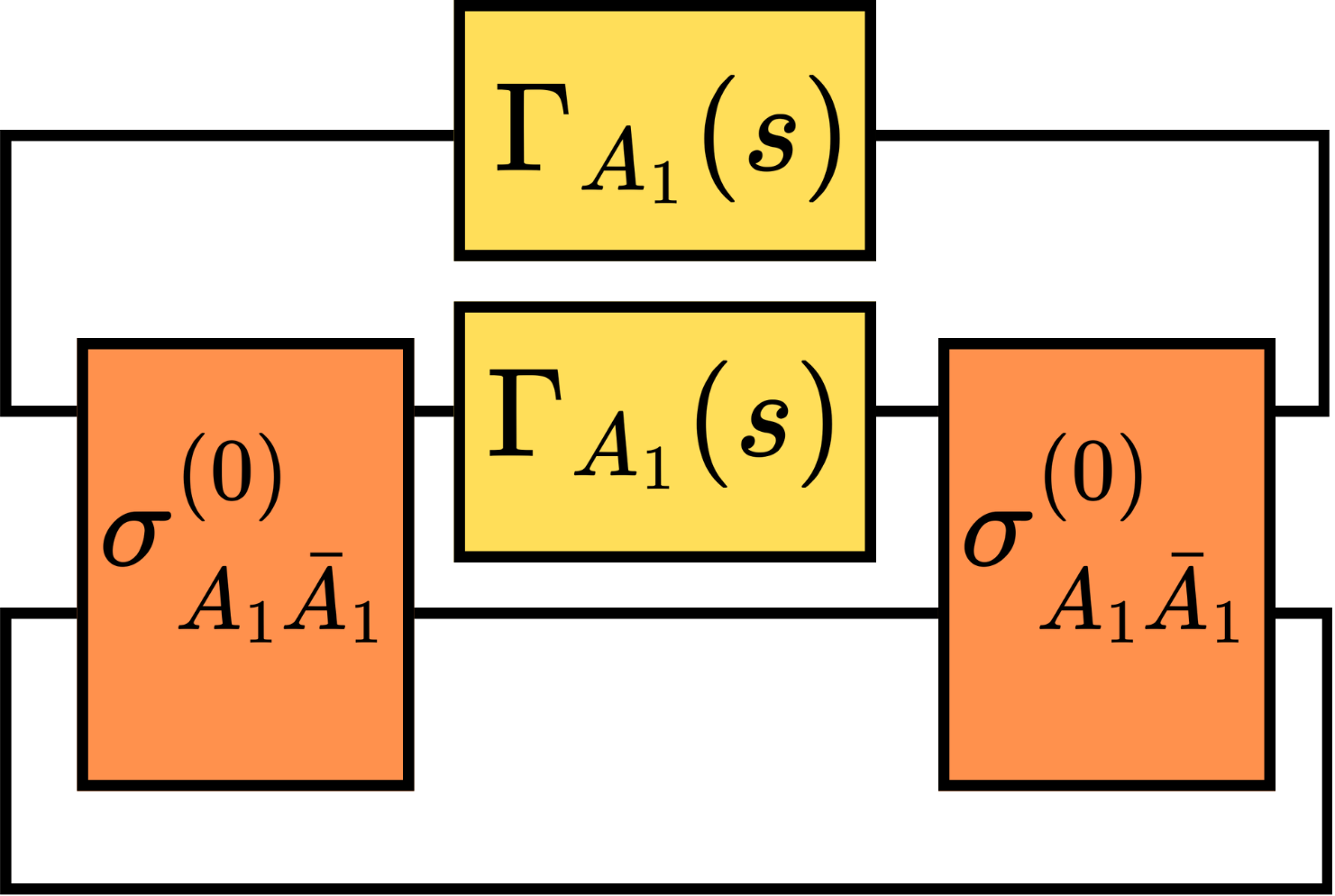} \\
$\widetilde{\mathscr{R}}_{6}$ &
\includegraphics[height=22mm]{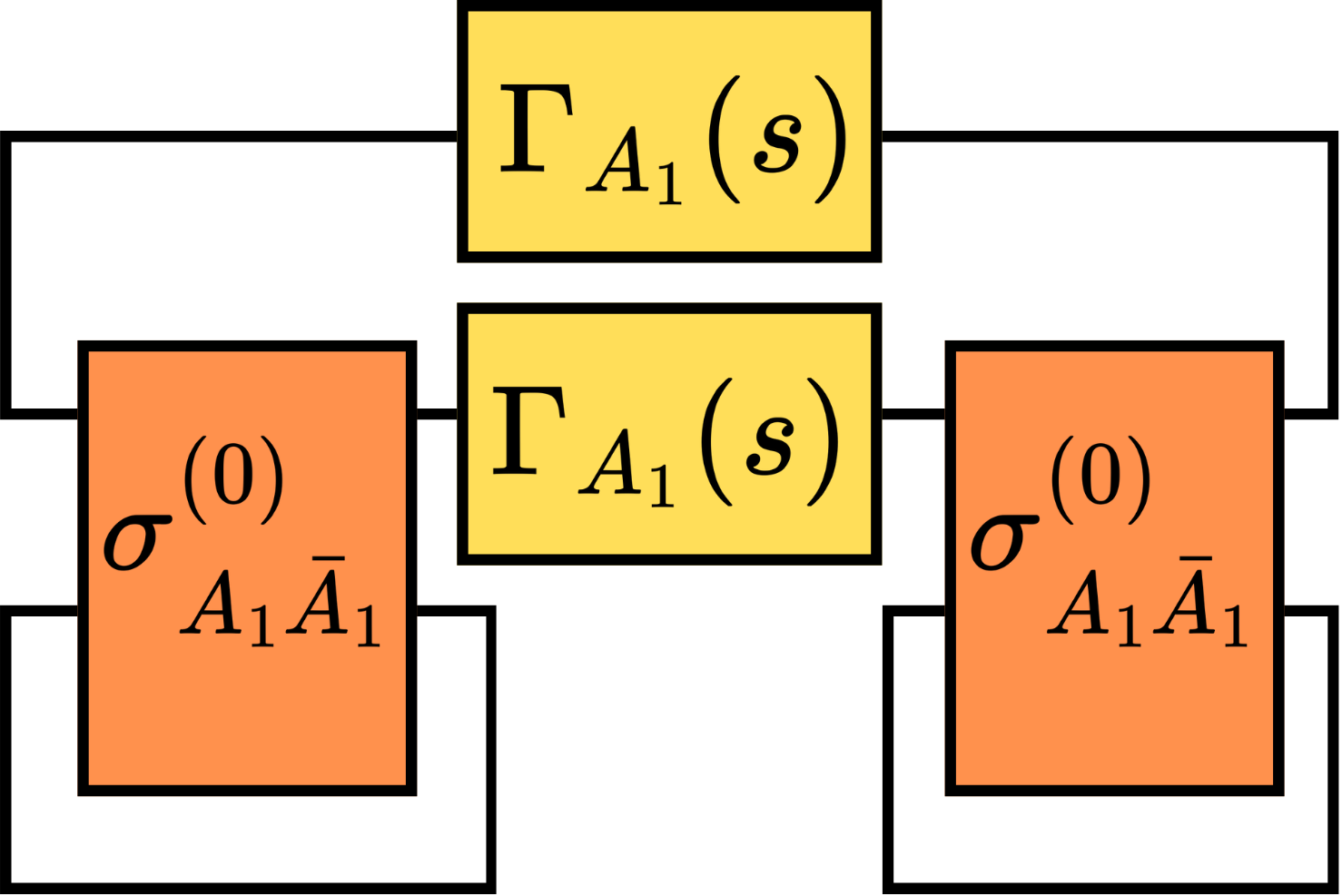} \\
\end{longtable}

\section{Technical details}
\label{logexpansion}
\subsection{Perturbative expansion of the entropy}
We are interested in the perturbative expansion of the entropy functional
\begin{equation}
S\!\left(A(\epsilon)\right)
:=
-\Tr\!\left(A(\epsilon)\log A(\epsilon)\right)
\label{eq:entropy-def-eps}
\end{equation}
about $\epsilon=0$. A key technical input is an explicit expression for the $\epsilon$-derivative of the operator logarithm in terms of resolvents. Before turning to the entropy expansion, we therefore recall (and for completeness derive) the standard resolvent representation for $\frac{d}{d\epsilon}\log A(\epsilon)$.
\begin{lemma}
For any differentiable family of positive operators $A(\epsilon)$, the derivative of the logarithm
admits the resolvent representation
\begin{equation}
\frac{d}{d\epsilon}\,\log A(\epsilon)
=
\int_{0}^{\infty} ds\,
\bigl[A(\epsilon)+s I\bigr]^{-1}
\frac{dA(\epsilon)}{d\epsilon}\,
\bigl[A(\epsilon)+s I\bigr]^{-1},
\label{eq:dlogA-resolvent}
\end{equation}
which follows from functional calculus and holds without any commutativity assumptions. 
\end{lemma}
\begin{proof}
  The following proof follows standard arguments based on functional calculus and
resolvent identities (see, e.g., \cite{Haber2018NotesOT}). We include it here for
completeness. We start from the standard integral identity for positive operators,
\begin{equation}
\ln(A+B)-\ln A
=
\int_{0}^{\infty} du\,
\Bigl[(A+uI)^{-1}-(A+B+uI)^{-1}\Bigr].
\label{eq:log-diff-resolvent}
\end{equation}
Let $A(t)>0$ be a differentiable family. By definition of the derivative,
\begin{equation}
\frac{d}{dt}\ln A(t)
=
\lim_{h\to 0}\frac{\ln A(t+h)-\ln A(t)}{h}.
\label{eq:dlog-def}
\end{equation}
Using the first-order expansion $A(t+h)=A(t)+h\dot A(t)+O(h^{2})$ and setting
$A:=A(t)$, $B:=h\dot A(t)$ in \eqref{eq:log-diff-resolvent}, we obtain
\begin{equation}
\frac{d}{dt}\ln A(t)
=
\lim_{h\to 0}\frac{1}{h}\int_{0}^{\infty}du\,
\Bigl[(A+uI)^{-1}-(A+h\dot A(t)+uI)^{-1}\Bigr].
\label{eq:dlog-resolvent-limit}
\end{equation}
For infinitesimal $h$, we expand the resolvent as follows:
\begin{align}
(A+h\dot A+uI)^{-1}
&=
\Bigl[(A+uI)\bigl(I+h(A+uI)^{-1}\dot A\bigr)\Bigr]^{-1}
\nonumber\\
&=
\bigl(I+h(A+uI)^{-1}\dot A\bigr)^{-1}(A+uI)^{-1}
\nonumber\\
&=
\Bigl(I-h(A+uI)^{-1}\dot A+O(h^{2})\Bigr)(A+uI)^{-1}
\nonumber\\
&=
(A+uI)^{-1}-h(A+uI)^{-1}\dot A\,(A+uI)^{-1}+O(h^{2}),
\label{eq:resolvent-expansion}
\end{align}
where $\dot A:=dA/dt$. Substituting \eqref{eq:resolvent-expansion} into
\eqref{eq:dlog-resolvent-limit} and taking the limit $h\to 0$ yields the
resolvent representation for the derivative of the logarithm:
\begin{equation}
\frac{d}{dt}\ln A(t)
=
\int_{0}^{\infty}du\,
(A(t)+uI)^{-1}\,\dot A(t)\,(A(t)+uI)^{-1}.
\label{eq:dlogA-resolvent}
\end{equation}
\end{proof}
Since $A(\epsilon)>0$ for $\epsilon$ in a neighborhood of $0$, the entropy admits
the Taylor expansion
\begin{equation}
S\!\left(A(\epsilon)\right)
=
S\!\left(A(0)\right)
+\epsilon\,\left.\frac{d}{d\epsilon}S\!\left(A(\epsilon)\right)\right|_{\epsilon=0}
+\frac{\epsilon^{2}}{2}\,\left.\frac{d^{2}}{d\epsilon^{2}}S\!\left(A(\epsilon)\right)\right|_{\epsilon=0}
+O(\epsilon^{3}).
\label{eq:S-Taylor-general}
\end{equation}

We now specialize to the perturbative ansatz
\begin{equation}
A(\epsilon)=A^{(0)}+\epsilon A^{(1)}+\epsilon^{2}A^{(2)}+O(\epsilon^{3}),
\label{eq:Aeps-expansion}
\end{equation}
with $A^{(0)}>0$. Using \eqref{eq:dlogA-resolvent}, the entropy admits the expansion
\begin{eqns}
S\!\left(A(\epsilon)\right)
&=
-\Tr\!\left[A^{(0)}\log A^{(0)}\right]
-\epsilon\,\Tr\!\left(
A^{(1)}\log A^{(0)}
+
A^{(0)}\,D_{\ln}\!\bigl(A^{(0)}\bigr)\!\left[A^{(1)}\right]
\right)
\\
&\quad
-\epsilon^{2}\,\Tr\Bigl(
A^{(2)}\log A^{(0)}
+
A^{(1)}\,D_{\ln}\!\bigl(A^{(0)}\bigr)\!\left[A^{(1)}\right]
+
A^{(0)}\,D_{\ln}\!\bigl(A^{(0)}\bigr)\!\left[A^{(2)}\right]
\Bigr)\\&\qquad\qquad
+O(\epsilon^{3}).
\label{eq:S-expansion}
\end{eqns}

Here $D_{\ln}(A)$ and $D^{2}_{\ln}(A)$ denote the first and second Fr\'echet
derivatives of the matrix logarithm, which admit the standard resolvent
representations
\begin{eqns}
D_{\ln}(A)[X]
&=
\int_{0}^{\infty}
(A+s I)^{-1}\,X\,(A+s I)^{-1}\,ds,
\\
D^{2}_{\ln}(A)[X,Y]
&=
\int_{0}^{\infty}
(A+s I)^{-1}\,X\,(A+s I)^{-1}\,Y\,(A+s I)^{-1}\,ds.
\label{eq:Dln-res}
\end{eqns}

We have also used the identity in Eq.~\eqref{eq:resolvent-theorem} to obtain Eq.~\eqref{eq:S-expansion}. All traces in \eqref{eq:S-expansion} are well defined provided the integrands are
trace class, which is automatically satisfied in finite-dimensional settings.

\
\subsection{Identity relating $D_{\ln}$ and $D_{\ln}^2$}
\label{subsec:AD2-identity-proof}
\begin{lemma}
Let $A>0$ be an invertible matrix and let $B$ be any matrix of the same dimension. Then
\begin{equation}
\Tr\!\left(A\,D^{2}_{\ln}(A)[B,B]\right)
=
\frac{1}{2}\,\Tr\!\left(B\,D_{\ln}(A)[B]\right).
\label{eq:AD2-identity}
\end{equation} 
\end{lemma}
\begin{proof}

Using the resolvent representations Eq.~\eqref{eq:Dln-res}, Eq.~\eqref{eq:AD2-identity} reduces to the resolvent trace identity
\begin{equation}
\,\Tr\!\left(
A\int_{0}^{\infty}\frac{1}{A+sI}\,B\,\frac{1}{A+sI}\,B\,\frac{1}{A+sI}\,ds
\right)
=
\Tr\!\left(
\int_{0}^{\infty}\frac{1}{A+sI}\,B\,\frac{1}{A+sI}\,B\,ds
\right).
\label{eq:resolvent-theorem}
\end{equation}
Now, we also have
\begin{eqns}
\Tr\int_{0}^{\infty}\frac{s}{(A+sI)^2}\,B\,\frac{1}{A+sI}\,B\,ds
&=
-\Tr\int_{0}^{\infty}s\,\frac{d}{ds}\!\left(\frac{1}{A+sI}\right)\,B\,\frac{1}{A+sI}\,B\,ds
\\
&=
-\Tr\left[\left.s\,\frac{1}{A+sI}\,B\,\frac{1}{A+sI}\,B\right|_{s=0}^{s=\infty}\right]
\\
&\quad
-\Tr\int_{0}^{\infty}\frac{s}{(A+sI)^2}\,B\,\frac{1}{A+sI}\,B\,ds\\& \quad
+\Tr\int_{0}^{\infty}\frac{1}{A+sI}\,B\,\frac{1}{A+sI}\,B\,ds.
\label{eq:ibp-core}
\end{eqns}
The boundary term vanishes: as $s\to0$ the prefactor $s$ kills the expression since $A>0$
implies $(A+sI)^{-1}$ bounded, and as $s\to\infty$ we have $(A+sI)^{-1}=O(1/s)$ so the traced
expression is $O(1/s)\to0$. Therefore \eqref{eq:ibp-core} implies
\begin{equation}
2\,\Tr\int_{0}^{\infty}\frac{s}{(A+sI)^2}\,B\,\frac{1}{A+sI}\,B\,ds
=
\Tr\int_{0}^{\infty}\frac{1}{A+sI}\,B\,\frac{1}{A+sI}\,B\,ds.
\label{eq:ibp-final}
\end{equation}

Now,
\begin{eqns}
&\,\Tr\!\left(
A\int_{0}^{\infty}\frac{1}{A+sI}\,B\,\frac{1}{A+sI}\,B\,\frac{1}{A+sI}\,ds
\right)
=
\,\Tr\int_{0}^{\infty}\frac{A}{(A+sI)^2}\,B\,\frac{1}{A+sI}\,B\,ds
\\
&=
\,\Tr\int_{0}^{\infty}\frac{1}{A+sI}\,B\,\frac{1}{A+sI}\,B\,ds
-
\,\Tr\int_{0}^{\infty}\frac{s}{(A+sI)^2}\,B\,\frac{1}{A+sI}\,B\,ds
\\
&=
\,\Tr\int_{0}^{\infty}\frac{1}{A+sI}\,B\,\frac{1}{A+sI}\,B\,ds
-
\frac{1}{2} \Tr\int_{0}^{\infty}\frac{1}{A+sI}\,B\,\frac{1}{A+sI}\,B\,ds
\\
&=
\frac{1}{2}\Tr\int_{0}^{\infty}\frac{1}{A+sI}\,B\,\frac{1}{A+sI}\,B\,ds,
\end{eqns}
 This proves \eqref{eq:resolvent-theorem}, and hence \eqref{eq:AD2-identity}.
\end{proof}
\subsection{Diagrammatic Representation and Haar Integration}
\label{WGcalculus}
We adopt a diagrammatic convention for representing operators, states, and unitaries. 
Each object in the computation is represented by a labeled box, with lines (or ``legs'') 
denoting the Hilbert spaces on which the object acts. The legs entering or exiting a box correspond 
to input and output indices of the operator, respectively. Connecting two legs represents a 
tensor contraction (i.e., a summation over the shared index), or equivalently, a partial trace 
over the associated Hilbert space.

For example, consider the product of three operators \( A \), \( B \), and \( C \) acting on the same space. 
In the diagrammatic notation, this corresponds to placing three boxes sequentially and connecting the output leg 
of \( A \) to the input leg of \( B \), and the output of \( B \) to the input of \( C \). 
The resulting chain of connections encodes the matrix multiplication \( ABC \). 
\[\includegraphics[width=0.27\linewidth]{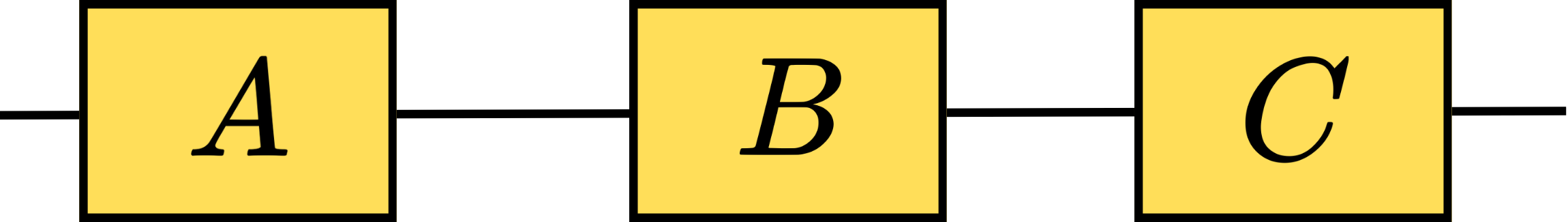}\]
If we take the trace of this product, \(\mathrm{Tr}(ABC)\), we close the remaining open legs by connecting 
the output leg of \( C \) back to the input leg of \( A \), forming a single closed loop.
\[\includegraphics[width=0.27\linewidth]{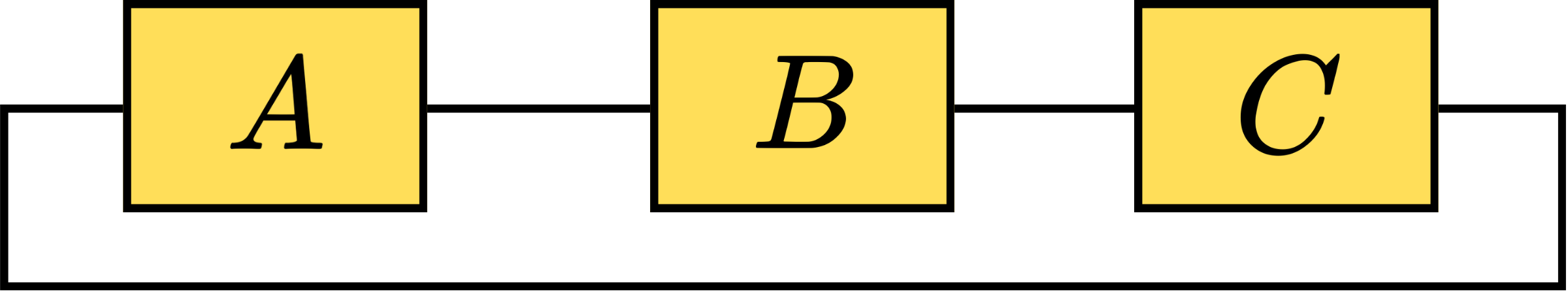}\]
This loop represents the summation over all contracted indices in the trace. 
We adopt the following diagrammatic conventions (used throughout).

\subsubsection*{Diagram Conventions}
\begin{itemize}
    \item  Every closed loop in a diagram corresponds to a trace, and every open leg represents a free index. 
    \item Each operator or state is represented by a rectangular box labeled by its symbol (e.g., $W_R$, $O$, $M$), with input and output legs corresponding to its domain and codomain.
    
    \item A unitary $U$ is represented by a black filled dot followed by a black open circle, while its adjoint $U^\dagger$ is represented by the reverse ordering in red (red open circle followed by a red filled dot). If multiple unitaries appear (e.g., $U$ and $V$), we label them explicitly by writing the corresponding symbol ($U$, $V$, etc.) above the associated dot--circle pair.
    \item The first, second, and third legs (from outermost to innermost) correspond to the Hilbert spaces $\mathcal{H}_{A_1}$, $\mathcal{H}_{A_2}$, and $\mathcal{H}_{\bar A}$, respectively. If a fourth leg is present, the ordering is taken to be $\mathcal{H}_{A_1}, \mathcal{H}_{A_2}, \mathcal{H}_{\bar A_2}, \mathcal{H}_{\bar A_1}$.
    \item Any extra legs that are not acted on by $W_R$, but only attach to external states or operators, are treated as reference systems in $\mathcal{H}_r$ and are labeled by $r$.
 \item A dotted segment is not a separate wire. It is the same wire as the adjoining solid segment, drawn dotted only over the portion where the wire crosses an operator box without attaching to it. In that region the wire simply passes by the operator, so none of its indices contract with the operator. After the box, the line becomes solid again to emphasize that the wire is continuing unchanged.
 \item When averaging over random unitaries, red dashed lines connect each unitary to its corresponding adjoint (e.g., $U$ to $U^\dagger$). These dashed connections represent the index contractions produced by the Haar integral.
\end{itemize}
Two complete example diagrams illustrating these conventions is shown below:
\paragraph{Example 1}
\[\Tr\Bigg[\Tr_{\bar A}\Big(W_R U_{A_1} \sigma_{A_1}^{(0)}\chi_{A_2\bar A}\Big)\Gamma_{A_1}U^\dagger_{A_1}Tr_{\bar A}\Big(W_R U_{A_1} \sigma_{A_1}^{(0)}\chi_{A_2\bar A}\Big)\Gamma_{A_1}U^\dagger_{A_1}\Bigg]\]
\begin{equation*}
\equiv \qquad \vcenter{\hbox{\includegraphics[width=0.65\linewidth]{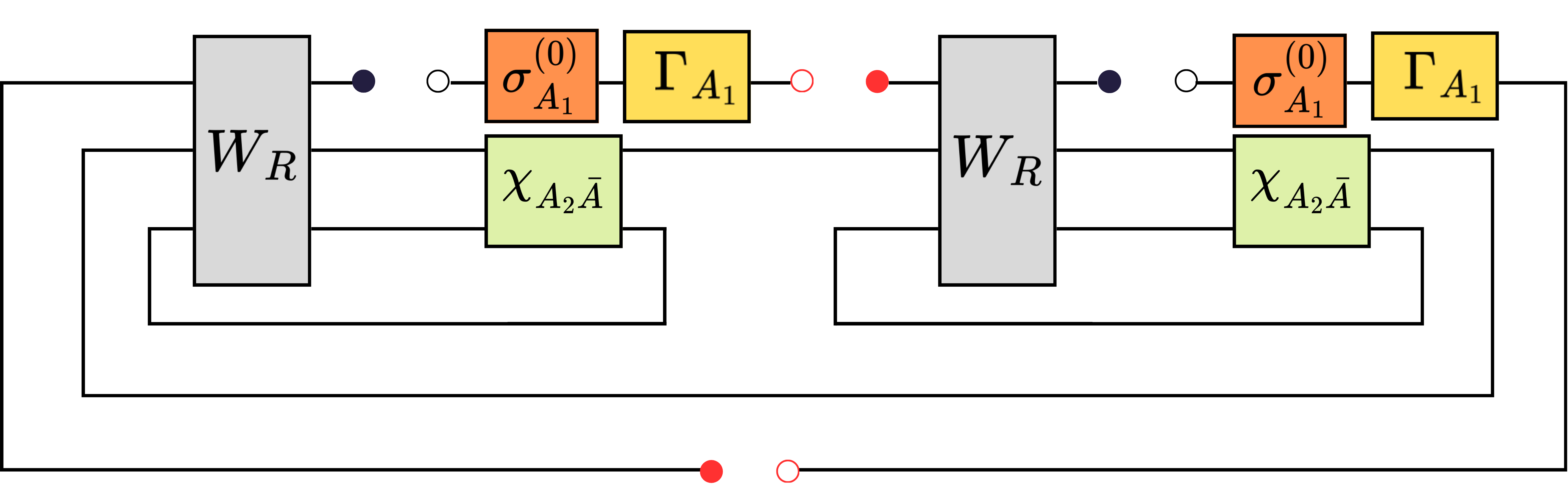}}}
\end{equation*}
\paragraph{Example 2}
\[\operatorname{Tr}\Bigg[\operatorname{Tr}_{A_1 \bar{A}_1}\left(W_R U_{A_1} V_{\bar{A}_1} \chi_{A_2 \bar{A}_2} \sigma_{A_1 \bar{A}_2}^{(0)}V^\dagger_{\bar{A}_1}\right) \Gamma_{A_1} U_{A_1}^\dagger \operatorname{Tr}_{A_1 \bar{A}_1}\left(W_R U_{A_1} V_{\bar{A}_1} \chi_{A_2 \bar{A}_2} \sigma_{A_1 \bar{A}_2}^{(0)}V^\dagger_{\bar{A}_1}\right) \Gamma_{A_1} U_{A_1}^\dagger\Bigg]\]
\begin{equation*}
\equiv \qquad \vcenter{\hbox{\includegraphics[width=0.65\linewidth]{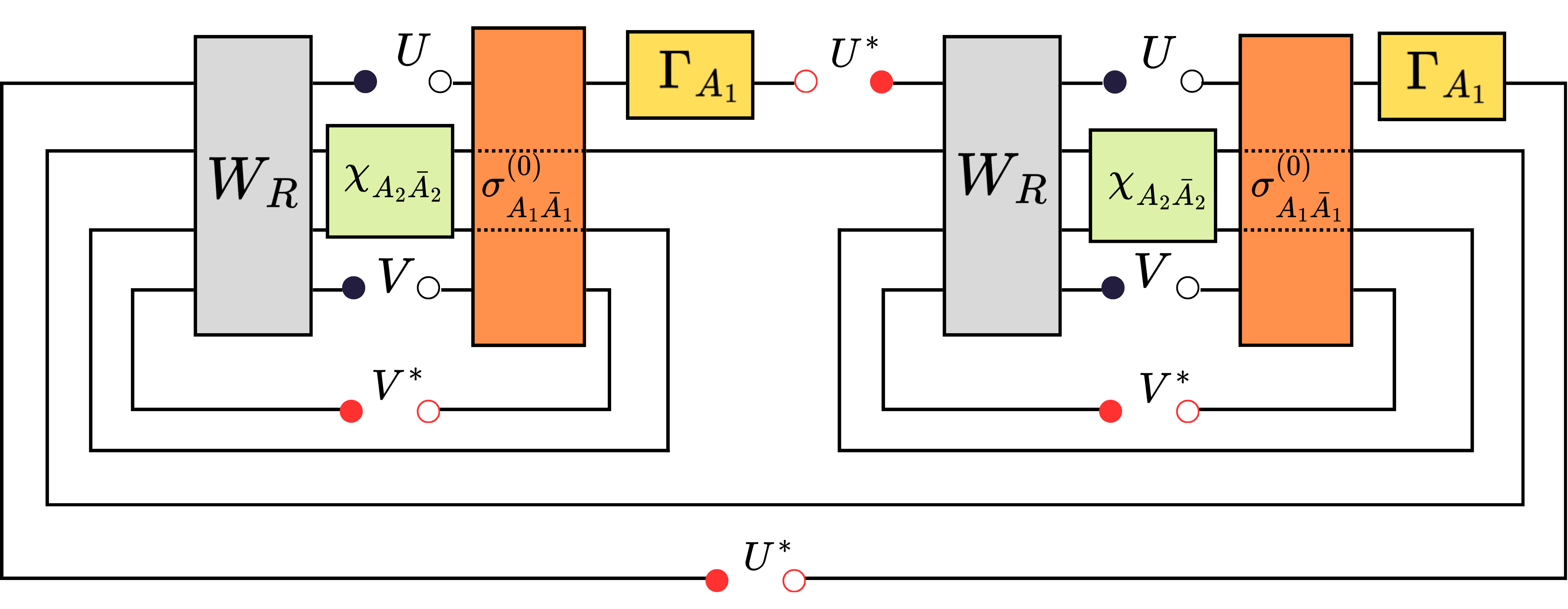}}}
\end{equation*}

\subsubsection{Weingarten Calculus}

To evaluate averages involving random unitaries, we use the Weingarten calculus, a method for computing polynomial integrals over the unitary group with respect to the Haar measure.  For a unitary matrix \( U \in \mathrm{U}(d) \), the expectation of a product of \( n \) matrix elements and their conjugates can be expressed in terms of permutations in the symmetric group \( S_n \):
\begin{equation*}
\left\langle
U_{a_1 b_1} \cdots U_{a_n b_n}
U^*_{\alpha_1 \beta_1} \cdots U^*_{\alpha_n \beta_n}
\right\rangle_U
= \sum_{\sigma, \tau \in S_n}
\delta_{a_1 \alpha_{\sigma(1)}} \cdots \delta_{a_n \alpha_{\sigma(n)}}
\delta_{b_1 \beta_{\tau(1)}} \cdots \delta_{b_n \beta_{\tau(n)}}
\, \mathrm{Wg}(\sigma^{-1}\tau, d),
\end{equation*}
where \( \mathrm{Wg}(\pi, d) \) is the Weingarten function, which depends on the permutation \( \pi \in S_n \) and the dimension \( d \).  
For \( n=2 \), corresponding to the second moment, \( S_2 \) has two elements: the identity \( e \) and the transposition \( (12) \). Their Weingarten values are
\begin{equation}
\mathrm{Wg}(e, d) = \frac{1}{d^2 - 1}, \qquad
\mathrm{Wg}((12), d) = -\frac{1}{d(d^2 - 1)}.
\end{equation}
Substituting these into the general formula gives
\begin{equation}
\begin{aligned}
\left\langle
U_{a_1 b_1} U_{a_2 b_2} U^*_{\alpha_1 \beta_1} U^*_{\alpha_2 \beta_2}
\right\rangle
= &
\ \mathrm{Wg}(e, d)
\big(
\delta_{a_1 \alpha_1}\delta_{a_2 \alpha_2}
\delta_{b_1 \beta_1}\delta_{b_2 \beta_2}
+
\delta_{a_1 \alpha_2}\delta_{a_2 \alpha_1}
\delta_{b_1 \beta_2}\delta_{b_2 \beta_1}
\big) \\
& +
\mathrm{Wg}((12), d)
\big(
\delta_{a_1 \alpha_2}\delta_{a_2 \alpha_1}
\delta_{b_1 \beta_1}\delta_{b_2 \beta_2}
+
\delta_{a_1 \alpha_1}\delta_{a_2 \alpha_2}
\delta_{b_1 \beta_2}\delta_{b_2 \beta_1}
\big).
\end{aligned}
\end{equation}

\subsubsection{Diagrammatic Haar Averaging Rules}
\label{WGcalculus-rules}
The diagrammatic procedure for performing a Haar average follows these steps \cite{Brouwer_1996}:
\begin{enumerate}
    \item Connect each black dot of a given unitary \(U^{(\ell)}\) to the corresponding red dot of its adjoint \(\big(U^{(\ell)}\big)^\dagger\) with a red dashed line, and likewise connect each black circle to the corresponding red circle. These red dashed lines represent the index contractions produced by the Haar average. Contractions only pair a unitary with its own adjoint (i.e., \(U^{(\ell)}\) contracts only with \(\big(U^{(\ell)}\big)^\dagger\), never with \(\big(U^{(m)}\big)^\dagger\) for \(\ell\neq m\)). In the present diagram there are four occurrences of the unitary/adjoint pair, yielding four admissible contraction patterns.
    \item Every closed loop composed of alternating solid lines and dotted lines is called a T-cycle. 
    Each T-cycle corresponds to a trace over the sequence of matrices encountered along that loop. 
    If a solid line is traversed in the direction opposite to its arrow, the corresponding operator appears transposed within the trace.

    \item Every closed loop composed of alternating gaps and dotted lines is called a U-cycle. 
    The length \( c_k \) of a U-cycle is defined as half the number of gap segments it contains. 
    The complete diagram gives rise to a collection of U-cycles, which together determine a coefficient 
    \( S_{c_1, \ldots, c_k} \) representing the weight of that particular diagram.

    \item The coefficient \( S_{c_1, \ldots, c_k} \) can be factorized into cumulants. 
    For the present case with four unitaries in the diagram, only two configurations arise:
    two U-cycles of length 1 each, contributing a weight
    \[
    S_{1,1} = \frac{1}{d^2 - 1},
    \]
    and a single U-cycle of length 2, contributing a weight
    \[
    S_2 = -\frac{1}{d(d^2 - 1)},
    \]
    where \( d \) is the dimension of the Hilbert space on which \( U \in \mathrm{SU}(d) \) acts. 
    In our case, \( d = 2 \). 
    The Haar average is obtained by summing all diagrams with their respective weights.

    \item Explicitly, the contraction identity used in the averaging procedure is
    \begin{equation}
    \begin{aligned}
    \left\langle U_{a_1 b_1} U_{a_2 b_2} U_{\alpha_1 \beta_1}^* U_{\alpha_2 \beta_2}^* \right\rangle
    =\;& 
    V_{1,1} \, \delta_{a_1 \alpha_1} \delta_{b_1 \beta_1} \delta_{a_2 \alpha_2} \delta_{b_2 \beta_2}
    + V_2 \, \delta_{a_1 \alpha_2} \delta_{b_1 \beta_1} \delta_{a_2 \alpha_1} \delta_{b_2 \beta_2} \\
    &+ V_2 \, \delta_{a_1 \alpha_1} \delta_{b_1 \beta_2} \delta_{a_2 \alpha_2} \delta_{b_2 \beta_1}
    + V_{1,1} \, \delta_{a_1 \alpha_2} \delta_{b_1 \beta_2} \delta_{a_2 \alpha_1} \delta_{b_2 \beta_1}.
    \end{aligned}
    \end{equation}
    Each term in this expression corresponds to one of the admissible contraction patterns between \( U \) and \( U^\dagger \) in the diagrammatic expansion.
    \item If the diagram contains several independent Haar-random unitaries \(U^{(1)},U^{(2)},\ldots\), the Haar measure factorizes,
\[
\Big\langle \,\cdots\, \Big\rangle \;=\;\int dU^{(1)}\int dU^{(2)}\cdots\,(\cdots),
\]
so the averaging can be performed separately for each unitary label. 
Equivalently, red dashed contractions only pair \(U^{(\ell)}\) with \(\big(U^{(\ell)}\big)^\dagger\); there are no admissible contractions that connect \(U^{(\ell)}\) to \(\big(U^{(m)}\big)^\dagger\) for \(\ell\neq m\). 
The weight of a full contraction pattern therefore factorizes as a product of the weights obtained from each unitary’s contraction pattern.

\end{enumerate}
\paragraph{Example:}
\begin{equation*}
 \vcenter{\hbox{\includegraphics[width=0.65\linewidth]{theory_draft_images/WGeg1.png}}}
\end{equation*}
The trace equation of the above diagram is:
\[\Tr\Bigg[\Tr_{\bar A}\Big(W_R U_{A_1} \sigma_{A_1}^{(0)}\chi_{A_2\bar A}\Big)\Gamma_{A_1}U^\dagger_{A_1}Tr_{\bar A}\Big(W_R U_{A_1} \sigma_{A_1}^{(0)}\chi_{A_2\bar A}\Big)\Gamma_{A_1}U^\dagger_{A_1}\Bigg]\]
Integration over \(U\) yields the 4 possible contractions:
\begin{align*}
  & \includegraphics[width=0.65\linewidth]{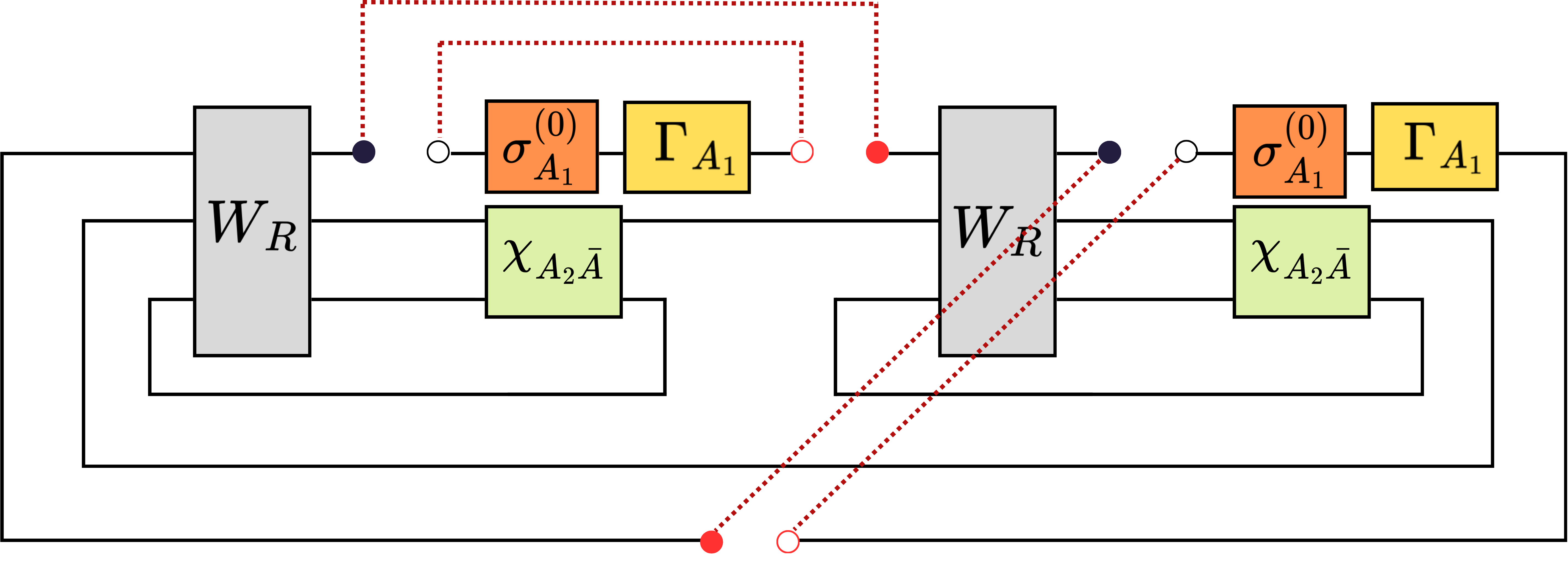} \\
   & \includegraphics[width=0.65\linewidth]{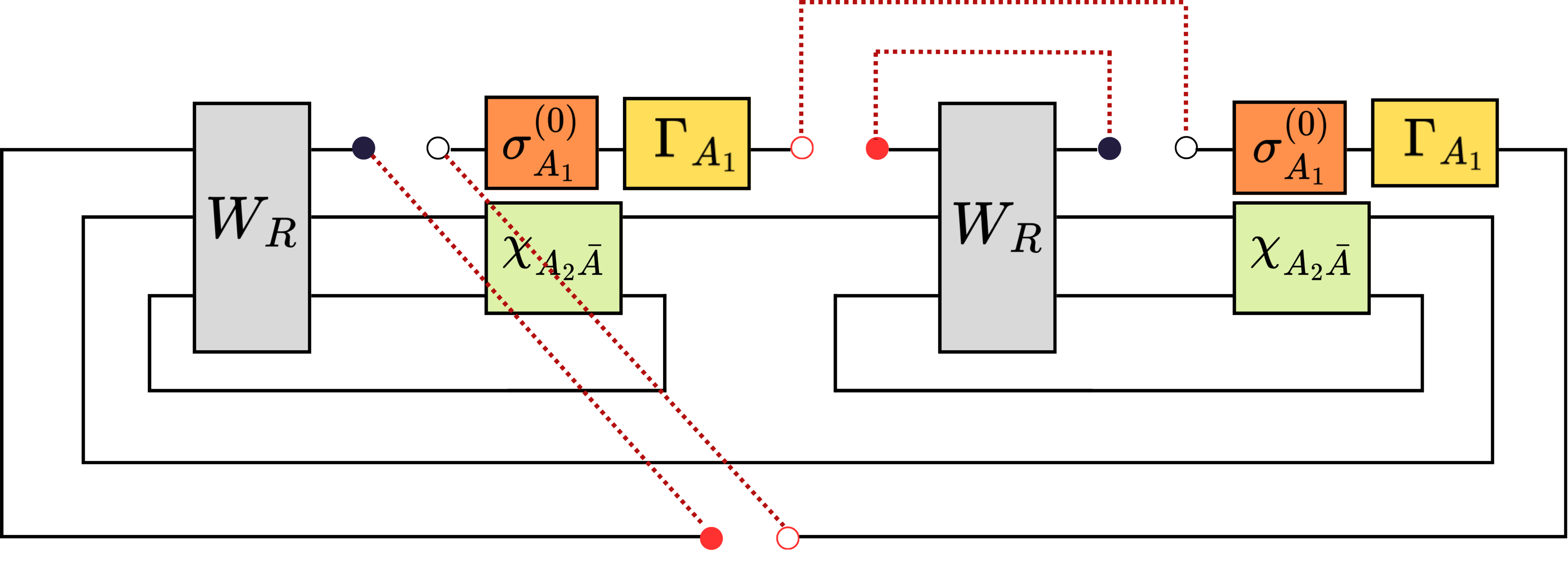} \\
      & \includegraphics[width=0.65\linewidth]{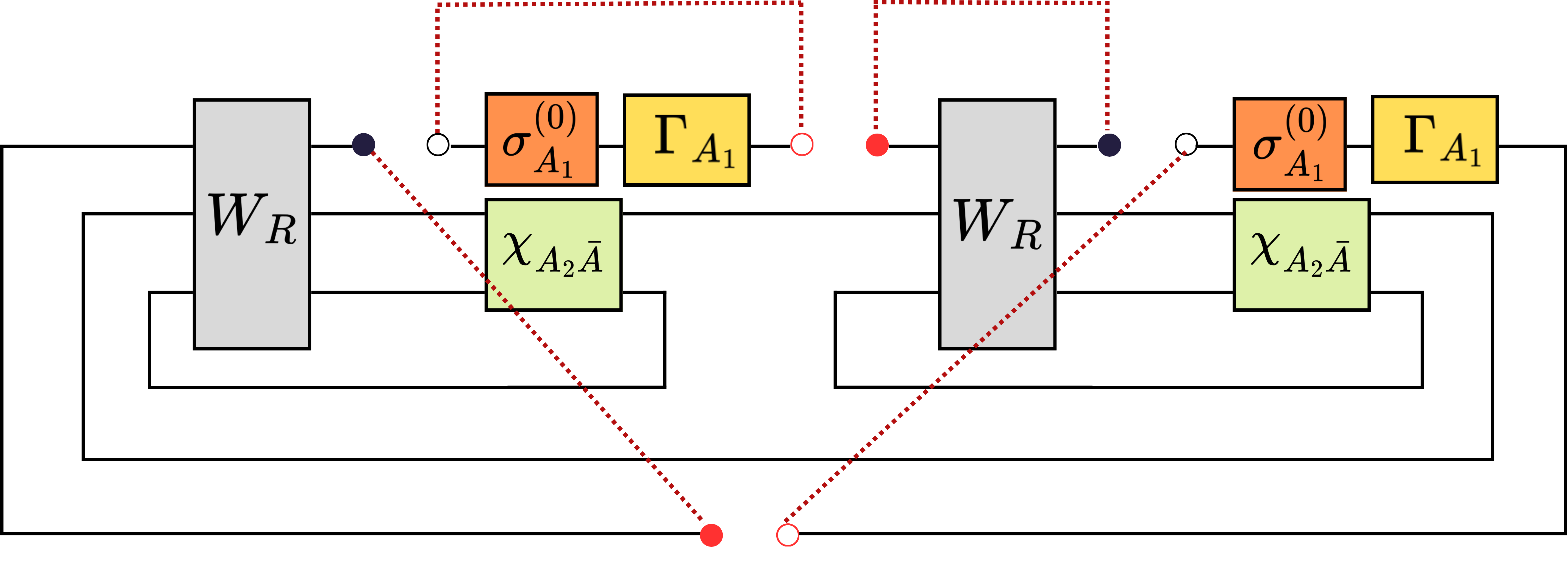} \\
         & \includegraphics[width=0.65\linewidth]{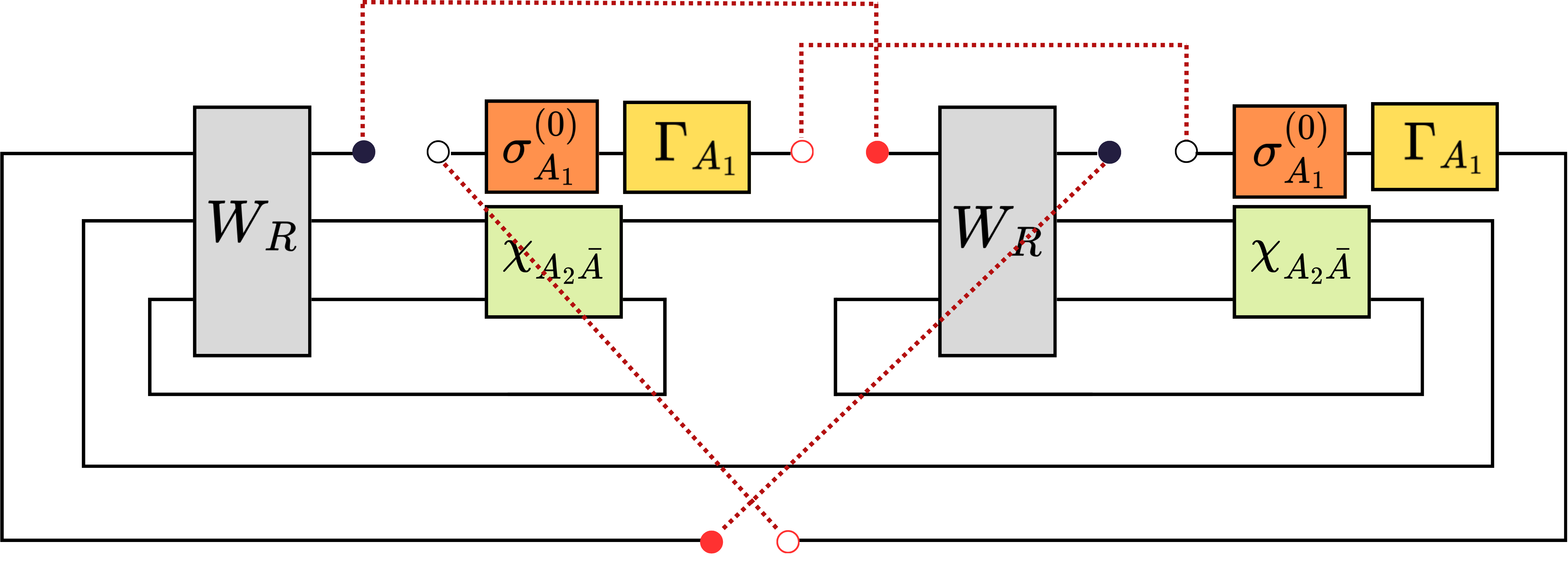} 
\end{align*}
After performing the Haar integral over \(U_{A_1}\), the expression reduces to a sum of contraction diagrams, each weighted by its corresponding Haar coefficient:
\begin{align*}
\frac{1}{d^2-1} \quad &\vcenter{\hbox{\includegraphics[width=0.5\linewidth]{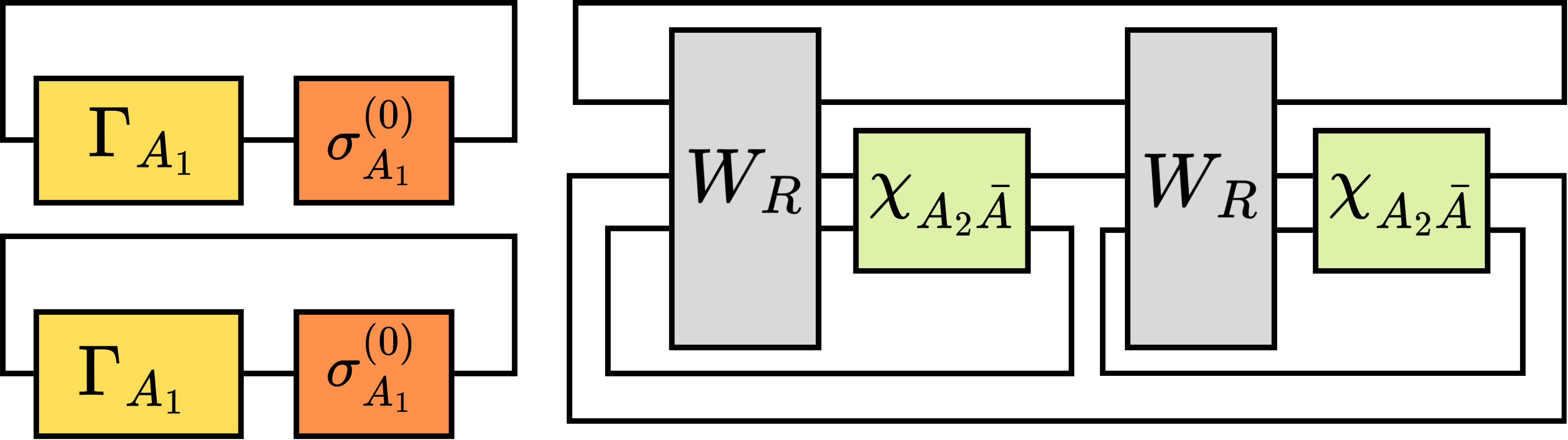}}}
\quad +\\[2mm]
\frac{1}{d^2-1}\quad&\vcenter{\hbox{\includegraphics[width=0.5\linewidth]{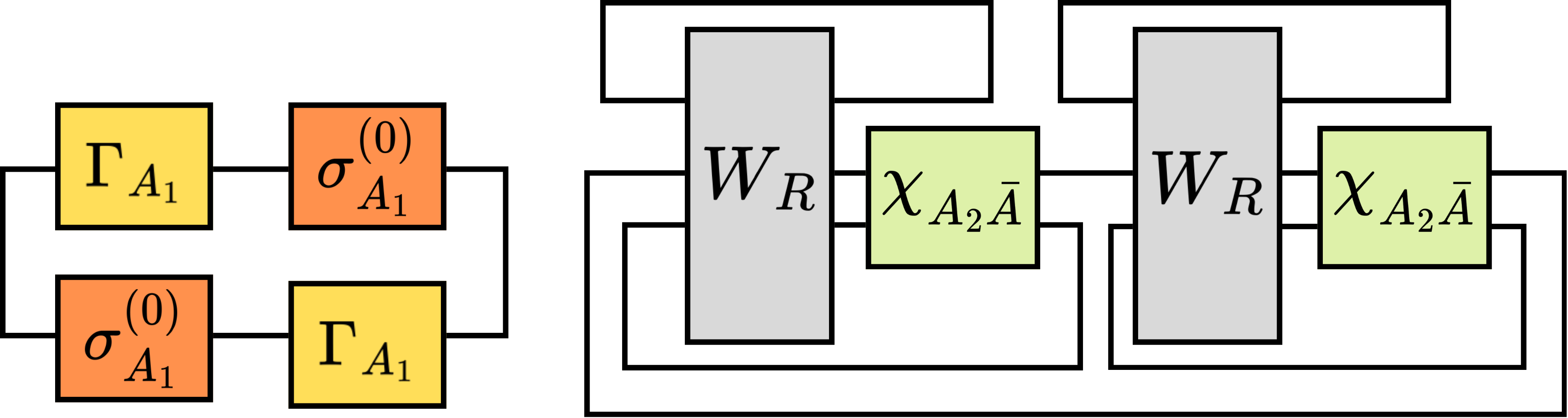}}}
\quad -\\[2mm]
\frac{1}{d(d^2-1)}\quad&\vcenter{\hbox{\includegraphics[width=0.5\linewidth]{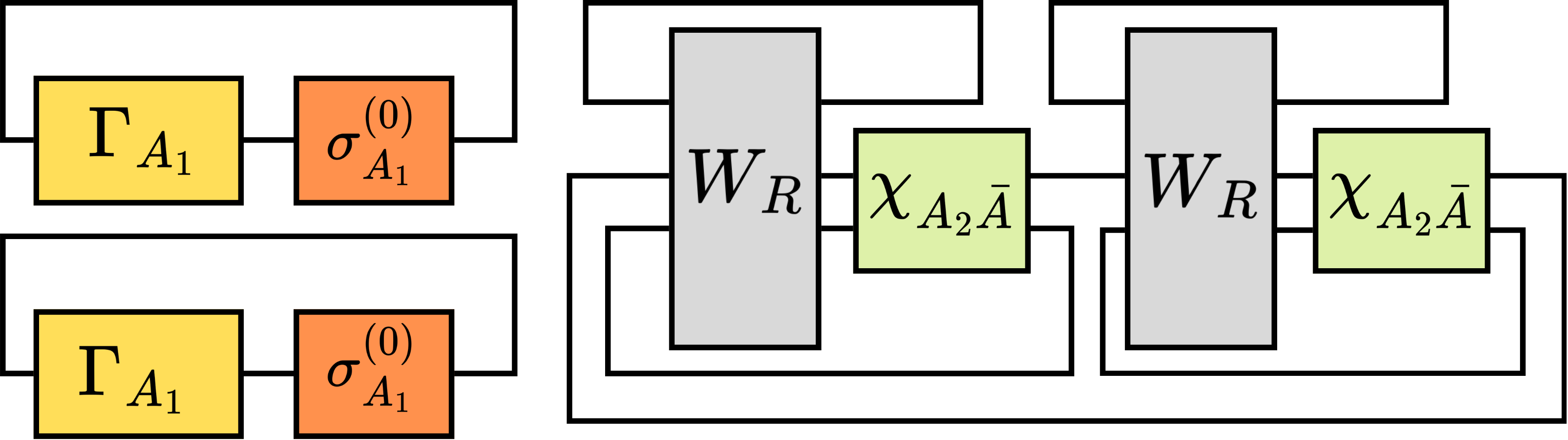}}}
\quad -\\[2mm]
\frac{1}{d(d^2-1)}\quad &\vcenter{\hbox{\includegraphics[width=0.5\linewidth]{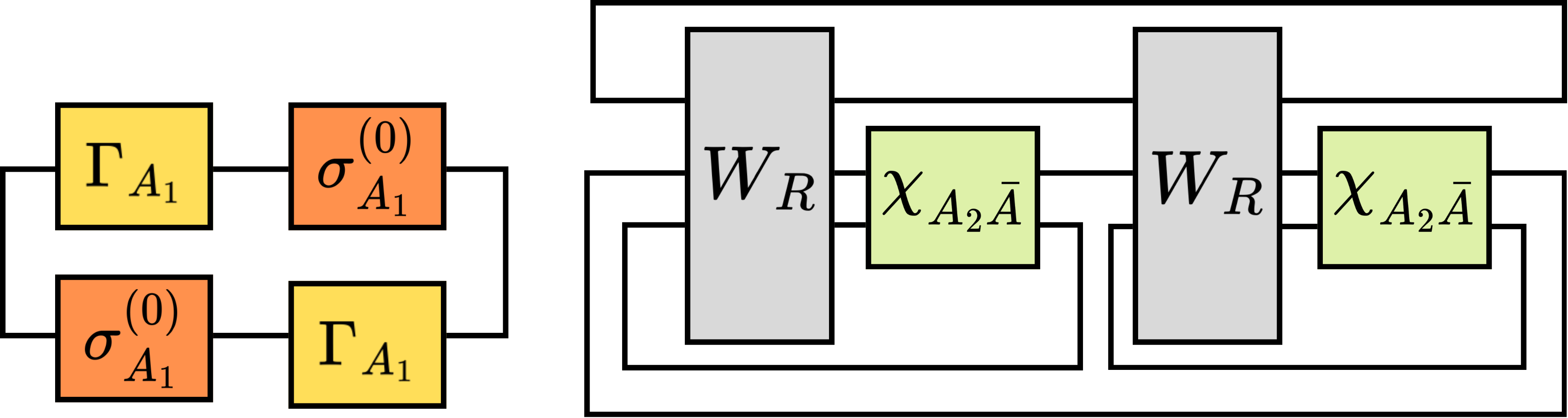}}}
\end{align*}
which can also be written as a trace equation:
\begin{align}
    &\frac{1}{d^2-1}\Bigg[\Tr\Big(\Gamma_{A_1}\sigma^{(0)}_{A_1}\Big)^2\Tr\Big[\Tr_{\bar A}\Big(W_R\chi_{A_2\bar A}\Big)^2\Big]+\Tr\Big[\Big(\Gamma_{A_1}\sigma^{(0)}_{A_1}\Big)^2\Big]\Tr\Big[\Tr_{A_1\bar A}\Big(W_R\chi_{A_2\bar A}\Big)^2\Big]\notag\\
    & - \frac{1}{d}\Tr\Big(\Gamma_{A_1}\sigma^{(0)}_{A_1}\Big)^2\Tr_{A_1\bar A}\Big(W_R\chi_{A_2\bar A}\Big)^2\Big]-\frac{1}{d}\Tr\Big[\Big(\Gamma_{A_1}\sigma^{(0)}_{A_1}\Big)^2\Big]\Tr\Big[\Tr_{\bar A}\Big(W_R\chi_{A_2\bar A}\Big)^2\Big]\Bigg]
\end{align}
For our calculations, we explicitly performed all computations using the \texttt{RTNI} package~\cite{Fukuda:2019pzs}, available at \url{https://github.com/MotohisaFukuda/RTNI}. 

\section{Matter-geometry entanglement is needed for gravity}
\label{app:semiclassical}

One can see why states in \eqref{eqn:statedep} are needed from the perspective of perturbative quantum gravity.
Suppose a matter field $\phi$ is coupled to gravitons $h_{\mu\nu}$ on a fixed background $g_{\mu\nu}$. For any coherent state of the matter, we can work in semi-classical gravity and obtain the solution for the metric perturbation, which can be written as some coherent state of gravitons. Let's denote such states as $|\phi_i\rangle |h_i\rangle$ for each solution of the linearized Einstein's equation where $|\phi_i\rangle, |h_i\rangle$ denote the states of the matter field and gravitons respectively, and $i$ is an abstract index that labels the different solutions with energy below some cutoff so that linearized gravity remains valid. These states, which are generally overcomplete, span a low energy subspace $\mathcal{H}_L=span\{|\phi_i\rangle|h_i\rangle\}$ of the effective field theory.
Suppose a background geometry $g_{\mu \nu} $ emerges from the entanglement structure of a code $\mathcal{C} \cong\mathcal{H}_L $ with subsystem complementary recovery, 
the logical degrees of freedom are now used to simulate a theory of the matter field coupled to gravitons on a fixed background $g_{\mu\nu}$. Even though the Hilbert space of a gauge theory is not factorizable due to non-trivial gauge constraints imposed by perturbative quantum gravity, it is embedable as a subspace of a factorizable Hilbert space by including the edge modes\footnote{In quantum reference frame language, we take $\mathcal{H}_L$ to be isomorphic to the kinematic Hilbert space of the gauge theory coupled to matter. When the gauge constraints are imposed, the theory lives in a gauge invariant, or physical subspace. However, to avoid confusing ``physical'' Hilbert space in QRF with that in QEC, we will work directly with the kinematic Hilbert space. More precisely, if we think of $|\phi_i\rangle|h_i\rangle$ as states spanning the physical gauge invariant subspace, then we need to first embed this space into $\mathcal{H}_L$ of a subsystem code that satisfies complementary recovery.}.

As the code supports subsystem complementary recovery on its logical qubits, we can write the encoded state by starting with the logical state, a shared entangled resource $|\chi\rangle$, and unitaries $U_A\otimes U_{\bar{A}}$ such that
\begin{equation}
    |\tilde{\psi}\rangle = U_A\otimes U_{\bar{A}} \Big(\sum_i c_i |\phi_i\rangle |h_i\rangle |\chi\rangle\Big).
\end{equation}

Instead of treating both $|\phi_i\rangle,|h_i\rangle$ as matter fields, we can now absorb $|h_i\rangle$ into the spacetime degrees of freedom, i.e., $|h_i\rangle|\chi\rangle\rightarrow |\chi_i\rangle$. Hence, we recover codewords of the form in eqn \eqref{eqn:statedep} and such a codeword then permits entanglement between matter and geometry, as it should be due to gravity. {The above re-identification can be understood as defining a new code which encode $span\{|\phi_i\rangle\}$ into a smaller isomorphic subspace $\mathcal C'$ such that $\mathcal{C}'\subset \mathcal{C}\subset \mH$. This corresponds to code concatenation where matter field $\phi$ is first encoded in the combined (kinematical) Hilbert space $\mathcal{C}$ of $\phi$ coupled to gravitons (outer code), which is then encoded in a subsystem code satisfying complementary recovery (inner code). }.

One can also relate such code concatenation with a heuristic picture of renormalization group, where the flow from UV to IR corresponds to the sequence of Hilbert spaces $\mathcal{H}\rightarrow \mathcal{C}\rightarrow \mathcal{C}'$ as we progressively ``integrate out'' the UV degrees of freedom associated with the higher energy physics from which the background spacetime geometry and the graviton degrees of freedom emerges. {For this reason, \cite{Cao_2018} also identifies entropy $S(\chi)$ with the UV/spacetime degrees of freedom whereas $S(\sigma_a)$ with that of the IR/matter field degrees of freedom.}

\section{Monotonic RT Lengths in AdS$_3$ with  Backreaction}
\label{app:perfect_fluid}
Here we provide an explicit example where the area of the (classical) extremal surface does increase with bulk entropy in the presence of gravity. We construct a static, circularly symmetric solution of the
Einstein equations in $2+1$ dimensions with negative cosmological constant
sourced by a static anisotropic fluid \cite{HerreraSantos1997, Cadogan:2024mcl, Cadoni:2020izk}, then compute the length of the
Ryu--Takayanagi (RT) geodesic anchored on a boundary interval.
We work in Schwarzschild-like coordinates $(t,r,\theta)$ and assume a static,
rotationally symmetric metric of the form
\begin{equation}
ds^2 = -A(r)^2\,dt^2 + B(r)^2\,dr^2 + r^2\,d\theta^2,
\label{eq:metric-ansatz}
\end{equation}
We now consider the two-parameter family
\begin{equation}
A(r)^2=1+\beta r^2,
\qquad
B(r)^2=\frac{1}{1+\alpha r^2},
\label{eq:AB-family-aniso}
\end{equation}
with constants \(\alpha>0\) and \(\beta\). The matter sector is taken to be an anisotropic fluid with energy density \(\mu(r)\), radial pressure \(p_r(r)\), and tangential pressure \(p_\theta(r)\). In the coordinate basis \((t,r,\theta)\), the stress-energy tensor is
\begin{equation}
T_{\mu\nu}
=
(\mu+p_\theta)\,u_\mu u_\nu
+
p_\theta\,g_{\mu\nu}
+
(p_r-p_\theta)\,\chi_\mu \chi_\nu,
\label{eq:T-aniso-fluid}
\end{equation}
where \(u^\mu\) is the fluid velocity, normalized by \(u_\mu u^\mu=-1\), and \(\chi^\mu\) is a unit spacelike vector in the radial direction, satisfying \(\chi_\mu\chi^\mu=1\) and \(u_\mu\chi^\mu=0\).
The fluid velocity and the unit spacelike vector in the radial direction are
\begin{equation*}
u^\mu=\left(\frac{1}{A(r)},\,0,\,0\right),
\quad
u_\mu=\bigl(-A(r),\,0,\,0\bigr),
\qquad
\chi^\mu=\left(0,\,\frac{1}{B(r)},\,0\right),
\quad
\chi_\mu=\bigl(0,\,B(r),\,0\bigr).
\label{eq:u-chi-aniso}
\end{equation*}
Hence, the nonzero components of the stress-energy tensor are
\begin{equation}
T_{tt}=\mu(r)\,A(r)^2,
\qquad
T_{rr}=p_r(r)\,B(r)^2,
\qquad
T_{\theta\theta}=p_\theta(r)\,r^2.
\label{eq:T-aniso-components}
\end{equation}
The Einstein equations with cosmological constant are
\begin{equation}
G_{\mu\nu}+\Lambda g_{\mu\nu}=\kappa T_{\mu\nu},
\label{eq:Einstein-aniso}
\end{equation}
where \(\kappa=8\pi G_3\) is the three-dimensional gravitational coupling.
Substituting the metric ansatz and the anisotropic stress tensor into the Einstein equations, we obtain three independent equations:
\begin{align}
&\frac{B'(r)}{r\,B(r)^3}-\Lambda = \kappa\,\mu(r),
\label{eq:aniso-eq1}\\
&\frac{A'(r)}{r\,A(r)}+\Lambda B(r)^2 = \kappa\,p_r(r)\,B(r)^2,
\label{eq:aniso-eq2}\\
&\frac{B(r)A''(r)-A'(r)B'(r)}{A(r)B(r)^3}+\Lambda = \kappa\,p_\theta(r).
\label{eq:aniso-eq3}
\end{align}
Solving these equations give:
\begin{align}
\mu(r)
&=
\mu
\equiv
\frac{-\,(\alpha+\Lambda)}{\kappa},
\label{eq:mu-solution-aniso}\\[4pt]
p_r(r)
&=
\frac{1}{\kappa}
\left[
\Lambda+\frac{\beta(1+\alpha r^2)}{1+\beta r^2}
\right],
\label{eq:pr-solution-aniso}\\[4pt]
p_\theta(r)
&=
\frac{1}{\kappa}
\left[
\Lambda+\frac{\beta\bigl(1+2\alpha r^2+\alpha\beta r^4\bigr)}{(1+\beta r^2)^2}
\right].
\label{eq:ptheta-solution-aniso}
\end{align}

Thus the metric is supported by an anisotropic fluid with constant energy density \(\mu\). The vacuum AdS\(_3\) limit is recovered by setting $ \mu=p_r=p_\theta=0$, 
which implies \( \alpha=\beta=-\Lambda \). The density is positive provided $-\Lambda>\alpha>0$. 
Moreover, the source is genuinely anisotropic unless \(\beta=\alpha\). In the special case \(\beta=\alpha\), one finds
\begin{equation}
p_r(r)=p_\theta(r)=\frac{\Lambda+\alpha}{\kappa}=-\mu,
\end{equation}
so the stress tensor becomes isotropic and behaves like a vacuum-energy source rather than a generic matter fluid. In what follows, we therefore restrict to the anisotropic branch \(\beta\neq\alpha\).

The conservation equation \(\nabla_\mu T^{\mu\nu}=0\) reduces to
\begin{equation}
p_r'(r)+(\mu+p_r)\frac{A'(r)}{A(r)}+\frac{p_r-p_\theta}{r}=0,
\label{eq:aniso-conservation-eq}
\end{equation}
and one verifies by direct substitution of \eqref{eq:mu-solution-aniso}--\eqref{eq:ptheta-solution-aniso} that it is satisfied identically.

For subsequent computations, it is convenient to rewrite the radial
coefficient in terms of
\begin{equation}
\alpha(\mu) \equiv -\bigl(\Lambda + \kappa\mu\bigr),
\end{equation}

Throughout, we restrict to this range so that $\alpha(\mu)>0$ and the geometry remains AdS-like rather than transitioning to a different asymptotic spacetime. Hence on a constant time slice the induced spatial metric is
\begin{equation}
dl^2 \equiv ds^2\big|_{dt=0}
= B(r)^2\,dr^2 + r^2\,d\theta^2
= \frac{dr^2}{1+\alpha r^2} + r^2\,d\theta^2.
\end{equation}

A spacelike curve anchored on the boundary interval can be
written as $r(\theta)$, and its length is
\begin{equation}
\mathcal{L}[r(\theta)]
= \int d\theta\,\mathcal{L}_{\text{geo}}(r,r'),
\qquad
\mathcal{L}_{\text{geo}}(r,r')
= \sqrt{\frac{{r'}^2}{1+\alpha r^2} + r^2},
\label{eq:Lgeodesic-functional}
\end{equation}
where
\begin{equation}
r' \equiv \frac{dr}{d\theta}.
\end{equation}
We view $\theta$ as the ``time'' variable for this one-dimensional
mechanical problem. The Lagrangian $\mathcal{L}_{\text{geo}}$ depends on $r$
and $r'$, but not explicitly on $\theta$, so there is a conserved Hamiltonian
(energy)
\begin{equation}
\mathcal{H}
= -\frac{r^2}{\mathcal{L}_{\text{geo}}}.
\label{eq:H-constant}
\end{equation}
To identify this constant, we use the turning point. We consider a geodesic
symmetric about $\theta=0$, with minimal radius $r=r_\ast$ at $\theta=0$.
At the turning point $r = r_\ast$ we have
\begin{eqns}
\mathcal{L}_{\text{geo}}(r_\ast,0) = r_\ast,
\qquad
\mathcal{H} = -\,\frac{r_\ast^2}{\mathcal{L}_{\text{geo}}(r_\ast,0)} = -\,r_\ast
.    
\end{eqns}
Now, Eq. \eqref{eq:H-constant} can be rewritten as the first-order constraint
\begin{equation}
\mathcal{L}_{\text{geo}}(r,r') = \frac{r^2}{r_\ast}.
\label{eq:Lgeo-rstar}
\end{equation}
This first-order Hamiltonian equation determines the geodesic profile
$r(\theta)$ and directly implies
\begin{equation}
\frac{d\theta}{dr}
= \frac{1}{r'} 
= \frac{r_\ast}{r\,\sqrt{1+\alpha r^2}\,\sqrt{r^2 - r_\ast^2}},
\end{equation}
We consider a geodesic symmetric about $\theta=0$, with endpoints at
$\theta = \pm \Delta\theta/2$ on the boundary. The half opening angle is
\begin{equation}
\frac{\Delta\theta}{2}
= \int_{r_\ast}^{\infty} \frac{d\theta}{dr}\,dr
= \int_{r_\ast}^{\infty}
\frac{r_\ast\,dr}{r\,\sqrt{1+\alpha r^2}\,\sqrt{r^2 - r_\ast^2}}.
\label{eq:theta-half-int-rstar}
\end{equation}
This integral can be evaluated exactly. One finds
\begin{equation}
\int_{r_\ast}^{\infty}
\frac{r_\ast\,dr}{r\,\sqrt{1+\alpha r^2}\,\sqrt{r^2 - r_\ast^2}}
= \tan^{-1}\!\left(\frac{1}{r_\ast\sqrt{\alpha}}\right),
\end{equation}
so that
\begin{equation}
\frac{\Delta\theta}{2}
= \tan^{-1}\!\left(\frac{1}{r_\ast\sqrt{\alpha}}\right)
\end{equation}
and hence
\begin{equation}
\tan\!\left(\frac{\Delta\theta}{2}\right)
= \frac{1}{r_\ast\sqrt{\alpha}}
\quad\Longrightarrow\quad
r_\ast(\Delta\theta,\alpha)
= \frac{1}{\sqrt{\alpha}\,\tan(\Delta\theta/2)}.
\label{eq:rstar-of-theta-alpha}
\end{equation}

The Hamiltonian relation Eq. \eqref{eq:Lgeo-rstar},
allows us to compute the RT length directly:
\begin{equation}
\mathcal{L}_{\mathrm{RT}}
= \int_{-\Delta\theta/2}^{\Delta\theta/2} \mathcal{L}_{\text{geo}}\,d\theta
= 2\int_0^{\Delta\theta/2} \frac{r(\theta)^2}{r_\ast}\,d\theta = 2\int_{r_\ast}^{r_{\max}}
\frac{r\,dr}{\sqrt{1+\alpha r^2}\,\sqrt{r^2 - r_\ast^2}}.
\end{equation}
This integral can be evaluated in closed form. The result is
\begin{equation}
\mathcal{L}_{\mathrm{RT}}\!\left(r_\ast,r_{\max},\alpha\right)
= \frac{2}{\sqrt{\alpha}}\,
\sinh^{-1}\!\left(
\frac{\sqrt{\alpha}\,\sqrt{r_{\max}^2 - r_\ast^2}}
     {\sqrt{1+\alpha r_\ast^2}}
\right).
\label{eq:LRT-final-rstar}
\end{equation}
Now, we fix the boundary interval $\Delta\theta$ and the UV cutoff $r_{\max}$ and
ask how $\mathcal{L}_{\mathrm{RT}}$ changes as we vary $\mu$.
In the UV regime $r_{\max} \gg r_\ast$, we may approximate
\begin{equation}
\sqrt{r_{\max}^2 - r_\ast^2} \simeq r_{\max},
\end{equation}
Substituting $r_\ast$, the RT length reduces to
\begin{equation}
\mathcal{L}_{\mathrm{RT}}(\alpha)
\simeq \frac{2}{\sqrt{\alpha}}\,
\sinh^{-1}\!\bigl(C\sqrt{\alpha}\bigr),
\qquad
C \equiv r_{\max}\sin(\Delta\theta/2) > 0.
\label{eq:LRT-UV}
\end{equation}
To study the $\mu$–dependence, it is enough to analyze $\partial\mathcal{L}_{\mathrm{RT}}/\partial\alpha$
for the approximate expression \eqref{eq:LRT-UV}.

Defining $\mathcal{X} \equiv C\sqrt{\alpha} \ge 0$, a short calculation gives
\begin{equation}
\frac{\partial\mathcal{L}_{\mathrm{RT}}}{\partial\alpha}
= \alpha^{-3/2}
\left[
\frac{\mathcal{X} }{\sqrt{1+\mathcal{X} ^2}} - \sinh^{-1}\mathcal{X} 
\right].
\label{eq:dL-dalpha-UV}
\end{equation}
Since $\alpha>0$, the sign is controlled by the bracket. Consider
\begin{equation}
\mathcal{F}(\mathcal{X}) \equiv \sinh^{-1}\mathcal{X} - \frac{\mathcal{X}}{\sqrt{1+\mathcal{X}^2}}
\end{equation}
We have $\mathcal{F}(0)=0$, and
\begin{equation}
\mathcal{F}'(\mathcal{X})
= \frac{1}{\sqrt{1+\mathcal{X}^2}} - \frac{1}{(1+\mathcal{X}^2)^{3/2}}
= \frac{\mathcal{X}^2}{(1+\mathcal{X}^2)^{3/2}} \ge 0,
\end{equation}
so $\mathcal{F}$ is monotonically increasing and
\begin{equation}
\sinh^{-1}\mathcal{X} \;\ge\; \frac{\mathcal{X}}{\sqrt{1+\mathcal{X}^2}}
\quad\text{for all } \mathcal{X} \ge 0.
\label{eq:arcsinh-ineq}
\end{equation}
It follows that
\begin{equation}
\frac{\partial \mathcal{L}_{\mathrm{RT}}}{\partial \alpha}
\le 0.
\end{equation}
Finally, since
\begin{equation}
\alpha(\mu) = -(\Lambda + \kappa\mu),
\qquad
\frac{d\alpha}{d\mu} = -\kappa < 0,
\end{equation}
we have
\begin{equation}
\frac{d\mathcal{L}_{\mathrm{RT}}}{d\mu}
= \frac{\partial\mathcal{L}_{\mathrm{RT}}}{\partial\alpha}\,
\frac{d\alpha}{d\mu}
= -\kappa\,
\frac{\partial\mathcal{L}_{\mathrm{RT}}}{\partial\alpha}
\;\ge\; 0.
\end{equation}
Thus, at fixed opening angle $\Delta\theta$ and in the UV limit $r_{\max}\gg r_\ast$,
the RT geodesic length $\mathcal{L}_{\mathrm{RT}}$ is a monotonically increasing
function of $\mu$.

It then follows from the first law of thermodynamics \cite{Gourgoulhon_2006,Ballesteros_2013} that we can rewrite the (comoving) entropy density in terms of the energy density and pressure contributions. For simplicity, we restrict ourselves to a family of the static solutions such that the anisotropic pressure profile is fixed when $\mu$ is varied. This then yields that $\partial s/\partial\mu =1/ T$.  Therefore, the spacelike geodesic length increases as a function of (thermal) entropy.
\bibliographystyle{JHEP}
\bibliography{ref}

\end{document}